\documentclass{birkau}

\usepackage{amsmath,amssymb,url,amsfonts}
\usepackage{mathptmx}
\usepackage{stmaryrd}
\usepackage{enumitem} 

\usepackage{hyperref}
\usepackage{tikz}
\usetikzlibrary{cd}

\usepackage{mathdots}
\usepackage{yhmath}
\usepackage{cancel}
\usepackage{color}
\usepackage{siunitx}
\usepackage{array}
\usepackage{multirow}
\usepackage{gensymb}
\usepackage{tabularx}
\usepackage{booktabs}
\usetikzlibrary{fadings}

\DeclareFontFamily{U}{mathx}{\hyphenchar\font45}
\DeclareFontShape{U}{mathx}{m}{n}{
      <5> <6> <7> <8> <9> <10>
      <10.95> <12> <14.4> <17.28> <20.74> <24.88>
      mathx10
      }{}
\DeclareSymbolFont{mathx}{U}{mathx}{m}{n}
\DeclareFontSubstitution{U}{mathx}{m}{n}
\DeclareMathAccent{\widecheck}{0}{mathx}{"71}
\DeclareMathAccent{\wideparen}{0}{mathx}{"75}

\numberwithin{equation}{section}

\theoremstyle{plain}
\newtheorem{theorem}{Theorem}[section]
\newtheorem{lemma}[theorem]{Lemma}

\theoremstyle{definition}
\newtheorem{definition}[theorem]{Definition}

\newtheorem{remark}[theorem]{Remark}

\newcommand{\sqcoversupset}{\raisebox{1.3ex}{\,\rotatebox{180}{$\sqsubseteq$}}\,}
\newcommand{\sqcoversubset}{\raisebox{1.3ex}{\,\rotatebox{180}{$\sqsupseteq$}}\,}
\newcommand{\existunique}{\exists \textit{\bf !}}

\begin{document}


\title[The logic of quantum mechanics]{The logic of quantum mechanics}

\author[Eric Buffenoir]{Eric Buffenoir}
\address{Universit\'e de la C\^ote d'Azur, CNRS, InPhyNi, FRANCE}
\email{eric.buffenoir@cnrs.fr}



\subjclass{81P10,81P13, 18B35}

\keywords{Logical foundations of quantum mechanics; quantum logic (quantum-theoretic aspects) / Contextuality in quantum theory / Preorders, orders, domains and lattices (viewed as categories)}

\begin{abstract}
The quantum logic program originated in a 1936 article by G. Birkhoff and J. von Neumann. This program is generally disregarded due to no-go theorems restricting the existence of the tensor product of elementary quantum logics and, above all, the impossibility of considering entangled states and Bell non-local states within the framework of these composite quantum logics. We revisit this study from the beginning and reverse the perspective. Here, the existence of a tensor product and a star involution are the only prerequisites for the definition of the state spaces. Surprisingly, the quantum logics constructed in this way turn out to have a close connection with irreducible Hilbert geometries, even though we did not impose this sort of structure ab initio.  Endly, the existence of some basic quantum-like properties is explicitly proven in our framework : contextuality,  no-broadcasting theorem, and Bell non-locality. These elements demonstrate that our quantum logic program is capable of achieving G. Birkhoff and J. von Neumann's initial ambition of founding quantum theory.
\end{abstract}

\maketitle


\section{Introduction}

The history of foundational research on quantum theory is as old as quantum mechanics itself: the ad hoc nature of the formalism (Hilbert space over the complex numbers, tensor product for composite objects, etc.) had prompted Von Neumann himself to seek a radically new approach:

\begin{quote}
"I would like to make a confession which may seem immoral: I do not believe absolutely in Hilbert space any more. After all, Hilbert space (as far as quantum mechanical things are concerned) was obtained by generalizing Euclidean space, based on the principle of ‘conserving the validity of all formal rules’ [$\cdots$]. Now we begin to believe that it is not the vectors that matter, but the lattice of all linear (closed) subspaces.  Because: 1) The vectors ought to represent the physical states, but they do so redundantly, up to a complex factor, only 2) and besides, the states are merely a derived notion, the primitive (phenomenologically given) notion being the qualities which correspond to the linear closed subspaces"\cite{BirkhoffA}
\end{quote}

The internal motivations for such axiomatic research in quantum mechanics are ultimately quite natural. First, it is a matter of justifying, on the basis of strictly operational principles, the “categorical” nature of the mathematical objects involved in the description of quantum systems (justification of linearity,  complex numbers, tensor products, etc.), which lead to the specific properties of quantum mechanics (entanglement, no-cloning, teleportation, etc.). The most significant contributions of this axiomatic approach fall within this framework. The next step is to place the standard description of quantum mechanics within a broader framework, making it possible to understand the specific singularity of quantum mechanics within this landscape of theories, and even to singularize possible generalizations.  Finally, the aim is to develop an {\em interpretation} of quantum theory, and it very quickly became apparent that this interpretation played a central role in the process of information acquisition by an observer about the observed quantum system.

The starting point for research into quantum logic was therefore provided by J. von Neumann and G. Birkhoff themselves  \cite{BirkhoffVonNeumann}. These authors conceived quantum logic in order to highlight the fundamental consequences of the notion of quantum superposition. In terms of the states of a physical system and the properties of that system, superposition means that the strongest property that is true for two distinct states is also true for states other than the two given states.  It easily follows that the distributive property of the lattice of properties collapses (whereas it is natural for the classical case). G. Birkhoff and J. von Neumann, as well as many others, believed that understanding the logical structure of superposition was the key to gaining a better understanding of quantum theory as a whole.

Subsequently, C. Piron's important work was to characterize this lattice of testable properties \cite{Piron1964} without any a priori reference to an underlying Hilbert space. From a conceptual point of view, Piron's description of the space of testable properties (unlike Mackey's program \cite{Mackey}) completely disregards the notion of quantum probability, which therefore appears to be a derived notion (obviously, quantum probabilities are “rediscovered” using results related to Gleason's Theorem \cite{Piron1972}). This lattice of testable properties appears {\it de facto} as a generalized Hilbert lattice, i.e., an orthomodular, atomistic, irreducible lattice satisfying the covering relation. It should be noted that some subsequent results \cite{Holland}\cite{Mayet} have been able to demonstrate, in infinite dimensions, the singular role played by the field of complex numbers in this construction. From a technical point of view, the program developed under Piron's leadership makes extensive use of the results established by Faure and Froelicher \cite{FaureFroelicher} on dualities in projective geometry (see \cite{Stubbe} for a shortened presentation), which gives these results a general character (in particular, these results prove to be relevant even for infinite-dimensional Hilbert spaces).  

The Achilles heel of quantum logic is that it fails {\it a priori} to elegantly capture the composition of quantum systems, i.e., the way we describe multipartite quantum experiments once is given the way to describe individual quantum systems. Indeed, Foulis and Randall \cite{FoulisRandall} have shown that there is no canonical object $\mathcal{L}_A \otimes \mathcal{L}_B$ in the category of ortho-algebras that: (1) would contain $\mathcal{L}_A$ and $\mathcal{L}_B$ as sublogics, (2) would preserve orthogonality relations, (3) would allow the representation of quantum correlations observed in practice.
This preliminary result was clarified by Pulmannova \cite{Pullmanova} and then by R. Hudson and S. Pullmanova \cite{HudsonPullmanova} who demonstrated this result including completions of the product of two quantum logics.  
In another work, Pykacz and Santos \cite{Pykacz} considered two quantum systems of dimension 2 (qubits) and constructed a tensor product $\mathcal{L}_A \otimes \mathcal{L}_B$ between the lattices of projectors on $\mathcal{H}_A$ and $\mathcal{H}_B$. They demonstrated that any measure on this tensor product corresponding to a separable state necessarily satisfies Bell’s inequalities. In conclusion, it seems that any logical product constructed from local lattices (even MacNeille’s completion of this product) contains only product states, never violating Bell’s inequalities.  Pure entangled states (such as Bell states) cannot therefore appear in these logics because they require correlations that are impossible to construct logically from local products. This is a fundamental obstruction to any attempt at local logic for quantum composites.

However, from the outset in 1935, Schrödinger promoted the idea that what truly characterizes quantum behavior is precisely the way quantum systems are composed. Over the past 40 years, we have seen ample evidence to support this claim. Quantum non-locality has been confirmed experimentally, and the emphasis on quantum information has been reinforced.

Under the impetus of S. Abramsky and B.Coecke, a new categorical approach to quantum mechanics has been developed over the last twenty years \cite{AbramskyCoecke}. It makes central use of the properties of $\dagger-$categories.  Its ability to describe the main quantum phenomena has been analyzed \cite{AbramskyCoecke2} and its link to the operational description of quantum mechanics has been clarified \cite{AbramskyHeunen}.  In the latter work the state/property duality, which is at the heart of the axiomatic construction, appears as a Chu duality\footnote{Chu categories are seemingly fairly elementary generalizations of topological spaces, which have proved fundamental in the study of $*$-autonomous categories, but which also allow many common categories to be formalized.}.  In \cite[Theorem 3.15]{Abramsky2012}, S. Abramsky makes explicit the fact that the {\em Projective quantum symmetry groupoid} $PSymmH$\footnote{The objects of this category are the Hilbert spaces of dimension greater than two, and the morphisms are the orbits on semi-unitary maps (i.e. unitary or anti-unitary) under the $U(1)$ group action, which are the relevant symmetries of Hilbert spaces from the point of view of quantum mechanics.} is fully and faithfully represented by the category $bmChu_{[0,1]}$, i.e., by the sub-category of the category of bi-extensional Chu spaces associated with the evaluation set $[0,1]$ obtained by restricting it to Chu morphisms $(f_\ast,f^\ast)$ for which $f_\ast$ is injective. 
This result suggests that Chu categories could have a central role in the construction of axiomatic quantum mechanics as they provide a natural characterization of the automorphisms of the theory.  More surprisingly, and interestingly for us,  S. Abramsky shows that the aforementioned representation of $PSymmH$ is 'already' full and faithfull if we replace the evaluation space of the Chu category by a \underline{three-element set}, where the three values represent "definitely yes", "definitely no" and "maybe" \cite[Theorem 4.4]{Abramsky2012}. This possibility clearly points to the semantics developed as part of Piron's program. However, even though he suggests it in this article, S. Abramsky did not take the reverse path that would lead from a semantics adapted to the operational description of quantum mechanics to a categorical description based on Chu categories. This is the path we are taking ourselves.

We cannot overlook the research program on Generalized Probabilistic Theories (GPT) (see \cite{Plavala} for an excellent introduction adapted to our purpose).  This prolific program has its origins in the initial work of Mackey \cite{Mackey}, Ludwig \cite{Ludwig} and Kraus \cite{Kraus} (see also \cite{Janotta} for a more recent review) and exploits the fundamental property of quantum phenomena, namely that the results of these tests carried out on collections of samples, prepared independently and in a similar manner, have reproducible relative frequencies. A physical state will therefore be defined as a probability vector corresponding to the results of a well-chosen set of tests. Programs for reconstructing quantum mechanics  \cite{Hardy1}\cite{Hardy2}\cite{Mueller}\cite{Chiribella} then proceed to select a set of postulates restricting the state space thus defined, and prescriptions for describing complex situations based on simple cases. While this formal framework seems sufficiently close to the usual view of quantum mechanics to allow the initial postulates to be meaningful from the point of view of the physics described, it has nonetheless distanced us from the initial topic for several reasons. First, while the observer contributes to giving substance to the notions of preparation, operation, measurement, etc., the process of information acquisition by the observer has completely disappeared \cite{Barrett}. Secondly, the notion of state has lost all meaning for a single sample from a given preparation; it has now only meaning for an infinite collection of samples prepared independently and in a similar manner. Finally, to clarify the underlying structure of the family of tests structuring the state space, the promoters of this program generally use combinatorial tools that remain deeply attached to the situation of finite-dimensional systems, which removes some of the hope of generality inextricably linked to any axiomatic approach.

Other research programs have attempted to circumvent these difficulties and devote particular attention to the process of information acquisition by an observer through a binary test on a singular sample \cite{Rovelli}\cite{Zeilinger}\cite{Spekkens}\cite{Hoehn}. To simplify these authors' arguments, quantum description refers to two types of fundamental principles. First, the amount of information that can be acquired about a quantum system is fundamentally limited. However, it is always possible to identify tests whose outcome is indeterminate and which will therefore produce new information when the associated measurement is performed. This new acquisition of information will be paid for by the destruction of an equal amount of previously acquired information. It should be noted, however, that these programs generally make use of complementary principles relating to the convexity of the space in question (these principles ultimately conceal the principle of superposition). It should also be noted that these programs exploit the ‘combinatorics’ linked to the incompatibility of the tests that can be performed on the system to construct the state space. In this respect, they are very similar to the work exploiting the particular orthomodular nature of the lattice of properties developed under the impetus of C. Piron.  However, these programs do not abandon the probabilistic perspective traditionally adopted for the states of a quantum system.

Simon Kochen and Ernst Specker \cite{KochenSpecker} introduced Partial Boolean Algebras (PBAs) as an alternative to the quantum logic of Birkhoff–von Neumann and Piron’s axiomatic reconstruction. Unlike von Neumann, who proposed a global non-distributive lattice of Hilbert space subspaces, and Piron, who aimed to axiomatize this lattice to recover Hilbert space, Kochen emphasizes that quantum logic should remain local: each family of compatible observables forms an ordinary Boolean algebra, but there is no single global structure unifying them.  PBAs capture precisely this patchwork of overlapping Boolean algebras, reflecting the contextuality shown by the Kochen–Specker theorem \cite{KochenSpecker}, and thus provide a more flexible account of the logic of quantum propositions.

My current research \cite{Buffenoir2021}\cite{Buffenoir2022}\cite{Buffenoir2025} aims to develop a strategy for circumventing the aforementioned no-go theorems, relating to the non-existence of a tensor product that can be used to describe composite systems based on the logical description of individual systems. This research program adopts a formalism very close to that initially adopted by G. Birkhoff and J. von Neumann and therefore completely abandons the categorical approach used by S. Abramsky and Al.   From a certain point of view, our approach is inspired by Kochen's, as it restricts involution to be only a partial application (the conjunction is global).\\
My program is exhaustive in that it covers all aspects of an operational approach to quantum mechanics, from general considerations of the operational approach to consequences relating to contextuality and the existence of non-local states according to Bell. Our quite long presentation therefore covers all these aspects.

In Section 2, we summarize the formalism of States/Effects Chu spaces and of corresponding morphisms. As mentioned before, these Chu spaces are GPT like descriptions based on the "boolean domain" as a target space (the three elements set equipped with some structures, see properties (\ref{orderB}) to (\ref{Binvol})). Then, we introduce our fundamental notion of {\em real structure} putting a distinction between {\em real states} and {\em hidden states} through the definition of a star involution. At the end of Section 2, we clarify our definitions of determinism/indeterminism in our framework. 

Section 3 and Section 5 are devoted to clarify the notion of real structure by showing how the generalized spaces of states owning hidden states are (i) obtained naturally by a completion procedure from real \underline{indeterministic} spaces of states (ii) linked to \underline{contextual} empirical models associated to the operational descriptions of these real spaces of states. 

Section 4 clarifies our requirements about tensor products of spaces of states in order to describe compound systems from the description of individual systems. We note that the basic solution to these requirements allows also the description of compound deterministic (classical) systems.  

Section 6 completes this analysis by giving a proposal for the description of compound indeterministic systems.  This proposal exploits our description of ontic completions presented along Section 3. 

However this description is rather general and Section 7 enters into a more detailed description of the iterated tensor product of elementary indeterministic systems (to describe several quantum bits).  This analysis exhibits a structure that is as close as possible from \underline{Hilbert geometries}.  

Section 8 shows that our description exhibits some fundamental quantum-like properties : \underline{no-broadcasting theorem} and  \underline{existence of {\em Bell non-local} states}.

\section{The operational formalism}\label{sectiongeneral}

Adopting the operational perspective on the considered experiments, we will introduce the following definitions.\\
A {\em preparation process} is an objectively defined, and thus 'repeatable', experimental sequence that allows singular samples of a certain physical system to be produced, in such a way that we are able to submit them to tests. 
We will denote by ${ \mathfrak{P}}$ the set of preparation processes (each element of ${ \mathfrak{P}}$ can be equivalently considered as the collection of samples produced through this preparation procedure). 
\\
For each  {\em property}, that the observer aims to test macroscopically on {\em any particular sample} of the considered micro-system, it will be assumed that the observer is able to define (i) some detailed 'procedure', in reference to the modes of use of some experimental apparatuses chosen to perform the operation/test, and (ii) a 'rule' allowing the answer 'yes' to be extracted if the macroscopic outcome of the experiment conforms with the expectation of the observer, when the test is performed on any input sample (as soon as this experimental procedure can be opportunely applied to this particular sample). 
These operations/tests, designed to determine the occurrence of a given property for a given sample, will be called {\em yes/no tests} {\em associated with this property}.  The set of 'yes/no tests' at the disposal of the observer will be denoted by ${ \mathfrak{T}}$. 

A yes/no test ${ \mathfrak{t}}\in { \mathfrak{T}}$ will be said to be {\em positive with certainty} (resp. {\em negative with certainty}) relatively to a preparation process ${ \mathfrak{p}}\in { \mathfrak{P}}$ iff the observer is led to affirm that the result of this test, realized on any of the particular samples that could be prepared according to this preparation process, would be 'positive with certainty' (resp. would be 'negative with certainty'), 'should' this test be effectuated. If the yes/no test can not be stated as 'certain', this yes/no test will be said to be {\em indeterminate}. \\ Concretely, the observer can establish the 'certainty' of the result of a given yes/no test on any given sample issued from a given preparation procedure, by running the same test on a sufficiently large (but finite) collection of samples issued from this same preparation process: if the outcome is always the same, the observer will be led to claim that similarly prepared 'new' samples would also produce the same result, if the experiment was effectuated. 
To summarize, for any preparation process ${ \mathfrak{p}}$ and any yes/no test ${ \mathfrak{t}}$, the evaluation ${ \mathfrak{e}}({ \mathfrak{p}},{ \mathfrak{t}})$ is an element of the set ${ \mathfrak{B} } := \{{ \bot}, { \rm \bf Y}, { \rm \bf N}\}$ \footnote{
The set ${ \mathfrak{B}}:=\{\textit{\bf Y},\textit{\bf N},\bot\}$ will be equipped with the following poset structure : 
\begin{eqnarray}
\forall u,v\in { \mathfrak{B} },&& (u\leq v)\; :\Leftrightarrow\; (u={ \bot}\;\;\textit{\rm or}\;\; u=v).\label{orderB}
\end{eqnarray} 
$({ \mathfrak{B}},\leq)$ is also an Inf semi-lattice which infima (resp. suprema) will be denoted $\bigwedge$ (resp. $\bigvee$). We have
\begin{eqnarray}
\forall x,y\in { \mathfrak{B} },&&  x \wedge y = \left\{\begin{array}{ll} x & \textit{\rm if}\;\;\; x=y\\
\bot & \textit{\rm if}\;\;\; x\not= y\end{array}\right.
\end{eqnarray}
We will also introduce a commutative monoid law denoted $\bullet$ and defined by
\begin{eqnarray}
\forall x\in { \mathfrak{B}},&& x \bullet \textit{\bf Y}=x, \;\;\;\;\; x \bullet \textit{\bf N}=\textit{\bf N}, \;\;\;\;\; \bot \bullet \bot = \bot. \label{expressionbullet}
\end{eqnarray}
This product law verifies the following properties
\begin{eqnarray}
\forall x\in { \mathfrak{B}},\forall B\subseteq { \mathfrak{B}} && x \bullet \bigwedge B=\bigwedge{}_{b\in B}(x\bullet b),\label{distributivitybullet}\\
\forall x\in { \mathfrak{B}},\forall C\subseteq_{Chain} { \mathfrak{B}} && x \bullet \bigvee C =\bigvee{}_{c\in C}(x\bullet c).\label{distributivitybullet2}
\end{eqnarray}
$({ \mathfrak{B} },\leq)$ will be also equipped with the following involution map :
\begin{eqnarray}
\overline{{ \bot}}:={ \bot} \;\;\;\;\;\;\;\; \overline{\textit{\bf Y}}:={\textit{\bf N}}\;\;\;\;\;\;\;\; \overline{\textit{\bf N}}:={\textit{\bf Y}}.\label{Binvol}
\end{eqnarray} 
$({ \mathfrak{B} },\leq)$ will be called {\em the boolean domain}.
} defined to be ${ \bot}$ (alternatively, ${ \rm \bf Y}$ or ${ \rm \bf N}$) if the outcome of the yes/no test ${ \mathfrak{t}}$ on any sample prepared according to the preparation procedure ${ \mathfrak{p}}$ is judged as 'indeterminate' ('positive with certainty' or 'negative with certainty', respectively) by the observer. We note that the partial order placed on ${ \mathfrak{B}}$ by equation (\ref{orderB}) characterizes the amount of information gathered by the observer. Precisely, when the determinacy of a yes/no test is established for an observer, we can consider that this observer possesses some elementary 'information' about the state of the system, whereas, in the 'indeterminate case', the observer has none.

This definition leads to a pre-order structure denoted by $\sqsubseteq_{{}_{ \mathfrak{P}}}$ on the space of preparations ${ \mathfrak{P}}$ and to a pre-order structure denoted by $\sqsubseteq_{{}_{ \mathfrak{T}}}$ on the space of tests ${ \mathfrak{T}}$ as shown in \cite{Buffenoir2021} :
\begin{eqnarray}
&& \forall { \mathfrak{p}}_1,{ \mathfrak{p}}_2\in { \mathfrak{P}}, \;\; ({ \mathfrak{p}}_1\sqsubseteq_{{}_{ \mathfrak{P}}} { \mathfrak{p}}_2)\;   :\Leftrightarrow  
(\; \forall { \mathfrak{t}}\in { \mathfrak{T}},\; { \mathfrak{e}}({ \mathfrak{p}}_1,{ \mathfrak{t}})\leq { \mathfrak{e}}({ \mathfrak{p}}_2,{ \mathfrak{t}}) \;)\\
&& \forall { \mathfrak{t}}_1, { \mathfrak{t}}_2\in { \mathfrak{T}},\;\; (\; { \mathfrak{t}}_1 \sqsubseteq_{{}_{ \mathfrak{T}}} { \mathfrak{t}}_2\;)  :\Leftrightarrow 
(\; \forall { \mathfrak{p}}\in { \mathfrak{P}},\; { \mathfrak{e}}({ \mathfrak{p}},{{ \mathfrak{t}}_1})\leq { \mathfrak{e}}({ \mathfrak{p}},{{ \mathfrak{t}}_2}) \;).
\end{eqnarray}

If we consider a collection of preparation processes, we can define a new preparation procedure called {\em mixture} as follows. The samples produced from the mixtured preparation procedure are obtained by a random mixing of the samples issued from the preparation processes of the considered collection indiscriminately.  As a consequence, the statements that the observer can establish after a sequence of tests on these samples produced through the mixtured procedure is given as the infimum of the statements that the observer can establish for the elements of the collection separately. \\
We will also assume that there exists a preparation process, unique from the point of view of the statements that can be produced about it, that can be interpreted as a 'randomly-selected' collection of 'un-prepared samples'. This element leads to complete indeterminacy for any yes/no test realized on it. \\
An equivalence relation can be defined from the previously defined pre-order on the set of preparations ${ \mathfrak{P}}$.  Two preparation processes are identified iff the statements established by the observer about the corresponding prepared samples are identical.  A {\em state} of the physical system is an equivalence class of preparation processes corresponding to the same informational content.
The set of equivalence classes will be called {\em space of states} and denoted ${ { \mathfrak{S}}}$.
\begin{eqnarray}
&& \forall { \mathfrak{p}}_1,{ \mathfrak{p}}_2\in { \mathfrak{P}}, \;\; ({ \mathfrak{p}}_1\sim_{{}_{ \mathfrak{P}}} { \mathfrak{p}}_2)\;   :\Leftrightarrow  
(\; \forall { \mathfrak{t}}\in { \mathfrak{T}},\; { \mathfrak{e}}({ \mathfrak{p}}_1,{ \mathfrak{t}})= { \mathfrak{e}}({ \mathfrak{p}}_2,{ \mathfrak{t}}) \;)\\
&& \lceil {{ \mathfrak{p}}} \rceil  :=  \{\, { \mathfrak{p}}'\in { \mathfrak{P}}\;\vert\; { \mathfrak{p}}'\sim_{{}_{ \mathfrak{P}}} { \mathfrak{p}}\,\}\\
&& { { \mathfrak{S}}}  :=  \{\, \lceil {{ \mathfrak{p}}} \rceil \;\vert\; { \mathfrak{p}}\in { \mathfrak{P}}\,\}
\end{eqnarray}
The space of states inherits a pointed partial order structure 
and an Inf semi-lattice structure
.\\
If we consider a collection of tests, we can define a new test called {\em mixture} as follows. The result obtained for the mixtured test is obtained by a random mixing of the results issued from the tests of the considered collection indiscriminately.  As a consequence, the statements that the observer can establish after a sequence of tests is given as the infimum of the statements that the observer can establish for each test separately.

An equivalence relation can be defined from the previously defined pre-order on the set of yes/no tests ${ \mathfrak{T}}$.  An {\em effect} of the physical system is an equivalence class of yes/no tests, i.e., a class of yes/no tests that are not distinguished from the point of view of the statements that the observer can produce by using these yes/no tests on finite collections of samples. The set of equivalence classes of yes/no tests will be denoted by ${ \mathfrak{E}}$. 
\begin{eqnarray}
&& \forall { \mathfrak{t}}_1, { \mathfrak{t}}_2\in { \mathfrak{T}},\;\; (\; { \mathfrak{t}}_1 \sim_{{}_{ \mathfrak{T}}} { \mathfrak{t}}_2\;)  :\Leftrightarrow 
(\; \forall { \mathfrak{p}}\in { \mathfrak{P}},\; { \mathfrak{e}}({ \mathfrak{p}},{{ \mathfrak{t}}_1})= { \mathfrak{e}}({ \mathfrak{p}},{{ \mathfrak{t}}_2}) \;)\\
&& \lfloor { \mathfrak{t}} \rfloor  :=  \{\, { \mathfrak{t}}'\in { \mathfrak{T}}\;\vert\; { \mathfrak{t}}'\sim_{{}_{ \mathfrak{T}}} { \mathfrak{t}}\,\}\\
&& { \mathfrak{E}}  :=  \{\,\lfloor { \mathfrak{t}} \rfloor \;\vert\; { \mathfrak{t}}\in { \mathfrak{T}}\,\}.
\end{eqnarray}
The space of effects inherits a partial order structure 
and an Inf semi-lattice structure
.\\
We then derive a map ${\epsilon}$ according to the following definition :
\begin{eqnarray}
\begin{array}{rcrcl}
{\epsilon}^{ \mathfrak{S}} & : &{ \mathfrak{E}} & \rightarrow & { \mathfrak{B} }{}^{ \mathfrak{S}} \\
& &{ \mathfrak{l}} & \mapsto & {\epsilon}^{ \mathfrak{S}}_{ { \mathfrak{l}}} \;\;\;\;\vert\;\;\;\;\; {\epsilon}^{ \mathfrak{S}}_{ \lfloor { \mathfrak{t}} \rfloor}(\lceil { \mathfrak{p}} \rceil):={ \mathfrak{e}}({ \mathfrak{p}},{ \mathfrak{t}}),\; \forall { \mathfrak{p}}\in { \mathfrak{P}},\forall { \mathfrak{t}}\in { \mathfrak{T}}.\label{defetilde}
\end{array}
\end{eqnarray}

\subsection{States/Effects Chu spaces}\label{subsectiongeneralized}

According to the previous considerations

\begin{definition}\label{definitiongeneralizedspaceofstates}
We will define {\em a space of states} (denoted ${ \mathfrak{S}}$) to be a poset satisfying
\begin{itemize}
\item ${ \mathfrak{S}}$ is a down complete Inf semi-lattice 
\item ${ \mathfrak{S}}$ admits a bottom element (denoted $\bot_{{}_{ \mathfrak{S}}}$).
\end{itemize}
The Infimum of a family $\{\,\sigma_i\;\vert\; i\in I\,\}\subseteq { \mathfrak{S}}$ will be denoted $\bigsqcap{}^{{}^{{ \mathfrak{S}}}}_{i\in I}\sigma_i$. The associated partial order on ${ \mathfrak{S}}$ will be denoted by $\sqsubseteq_{{}_{{ \mathfrak{S}}}}$.
\end{definition}

Here and in the following, $\widehat{\Sigma\Sigma'}{}^{{}^{ \mathfrak{S}}}$ will mean that $\Sigma$ and $\Sigma'$ have a common upper-bound in ${ \mathfrak{S}}$, and $\neg \widehat{\Sigma\Sigma'}{}^{{}^{ \mathfrak{S}}}$ means they have none.  More generally, for any subset $S\subseteq { \mathfrak{S}}$, the logical value $\widehat{\;S\;}{}^{{}^{ \mathfrak{S}}}$ will be "TRUE" iff the elements of $S$ admit a common upper-bound in ${ \mathfrak{S}}$.\\
The supremum of two elements $\sigma$ and $\sigma'$ in ${ \mathfrak{S}}$, when it exists (i.e. when $\sigma$ and $\sigma'$ have a common upper-bound), will be denoted $\sigma \sqcup_{{}_{{ \mathfrak{S}}}} \sigma'$. \\
We will also adopt the following notations :   $\uparrow^{{}^{ \mathfrak{S}}}\!\!\Sigma$ (resp. $\downarrow_{{}_{ \mathfrak{S}}}\!\!\Sigma$) for the upper subset (resp. the lower subset) $\{ \sigma\in { \mathfrak{S}} \;\vert\; \sigma\sqsupseteq_{{}_{{ \mathfrak{S}}}}\Sigma\}$ (resp. $\{ \sigma\in { \mathfrak{S}} \;\vert\; \sigma\sqsubseteq_{{}_{{ \mathfrak{S}}}}\Sigma\}$). We will also denote $\sigma \parallel_{{}_{ \mathfrak{S}}}\sigma'$ for the logical value ($\sigma \not\sqsubseteq_{{}_{ \mathfrak{S}}}\sigma'$ and $\sigma' \not\sqsubseteq_{{}_{ \mathfrak{S}}}\sigma$).\\
Endly, we will introduce the "covering" relation denoted $\sqcoversubset_{{}_{{ \mathfrak{S}}}}$ which means :
\begin{eqnarray}
\forall \alpha,\beta\in { \mathfrak{S}},\;\;\alpha \sqcoversubset_{{}_{{ \mathfrak{S}}}}\beta &:\Leftrightarrow & (\,\alpha \sqsubset_{{}_{{ \mathfrak{S}}}}\beta\;\;\textit{\rm and}\;\; \forall \gamma\in { \mathfrak{S}},\;\;\;\alpha\sqsubset_{{}_{{ \mathfrak{S}}}}\gamma\sqsubseteq_{{}_{{ \mathfrak{S}}}}\beta\;\;\Rightarrow\;\;\gamma=\beta\,).\;\;\;\;\;\;\;\;\;\;\;\;\;\;\;\label{defcoversubset}
\end{eqnarray}

\begin{definition}\label{FiniteRankCondition}
We will say that the space of states ${ \mathfrak{S}}$ satisfies the {\em Finite Rank Condition} iff 
\begin{eqnarray}
\forall U\subseteq { \mathfrak{S}}\;\vert\; \widehat{\;U\;}{}^{{}^{{ \mathfrak{S}}}},&&\exists \;\textit{\rm $V$ finite subset of ${ \mathfrak{S}}$ with}\nonumber\\&& (\forall \sigma\in V,\exists \sigma'\in U, \sigma\sqsubseteq_{{}_{{ \mathfrak{S}}}}\sigma') \;\textit{\rm and}\;(\bigsqcup{}^{{}^{{ \mathfrak{S}}}}  U=\bigsqcup{}^{{}^{{ \mathfrak{S}}}}  V).\;\;\;\;\;\;\;\;\;\;\;\;\;\;\label{finiterank}
\end{eqnarray}
\end{definition}

\begin{definition}\label{defpurestates}
We will say that {\em the space of states ${ \mathfrak{S}}$ admits a description in terms of pure states} iff we have moreover
that the set of complely meet-irreducible elements of ${ \mathfrak{S}}$, denoted ${ \mathfrak{S}}^{{}^{pure}}$ and called {\em set of pure states}, is equal to the set of maximal elements $Max({ \mathfrak{S}})$ and it 
 is a generating set for ${ \mathfrak{S}}$, i.e.  
\begin{eqnarray}
&&\forall \sigma \in { \mathfrak{S}}, \;\; \sigma= \bigsqcap{}^{{}^{ { \mathfrak{S}}}}  \underline{\sigma}_{{}_{ { \mathfrak{S}}}} ,\;\;\textit{\rm where}\;\;
\underline{\sigma}_{{}_{ { \mathfrak{S}}}}:=
({ \mathfrak{S}}^{{}^{pure}} \cap (\uparrow^{{}^{ { \mathfrak{S}}}}\!\!\!\! \sigma) ) \;\;\textit{\rm and}\;\; { \mathfrak{S}}^{{}^{pure}}=Max({ \mathfrak{S}}).\;\;\;\;\;\;\;\;\;\;\;\;
\label{completemeetirreducibleordergenerating}
\end{eqnarray}
Except if explicitly mentioned, we will not explicitly assume that the considered spaces of states admit a description in terms of pure states.
\end{definition}

\begin{definition}\label{defgeneralizedeffectspace}
The {\em space of effects}, denoted ${ \mathfrak{E}}_{ \mathfrak{S}}$ is defined to be 
\begin{eqnarray}
{ \mathfrak{E}}_{ \mathfrak{S}}:=\{ { \mathfrak{l}}_{{}_{(\sigma,\sigma')}}\;\vert\; \sigma,\sigma'\in { \mathfrak{S}},\;\neg \widehat{\sigma\sigma'}{}^{{}^{ \mathfrak{S}}}\;\}\cup \{ { \mathfrak{l}}_{{}_{(\sigma,\centerdot)}}\;\vert\; \sigma\in { \mathfrak{S}}\;\}\cup \{ { \mathfrak{l}}_{{}_{(\centerdot,\sigma)}}\;\vert\; \sigma\in { \mathfrak{S}}\;\}\cup \{\, { \mathfrak{l}}_{{}_{(\centerdot,\centerdot)}}\,\} &&\;\;\;\;\;\;\;\;
\end{eqnarray}
as a set, with the following Inf semi-lattice law
\begin{eqnarray}
{ \mathfrak{l}}_{{}_{(\sigma_1,\sigma'_1)}} \sqcap_{{}_{{ \mathfrak{E}}_{ \mathfrak{S}}}} { \mathfrak{l}}_{{}_{(\sigma_2,\sigma'_2)}} &=& \left\{ \begin{array}{lcl}
{ \mathfrak{l}}_{{}_{(\sigma_1\sqcup_{{}_{{ \mathfrak{S}}}} \sigma_2,\sigma'_1 \sqcup_{{}_{{ \mathfrak{S}}}} \sigma'_2)}} & \textit{\rm if} & \widehat{\sigma_1\sigma_2}{}^{{}^{ \mathfrak{S}}}\;\textit{\rm and}\; \widehat{\sigma'_1\sigma'_2}{}^{{}^{ \mathfrak{S}}}\\
{ \mathfrak{l}}_{{}_{(\;\centerdot\;,\sigma'_1 \sqcup_{{}_{{ \mathfrak{S}}}} \sigma'_2)}} & \textit{\rm if} & \neg \widehat{\sigma_1\sigma_2}{}^{{}^{ \mathfrak{S}}}\;\textit{\rm and}\; \widehat{\sigma'_1\sigma'_2}{}^{{}^{ \mathfrak{S}}}\\
{ \mathfrak{l}}_{{}_{(\sigma_1\sqcup_{{}_{{ \mathfrak{S}}}} \sigma_2,\;\centerdot\;)}} & \textit{\rm if} & \widehat{\sigma_1\sigma_2}{}^{{}^{ \mathfrak{S}}}\;\textit{\rm and}\; \neg\widehat{\sigma'_1\sigma'_2}{}^{{}^{ \mathfrak{S}}}\\
{ \mathfrak{l}}_{{}_{(\centerdot,\centerdot)}} & \textit{\rm if} & \neg\widehat{\sigma_1\sigma_2}{}^{{}^{ \mathfrak{S}}}\;\textit{\rm and}\; \neg\widehat{\sigma'_1\sigma'_2}{}^{{}^{ \mathfrak{S}}}
\end{array}\right.\label{defcapES}
\end{eqnarray}
This expression is naturally extended to the whole set of effects (i.e. including the elements of the form ${ \mathfrak{l}}_{{}_{(\sigma,\centerdot)}}$, ${ \mathfrak{l}}_{{}_{(\centerdot,\sigma)}}$ and ${ \mathfrak{l}}_{{}_{(\centerdot,\centerdot)}}$) as long as we adopt the convention defining 
\begin{eqnarray}
\forall \sigma\in { \mathfrak{S}},\;\;\;\;\neg\widehat{\;\centerdot\;\sigma}{}^{{}^{ \mathfrak{S}}}:= \textit{\rm TRUE}\;\;\;\;\;\;\; &\textit{\rm and}& \;\;\;\;\;\;\; \neg\widehat{\;\centerdot\;\centerdot\;}{}^{{}^{ \mathfrak{S}}}:= \textit{\rm TRUE}.
\end{eqnarray}
\end{definition}

\begin{definition}
The evaluation map $\epsilon^{ \mathfrak{S}}$ is defined by
\begin{eqnarray}
\epsilon^{ \mathfrak{S}}_{{ \mathfrak{l}}_{{}_{(\sigma,\sigma')}}}(\sigma'') &:= & \left\{ \begin{array}{lcl} \textit{\bf Y} & \textit{\rm if} & \sigma\sqsubseteq_{{}_{ \mathfrak{S}}} \sigma''\\\textit{\bf N} & \textit{\rm if} & \sigma'\sqsubseteq_{{}_{ \mathfrak{S}}} \sigma''\\ \bot & \textit{\rm otherwise } &\end{array}\right.\label{defepsilonS}
\end{eqnarray}
this expression is naturally extended to the whole set of effects by adopting the following convention
\begin{eqnarray}
\forall \sigma\in { \mathfrak{S}},\;\; (\;\sigma\;\sqsubseteq_{{}_{{ \mathfrak{S}}}}\centerdot ) & :=&\textit{\rm TRUE}.
\end{eqnarray}
\end{definition}

\begin{definition}
We will adopt the following notations :
\begin{eqnarray}
\bot_{{}_{{ \mathfrak{E}}_{ \mathfrak{S}}}} & := & { \mathfrak{l}}_{{}_{(\centerdot,\centerdot)}}\\
{ \mathfrak{Y}}_{{}_{{ \mathfrak{E}}_{ \mathfrak{S}}}} & := & { \mathfrak{l}}_{{}_{(\bot_{{}_{{ \mathfrak{S}}}},\centerdot)}}\\
\forall \sigma_1,\sigma_2\in { \mathfrak{S}}\cup\{\centerdot\},\;\;\;\;\;\; \overline{{ \mathfrak{l}}_{{}_{(\sigma_1,\sigma_2)}}}& := & { \mathfrak{l}}_{{}_{(\sigma_2,\sigma_1)}}
\end{eqnarray}
\end{definition}

\begin{lemma}
The evaluation map $\epsilon^{ \mathfrak{S}}$ satisfies the following properties :
\begin{eqnarray}
\forall { \mathfrak{l}}\in { \mathfrak{E}}_{ \mathfrak{S}},\forall S\subseteq { \mathfrak{S}}, && { \epsilon}^{ \mathfrak{S}}_{ { \mathfrak{l}}}(\bigsqcap{}^{{}^{ \mathfrak{S}}} S)=\bigwedge{}_{{}_{\sigma\in S}}\; { \epsilon}^{ \mathfrak{S}}_{ { \mathfrak{l}}}(\sigma),
 \label{axiomsigmainfsemilattice}\\
 \forall \sigma\in { \mathfrak{S}},\forall E\subseteq { \mathfrak{E}}_{ \mathfrak{S}},&&{ \epsilon}^{ \mathfrak{S}}_{\bigsqcap{}^{{}^{ \mathfrak{E}_{ \mathfrak{S}}}} E}(\sigma)=\bigwedge{}_{{}_{{ \mathfrak{l}}\in E}}\; { \epsilon}^{ \mathfrak{S}}_{ { \mathfrak{l}}}(\sigma),\label{axiomEinfsemilattice}\\
 \forall { \mathfrak{l}}\in { \mathfrak{E}}_{ \mathfrak{S}}, \forall \sigma\in { \mathfrak{S}}, & & { \epsilon}^{ \mathfrak{S}}_{ \overline{\; \mathfrak{l}\;}}(\sigma)= \overline{{ \epsilon}^{ \mathfrak{S}}_{ { \mathfrak{l}}}(\sigma)},\label{etbar}\\
\forall \sigma\in { \mathfrak{S}}, & & { \epsilon}^{ \mathfrak{S}}_{{ \mathfrak{Y}}_{ \mathfrak{E}_{ \mathfrak{S}}}}(\sigma)= \textit{\bf Y},\label{ety}\\
\forall \sigma\in { \mathfrak{S}},  & & { \epsilon}^{ \mathfrak{S}}_{\bot_{ \mathfrak{E}_{ \mathfrak{S}}}}(\sigma)= \bot.
\end{eqnarray}
and 
\begin{eqnarray}
&& \forall { \mathfrak{l}},{ \mathfrak{l}}'\in { \mathfrak{E}}_{ \mathfrak{S}},\;\;\;\;\;\;\;\;\;\;\;\;(\, \forall \sigma\in { \mathfrak{S}},\; { \epsilon}^{ \mathfrak{S}}_{ { \mathfrak{l}}}(\sigma)={ \epsilon}^{ \mathfrak{S}}_{ { \mathfrak{l}}'}(\sigma) \,) \Leftrightarrow  (\, { \mathfrak{l}}= { \mathfrak{l}}' \,),\label{Chuextensional}\\
&& \forall \sigma,\sigma'\in { \mathfrak{S}},\;\;\;\;\;\;\;\;\;\;\;\;(\, \forall { \mathfrak{l}}\in { \mathfrak{E}}_{ \mathfrak{S}},\; { \epsilon}^{ \mathfrak{S}}_{ { \mathfrak{l}}}(\sigma)={ \epsilon}^{ \mathfrak{S}}_{ { \mathfrak{l}}}(\sigma') \,) \Leftrightarrow  (\, \sigma= \sigma' \,)\label{Chuseparated}
\end{eqnarray}
\end{lemma}

$({ \mathfrak{S}},{ \mathfrak{E}}_{ \mathfrak{S}},\epsilon^{ \mathfrak{S}})$ will be designated as the generalized States/Effects Chu space associated to ${ \mathfrak{S}}$.\\

We also note the following useful property. 
\begin{lemma}\label{lemma2compatiblemeasurements}
If $(\sigma_1\sqcap_{{}_{{ \mathfrak{S}}}} \sigma_2)$ and $(\sigma'_1\sqcap_{{}_{{ \mathfrak{S}}}} \sigma'_2)$ admit no common upper-bound in ${ \mathfrak{S}}$, then ${ \mathfrak{l}}_{{}_{(\sigma_1,\sigma'_1)}}$ and ${ \mathfrak{l}}_{{}_{(\sigma_2,\sigma'_2)}}$ admit a common upper-bound in ${ \mathfrak{E}}_{ \mathfrak{S}}$ which is ${ \mathfrak{l}}_{{}_{(\sigma_1\sqcap_{{}_{{ \mathfrak{S}}}} \sigma_2,\sigma'_1\sqcap_{{}_{{ \mathfrak{S}}}} \sigma'_2)}}$.  As a consequence, the supremum $({ \mathfrak{l}}_{{}_{(\sigma_1,\sigma'_1)}} \sqcup_{{}_{{ \mathfrak{E}}_{ \mathfrak{S}}}} { \mathfrak{l}}_{{}_{(\sigma_2,\sigma'_2)}})$ exists in ${ \mathfrak{E}}_{ \mathfrak{S}}$ and is given by
\begin{eqnarray}
{ \mathfrak{l}}_{{}_{(\sigma_1,\sigma'_1)}} \sqcup_{{}_{{ \mathfrak{E}}_{ \mathfrak{S}}}} { \mathfrak{l}}_{{}_{(\sigma_2,\sigma'_2)}}  &=& { \mathfrak{l}}_{{}_{(\sigma_1\sqcap_{{}_{{ \mathfrak{S}}}} \sigma_2,\sigma'_1\sqcap_{{}_{{ \mathfrak{S}}}} \sigma'_2)}}.
\end{eqnarray}
This result can be extended to the whole set of effects as long as we adopt the following conventions  
\begin{eqnarray}
(\;\centerdot\;\sqcap_{{}_{{ \mathfrak{S}}}} \sigma) := \sigma  \;\;\;\;\;\;\; &\textit{\rm and}& \;\;\;\;\;\;\;  (\;\centerdot\;\sqcap_{{}_{{ \mathfrak{S}}}} \;\centerdot\;) := \;\centerdot\;
\end{eqnarray}
\end{lemma}

\subsection{Morphisms}\label{subsectionchannels}

We turn the collection of States/Effects Chu spaces into a category by defining the following morphisms.

\begin{definition}
We will consider the morphisms from a States/Effects Chu space $({ \mathfrak{S}}_{A},{ \mathfrak{E}}_{{ \mathfrak{S}}_A},\epsilon^{{ \mathfrak{S}}_A})$ to another States/Effects Chu space $({ \mathfrak{S}}_{B},{ \mathfrak{E}}_{{ \mathfrak{S}}_B},\epsilon^{{ \mathfrak{S}}_B})$, i.e.  pairs of maps $f : { \mathfrak{S}}_{A} \rightarrow { \mathfrak{S}}_{B}$ and $f^{\ast} : { \mathfrak{E}}_{{ \mathfrak{S}}_B}\rightarrow { \mathfrak{E}}_{{ \mathfrak{S}}_A}$ satisfying the following properties  
\begin{eqnarray} 
\forall \sigma_A\in { \mathfrak{S}}_{A}, \forall { \mathfrak{l}}_B \in { \mathfrak{E}}_{{ \mathfrak{S}}_B}&&\epsilon^{{ \mathfrak{S}}_B}_{{ \mathfrak{l}}_B}(f(\sigma_A))=\epsilon^{{ \mathfrak{S}}_A}_{f^{\ast}({ \mathfrak{l}}_B)}(\sigma_A).\label{defchumorphism}
\end{eqnarray}
\end{definition}

\begin{remark}
Note that, the eventual surjectivity of $f^{\ast}$ implies the injectivity of $f$.  This point uses the property (\ref{Chuseparated}).  Explicitly,
\begin{eqnarray}
\forall \sigma_A,\sigma'_A\in { \mathfrak{S}}_{A},\;\; f(\sigma_A)=f(\sigma'_A) &\Leftrightarrow & (\,\forall { \mathfrak{l}}_B\in { \mathfrak{E}}_{{ \mathfrak{S}}_B},\;\epsilon^{{ \mathfrak{S}}_B}_{{ \mathfrak{l}}_B}(f(\sigma_A))=\epsilon^{{ \mathfrak{S}}_B}_{{ \mathfrak{l}}_B}(f(\sigma'_A))\,)\nonumber\\
&\Leftrightarrow & (\,\forall { \mathfrak{l}}_B\in { \mathfrak{E}}_{{ \mathfrak{S}}_B},\;\epsilon^{{ \mathfrak{S}}_B}_{f^{\ast}({ \mathfrak{l}}_B)}(\sigma_A)=\epsilon^{{ \mathfrak{S}}_B}_{f^{\ast}({ \mathfrak{l}}_B)}(\sigma'_A)\,)\nonumber\\
&\Leftrightarrow & (\,\forall { \mathfrak{l}}'_B\in { \mathfrak{E}}_{{ \mathfrak{S}}_B},\;\epsilon^{{ \mathfrak{S}}_B}_{{ \mathfrak{l}}'_B}(\sigma_A)=\epsilon^{{ \mathfrak{S}}_B}_{{ \mathfrak{l}}'_B}(\sigma'_A)\,)\nonumber\\
&\Leftrightarrow & (\, \sigma_A=\sigma'_A\,).
\end{eqnarray}
In the same way,  using the properties (\ref{Chuextensional}) and the eventual surjectivity of $f$, we can deduce the injectivity of $f^{\ast}$. 
\end{remark}

The duality property (\ref{defchumorphism}) suffices to deduce the following properties.
\begin{theorem} The left-component $f$ of a Chu morphism from $({ \mathfrak{S}}_A,{ \mathfrak{E}}_{{ \mathfrak{S}}_A},\epsilon^{{ \mathfrak{S}}_A})$ to $({ \mathfrak{S}}_B,{ \mathfrak{E}}_{{ \mathfrak{S}}_B},\epsilon^{{ \mathfrak{S}}_B})$ satisfies
\begin{eqnarray}
\forall S\subseteq { \mathfrak{S}}_A,&& f(\bigsqcap{}^{{}^{{ \mathfrak{S}}_A}} \! S) = \bigsqcap{}^{{}^{{ \mathfrak{S}}_B}}_{{}_{\sigma\in S}} \; f(\sigma)\label{f12cap}
\end{eqnarray}
The right-component $f^{\ast}$ of a Chu morphism from $({ \mathfrak{S}}_A,{ \mathfrak{E}}_{{ \mathfrak{S}}_A},\epsilon^{{ \mathfrak{S}}_A})$ to $({ \mathfrak{S}}_B,{ \mathfrak{E}}_{{ \mathfrak{S}}_B},\epsilon^{{ \mathfrak{S}}_B})$ satisfies
\begin{eqnarray}
\forall E\subseteq { \mathfrak{E}}_{{ \mathfrak{S}}_B},&& f^{\ast}(\bigsqcap{}^{{}^{{ \mathfrak{E}}_{{ \mathfrak{S}}_B}}} \! E) = \bigsqcap{}^{{}^{{ \mathfrak{E}}_{{ \mathfrak{S}}_A}}}_{{}_{{ \mathfrak{l}}\in E}} \; f^{\ast}({ \mathfrak{l}})\label{f21cap}\\
\forall { \mathfrak{l}} \in { \mathfrak{E}}_{{ \mathfrak{S}}_B},&& f^{\ast}(\; \overline{\,{ \mathfrak{l}}\,} \;) = \overline{f^{\ast}({ \mathfrak{l}})}\label{f21bar}\\
&& f^{\ast}(\;{ \mathfrak{Y}}_{{ \mathfrak{E}}_{{ \mathfrak{S}}_B}}\;) = { \mathfrak{Y}}_{{ \mathfrak{E}}_{{ \mathfrak{S}}_A}}.\label{f21Y}
\end{eqnarray}
\end{theorem}
\begin{proof}
All proofs follow the same trick based on the duality relation (\ref{defchumorphism}) and the separation property (\ref{Chuseparated}).  For example, for any $S\subseteq { \mathfrak{S}}_A$ and any ${ \mathfrak{l}} \in { \mathfrak{E}}_{{ \mathfrak{S}}_B}$, we have, using (\ref{defchumorphism}) and (\ref{axiomEinfsemilattice}) :
\begin{eqnarray}
\epsilon^{{ \mathfrak{S}}_B}_{{ \mathfrak{l}}}(f(\bigsqcap{}^{{}^{{ \mathfrak{S}}_A}} \!\! S)) &=& \epsilon^{{ \mathfrak{S}}_A}_{f^{\ast}({ \mathfrak{l}})}(\bigsqcap{}^{{}^{{ \mathfrak{S}}_A}} \!\! S)\nonumber\\
&=& \bigwedge{}_{\sigma\in S} \; \epsilon^{{ \mathfrak{S}}_A}_{f^{\ast}({ \mathfrak{l}})}(\sigma)\nonumber\\
&=&  \bigwedge{}_{\sigma\in S} \; \epsilon^{{ \mathfrak{S}}_B}_{{ \mathfrak{l}}}(f(\sigma))\nonumber\\
&=&  \epsilon^{{ \mathfrak{S}}_B}_{{ \mathfrak{l}}}(\bigsqcap{}^{{}^{{ \mathfrak{S}}_B}}_{\sigma\in S} f(\sigma))
\end{eqnarray}
We now use the property (\ref{Chuseparated}) to deduce (\ref{f12cap}).
\end{proof}

\begin{theorem}
The composition of a morphism $(f ,f^{\ast})$ from $({ \mathfrak{S}}_A,{ \mathfrak{E}}_{{ \mathfrak{S}}_A},\epsilon^{{ \mathfrak{S}}_A})$ to $({ \mathfrak{S}}_B,{ \mathfrak{E}}_{{ \mathfrak{S}}_B},\epsilon^{{ \mathfrak{S}}_B})$ by another morphism $(g ,g^{\ast})$ defined from $({ \mathfrak{S}}_B,{ \mathfrak{E}}_{{ \mathfrak{S}}_B},\epsilon^{{ \mathfrak{S}}_B})$ to $({ \mathfrak{S}}_C,{ \mathfrak{E}}_{{ \mathfrak{S}}_C},\epsilon^{{ \mathfrak{S}}_C})$ is given by 
$( g\circ f ,f^{\ast} \circ g^{\ast})$ defining a valid morphism from $({ \mathfrak{S}}_A,{ \mathfrak{E}}_{{ \mathfrak{S}}_A},\epsilon^{{ \mathfrak{S}}_A})$ to $({ \mathfrak{S}}_C,{ \mathfrak{E}}_{{ \mathfrak{S}}_C},\epsilon^{{ \mathfrak{S}}_C})$.
\end{theorem}
\begin{proof}Using two times the duality property, we obtain
\begin{eqnarray} 
\epsilon^{{ \mathfrak{S}}_C}_{{ \mathfrak{l}}_C}(g\circ f (\sigma_A))&=&
\epsilon^{{ \mathfrak{S}}_B}_{g^{\ast}({ \mathfrak{l}}_C)}(f (\sigma_A))=
\epsilon^{{ \mathfrak{S}}_A}_{f^{\ast} \circ g^{\ast}({ \mathfrak{l}}_C)}(\sigma_A).
\end{eqnarray}
\end{proof}

\begin{lemma}  \label{aepsilonsigma}
Let us consider a states/effects Chu space $({ \mathfrak{S}},{ \mathfrak{E}}_{{ \mathfrak{S}}},\epsilon^{{ \mathfrak{S}}})$. 
Let us consider a map $(A : { \mathfrak{S}}\longrightarrow { \mathfrak{B}},  \sigma \mapsto { \mathfrak{a}}_{\sigma})$ satisfying
\begin{eqnarray}
\forall \{\sigma_i\;\vert\; i\in I\}\subseteq { \mathfrak{S}},&&  { \mathfrak{a}}_{{\bigsqcap{}^{{}^{ \mathfrak{S}}}_{{}_{i\in i}}\sigma_i}} = \bigwedge{}_{{i\in I}} \;{ \mathfrak{a}}_{{\sigma_i}},\label{theorema2}
\end{eqnarray}
Then, we have
\begin{eqnarray}
\existunique \;{ \mathfrak{l}}\in { \mathfrak{E}}_{ \mathfrak{S}} & \vert & \forall \sigma\in { \mathfrak{S}}, \; { \epsilon}^{ \mathfrak{S}}_{{ \mathfrak{l}}}(\sigma)={ \mathfrak{a}}_{\sigma}.
\end{eqnarray}
\end{lemma}
\begin{proof}
Straightforward. If $\{\, \sigma\;\vert\; { \mathfrak{a}}_{\sigma}=\textit{\bf Y}\,\}$ and $\{\, \sigma\;\vert\; { \mathfrak{a}}_{\sigma}=\textit{\bf N}\,\}$ are not empty, it suffices to define $\Sigma_A:=\bigsqcap^{{}^{{ \mathfrak{S}}}}\{\, \sigma\;\vert\; { \mathfrak{a}}_{\sigma}=\textit{\bf Y}\,\}$,  $\Sigma_A':=\bigsqcap^{{}^{{ \mathfrak{S}}}}\{\, \sigma\;\vert\; { \mathfrak{a}}_{\sigma}=\textit{\bf N}\,\}$ and ${ \mathfrak{l}}:={ \mathfrak{l}}_{(\Sigma_A,\Sigma_A')} $ (the case where some or all of these subsets are empty is treated immediately).\\
The uniqueness of ${ \mathfrak{l}}$ is guarantied by the property (\ref{Chuextensional}).
\end{proof}

\begin{theorem}
For any map $f:{ \mathfrak{S}}_A \longrightarrow { \mathfrak{S}}_B$ satisfying 
\begin{eqnarray}
\forall \{\sigma_i\;\vert\; i\in I\}\subseteq { \mathfrak{S}}_A, \; f({\bigsqcap{}^{{}^{ \mathfrak{S}_A}}_{{}_{i\in i}}\sigma_i}) = \bigsqcap{}^{{}^{{ \mathfrak{S}}_B}}_{{i\in I}} \;f(\sigma_i), 
\end{eqnarray}
there exists a unique map $f^\ast: { \mathfrak{E}}_{ \mathfrak{S}_B} \longrightarrow { \mathfrak{E}}_{ \mathfrak{S}_A}$ such that $(f,f^{\ast})$ is a morphism from the States/Effects Chu space $({ \mathfrak{S}}_{A},{ \mathfrak{E}}_{{ \mathfrak{S}}_A},\epsilon^{{ \mathfrak{S}}_A})$ to another States/Effects Chu space $({ \mathfrak{S}}_{B},{ \mathfrak{E}}_{{ \mathfrak{S}}_B},\epsilon^{{ \mathfrak{S}}_B})$.
\end{theorem}
\begin{proof}
Direct consequence of Lemma \ref{aepsilonsigma}. It suffices to define $f^{\ast}$ as the map which associates to any ${ \mathfrak{l}}\in { \mathfrak{E}}_{ \mathfrak{S}_B}$ the unique element ${ \mathfrak{l}}'\in { \mathfrak{E}}_{ \mathfrak{S}_A}$ such that $\forall \sigma\in { \mathfrak{S}}, \; { \epsilon}^{ \mathfrak{S}_A}_{{ \mathfrak{l}'}}(\sigma)={ \epsilon}^{ \mathfrak{S}_B}_{{ \mathfrak{l}}}(f(\sigma))$.
\end{proof}

\begin{theorem}  \label{blepsilonsigma}
Let us consider a states/effects Chu space $({ \mathfrak{S}},{ \mathfrak{E}}_{{ \mathfrak{S}}},\epsilon^{{ \mathfrak{S}}})$. 
Let us consider a map $(B : { \mathfrak{E}}_{{ \mathfrak{S}}}\longrightarrow { \mathfrak{B}},{ \mathfrak{l}}\mapsto { \mathfrak{b}}_{{ \mathfrak{l}}})$ satisfying
\begin{eqnarray}
\forall \{{ \mathfrak{l}}_i\;\vert\; i\in I\}\subseteq { \mathfrak{E}}_{{ \mathfrak{S}}},&&  { \mathfrak{b}}_{{\bigsqcap{}^{{}^{{ \mathfrak{E}}_{{ \mathfrak{S}}}}}_{{}_{i\in i}}{ \mathfrak{l}}_i}} = \bigwedge{}_{{i\in I}} \;{ \mathfrak{b}}_{{{ \mathfrak{l}}_i}},\label{theorembl2}\\
\forall { \mathfrak{l}}\in { \mathfrak{E}}_{{ \mathfrak{S}}},&& { \mathfrak{b}}_{\overline{\, \mathfrak{l}\,}}=\overline{{ \mathfrak{b}}_{{ \mathfrak{l}}}},\label{theorembl3}\\
& & { \mathfrak{b}}_{{ \mathfrak{Y}}_{{ \mathfrak{E}}_{{ \mathfrak{S}}}}}=\textit{\bf Y}.\label{theorembl4}
\end{eqnarray}
Then, we have
\begin{eqnarray}
\existunique \; \sigma\in { \mathfrak{S}} & \vert & \forall { \mathfrak{l}}\in { \mathfrak{E}}_{{ \mathfrak{S}}}, \; { \epsilon}^{ \mathfrak{S}}_{{ \mathfrak{l}}}(\sigma)={ \mathfrak{b}}_{{ \mathfrak{l}}}.
\end{eqnarray}
\end{theorem}
\begin{proof}
Let us consider 
\begin{eqnarray}
{\mathfrak{l}}_B &:=&\bigsqcap{}^{{}^{{{ \mathfrak{E}}_{{ \mathfrak{S}}}}}}\{\, { \mathfrak{l}}\in { \mathfrak{E}}_{{ \mathfrak{S}}} \;\vert\; { \mathfrak{b}}_{{ \mathfrak{l}}}=\textit{\bf Y}\,\}
\end{eqnarray}  
Note that ${ \mathfrak{l}}_B$ exists in ${ \mathfrak{E}}_{{ \mathfrak{S}}}$ because ${ \mathfrak{E}}_{{ \mathfrak{S}}}$ is a down-complete Inf semi-lattice and the subset $\{\, { \mathfrak{l}}\in { \mathfrak{E}}_{{ \mathfrak{S}}} \;\vert\; { \mathfrak{b}}_{{ \mathfrak{l}}}=\textit{\bf Y}\,\}$ contains at least the element ${ \mathfrak{Y}}_{{ \mathfrak{E}}_{{ \mathfrak{S}}}}$.  Moreover, ${ \mathfrak{b}}_{{ \mathfrak{l}}_B}=\textit{\bf Y}$ because of the relation (\ref{theorembl2}). Note also that for any ${ \mathfrak{l}}$ in ${ \mathfrak{E}}_{{ \mathfrak{S}}}$, we have ${ \mathfrak{l}} \sqsupseteq_{{}_{{{ \mathfrak{E}}_{{ \mathfrak{S}}}}}} { \mathfrak{l}}_B$ implies ${ \mathfrak{b}}_{{ \mathfrak{l}}}=\textit{\bf Y}$ because the property (\ref{theorembl2}) implies the relation 
\begin{eqnarray}
\forall { \mathfrak{l}}, { \mathfrak{l}}'\in { \mathfrak{E}}_{{ \mathfrak{S}}},&& ({ \mathfrak{l}} \sqsubseteq_{{}_{{ \mathfrak{E}}_{{ \mathfrak{S}}}}} { \mathfrak{l}}')\Rightarrow ({ \mathfrak{b}}_{{ \mathfrak{l}}} \leq { \mathfrak{b}}_{{ \mathfrak{l}}'}).\label{theorembl1}
\end{eqnarray}
Conversely,  for any ${ \mathfrak{l}}$ in ${ \mathfrak{E}}_{{ \mathfrak{S}}}$,  the property ${ \mathfrak{b}}_{{ \mathfrak{l}}}=\textit{\bf Y}$ implies ${ \mathfrak{l}} \sqsupseteq_{{}_{{{ \mathfrak{E}}_{ \mathfrak{S}}}}} { \mathfrak{l}}_B$ due to the definition of ${ \mathfrak{l}}_B$. \\
Let us now introduce 
\begin{eqnarray}
\Sigma_{{}_{{ \mathfrak{l}}_B}} &:=& \bigsqcap{}^{{}^{ \mathfrak{S}}}({\epsilon}^{ \mathfrak{S}}_{{ \mathfrak{l}}_B})^{\;-1}(\textit{\bf Y}). 
\end{eqnarray}
We could suppose that ${ \mathfrak{l}}_B={ \mathfrak{l}}_{(\Sigma_{{ \mathfrak{l}}_B},\Sigma'_{{ \mathfrak{l}}_B})}$ for a certain $\Sigma'_{{ \mathfrak{l}}_B}\in {\mathfrak{S}}$. However, we note that, because of (\ref{theorembl2}) and (\ref{theorembl4}), we have ${ \mathfrak{l}}_{(\Sigma_{{ \mathfrak{l}}_B},\centerdot)}\sqsubset_{{}_{{{ \mathfrak{E}}_{ \mathfrak{S}}}}} { \mathfrak{l}}_{(\Sigma_{{ \mathfrak{l}}_B},\Sigma'_{{ \mathfrak{l}}_B})}$ and ${ \mathfrak{b}}_{{ \mathfrak{l}}_{(\Sigma_{{ \mathfrak{l}}_B},\centerdot)}}={ \mathfrak{b}}_{{ \mathfrak{l}}_{(\Sigma_{{ \mathfrak{l}}_B},\Sigma'_{{ \mathfrak{l}}_B})}\sqcap_{{}_{{{ \mathfrak{E}}_{{ \mathfrak{S}}}}}} { \mathfrak{Y}}_{{ \mathfrak{E}}_{{ \mathfrak{S}}}} }= { \mathfrak{b}}_{{ \mathfrak{l}}_{(\Sigma_{{ \mathfrak{l}}_B},\Sigma'_{{ \mathfrak{l}}_B})}}\wedge { \mathfrak{b}}_{{ \mathfrak{Y}}_{{ \mathfrak{E}}_{{ \mathfrak{S}}}} }=\textit{\bf Y}$ which would contradict the definition of ${ \mathfrak{l}}_B$.  Hence, we have to accept that ${ \mathfrak{l}}_B={ \mathfrak{l}}_{(\Sigma_{{ \mathfrak{l}}_B},\centerdot)}$. \\
Due to this peculiar form of ${ \mathfrak{l}}_B$, we deduce that, for any ${ \mathfrak{l}}_{(\Sigma,\Sigma')}$ in ${ \mathfrak{E}}_{{ \mathfrak{S}}}$, the property ${ \mathfrak{l}}_{(\Sigma,\Sigma')} \not\sqsupseteq_{{}_{{{ \mathfrak{E}}_{ \mathfrak{S}}}}} { \mathfrak{l}}_B$ is then equivalent to the property $\Sigma\not\sqsupseteq_{{}_{{ \mathfrak{S}}}} \Sigma_{{ \mathfrak{l}}_B}$. Then, if ${ \mathfrak{l}}_{(\Sigma,\Sigma')} \not\sqsupseteq_{{}_{{{ \mathfrak{E}}_{{ \mathfrak{S}}}}}} { \mathfrak{l}}_B$ we cannot have $\epsilon^{ \mathfrak{S}}_{{ \mathfrak{l}}_{(\Sigma,\Sigma')}}(\Sigma_{{}_{{ \mathfrak{l}}_B}})=\textit{\bf Y}$ (by the very definition of $\epsilon^{ \mathfrak{S}}$). \\
On another part, for any ${ \mathfrak{l}}$ in ${ \mathfrak{E}}_{{ \mathfrak{S}}}$ such that ${ \mathfrak{l}} \sqsupseteq_{{}_{{{ \mathfrak{E}}_{{ \mathfrak{S}}}}}} { \mathfrak{l}}_B$, we have $\epsilon^{ \mathfrak{S}}_{{ \mathfrak{l}}}(\Sigma_{{}_{{ \mathfrak{l}}_B}})\geq \epsilon^{ \mathfrak{S}}_{{ \mathfrak{l}}_B}(\Sigma_{{}_{{ \mathfrak{l}}_B}}) = \textit{\bf Y}$, i.e. $\epsilon^{ \mathfrak{S}}_{{ \mathfrak{l}}}(\Sigma_{{}_{{ \mathfrak{l}}_B}}) = \textit{\bf Y}$. \\ 
We then conclude that the property $\epsilon^{ \mathfrak{S}}_{{ \mathfrak{l}}}(\Sigma_{{}_{{ \mathfrak{l}}_B}}) = \textit{\bf Y}$ is equivalent to the property ${ \mathfrak{l}} \sqsupseteq_{{}_{{{ \mathfrak{E}}_{{ \mathfrak{S}}}}}} { \mathfrak{l}}_B$, or in other words 
\begin{eqnarray}
\forall { \mathfrak{l}}\in { \mathfrak{E}}_{{ \mathfrak{S}}},\;\;\;\; \epsilon^{ \mathfrak{S}}_{{ \mathfrak{l}}}(\Sigma_{{}_{{ \mathfrak{l}}_B}}) = \textit{\bf Y} &\Leftrightarrow & { \mathfrak{b}}_{{ \mathfrak{l}}}=\textit{\bf Y}. 
 \end{eqnarray}
  Using (\ref{theorembl3}) and (\ref{etbar}),  we then deduce that $(\epsilon^{ \mathfrak{S}}_{{ \mathfrak{l}}}(\Sigma_{{}_{{ \mathfrak{l}}_B}}) = \textit{\bf N})\Leftrightarrow (\epsilon^{ \mathfrak{S}}_{\overline{\, \mathfrak{l}\,}}(\Sigma_{{}_{{ \mathfrak{l}}_B}}) = \textit{\bf Y})\Leftrightarrow ({ \mathfrak{b}}_{\overline{\, \mathfrak{l}\,}}=\textit{\bf Y}) \Leftrightarrow ({ \mathfrak{b}}_{{ \mathfrak{l}}}=\textit{\bf N})$. \\
  As a final conclusion, we have for any ${ \mathfrak{l}}\in { \mathfrak{E}}_{{ \mathfrak{S}}}$ the equality $\epsilon^{ \mathfrak{S}}_{{ \mathfrak{l}}}(\Sigma_{{}_{{ \mathfrak{l}}_B}})={ \mathfrak{b}}_{{ \mathfrak{l}}}$. \\
Let us finally suppose that there exists $\sigma_1,\sigma_2\in { \mathfrak{S}}$ such that $\forall { \mathfrak{l}}\in { \mathfrak{E}}_{{ \mathfrak{S}}}, \; \epsilon_{{ \mathfrak{l}}}(\sigma_1)={ \mathfrak{b}}_{{ \mathfrak{l}}}=\epsilon_{{ \mathfrak{l}}}(\sigma_2)$. We then obtain $\sigma_1=\sigma_2$. The state $\sigma$ is then uniquely fixed to be $\Sigma_{{}_{{ \mathfrak{l}}_B}}$.
\end{proof}

\begin{theorem}
For any map $f^\ast:{ \mathfrak{E}}_{ \mathfrak{S}_B} \longrightarrow { \mathfrak{E}}_{{ \mathfrak{S}}_A}$ satisfying
\begin{eqnarray}
\forall E\subseteq { \mathfrak{E}}_{{ \mathfrak{S}}_B},&& f^{\ast}(\bigsqcap{}^{{}^{{ \mathfrak{E}}_{{ \mathfrak{S}}_B}}} \! E) = \bigsqcap{}^{{}^{{ \mathfrak{E}}_{{ \mathfrak{S}}_A}}}_{{}_{{ \mathfrak{l}}\in E}} \; f^{\ast}({ \mathfrak{l}})\\
\forall { \mathfrak{l}} \in { \mathfrak{E}}_{{ \mathfrak{S}}_B},&& f^{\ast}(\; \overline{\,{ \mathfrak{l}}\,} \;) = \overline{f^{\ast}({ \mathfrak{l}})}\\
&& f^{\ast}(\;{ \mathfrak{Y}}_{{ \mathfrak{E}}_{{ \mathfrak{S}}_B}}\;) = { \mathfrak{Y}}_{{ \mathfrak{E}}_{{ \mathfrak{S}}_A}},
\end{eqnarray}
there exists a unique map $f:{ \mathfrak{S}}_A \longrightarrow { \mathfrak{S}}_B$ such that $(f,f^{\ast})$ is a morphism from the States/Effects Chu space $({ \mathfrak{S}}_{A},{ \mathfrak{E}}_{{ \mathfrak{S}}_A},\epsilon^{{ \mathfrak{S}}_A})$ to the States/Effects Chu space $({ \mathfrak{S}}_{B},{ \mathfrak{E}}_{{ \mathfrak{S}}_B},\epsilon^{{ \mathfrak{S}}_B})$.
\end{theorem}
\begin{proof}
Direct consequence of Lemma \ref{blepsilonsigma}.  It suffices to define $f$ as the map which associates to any $\sigma \in { \mathfrak{S}_A}$ the unique element $\sigma'\in { \mathfrak{S}_B}$ such that $\forall { \mathfrak{l}}\in { \mathfrak{E}}_{ \mathfrak{S}_B}, \; { \epsilon}^{ \mathfrak{S}_B}_{{ \mathfrak{l}}}(\sigma')={ \epsilon}^{ \mathfrak{S}_A}_{f^{\ast}({ \mathfrak{l}})}(\sigma)$ (in the application of the Lemma \ref{blepsilonsigma}, we simply choose ${ \mathfrak{E}}:={ \mathfrak{E}}_{{ \mathfrak{S}}_B}$). 
\end{proof}

As a consequence of the two last theorems, the couple of maps $(f,f^{\ast})$ defining a morphism from $({ \mathfrak{S}}_A,{ \mathfrak{E}}_{{ \mathfrak{S}}_A},\epsilon^{{ \mathfrak{S}}_A})$ to $({ \mathfrak{S}}_B,{ \mathfrak{E}}_{{ \mathfrak{S}}_B},\epsilon^{{ \mathfrak{S}}_B})$ can be reduced to the single data $f$ (or to the single data $f^\ast$ as well). We will then speak shortly of "the morphism $f$ from the space of states ${ \mathfrak{S}}_{A}$ to the space of states ${ \mathfrak{S}}_{B}$" rather than "the morphism from the states/effects Chu space $({ \mathfrak{S}}_A,{ \mathfrak{E}}_{{ \mathfrak{S}}_A},\epsilon^{{ \mathfrak{S}}_A})$ to the states/effects Chu space $({ \mathfrak{S}}_B,{ \mathfrak{E}}_{{ \mathfrak{S}}_B},\epsilon^{{ \mathfrak{S}}_B})$".

\begin{definition} The set of morphisms from ${ \mathfrak{S}}_A$ to ${ \mathfrak{S}}_B$ will be denoted ${\mathfrak{C}}({ \mathfrak{S}}_A,{ \mathfrak{S}}_B)$.
\end{definition}

\begin{definition}
We define the infimum of two maps $f$ and $g$ satisfying (\ref{f12cap}) (resp. two maps $f^\ast$ and $g^\ast$ satisfying (\ref{f21cap})) by $\forall \sigma \in { \mathfrak{S}}_A, (f\sqcap_{{}_{{\mathfrak{C}}({ \mathfrak{S}}_A,{ \mathfrak{S}}_B)}} g)(\sigma):= f(\sigma) \sqcap_{{}_{{ \mathfrak{S}}_B}} g(\sigma)$.
\end{definition}

\begin{theorem}
The infimum of a morphism $(f ,f^{\ast})$ from $({ \mathfrak{S}}_A,{ \mathfrak{E}}_{{ \mathfrak{S}}_A},\epsilon^{{ \mathfrak{S}}_A})$ to $({ \mathfrak{S}}_B,{ \mathfrak{E}}_{{ \mathfrak{S}}_B},\epsilon^{{ \mathfrak{S}}_B})$ with another morphism $(g ,g^{\ast})$ defined from $({ \mathfrak{S}}_A,{ \mathfrak{E}}_{{ \mathfrak{S}}_A},\epsilon^{{ \mathfrak{S}}_A})$ to $({ \mathfrak{S}}_B,{ \mathfrak{E}}_{{ \mathfrak{S}}_B},\epsilon^{{ \mathfrak{S}}_B})$ defines a valid morphism from $({ \mathfrak{S}}_A,{ \mathfrak{E}}_{{ \mathfrak{S}}_A},\epsilon^{{ \mathfrak{S}}_A})$ to $({ \mathfrak{S}}_B,{ \mathfrak{E}}_{{ \mathfrak{S}}_B},\epsilon^{{ \mathfrak{S}}_B})$.
\end{theorem}
\begin{proof}
Direct consequence of the duality property.
\end{proof}

In the following, we will be particularly interested to a certain class of morphisms called "measurements".

\begin{definition} A {\em measurement} is a morphism ${ \mathfrak{m}}$ from the state space ${ \mathfrak{S}}$ to the space of evaluations ${ \mathfrak{B}}$.  
The space of measurements associated to ${ \mathfrak{S}}$ will be denoted ${\mathfrak{M}}_{{}_{{ \mathfrak{S}}}}$.
\end{definition}

\begin{theorem}
The space of measurements ${\mathfrak{M}}_{{}_{{ \mathfrak{S}}}}$ is isomorphic to the space of effects ${ \mathfrak{E}}_{ \mathfrak{S}}$. 
\end{theorem}
\begin{proof}
An effect ${ \mathfrak{l}}\in { \mathfrak{E}}_{ \mathfrak{S}}$ being given, the map ${ \mathfrak{m}}_{{ \mathfrak{l}}}$ defined by
\begin{eqnarray}
&& { \mathfrak{m}}_{{ \mathfrak{l}}}(\sigma) := \epsilon^{ \mathfrak{S}}_{ \mathfrak{l}}(\sigma),\;\;\;\forall \sigma\in { \mathfrak{S}}.\label{measurementeffect}
\end{eqnarray}
 is a well defined homomorphism from ${ \mathfrak{S}}$ to ${ \mathfrak{B}}$.  Hence, the map
\begin{eqnarray}
&&\begin{array}{rcrcl}
\varphi & : & { \mathfrak{E}}_{ \mathfrak{S}} & \longrightarrow & {\mathfrak{M}}_{{}_{{ \mathfrak{S}}}}\\
& & { \mathfrak{l}} & \mapsto & { \mathfrak{m}}_{{ \mathfrak{l}}}\label{defvarphi}
\end{array}
\end{eqnarray}
is an homomorphism.  It is moreover bijective. Indeed, let us consider any ${ \mathfrak{m}}$ in ${\mathfrak{M}}_{{}_{{ \mathfrak{S}}}}$.  If $\{\, \sigma\;\vert\; { \mathfrak{m}}({\sigma})=\textit{\bf Y}\,\}$ and $\{\, \sigma\;\vert\; { \mathfrak{m}}({\sigma})=\textit{\bf N}\,\}$ are not empty, they are principal filters, and it suffices to associate to ${ \mathfrak{m}}$ the states $\sigma_A:=\bigsqcap^{{}^{{ \mathfrak{S}}}}\{\, \sigma\;\vert\; { \mathfrak{m}}({\sigma})=\textit{\bf Y}\,\}$,  $\sigma_A':=\bigsqcap^{{}^{{ \mathfrak{S}}}}\{\, \sigma\;\vert\; { \mathfrak{m}}({\sigma})=\textit{\bf N}\,\}$ and the effect ${ \mathfrak{l}}:={ \mathfrak{l}}_{(\sigma_A,\sigma_A')} $ (the case where some or all of these subsets are empty is treated immediately).\\

As usual,  the data ${ \mathfrak{m}}_{{ \mathfrak{l}}}$ suffices to define a Chu morphism $({ \mathfrak{m}}_{{ \mathfrak{l}}},{ \mathfrak{m}}^\ast_{{ \mathfrak{l}}})$. Indeed, the map ${ \mathfrak{m}}^\ast_{{ \mathfrak{l}}} : { \mathfrak{E}}_{ \mathfrak{B}} \longrightarrow { \mathfrak{E}}_{ \mathfrak{S}}$ is defined by the duality relation (\ref{defchumorphism}), i.e.
 \begin{eqnarray}
 \forall \sigma\in { \mathfrak{S}}, \forall { \mathfrak{u}}\in { \mathfrak{E}}_{ \mathfrak{B}},&& \epsilon^{ \mathfrak{S}}_{{ \mathfrak{m}}^\ast_{{ \mathfrak{l}}}({ \mathfrak{u}})}(\sigma)=\epsilon^{ \mathfrak{B}}_{{ \mathfrak{u}}}({ \mathfrak{m}}_{{ \mathfrak{l}}}(\sigma))=\epsilon^{ \mathfrak{B}}_{{ \mathfrak{u}}}( \epsilon^{ \mathfrak{S}}_{ \mathfrak{l}}(\sigma))\label{appliedduality}
\end{eqnarray}
The expression of ${ \mathfrak{m}}^\ast_{{ \mathfrak{l}}}({ \mathfrak{u}})$ for any ${ \mathfrak{u}}$ in ${ \mathfrak{E}}_{ \mathfrak{B}}$ is reduced by the application of the three relations (\ref{f21cap})(\ref{f21bar})(\ref{f21Y}) to the property 
\begin{eqnarray}
{ \mathfrak{m}}^\ast_{{ \mathfrak{l}}}({ \mathfrak{l}}_{(\textit{\bf Y},\textit{\bf N})}):={ \mathfrak{l}}\label{mastu}
\end{eqnarray}
 which validity can be checked directly on (\ref{appliedduality}).  
\end{proof}

\subsection{Real structures}\label{subsectionreal}

During this subsection, we will consider the states/effects Chu Space $({ \mathfrak{S}},{ \mathfrak{E}}_{{ \mathfrak{S}}},\epsilon^{{ \mathfrak{S}}})$.

\begin{definition}\label{definitionstarstructure}
A {\em real structure} on a space of states ${ \mathfrak{S}}$ is a pair $(\overline{ \mathfrak{S}},\star)$ where
\begin{itemize}
\item $\overline{ \mathfrak{S}}$ is a subset of ${ \mathfrak{S}}$ satisfying
\begin{eqnarray}
\bot_{{}_{ \mathfrak{S}}} & \in & \overline{ \mathfrak{S}}\\
\forall S\subseteq \overline{ \mathfrak{S}},\;\;\;\; \bigsqcap{}^{{}^{{ \mathfrak{S}}}}S & \in & \overline{ \mathfrak{S}},
\end{eqnarray}
\item $\overline{ \mathfrak{S}}$ is {\em generated by its maximal elements}, i.e. 
 \begin{eqnarray}
\hspace{-1cm}&&\forall \sigma \in \overline{ \mathfrak{S}}, \;\; \sigma= \bigsqcap{}^{{}^{ \overline{ \mathfrak{S}}}}  \underline{\sigma}_{{}_{\overline{ \mathfrak{S}}}} \;\;\textit{\rm where}\;\; \underline{\sigma}_{{}_{ \overline{ \mathfrak{S}}}}:=\{\, \sigma'\in \overline{ \mathfrak{S}}{}^{{}^{pure}}\;\vert\; \sigma'\sqsupseteq_{{}_{\overline{ \mathfrak{S}}}} \sigma\;\} \;\;\textit{\rm and}\;\; \overline{ \mathfrak{S}}{}^{{}^{pure}}:=Max(\overline{ \mathfrak{S}})\;\;\;\;\;\;\;\;\;\;\;\;
\label{completemeetirreducibleordergeneratingS}
\end{eqnarray}
\item $\star$ is a map from $\overline{ \mathfrak{S}}\smallsetminus \{\bot_{{}_{ \mathfrak{S}}}\}$ to $\overline{ \mathfrak{S}}\smallsetminus \{\bot_{{}_{ \mathfrak{S}}}\}$ satisfying
\begin{eqnarray}
\forall \sigma\in \overline{ \mathfrak{S}}\smallsetminus \{\bot_{{}_{ \mathfrak{S}}}\},&& (\sigma^{\star})^{\star}=\sigma \label{involutive}\\
\forall \sigma_1,\sigma_2\in \overline{ \mathfrak{S}}\smallsetminus \{\bot_{{}_{ \mathfrak{S}}}\},&& \sigma_1\sqsubseteq_{{}_{\overline{ \mathfrak{S}}}}\sigma_2\;\;\Rightarrow\;\; \sigma_2^\star \sqsubseteq_{{}_{\overline{ \mathfrak{S}}}} \sigma_1^\star \label{orderreversing}\\
\forall \sigma\in \overline{ \mathfrak{S}}\smallsetminus \{\bot_{{}_{ \mathfrak{S}}}\},&&\neg\; \widehat{\sigma\sigma^\star}{}^{{}^{{ \mathfrak{S}}}},\label{starcomplement}
\end{eqnarray}
\item the sub Inf semi-lattice of ${ \mathfrak{E}}_{ \mathfrak{S}}$ formed by the elements of the following set
\begin{eqnarray}
\hspace{-0.5cm}\{\,{ \mathfrak{l}}_{(\sigma,\sigma')}\;\vert\; \sigma,\sigma'\in \overline{ \mathfrak{S}}\smallsetminus \{\bot_{{}_{{ \mathfrak{S}}}}\}, \sigma'\sqsupseteq_{{}_{\overline{ \mathfrak{S}}}} \sigma^\star\,\} \cup \{ { \mathfrak{l}}_{{}_{(\sigma,\centerdot)}}\;\vert\; \sigma\in \overline{ \mathfrak{S}}\;\}\cup \{ { \mathfrak{l}}_{{}_{(\centerdot,\sigma)}}\;\vert\; \sigma\in \overline{ \mathfrak{S}}\;\}\cup \{\, { \mathfrak{l}}_{{}_{(\centerdot,\centerdot)}}\,\}\;\;\;\;\;\;\label{reducedspaceofeffects}
\end{eqnarray}
will be denoted $\overline{ \mathfrak{E}}_{ \mathfrak{S}}$ and will be assumed to satisfy
\begin{eqnarray}
&& \forall \sigma,\sigma'\in { \mathfrak{S}},\;\;\;\;\;\;\;\;\;\;\;\;(\, \forall { \mathfrak{l}}\in \overline{ \mathfrak{E}}_{ \mathfrak{S}},\; { \epsilon}^{ \mathfrak{S}}_{ { \mathfrak{l}}}(\sigma)={ \epsilon}^{ \mathfrak{S}}_{ { \mathfrak{l}}}(\sigma') \,) \Leftrightarrow  (\, \sigma= \sigma' \,).\label{realChuseparated}
\end{eqnarray}
\end{itemize}

\vspace{0.3cm}
The elements of $\overline{ \mathfrak{S}}$ will be called {\em real states}, whereas the elements of ${ \mathfrak{S}}\smallsetminus \overline{ \mathfrak{S}}$ will be called {\em hidden states}. \\ 
The elements of $\overline{ \mathfrak{S}}^{{}^{pure}}$ will be called {\em pure states}.\\
The elements of $\overline{ \mathfrak{E}}_{ \mathfrak{S}}$ will be called {\em real effects}.
\end{definition}

\begin{remark}
We note without proof the following basic results
\begin{eqnarray}
&& \forall { \mathfrak{l}},{ \mathfrak{l}}'\in \overline{ \mathfrak{E}}_{ \mathfrak{S}},\;\;\;\;\;\;\;\;\;\;\;\;(\, \forall \sigma\in { \mathfrak{S}},\; { \epsilon}^{ \mathfrak{S}}_{ { \mathfrak{l}}}(\sigma)={ \epsilon}^{ \mathfrak{S}}_{ { \mathfrak{l}}'}(\sigma) \,) \Leftrightarrow  (\, { \mathfrak{l}}= { \mathfrak{l}}' \,),\label{realChuextensional}\\
&&\forall L\subseteq \overline{ \mathfrak{E}}_{ \mathfrak{S}},\;\;\;\;( \bigsqcap{}^{{}^{\overline{ \mathfrak{E}}_{ \mathfrak{S}}}}L) \in \overline{ \mathfrak{E}}_{ \mathfrak{S}},\label{axiomreduc2}\\
&& \forall { \mathfrak{l}}\in \overline{ \mathfrak{E}}_{ \mathfrak{S}},\;\;\;\;\overline{\,{ \mathfrak{l}}\,} \in \overline{ \mathfrak{E}}_{ \mathfrak{S}},\label{axiomreduc3}\\
&&{ \mathfrak{Y}}_{{}_{{{ \mathfrak{E}}}_{ \mathfrak{S}}}} \in  \overline{ \mathfrak{E}}_{ \mathfrak{S}} \;\;\;\;\;\;\textit{\rm and} \;\;\;\;\;\; \bot_{{}_{{{ \mathfrak{E}}}_{ \mathfrak{S}}}} \in  \overline{ \mathfrak{E}}_{ \mathfrak{S}}.\label{axiomreduc4}
\end{eqnarray}
\end{remark}

\subsection{First results on real structures}\label{firstresultsrealstructures}

\begin{definition}
We will define a binary relation on ${}{\mathfrak{S}}\smallsetminus \{\bot_{{}_{{ \mathfrak{S}}}}\}$, denoted by $\underline{\perp}$ and called {\em orthogonality relation}, as follows :
\begin{eqnarray}
\hspace{-0.5cm}\forall \sigma,\sigma'\in {}{\mathfrak{S}}\smallsetminus \{\bot_{{}_{{ \mathfrak{S}}}}\},&& \sigma \underline{\perp} \sigma'\;\;:\Leftrightarrow\;\; \left( \exists \omega\in \overline{ \mathfrak{S}}\smallsetminus \{\bot_{{}_{\overline{ \mathfrak{S}}}}\},\; \sigma\sqsupseteq_{{}_{{ \mathfrak{S}}}}\omega \;\;\textit{\rm and}\;\; \sigma'\sqsupseteq_{{}_{{ \mathfrak{S}}}}\omega^\star\right).\label{definorthoS}\;\;\;\;\;\;\;\;\;\;\;\;\;\;
\end{eqnarray}
\end{definition}


\begin{definition}
We will adopt the following standard notation for the orthogonal of a subset of ${}{\mathfrak{S}}$
\begin{eqnarray}
\forall S\subseteq {}{\mathfrak{S}} & & S^{\underline{\perp}}:=\{\,\sigma\in {}{\mathfrak{S}} \;\vert\; \forall \sigma'\in S, \; \sigma \underline{\perp} \sigma'\,\}.\label{deforthosubset}
\end{eqnarray}
\end{definition}

\begin{lemma}
The operator defined on ${\mathcal{P}}({}{\mathfrak{S}})$ which maps $S$ to $(S^{\underline{\perp}})^{\underline{\perp}}$ (which will eventually be denoted by $S^{{\underline{\perp}}{\underline{\perp}}}$) owns the following properties
\begin{eqnarray}
\forall S\subseteq {}{\mathfrak{S}},&& S\subseteq S^{{\underline{\perp}}{\underline{\perp}}},\label{orthoorthoexpansive}\\
\forall S_1,S_2\subseteq {\mathfrak{S}},&& S_1\subseteq S_2\;\Rightarrow\; S_1^{{\underline{\perp}}{\underline{\perp}}}\subseteq S_2^{{\underline{\perp}}{\underline{\perp}}}.\;\;\;\;\;\;\;\label{orthoorthomonotone}
\end{eqnarray}
We have moreover
\begin{eqnarray}
\forall S\subseteq {}{\mathfrak{S}},&& ((S^{\underline{\perp}})^{\underline{\perp}})^{\underline{\perp}}=(S)^{\underline{\perp}}.\label{ortho3=ortho}
\end{eqnarray}
and then
\begin{eqnarray}
\forall S\subseteq {}{\mathfrak{S}},&& (S^{{\underline{\perp}}{\underline{\perp}}})^{{\underline{\perp}}{\underline{\perp}}}=S^{{\underline{\perp}}{\underline{\perp}}}.\label{ortho4=ortho2}
\end{eqnarray}
\end{lemma}
\begin{proof}
The properties (\ref{orthoorthoexpansive}) and (\ref{orthoorthomonotone}) are obtained trivially. \\
Concerning property (\ref{ortho3=ortho}), we have the following analysis. $((B^{\underline{\perp}})^{\underline{\perp}})^{\underline{\perp}}\supseteq B^{\underline{\perp}}$ is a consequence of the fact that the property (\ref{orthoorthoexpansive}) for $B^{\underline{\perp}}$, and $((B^{\underline{\perp}})^{\underline{\perp}})^{\underline{\perp}}\subseteq B^{\underline{\perp}}$ is a consequence of the action of the order inversing property of $\underline{\perp}$ on the property (\ref{orthoorthoexpansive}) for $B$.\\
The property (\ref{ortho4=ortho2}) is a direct consequence of property (\ref{ortho3=ortho})
\end{proof}

\begin{definition}
A subset $S$ of ${}{\mathfrak{S}}$ will be said to be {\em orthoclosed} iff $(S^{\underline{\perp}})^{\underline{\perp}}=S$. The set of orthoclosed subsets of ${}{\mathfrak{S}}$ will be denoted ${ \mathfrak{H}}({}{\mathfrak{S}}, {\underline{\perp}})$.  
\end{definition}

\begin{lemma}\label{lemmavarnothing&E}
$\varnothing$ and $(\uparrow^{{}^{{}{\mathfrak{S}}}}\!\!\!\bot_{{}_{{\mathfrak{S}}}})$ are both elements of ${ \mathfrak{H}}({}{\mathfrak{S}}, {\underline{\perp}})$. 
\end{lemma}

\begin{lemma}\label{lemmaorthocapsemi-lattice}
Equipped with the inclusion $\subseteq$, ${ \mathfrak{H}}({}{\mathfrak{S}}, {\underline{\perp}})$ is a partially ordered set. With the intersection $\cap$, ${ \mathfrak{H}}({}{\mathfrak{S}}, {\underline{\perp}})$ is a down complete Inf semi-lattice.
\end{lemma}
\begin{proof}
The first assertion is trivial. Concerning the second assertion, let us consider two orthoclosed subsets $H,H'\in { \mathfrak{H}}({}{\mathfrak{S}}, {\underline{\perp}})$. The properties $H\cap H'\subseteq H,H'$ imply $(H\cap H')^{\underline{\perp}\underline{\perp}}\subseteq (H)^{\underline{\perp}\underline{\perp}}=H,(H')^{\underline{\perp}\underline{\perp}}=H'$ (due to the monotonicity of the orthogonal closure) and then $(H\cap H')^{\underline{\perp}\underline{\perp}}\subseteq (H \cap H')$. Using the first axiom of closure, we have also $(H\cap H')^{\underline{\perp}\underline{\perp}}\supseteq (H \cap H')$. This concludes the proof.
\end{proof}

\begin{lemma}
As a consequence of property (\ref{ortho3=ortho}), we have that,  for any $S$ orthoclosed subset of ${}{\mathfrak{S}}$, there exists a subset $S'$ of ${}{\mathfrak{S}}$ such that $(S')^{\underline{\perp}}=S$, and conversely, any subset of the form $S^{\underline{\perp}}$ is orthoclosed.\\
As a conclusion, the image of ${ \mathcal{P}}({}{\mathfrak{S}})$ by the orthogonality operator is exactly ${ \mathfrak{H}}({}{\mathfrak{S}}, {\underline{\perp}})$.
\end{lemma}

\begin{theorem}\label{orthocomplem0}
The set of orthoclosed subsets satisfies the following properties :
\begin{eqnarray}
\forall H\in { \mathfrak{H}}({}{\mathfrak{S}}, {\underline{\perp}}),&& (H^{\underline{\perp}})^{\underline{\perp}}=H,\\
\forall H_1,H_2\in { \mathfrak{H}}({}{\mathfrak{S}}, {\underline{\perp}}),&& (H_1\subseteq H_2)\;\Rightarrow \; (H_2^{\underline{\perp}} \subseteq H_1^{\underline{\perp}}).
\end{eqnarray}
\end{theorem}
\begin{proof}
The second relation is trivially deduced from Definition \ref{deforthosubset} of the orthogonal of a subset. The first relation is (\ref{ortho4=ortho2}).
\end{proof}

\begin{definition}
Using the results of the previous lemma, it appears natural to introduce the following orthogonality relation on the set of orthoclosed subsets ${ \mathfrak{H}}({\mathfrak{S}}, {\underline{\perp}})$
\begin{eqnarray}
\forall H_1,H_2\in { \mathfrak{H}}({}{\mathfrak{S}}, {\underline{\perp}}),\;\; H_1\;\underline{\perp}\; H_2 & :\Leftrightarrow & H_1^{\underline{\perp}} \supseteq H_2.\label{deforthoHSorthog}
\end{eqnarray} 
We use the same notation as the orthogonality relation of ${}{\mathfrak{S}}$ because ${}{\mathfrak{S}}$ is naturally injected into ${ \mathfrak{H}}({}{\mathfrak{S}}, {\underline{\perp}})$.\\
This orthogonality relation is obviously symmetric and irreflexive. 
\end{definition}

\begin{definition}
We will adopt the following notation for the suprema :
\begin{eqnarray}
\forall H_1,H_2\in { \mathfrak{H}}({}{\mathfrak{S}}, {\underline{\perp}}),&& H_1 \curlyvee H_2 := (H_1^{\underline{\perp}} \cap H_2^{\underline{\perp}})^{\underline{\perp}}.\label{demorgan}
\end{eqnarray}
\end{definition}

\begin{definition}
During this subsection, the map $\Theta^{\overline{ \mathfrak{S}}}$ from ${ \mathfrak{S}}$ to ${\mathcal{P}}(\overline{ \mathfrak{S}})$ that will be used extensively, is defined by $\Theta^{\overline{ \mathfrak{S}}}(\zeta):=Max((\downarrow_{{}_{{ \mathfrak{S}}}}\zeta)\cap \overline{ \mathfrak{S}})$. For a detailed analysis of the properties of $\Theta^{\overline{{ \mathfrak{S}}}}$, see subsection \ref{subsectionpreliminaryhidden}.
\end{definition}

\begin{lemma}\label{lemmaSorthogonal}
\begin{eqnarray}
\forall \sigma\in {}{\mathfrak{S}},&& \{\sigma\}^{\underline{\perp}} = \bigcup{}_{\omega\in \Theta^{\overline{{ \mathfrak{S}}}}(\sigma)}(\uparrow^{{}^{{}{\mathfrak{S}}}}\!\!\omega^\star),\\
\forall S\subseteq {}{\mathfrak{S}},&& S^{\underline{\perp}} = \bigcap{}_{\sigma\in S}\;\{\sigma\}^{\underline{\perp}}.
\end{eqnarray}
\end{lemma}
\begin{proof}
Trivial.   
\end{proof}

\begin{theorem}\label{expressionorthogonal}
For any $S\subseteq {}{\mathfrak{S}}$, we define $U_S$ as the set of maps from $S$ to ${\overline{{ \mathfrak{S}}}}$ mapping any $\sigma\in S$ to an element of $\Theta^{\overline{{ \mathfrak{S}}}}(\sigma)$.  In other words,
\begin{eqnarray}
U_S &:=& \{\,\phi\in {\overline{{ \mathfrak{S}}}}^S \;\vert\; \forall \sigma\in S,\;\phi(\sigma)\in  \Theta^{\overline{{ \mathfrak{S}}}}(\sigma)\,\}.\label{defUS}
\end{eqnarray}
We then define, for any $S\subseteq {}{\mathfrak{S}}$,  an element $V_S$ in $\widecheck{ \mathfrak{S}}\cup \{\centerdot\}$ by
\begin{eqnarray}
V_S &:=& \{\,\bigsqcup{}^{{}^{{ \mathfrak{S}}}}_{\sigma\in S} (\phi(\sigma))^\star\;\vert\; \phi\in U_S\,\}.
\end{eqnarray}
Endly, we have
\begin{eqnarray}
S^{\underline{\perp}} &=& \bigcup{}_{\alpha\in V_S}(\uparrow^{{}^{{}{\mathfrak{S}}}}\!\!\alpha).\label{expressionSortho}
\end{eqnarray}
with the obvious convention $(\uparrow^{{}^{{}{\mathfrak{S}}}}\!\!\centerdot) :=\varnothing$.
\end{theorem}
\begin{proof}
Direct consequence of Lemma \ref{lemmaSorthogonal}.
\end{proof}

\begin{lemma}
For any $S\subseteq {}{\mathfrak{S}}$, we define $U_S$ as in \ref{defUS}, and we have
\begin{eqnarray}
(S^{\underline{\perp}})^{\underline{\perp}} &=& \bigcap{}_{\phi\in U_S}\left( \bigcup{}_{\rho\in \Theta^{\overline{{ \mathfrak{S}}}}\left(\bigsqcup{}^{{}^{{}{\mathfrak{S}}}}_{\sigma\in S} (\phi(\sigma))^\star\right)}(\uparrow^{{}^{{}{\mathfrak{S}}}}\!\!\rho^\star)\right).\label{expressionSorthoortho}
\end{eqnarray}
\end{lemma}
\begin{proof}
Direct consequence of Theorem \ref{expressionorthogonal} and Lemma \ref{lemmaSorthogonal}.
\end{proof}

\begin{theorem}\label{theoremorthocomplementation}
We have the following ortho-complementation property :
\begin{eqnarray}
\forall H\in { \mathfrak{H}}({\mathfrak{S}}, {\underline{\perp}}),&&\left(H^{\underline{\perp}} \cap H  \right)=\varnothing.
\end{eqnarray}
\end{theorem}
\begin{proof}
Let us fix $S\subseteq {\mathfrak{S}}$.\\
Let us suppose that there exists $\sigma\in {\mathfrak{S}}$ such that $\sigma\in \left((S^{\underline{\perp}})^{\underline{\perp}} \cap (S^{\underline{\perp}})  \right)$.  We intent to exhibit a contradiction.\\
Using expressions (\ref{expressionSortho}) and (\ref{expressionSorthoortho}), we deduce that 
\begin{itemize}
\item there exists $\phi\in U_S$ such that $\left(\bigsqcup{}^{{}^{{}{\mathfrak{S}}}}_{\sigma\in S} (\phi(\sigma))^\star\right) \sqsubseteq_{{}_{{}{\mathfrak{S}}}} \sigma$,
\item for any $\psi\in U_S$, there exists $\rho\in \Theta^{\overline{{ \mathfrak{S}}}}\left(\bigsqcup{}^{{}^{{ \mathfrak{S}}}}_{\sigma\in S} (\psi(\sigma))^\star\right)$ such that $\rho^\star \sqsubseteq_{{}_{{}{\mathfrak{S}}}} \sigma$.
\end{itemize}
From the first point, we deduce that for any $\kappa \in \Theta^{\overline{{ \mathfrak{S}}}}\left(\bigsqcup{}^{{}^{{}{\mathfrak{S}}}}_{\sigma\in S} (\phi(\sigma))^\star\right)$ we have $\kappa \sqsubseteq_{{}_{{}{\mathfrak{S}}}} \sigma$. Let us now choose for $\psi$ in the second point the element $\phi\in U_S$ fixed in the first point, and let us choose in the second point $\kappa$ to be equal to the $\rho$ fixed by the second point, we have then simultaneously $\rho \sqsubseteq_{{}_{{}{\mathfrak{S}}}} \sigma$ and $\rho^\star \sqsubseteq_{{}_{{}{\mathfrak{S}}}} \sigma$. We have then obtained that there exists an element $\rho\in \overline{\mathfrak{S}}$ such that $\rho$ and $\rho^\star$ admit a common upper bound in ${}{\mathfrak{S}}$, which is impossible. We have then obtained the announced contradiction. This concludes the proof.
\end{proof}

\begin{theorem}
${ \mathfrak{H}}({}{\mathfrak{S}}, {\underline{\perp}})$ is a complete ortho-lattice.
\end{theorem}
\begin{proof}
This result summarizes the results obtained in Lemma \ref{lemmavarnothing&E}, Lemma \ref{lemmaorthocapsemi-lattice}, Theorem \ref{orthocomplem0} and Theorem \ref{theoremorthocomplementation}.
\end{proof}

\begin{definition}\label{deflangleSrangle}
We will define the set $\langle {}{\mathfrak{S}}\rangle$, as follows
\begin{eqnarray}
\hspace{-1.8cm}&&\langle{}{\mathfrak{S}}\rangle :=  \{\, \{\sigma\}\;\vert\; \sigma\in ({}{\mathfrak{S}}\smallsetminus \overline{\mathfrak{S}})\,\} \cup \underbrace{{}{\mathfrak{S}}} \cup\{\bot_{{}_{\mathfrak{S}}} \}\;\;\textit{\rm with}\\
\hspace{-1.8cm}&&\underbrace{{}{\mathfrak{S}}}:= \{\,S\subseteq \overline{{ \mathfrak{S}}}\smallsetminus \{\bot_{{}_{\mathfrak{S}}} \}\;\vert\; \exists \sigma\in {{}{\mathfrak{S}}}\smallsetminus \{\bot_{{}_{\mathfrak{S}}} \}\;\;\textit{\rm s.t.}\;\; \Theta^{\overline{{ \mathfrak{S}}}}(\sigma)=\{\,\alpha^\star\;\vert\; \alpha\in S\,\}\,\}.
\end{eqnarray}
We note obviously that ${}{\mathfrak{S}}\subseteq \langle{}{\mathfrak{S}}\rangle$.
\end{definition}

\begin{lemma}
The set ${ \mathfrak{H}}({}{\mathfrak{S}}, {\underline{\perp}})$ of orthoclosed subsets of ${}{\mathfrak{S}}$ is explicitly described in terms of $\langle {}{\mathfrak{S}}\rangle$, as follows :
\begin{eqnarray}
{ \mathfrak{H}}({}{\mathfrak{S}}, {\underline{\perp}}) &=& \{\, (\uparrow^{{}^{{{}{\mathfrak{S}}}}}S)\;\vert\; S\in \langle{}{\mathfrak{S}}\rangle\cup \{\varnothing\}\,\}.
\end{eqnarray}
\end{lemma}
\begin{proof}
Direct consequence of formula (\ref{expressionSorthoortho}).
\end{proof}

\begin{lemma}\label{rangleHsubsetH}
The action of the orthogonality operator on the elements of ${ \mathfrak{H}}({}{\mathfrak{S}}, {\underline{\perp}})$ is explicitly given by
\begin{eqnarray}
\varnothing^{\underline{\perp}} &=& (\uparrow^{{}^{{{}{\mathfrak{S}}}}}\bot_{{}_{{ \mathfrak{S}}}} ),\\
(\uparrow^{{}^{{{}{\mathfrak{S}}}}}\bot_{{}_{{ \mathfrak{S}}}} )^{\underline{\perp}} &=& \varnothing,\\
\forall \sigma\in \overline{ \mathfrak{S}}\smallsetminus \{\bot_{{}_{{ \mathfrak{S}}}}\},\;\;(\uparrow^{{}^{{{}{\mathfrak{S}}}}}\sigma )^{\underline{\perp}} &=&(\uparrow^{{}^{{{}{\mathfrak{S}}}}}\sigma^\star ),\\
\forall \sigma\in {}{\mathfrak{S}}\smallsetminus \overline{ \mathfrak{S}},\;\; (\uparrow^{{}^{{{}{\mathfrak{S}}}}}\sigma )^{\underline{\perp}} &=& \uparrow^{{}^{{{}{\mathfrak{S}}}}} \{\, \alpha^\star\;\vert\; \alpha\in \Theta^{\overline{{ \mathfrak{S}}}}(\sigma)\;\},\\
\forall S\in \underbrace{{}{\mathfrak{S}}},\;\;
(\uparrow^{{}^{{{}{\mathfrak{S}}}}}S)^{\underline{\perp}} &=& \uparrow^{{}^{{{}{\mathfrak{S}}}}} \bigsqcup{}^{{}^{{{}{\mathfrak{S}}}}}\{\, \alpha^\star\;\vert\; \alpha\in S\;\}.
\end{eqnarray}
\end{lemma}
\begin{proof}
Direct consequence of Lemma \ref{lemmaSorthogonal}.
\end{proof}

\subsection{Real morphisms}

\begin{definition} Let us consider that the states/effects Chu Spaces $({ \mathfrak{S}}_A,{ \mathfrak{E}}_{{ \mathfrak{S}}_A},\epsilon^{{ \mathfrak{S}}_A})$ and $({ \mathfrak{S}}_B,{ \mathfrak{E}}_{{ \mathfrak{S}}_B},\epsilon^{{ \mathfrak{S}}_B})$ are respectively equipped with real structures $(\overline{ \mathfrak{S}}_A,\star)$ and $(\overline{ \mathfrak{S}}_B,\star)$ (we use the same notations for the star involutions as long as there can be no ambiguities). \\
Any morphism $(f ,f^{\ast})$ from the states/effects Chu Space $({ \mathfrak{S}}_A,{ \mathfrak{E}}_{{ \mathfrak{S}}_A},\epsilon^{{ \mathfrak{S}}_A})$ to the states/effects Chu Space $({ \mathfrak{S}}_B,{ \mathfrak{E}}_{{ \mathfrak{S}}_B},\epsilon^{{ \mathfrak{S}}_B})$ satisfying
\begin{eqnarray}
f^{\ast}(\overline{ \mathfrak{E}}_{{ \mathfrak{S}}_B}) & \subseteq & \overline{ \mathfrak{E}}_{{ \mathfrak{S}}_A}.\label{defrealmorphism}
\end{eqnarray}
will be said to be {\em a real-morphism}.\\
The set of real morphisms from ${ \mathfrak{S}}_A$ to ${ \mathfrak{S}}_B$ will be denoted $\overline{\mathfrak{C}}({ \mathfrak{S}}_A,{ \mathfrak{S}}_B)$.
\end{definition}

\begin{definition}
If ${ \mathfrak{S}}$ admits a real structure, we can define the set of {\em real measurements,} denoted by $ \overline{\mathfrak{M}}_{{}_{{ \mathfrak{S}}}}$, as the image by $\varphi$ (defined in (\ref{defvarphi})) of $\overline{ \mathfrak{E}}_{ \mathfrak{S}}$.  

\end{definition}

\begin{lemma}
Real measurements (as defined in previous definition) are obviously real morphisms from ${ \mathfrak{S}}$ to ${ \mathfrak{B}}$.  
\end{lemma}
\begin{proof}
From the defining equation (\ref{mastu}), we deduce immediately the property (\ref{defrealmorphism}) as soon as the considered measurement is a real measurement.
\end{proof}

\subsection{Determinism vs.  Indeterminism}\label{particularspacesofstates}

In this subsection, we will consider a space of states, denoted  ${ \mathfrak{S}}$, which admits a real structure $(\overline{ \mathfrak{S}},\star)$.

\begin{definition}\label{deterministicmeasurement}
A real measurement ${\mathfrak{m}}\in \overline{\mathfrak{M}}_{{}_{{ \mathfrak{S}}}}$ satisfying 
\begin{eqnarray}
{\mathfrak{m}}(\overline{ \mathfrak{S}}{}^{{}^{pure}}) &\subseteq & \{\textit{\bf Y},\textit{\bf N}\}
\end{eqnarray}
is said to be {\em a deterministic real measurement}.
\end{definition}

\begin{definition}\label{deterministicspace}
The space of states ${ \mathfrak{S}}$ equipped with its real structure $(\overline{ \mathfrak{S}},\star)$ is said to be {\em a deterministic space of states} (or {\em a simplex space of states}) iff ${\mathfrak{m}}_{{ \mathfrak{l}}}$ is a deterministic real measurement for every ${{ \mathfrak{l}}} \in  \{\,{ \mathfrak{l}}_{(\sigma,\sigma^\star)}\;\vert\; \sigma\in \overline{ \mathfrak{S}}{}^{{}^{pure}}\,\}$.
\begin{eqnarray}
(\overline{ \mathfrak{S}},\star)\;\textit{\rm deterministic} & \Leftrightarrow & \forall \sigma\in \overline{ \mathfrak{S}}{}^{{}^{pure}},\; {\mathfrak{m}}_{{ \mathfrak{l}}_{(\sigma,\sigma^\star)}}(\overline{ \mathfrak{S}}{}^{{}^{pure}}) \subseteq  \{\textit{\bf Y},\textit{\bf N}\}
\end{eqnarray}
\end{definition}

\begin{lemma}\label{lemmaUsimplex}
The space of states ${ \mathfrak{S}}$ equipped with its real structure $(\overline{ \mathfrak{S}},\star)$ is a deterministic space of states iff 
\begin{eqnarray}
\forall \sigma \in \overline{ \mathfrak{S}},\existunique \;U_\sigma\subseteq \overline{ \mathfrak{S}}{}^{{}^{pure}} & \vert & \sigma=\bigsqcap{}^{{}^{\overline{ \mathfrak{S}}}}U_\sigma.\label{simplexdecompunique}
\end{eqnarray}
(we note that necessarily $U_\sigma=\underline{\sigma}_{{}_{\overline{ \mathfrak{S}}}}$)
\end{lemma}
\begin{proof}
Let us consider that ${ \mathfrak{S}}$ equipped with its real structure $(\overline{ \mathfrak{S}},\star)$ is a deterministic space of states, and let us suppose that there exists $\sigma_1,\sigma_2\in \overline{ \mathfrak{S}}\smallsetminus \{\bot_{{}_{{ \mathfrak{S}}}}\}$ and $\sigma_3\in \overline{ \mathfrak{S}}{}^{{}^{pure}}$ such that 
\begin{eqnarray}
\sigma_3 \sqsupseteq_{{}_{\overline{ \mathfrak{S}}}} (\sigma_1\sqcap_{{}_{\overline{ \mathfrak{S}}}}\sigma_2)\;\;\textit{\rm and} \;\; \sigma_3\not\sqsupseteq_{{}_{\overline{ \mathfrak{S}}}}\sigma_1\;\;\textit{\rm and}\;\;   \sigma_3\not\sqsupseteq_{{}_{\overline{ \mathfrak{S}}}}\sigma_2\;\;\textit{\rm and}\;\;  \sigma_1\parallel_{{}_{\overline{ \mathfrak{S}}}}\sigma_2.\label{simplexdecompuniqueproof1}
\end{eqnarray}
Let us exhibit a contradiction. As long as $\sigma_3\not\sqsupseteq_{{}_{\overline{ \mathfrak{S}}}}\sigma_1$ we know that $\forall \omega\in \underline{\sigma_1}_{{}_{\overline{ \mathfrak{S}}}},  \omega\not\sqsupseteq_{{}_{\overline{ \mathfrak{S}}}}\sigma_3$ and then ${\mathfrak{m}}_{{ \mathfrak{l}}_{(\sigma_3,\sigma_3^\star)}}(\omega)\not=\textit{\bf Y}$ and then ${\mathfrak{m}}_{{ \mathfrak{l}}_{(\sigma_3,\sigma_3^\star)}}(\omega)=\textit{\bf N}$ (here we use ${\mathfrak{m}}_{{ \mathfrak{l}}_{(\sigma_3,\sigma_3^\star)}}(\overline{ \mathfrak{S}}{}^{{}^{pure}}) \subseteq  \{\textit{\bf Y},\textit{\bf N}\}$). As a consequence, using the homomorphic property of the measurement, we obtain ${\mathfrak{m}}_{{ \mathfrak{l}}_{(\sigma_3,\sigma_3^\star)}}(\sigma_1)=\textit{\bf N}$, or in other words $\sigma_1\sqsupseteq_{{}_{\overline{ \mathfrak{S}}}}\sigma_3^\star$. Analogously, we obtain $\sigma_2\sqsupseteq_{{}_{\overline{ \mathfrak{S}}}}\sigma_3^\star$. As a result, we obtain $(\sigma_1\sqcap_{{}_{\overline{ \mathfrak{S}}}}\sigma_2)\sqsupseteq_{{}_{\overline{ \mathfrak{S}}}}\sigma_3^\star$. But we have assumed $\sigma_3 \sqsupseteq_{{}_{\overline{ \mathfrak{S}}}} (\sigma_1\sqcap_{{}_{\overline{ \mathfrak{S}}}}\sigma_2)$ which leads to the announced contradiction.  We have then proved the negation of the property (\ref{simplexdecompuniqueproof1}).\\
Let us now imagine that the property (\ref{simplexdecompunique}) is not satisfied.  Then, there exists a subset $U$ of $\underline{\sigma}_{{}_{\overline{ \mathfrak{S}}}}$ which is minimal for the inclusion among the subsets of $\underline{\sigma}_{{}_{\overline{ \mathfrak{S}}}}$ satisfying $\sigma=\bigsqcap{}^{{}^{\mathfrak{S}}} U$, and there exists $\sigma'$ element of $\underline{\sigma}_{{}_{\overline{ \mathfrak{S}}}} \smallsetminus U$.\\ Let us then introduce $U_1$ and $U_2$ such that $U_1\cap U_2=\varnothing$ and $U_1\cup U_2=U$. Now, we define $\sigma_1:=\bigsqcap{}^{{}^{\mathfrak{S}}} U_1$, $\sigma_2:=\bigsqcap{}^{{}^{\mathfrak{S}}} U_2$ and $\sigma_3:=\sigma'$. They satisfy (\ref{simplexdecompuniqueproof1}) which is contradictory. Then, we have proved the property (\ref{simplexdecompunique}).\\
Reciprocally, let us suppose that the property (\ref{simplexdecompunique}) is satisfied.  We have then for any $\sigma\in \overline{ \mathfrak{S}}\smallsetminus \{\bot_{{}_{{ \mathfrak{S}}}}\}$ the basic result $U_{\bot_{{}_{{ \mathfrak{S}}}}}=U_{\sigma} \cup U_{\sigma^\star}$, which leads to the property ${\mathfrak{m}}_{{ \mathfrak{l}}_{(\sigma,\sigma^\star)}}(\overline{ \mathfrak{S}}{}^{{}^{pure}}) \subseteq  \{\textit{\bf Y},\textit{\bf N}\}$. This concludes the proof.
\end{proof}
\begin{lemma}
$(\overline{ \mathfrak{S}},\star)$ is deterministic iff $(\overline{ \mathfrak{S}}{}^{{}^{pure}},\underline{\perp})$ is completely reducible, i.e. iff
\begin{eqnarray}
\forall A\subseteq \overline{ \mathfrak{S}}{}^{{}^{pure}}\;\vert\; A=A^{\underline{\perp}\underline{\perp}},&\textit{\rm we have}& A\cup A^{\underline{\perp}}=\overline{ \mathfrak{S}}{}^{{}^{pure}},\;\;\;\;A\cap A^{\underline{\perp}}=\varnothing.
\end{eqnarray}
Here, the operator $(\cdot)^{\underline{\perp}}$ designates the operator from $\mathcal{P}(\overline{ \mathfrak{S}}{}^{{}^{pure}})$ to itself which maps $A$ to $\{\sigma\in \overline{ \mathfrak{S}}{}^{{}^{pure}}\;\vert\; \forall \sigma'\in A, \sigma\underline{\perp}\sigma'\}$.
\end{lemma}
\begin{proof}
We note that, at the end of the proof of Lemma \ref{lemmaUsimplex}, we have shown the following result
\begin{eqnarray}
\hspace{-1cm}(\overline{ \mathfrak{S}},\star)\;\textit{\rm deterministic} & \Leftrightarrow & \forall \sigma\in \overline{ \mathfrak{S}}\smallsetminus \{\bot_{{}_{{ \mathfrak{S}}}}\},\; {\mathfrak{m}}_{{ \mathfrak{l}}_{(\sigma,\sigma^\star)}}(\overline{ \mathfrak{S}}{}^{{}^{pure}}) \subseteq  \{\textit{\bf Y},\textit{\bf N}\}
\end{eqnarray}
the announced result is then obtained.
\end{proof}
\begin{lemma}
We have also obtained the following characterization
\begin{eqnarray}
\hspace{-1cm}(\overline{ \mathfrak{S}},\star)\;\textit{\rm NOT deterministic} & \Leftrightarrow &\exists\sigma_1,\sigma_2\in \overline{ \mathfrak{S}}\smallsetminus \{\bot_{{}_{{ \mathfrak{S}}}}\},\exists\sigma_3\in \overline{ \mathfrak{S}}{}^{{}^{pure}}\;\;\vert\;\; \label{threestatesnotdeterministic}\\
&&\sigma_3 \sqsupseteq_{{}_{\overline{ \mathfrak{S}}}} (\sigma_1\sqcap_{{}_{\overline{ \mathfrak{S}}}}\sigma_2)\;\;\textit{\rm and} \;\; \sigma_3\not\sqsupseteq_{{}_{\overline{ \mathfrak{S}}}}\sigma_1\;\;\textit{\rm and}\;\;   \sigma_3\not\sqsupseteq_{{}_{\overline{ \mathfrak{S}}}}\sigma_2\;\;\textit{\rm and}\;\;  \sigma_1\parallel_{{}_{\overline{ \mathfrak{S}}}}\sigma_2.\nonumber
\end{eqnarray}
\end{lemma}
\begin{proof}
This result has already been proved along the proof of Lemma \ref{lemmaUsimplex}.
\end{proof}


\begin{theorem}
Let us suppose that $\overline{ \mathfrak{S}}$ is a simplex space of states and let $\star$ be defined by
\begin{eqnarray}
\forall \sigma\in \overline{ \mathfrak{S}},&&\sigma^\star := \bigsqcap{}^{{}^{\overline{ \mathfrak{S}}}}_{\alpha\in \overline{ \mathfrak{S}}{}^{{}^{pure}}\!\!\!\!\smallsetminus \underline{\sigma}_{{}_{\overline{ \mathfrak{S}}}}} \alpha.\label{starsimplex}
\end{eqnarray}
Then, $(\overline{ \mathfrak{S}}, \star)$ is a real structure of $\overline{ \mathfrak{S}}$.\\
Conversely, let us consider ${ \mathfrak{S}}$ a space of states which admits a real structure $(\overline{ \mathfrak{S}}, \star)$. Let us suppose that $\overline{ \mathfrak{S}}$ is a simplex space of states.  Then, we have necessarily ${ \mathfrak{S}}=\overline{ \mathfrak{S}}$.
\end{theorem}
\begin{proof}
As long as $\overline{\mathfrak{S}}$ is a simplex space of states, the star map $\star$ defined on the real structure of ${\mathfrak{S}}$ is necessarily given by the formula (\ref{starsimplex}).\\
The first part of the theorem is trivial to check. \\
Let us now assume that the real structure of ${ \mathfrak{S}}$ is a simplex space of states and let us prove that ${ \mathfrak{S}}=\overline{ \mathfrak{S}}$.\\
Let us then consider two elements $\sigma$ and $\sigma'$ in $\overline{\mathfrak{S}}$ such that $\neg \;\widehat{\sigma \sigma'}{}^{{}^{\overline{\mathfrak{S}}}}$.  Let us suppose that $\widehat{\sigma \sigma'}{}^{{}^{{\mathfrak{S}}}}$.\\
As long as $\overline{\mathfrak{S}}$ is a simplex, we must necessarily have $\sigma^\star \sqsubseteq_{{}_{\overline{\mathfrak{S}}}} \sigma'$.  As a consequence, we have $\widehat{\sigma\,\sigma^\star}{}^{{}^{{\mathfrak{S}}}}$ which contradicts the condition $\neg\; \widehat{\sigma\sigma^\star}{}^{{}^{{ \mathfrak{S}}}}$ imposed for a real structure. \\ 
As a conclusion,  ${ \mathfrak{S}}\smallsetminus \overline{\mathfrak{S}}$ is empty.
\end{proof}

\begin{definition}
Traditionally, an Inf semi-lattice ${ \mathfrak{S}}$ is said to be {\em distributive} iff 
\begin{eqnarray}
&&\hspace{-2cm} \forall \sigma,\sigma_1,\sigma_2\in { \mathfrak{S}}\;\vert\; \sigma\not= \sigma_1,\sigma_2 ,\;\;\;\;\;\; (\sigma_1 \sqcap_{{}_{ \mathfrak{S}}} \sigma_2) \sqsubseteq_{{}_{ \mathfrak{S}}}  \sigma  \Rightarrow \nonumber\\ 
&&\exists \sigma'_1,\sigma'_2\in { \mathfrak{S}}\;\vert\; (\,\sigma_1 \sqsubseteq_{{}_{ \mathfrak{S}}} \sigma'_1,\;\;\; \sigma_2 \sqsubseteq_{{}_{ \mathfrak{S}}} \sigma'_2\;\;\;\textit{\rm and}\;\;\; \sigma = \sigma'_1 \sqcap_{{}_{ \mathfrak{S}}} \sigma'_2\,).\;\;\;\;\;\;\;\;\;\;\;\;
\end{eqnarray} 
\end{definition}

\begin{lemma}
A simplex space of states is necessarily distributive as an Inf semi-lattice. 
\end{lemma}
\begin{proof}
Let us first assume that ${\mathfrak{S}}$ is a simplex space of states. 
Let us consider $\sigma,\sigma_1,\sigma_2\in { \mathfrak{S}}$ such that $\sigma\not= \sigma_1,\sigma_2 $ and $ (\sigma_1 \sqcap_{{}_{ \mathfrak{S}}} \sigma_2) \sqsubseteq_{{}_{ \mathfrak{S}}}  \sigma$. We have then $\underline{\sigma}_{{}_{{ \mathfrak{S}}}} \subseteq (\underline{\sigma_1}_{{}_{{ \mathfrak{S}}}}\cup \underline{\sigma_2}_{{}_{{ \mathfrak{S}}}})$. If we define $\sigma'_1:=\bigsqcap{}^{{}^{\mathfrak{S}}}(\underline{\sigma_1}_{{}_{{ \mathfrak{S}}}}\cap \underline{\sigma}_{{}_{{ \mathfrak{S}}}})$ and $\sigma'_2:=\bigsqcap{}^{{}^{\mathfrak{S}}}(\underline{\sigma_2}_{{}_{{ \mathfrak{S}}}}\cap \underline{\sigma}_{{}_{{ \mathfrak{S}}}})$ we check immediately that $\sigma_1 \sqsubseteq_{{}_{ \mathfrak{S}}} \sigma'_1,\;\;\; \sigma_2 \sqsubseteq_{{}_{ \mathfrak{S}}} \sigma'_2$ and $\sigma = \sigma'_1 \sqcap_{{}_{ \mathfrak{S}}} \sigma'_2$. As a result, ${\mathfrak{S}}$ is distributive.
\end{proof}

Let us now exhibit the two simplest examples of simplex spaces of states called respectively ${ \mathfrak{Z}}_2 \cong {\mathfrak{B}}$ and ${ \mathfrak{Z}}_3$ (they are given with their star operation):
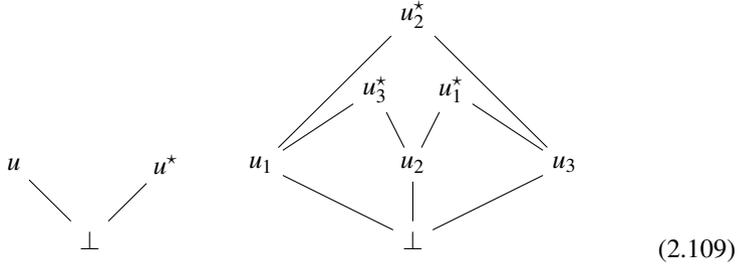
\begin{eqnarray}
\begin{tikzpicture}
  \node (d) at (-1,0) {$u$};
  \node (f) at (1,0) {$u^\star$};
  \node (min) at (0,-1) {$\bot$};
  \draw (min) -- (d)  (min) -- (f) ;
\end{tikzpicture} 
&&\begin{tikzpicture}
  \node (a) at (-0.5,1) {$u_3^\star$};
  \node (b) at (0,2) {$u_2^\star$};
  \node (c) at (0.5,1) {$u_1^\star$};
  \node (e) at (0,0) {$u_2$};
  \node (g) at (-2,0) {$u_1$};
  \node (h) at (2,0) {$u_3$};
  \node (min) at (0,-1) {$\bot$};
  \draw (min) -- (g) (min) -- (h)
  (min) -- (e) (a) -- (g) (c) -- (h) (g) -- (b) -- (h) (a) -- (e) -- (c);
\end{tikzpicture} 
\end{eqnarray}

Although many type of indeterministic theories can be build, we fix a principle that will describe a rather general class of them. This choice is principally motivated by Lemma \ref{lemmairreducibility} and Lemma \ref{theoremirreducibilitytensor}.

\begin{definition}\label{completelyindeterministicspace}
The space of states ${ \mathfrak{S}}$ equipped with its real structure $(\overline{ \mathfrak{S}},\star)$ is said to be {\em a completely indeterministic space of states} iff,  
\begin{eqnarray}  
\forall \sigma,\lambda\in \overline{ \mathfrak{S}}{}^{{}^{pure}}\;\vert\; \sigma\sqsupseteq_{{}_{ \mathfrak{S}}}\lambda^\star, && \exists \kappa \in \overline{ \mathfrak{S}}{}^{{}^{pure}}\;\vert\; \kappa\not\sqsupseteq_{{}_{ \mathfrak{S}}}\sigma^\star\;\textit{\rm and}\; \kappa\not\sqsupseteq_{{}_{ \mathfrak{S}}}\lambda^\star.
\end{eqnarray}
\end{definition}


\begin{lemma}\label{lemmairreducibility}
If $(\overline{ \mathfrak{S}},\star)$ is completely indeterministic then $(\overline{ \mathfrak{S}}{}^{{}^{pure}},\underline{\perp})$ is irreducible, i.e.  
\begin{eqnarray}
\nexists A\subseteq \overline{ \mathfrak{S}}{}^{{}^{pure}}\;\vert\; (\, A=A^{\underline{\perp}\underline{\perp}},\;\;\;\;\varnothing \varsubsetneq A \varsubsetneq \overline{ \mathfrak{S}}{}^{{}^{pure}},\;\;\;\; A\cup A^{\underline{\perp}}=\overline{ \mathfrak{S}}{}^{{}^{pure}},\;\;\;\;A\cap A^{\underline{\perp}}=\varnothing\,).\;\;\;\;\;\;\;\;
\end{eqnarray}
Here again, the operator $(\cdot)^{\underline{\perp}}$ designates the operator from $\mathcal{P}(\overline{ \mathfrak{S}}{}^{{}^{pure}})$ to itself which maps $A$ to $\{\sigma\in \overline{ \mathfrak{S}}{}^{{}^{pure}}\;\vert\; \forall \sigma'\in A, \sigma\underline{\perp}\sigma'\}$.
\end{lemma}
\begin{proof}
Let us consider a completely indeterministic $(\overline{ \mathfrak{S}},\star)$. Let us suppose that $(\overline{ \mathfrak{S}}{}^{{}^{pure}},\underline{\perp})$ is reducible and let us exhibit a contradiction. Let us then consider $A\subseteq \overline{ \mathfrak{S}}{}^{{}^{pure}}$ such that $A=A^{\underline{\perp}\underline{\perp}}$ and $\varnothing \varsubsetneq A \varsubsetneq \overline{ \mathfrak{S}}{}^{{}^{pure}}$ and $A\cup A^{\underline{\perp}}=\overline{ \mathfrak{S}}{}^{{}^{pure}}$ and $A\cap A^{\underline{\perp}}=\varnothing$. Let us fix $\sigma$ an element of $A$. For any $\omega$ satisfying $\omega\not\!\!\underline{\perp} \sigma$ we cannot have $\omega\in A^{\underline{\perp}}$ and then, by assumption, $\omega\in A$. Let us now consider $\lambda \in A^{\underline{\perp}}$. As long as $\lambda\underline{\perp}\sigma$ and $(\overline{ \mathfrak{S}},\star)$ being completely indeterministic, there exists $\kappa\in \overline{ \mathfrak{S}}{}^{{}^{pure}}$ such that $\kappa\not\!\!\!\underline{\perp}\sigma$ and $\kappa\not\!\!\!\underline{\perp}\lambda$. From $\kappa\not\!\!\!\underline{\perp}\sigma$ we deduce that $\kappa\in A$. However, from $\kappa\not\!\!\!\underline{\perp}\lambda$ we deduce that $\lambda\notin A^{\underline{\perp}}$ which is contradictory. This concludes the proof of the assertion.
\end{proof}

Although the previous notion of indeterminism is quite appealing, we will introduce another notion called linearity which emphasizes the notion of "superposition" of states in the notion of indeterminism.

\begin{definition}\label{linearindeterministicspace}
The space of states ${ \mathfrak{S}}$ equipped with its real structure $(\overline{ \mathfrak{S}},\star)$ is said to be {\em linear} iff, 
\begin{eqnarray}
&&\hspace{-1cm}\forall \sigma_1,\sigma_2\in \overline{ \mathfrak{S}}{}^{{}^{pure}}\;\vert\; 
(\sigma_1\sqcap_{{}_{\overline{ \mathfrak{S}}}}\sigma_2)\sqcoversubset_{{}_{\overline{ \mathfrak{S}}}}\sigma_1,\sigma_2, \nonumber\\
&& \exists \sigma_3 \in { \mathfrak{S}}\;\vert\; (\,\sigma_3\not\sqsupseteq_{{}_{{ \mathfrak{S}}}}\sigma_1\;\textit{\rm and}\; \sigma_3\not\sqsupseteq_{{}_{{ \mathfrak{S}}}}\sigma_2\;\textit{\rm and}\; \sigma_3\sqcoversupset_{{}_{{ \mathfrak{S}}}}(\sigma_1\sqcap_{{}_{\overline{ \mathfrak{S}}}}\sigma_2)\,).\;\;\;\;\;\;\;\;\;\;\;\;\;\;\;\;\;\;
\end{eqnarray}
\end{definition}

A generic class of linear and completely indeterministic spaces of states called ${ \mathfrak{Z}}'_N$ (with $N\geq 2$) is given as follows 
\begin{eqnarray}
\begin{tikzpicture}
  \node (a) at (-3,0) {$\sigma_1$};
   \node (e) at (-2,0) {$\cdots$};
  \node (b) at (-1,0) {$\sigma_N$};
  \node (c) at (1,0) {$\sigma_1^\star$};
   \node (f) at (2,0) {$\cdots$};
    \node (d) at (3,0) {$\sigma_N^\star$};
  \node (min) at (0,-1) {$\bot$};
  \draw (min) -- (a)  (min) -- (b) (min) -- (c)  (min) -- (d);
\end{tikzpicture} 
\end{eqnarray}
A space of states of this class will be said to be a {\em one-dimensional indeterministic space of states} (this space of states will be eventually called "quantum bit" space of states).

\section{Characterization of hidden states}\label{sectioncharacterizationhiddenstates}

In the present section, we intent to clarify the notion of real structure and in particular to explore the possibility to derive ${ \mathfrak{S}}$ as a "completion" of its real structure $(\overline{{ \mathfrak{S}}},\star)$, i.e. to build hidden states in a suitable completion of the space of real states.

\subsection{Preliminary remarks}\label{subsectionpreliminaryhidden}

Let us introduce some notations.\\
We will denote by ${ \mathcal{F}}(\overline{{ \mathfrak{S}}})$ the set of chain complete lower subsets in $\overline{{ \mathfrak{S}}}$. Explicitly,we have for any ${ \mathfrak{F}}\in { \mathcal{F}}(\overline{{ \mathfrak{S}}})$
\begin{eqnarray}
&&{ \mathfrak{F}}\subseteq \overline{{ \mathfrak{S}}},\\
&& \forall C\subseteq_{Chain}{ \mathfrak{F}}, (\bigsqcup{}^{{}^{\overline{{ \mathfrak{S}}}}} C)\in { \mathfrak{F}},\\
&& \forall \sigma,\sigma'\in \overline{{ \mathfrak{S}}}, (\sigma\in { \mathfrak{F}}\;\;\textit{\rm and}\;\;\sigma\sqsupseteq_{{}_{\overline{{ \mathfrak{S}}}}}\sigma')\Rightarrow (\sigma'\in { \mathfrak{F}}).
\end{eqnarray}
${ \mathcal{F}}(\overline{{ \mathfrak{S}}})$ equipped with $\subseteq$ is a poset, and equipped with $\cap$ is an Inf semi-lattice.\\
Due to Zorn lemma and the chain completeness of any element ${ \mathfrak{F}}$ of ${ \mathcal{F}}(\overline{{ \mathfrak{S}}})$, we have $Max({ \mathfrak{F}})\not=\varnothing$. 
We then define the map $\Omega$ sending any element of ${ \mathcal{F}}(\overline{{ \mathfrak{S}}})$ to its set of maximal elements :
\begin{eqnarray}
\begin{array}{lclcl}
\Omega & : & { \mathcal{F}}(\overline{{ \mathfrak{S}}}) & \longrightarrow & {\mathcal{P}}(\overline{{ \mathfrak{S}}})\\
& & { \mathfrak{F}} & \mapsto & Max({ \mathfrak{F}}).
\end{array}
\end{eqnarray}

We define on ${\mathcal{P}}(\overline{{ \mathfrak{S}}})$ a pre-order denoted $\sqsubseteq$ by
\begin{eqnarray}
\forall F,F'\in {\mathcal{P}}(\overline{{ \mathfrak{S}}}),&& F\sqsubseteq F' \;\;:\Leftrightarrow \;\; \forall f\in F,\exists f'\in F'\;\vert\; f \sqsubseteq_{{}_{\overline{{ \mathfrak{S}}}}} f'.\label{DefinitionpreorderPSbar}
\end{eqnarray}
Then $\Omega$ satisfies
\begin{lemma}
\begin{eqnarray}
\forall { \mathfrak{F}},{ \mathfrak{F}}'\in { \mathcal{F}}(\overline{{ \mathfrak{S}}}),&& { \mathfrak{F}}\subseteq { \mathfrak{F}}' \;\;\Rightarrow\;\; \Omega({ \mathfrak{F}})\sqsubseteq \Omega({ \mathfrak{F}}').
\end{eqnarray}
\end{lemma}
\begin{proof}
Trivial.
\end{proof}
\begin{lemma}
\begin{eqnarray}
\forall \sigma\in {{ \mathfrak{S}}},&& \{\,\sigma'\in \overline{{ \mathfrak{S}}}\;\vert\; \sigma'\sqsubseteq_{{}_{{{ \mathfrak{S}}}}}\sigma\,\} \in { \mathcal{F}}(\overline{{ \mathfrak{S}}}).
\end{eqnarray}
\end{lemma}
\begin{proof}
The only non trivial part is the chain-completeness of $\{\,\sigma'\in \overline{{ \mathfrak{S}}}\;\vert\; \sigma'\sqsubseteq_{{}_{{{ \mathfrak{S}}}}}\sigma\,\}$. It derives from the chain-completeness of $\overline{{ \mathfrak{S}}}$ which can be derived from the down-completeness of $\overline{{ \mathfrak{S}}}$ using the star operation.
\end{proof}

\begin{definition}
We will denote by ${\widehat{{\mathcal{P}}(\overline{{ \mathfrak{S}}})}}$ the subset of ${\mathcal{P}}(\overline{{ \mathfrak{S}}})$ formed by the collections $J\in {\mathcal{P}}(\overline{{ \mathfrak{S}}})$ which admit a common upper-bound in ${ \mathfrak{S}}$.  Due to the down-completeness of ${ \mathfrak{S}}$, the elements of such a collection admit a supremum in ${ \mathfrak{S}}$.
\end{definition}

\begin{lemma}\label{PastlowersubsetP}
${\widehat{{\mathcal{P}}(\overline{{ \mathfrak{S}}})}} $ is a lower subset (i.e.  it is downward closed) of $({\mathcal{P}}(\overline{{ \mathfrak{S}}}),\sqsubseteq)$ containing every elements of the form $\{\sigma\}$ for $\sigma\in \overline{{ \mathfrak{S}}}$.
\end{lemma}
\begin{proof}
Trivial.
\end{proof}

\begin{definition}\label{definitionTheta}
We have already introduced the following map
\begin{eqnarray}
\begin{array}{lclcl}
\Theta^{\overline{{ \mathfrak{S}}}} & : & {{ \mathfrak{S}}} & \longrightarrow & {\widehat{{\mathcal{P}}(\overline{{ \mathfrak{S}}})}}\\
& & \sigma & \mapsto & \Omega(\{\,\sigma'\in \overline{{ \mathfrak{S}}}\;\vert\; \sigma'\sqsubseteq_{{}_{{{ \mathfrak{S}}}}}\sigma\,\}).\label{defThetaS}
\end{array}
\end{eqnarray}
\end{definition}

\begin{definition}\label{definitionLambda}
We also introduce the following map 
\begin{eqnarray}
\begin{array}{lclcl}
\Lambda^{\overline{{ \mathfrak{S}}}} & : & {\widehat{{\mathcal{P}}(\overline{{ \mathfrak{S}}})}} & \longrightarrow & { \mathfrak{S}}\\
& & J & \mapsto &\Lambda^{\overline{{ \mathfrak{S}}}}(J):=\bigsqcup{}^{{}^{{ \mathfrak{S}}}}J.
\end{array}
\end{eqnarray}
\end{definition}

\begin{lemma}\label{lemmaLambda}
\begin{eqnarray}
\forall \sigma\in {{ \mathfrak{S}}},\exists J:=\Theta^{\overline{{ \mathfrak{S}}}}(\sigma)\in {\widehat{{\mathcal{P}}(\overline{{ \mathfrak{S}}})}} & \vert & \Lambda^{\overline{{ \mathfrak{S}}}}(J)=\sigma.
\end{eqnarray}
\end{lemma}
\begin{proof}
Let us consider $\sigma\in {{ \mathfrak{S}}}$ and let us introduce $\sigma':=\bigsqcup{}^{{}^{{ \mathfrak{S}}}}\{\,\alpha\in \overline{{ \mathfrak{S}}}\;\vert\; \alpha\sqsubseteq_{{}_{\overline{{ \mathfrak{S}}}}}\sigma\,\}$.  We note that for any ${ \mathfrak{l}}\in \overline{ \mathfrak{E}}_{ \mathfrak{S}}$, we have $\epsilon^{{ \mathfrak{S}}}_{ \mathfrak{l}}(\sigma)=\epsilon^{{ \mathfrak{S}}}_{ \mathfrak{l}}(\sigma')$ and then, due to the condition (\ref{realChuseparated}), we then conclude that $\sigma=\sigma'$.
\end{proof}

\begin{lemma}\label{galoisthetaSbar}
The pair $(\Theta^{\overline{{ \mathfrak{S}}}},\Lambda^{\overline{{ \mathfrak{S}}}})$ is a Galois connection, i.e.
\begin{eqnarray}
\forall J\in {\widehat{{\mathcal{P}}(\overline{{ \mathfrak{S}}})}},\forall \sigma\in {{ \mathfrak{S}}},&& J \sqsubseteq \Theta^{\overline{{ \mathfrak{S}}}}(\sigma) \;\; \Leftrightarrow \;\; \Lambda^{\overline{{ \mathfrak{S}}}}(J)\sqsubseteq_{{}_{{ \mathfrak{S}}}} \sigma.
\end{eqnarray}
\end{lemma}

\begin{definition}\label{definitionJSbar}
We will denote by $cl^{\overline{{ \mathfrak{S}}}}$ the closure associated to the previous Galois connection, i.e.
\begin{eqnarray}
\begin{array}{lclcl}
cl^{\overline{{ \mathfrak{S}}}} & : & {\widehat{{\mathcal{P}}(\overline{{ \mathfrak{S}}})}} & \longrightarrow & {\widehat{{\mathcal{P}}(\overline{{ \mathfrak{S}}})}}\\
& & J & \mapsto &\Theta^{\overline{{ \mathfrak{S}}}} \circ \Lambda^{\overline{{ \mathfrak{S}}}}(J).
\end{array}\label{expressioncl}
\end{eqnarray}
We will denote by ${ \mathfrak{J}}_{\overline{{ \mathfrak{S}}}}$ the set of closed elements in ${\widehat{{\mathcal{P}}(\overline{{ \mathfrak{S}}})}}$, i.e. ${ \mathfrak{J}}_{\overline{{ \mathfrak{S}}}}:=cl^{\overline{{ \mathfrak{S}}}}({\widehat{{\mathcal{P}}(\overline{{ \mathfrak{S}}})}})$.
\end{definition}

\begin{theorem}
We have as usual
\begin{eqnarray}
\forall J\in {\widehat{{\mathcal{P}}(\overline{{ \mathfrak{S}}})}},&& J\sqsubseteq cl^{{{ \mathfrak{S}}'}}(J),\label{closure1}\\
\forall J,J'\in {\widehat{{\mathcal{P}}(\overline{{ \mathfrak{S}}})}},&& J\sqsubseteq J' \;\;\Rightarrow\;\;cl^{{{ \mathfrak{S}}'}}(J)\sqsubseteq cl^{{{ \mathfrak{S}}'}}(J'),\label{closure2}\\
\forall J\in {\widehat{{\mathcal{P}}(\overline{{ \mathfrak{S}}})}},&& cl^{{{ \mathfrak{S}}'}}\circ cl^{{{ \mathfrak{S}}'}}(J) = cl^{{{ \mathfrak{S}}'}}(J).\label{closure3}
\end{eqnarray}
\end{theorem}

\begin{lemma}\label{JSbarfirstcondition}
\begin{eqnarray}
\forall J\in { \mathfrak{J}}_{\overline{{ \mathfrak{S}}}},\forall \sigma,\sigma'\in J,&& \sigma\not=\sigma' \;\;\Rightarrow\;\; \neg \widehat{\sigma\sigma'}{}^{{}^{\overline{{ \mathfrak{S}}}}}.\label{defJSbar2}
\end{eqnarray}
\end{lemma}
\begin{proof}
Let us consider $ \sigma,\sigma'\in \Theta^{\overline{{ \mathfrak{S}}}}(\alpha)$ with $\sigma\not=\sigma'$. Let us suppose that $\widehat{\sigma\sigma'}{}^{{}^{\overline{{ \mathfrak{S}}}}}$. We have then $(\sigma\sqcup_{{}_{{{ \mathfrak{S}}}}}\sigma')\in \overline{{ \mathfrak{S}}}$ and also $(\sigma\sqcup_{{}_{{{ \mathfrak{S}}}}}\sigma')\sqsubseteq_{{}_{{{ \mathfrak{S}}}}}\alpha$, and we have obviously $\sigma\sqsubset_{{}_{{{ \mathfrak{S}}}}} (\sigma\sqcup_{{}_{{{ \mathfrak{S}}}}}\sigma')$. These points lead to a contradiction as we have assumed that $ \sigma$ is a  {maximal element} of $\{\,\alpha'\in \overline{{ \mathfrak{S}}}\;\vert\; \alpha'\sqsubseteq_{{}_{{{ \mathfrak{S}}}}}\alpha\,\}$.
\end{proof}

\begin{lemma}\label{JSbarsecondcondition}
\begin{eqnarray}
\forall J\in { \mathfrak{J}}_{\overline{{ \mathfrak{S}}}},\forall \sigma,\sigma'\in J,&& \sigma'\not\sqsupseteq_{{}_{\overline{{ \mathfrak{S}}}}}\sigma^\star.\label{defJSbar3}
\end{eqnarray}
Note that this condition is equivalent to the following one
\begin{eqnarray}
\forall J\in { \mathfrak{J}}_{\overline{{ \mathfrak{S}}}},\nexists \omega\in \overline{{ \mathfrak{S}}} & \vert & \{\omega,\omega^\star\}\sqsubseteq J.\label{defJSbar3b}
\end{eqnarray}
\end{lemma}
\begin{proof}
Let us consider $\sigma,\sigma'\in \Theta^{\overline{{ \mathfrak{S}}}}(\alpha)$. 
If we had $\sigma'\sqsupseteq_{{}_{{{ \mathfrak{S}}}}}\sigma^\star$, $\alpha$ would be a common upper bound of $\sigma$ and $\sigma^\star$ in ${{ \mathfrak{S}}}$. This point would contradict the fact that $\neg\widehat{\sigma\sigma^\star}{}^{{}^{{{ \mathfrak{S}}}}}$ by the definition of the $\star$ map. \\
The analogous formulation (\ref{defJSbar3b}) is derived from (\ref{defJSbar3}) using the order-reversing property of the star. 
\end{proof}

\begin{lemma}\label{lemmaposetJSbar}
The pre-order $\sqsubseteq$ defined on ${ \mathcal{P}}(\overline{{ \mathfrak{S}}})$ is in fact a partial order on ${ \mathfrak{J}}_{\overline{{ \mathfrak{S}}}}$.
\end{lemma}
\begin{proof}
Let us consider $J$ and $J'$ two elements of ${ \mathfrak{J}}_{\overline{{ \mathfrak{S}}}}$ and let us assume that $J\sqsubseteq J'$ and $J'\sqsubseteq J$. Due to $J\sqsubseteq J'$, we deduce that for any $ \sigma\in J$ there exists $\sigma'\in J'$ such that $\sigma \sqsubseteq_{{}_{\overline{{ \mathfrak{S}}}}} \sigma'$. For this $\sigma'$, the property $J'\sqsubseteq J$ implies that there exists $\sigma''\in J$ such that $\sigma' \sqsubseteq_{{}_{\overline{{ \mathfrak{S}}}}} \sigma''$. However, due to the property (\ref{defJSbar2}) we necessarily have $\sigma=\sigma''$, and then $\sigma=\sigma'$ by antisymmetry of $\sqsupseteq_{{}_{\overline{{ \mathfrak{S}}}}}$. This concludes the proof.
\end{proof}

\begin{lemma}\label{infimumJ}
The infimum of two elements $J$ and $J'$ of ${ \mathfrak{J}}_{\overline{{ \mathfrak{S}}}}$, denoted $J\sqcap_{{}_{{ \mathfrak{J}}_{\overline{{ \mathfrak{S}}}}}} J'$ is given as follows
\begin{eqnarray}
J\sqcap_{{}_{{ \mathfrak{J}}_{\overline{{ \mathfrak{S}}}}}} J' &=& Max \left( (\downarrow_{{}_{ \mathfrak{S}}}\bigsqcup{}^{{}^{{{ \mathfrak{S}}}}}J) \cap (\downarrow_{{}_{ \mathfrak{S}}}\bigsqcup{}^{{}^{{{ \mathfrak{S}}}}}J') \cap {\overline{{ \mathfrak{S}}}}
\right)\\
&=& J\sqcap J'\\
&=& Max \{\, \sigma\sqcap_{{}_{{\overline{ \mathfrak{S}}}}}\sigma'\;\vert\; \sigma\in J,\sigma'\in J'\,\}.\label{infimumJprop}
\end{eqnarray}
\end{lemma}
\begin{proof}
Let us consider $J$ and $J'$ two elements of ${ \mathfrak{J}}_{\overline{{ \mathfrak{S}}}}$ (i.e. $cl^{\overline{{ \mathfrak{S}}}}(J)=J$ and $cl^{\overline{{ \mathfrak{S}}}}(J')=J'$). We first note that $(J\sqcap J') \sqsubseteq J$ and $J$ is in ${\widehat{{\mathcal{P}}(\overline{{ \mathfrak{S}}})}}$, and then, using Lemma \ref{PastlowersubsetP} we deduce that $(J\sqcap J')$ is in ${\widehat{{\mathcal{P}}(\overline{{ \mathfrak{S}}})}}$.\\
Now, let us consider $J''$ in ${ \mathfrak{J}}_{\overline{{ \mathfrak{S}}}}$ such that $J''\sqsubseteq_{{}_{{ \mathfrak{J}}_{\overline{{ \mathfrak{S}}}}}}J$ and $J''\sqsubseteq_{{}_{{ \mathfrak{J}}_{\overline{{ \mathfrak{S}}}}}}J'$. We have that, for any $\sigma''\in J''$, there exists $\sigma\in J$,$\sigma'\in J'$ such that $\sigma''\sqsubseteq_{{}_{{\overline{{ \mathfrak{S}}}}}}\sigma$ and $\sigma''\sqsubseteq_{{}_{{\overline{{ \mathfrak{S}}}}}}\sigma'$, and then such that $\sigma''\sqsubseteq_{{}_{{\overline{{ \mathfrak{S}}}}}}(\sigma\sqcap_{{}_{{\overline{{ \mathfrak{S}}}}}}\sigma')$. As a result, we have $J''\sqsubseteq Max\{\, \sigma\sqcap_{{}_{{\overline{ \mathfrak{S}}}}}\sigma'\;\vert\; \sigma\in J,\sigma'\in J'\,\}$. 
\end{proof}

\begin{lemma}\label{supremumJ}
Let $I$ and $I'$ be two elements of ${ \mathfrak{J}}_{\overline{{ \mathfrak{S}}}}$ such that there exists $I''\in { \mathfrak{J}}_{\overline{{ \mathfrak{S}}}}$ with $I\sqsubseteq I''$ and $I'\sqsubseteq I''$.\\
The supremum of $I$ and $I'$ in ${{ \mathfrak{J}}_{\overline{{ \mathfrak{S}}}}}$ (denoted $I\sqcup_{{}_{{ \mathfrak{J}}_{\overline{{ \mathfrak{S}}}}}} I'$) exists and is simply given by
\begin{eqnarray}
I\sqcup_{{}_{{ \mathfrak{J}}_{\overline{{ \mathfrak{S}}}}}} I' &=& cl^{\overline{{ \mathfrak{S}}}}(I \cup I').
\end{eqnarray}
\end{lemma}
\begin{proof}
Let $I$ and $I'$ be two elements of ${ \mathfrak{J}}_{\overline{{ \mathfrak{S}}}}$.  Because of the down-completeness of ${ \mathfrak{J}}_{\overline{{ \mathfrak{S}}}}$ and of the fact that $I$ and $I'$ admit a common-upper-bound, the supremum of $I$ and $I'$ do exist.\\
The supremum $I\sqcup_{{}_{{ \mathfrak{J}}_{\overline{{ \mathfrak{S}}}}}} I'$ must be the smallest closed element upper-bounding for $\sqsubseteq$ the elements $I$ and $I'$. This supremum is then given by $cl^{\overline{{ \mathfrak{S}}}}(I \cup I')$.  Note that the expression $cl^{\overline{{ \mathfrak{S}}}}(I \cup I')$ is well defined because of the existence of $I''$.
\end{proof}

\begin{theorem}
The map $\Theta^{\overline{{ \mathfrak{S}}}}$ is an isomorphism from ${ \mathfrak{S}}$ to ${ \mathfrak{J}}_{\overline{{ \mathfrak{S}}}}$. 
\end{theorem}
\begin{proof}
In order to prove the injective property for $\Theta^{\overline{{ \mathfrak{S}}}}$, it suffices to note that, for any $\sigma,\sigma'\in { \mathfrak{S}}$, the property $\Theta^{\overline{{ \mathfrak{S}}}}(\sigma)=\Theta^{\overline{{ \mathfrak{S}}}}(\sigma')$ implies $\Lambda^{\overline{{ \mathfrak{S}}}}\circ \Theta^{\overline{{ \mathfrak{S}}}}(\sigma)=\Lambda^{\overline{{ \mathfrak{S}}}}\circ \Theta^{\overline{{ \mathfrak{S}}}}(\sigma')$, i.e.  $\sigma=\sigma'$ (see Lemma \ref{lemmaLambda}).\\
The surjectivity property for $\Theta^{\overline{{ \mathfrak{S}}}}$ is a consequence of the injectivity of $\Lambda^{\overline{{ \mathfrak{S}}}}$ on ${ \mathfrak{J}}_{\overline{{ \mathfrak{S}}}}$.\\
As a left component of a Galois connection, the map $\Theta^{\overline{{ \mathfrak{S}}}}$ preserves infima. 
\end{proof}

\begin{lemma} 
$\{\,J\in { \mathfrak{J}}_{\overline{{ \mathfrak{S}}}} \;\vert\; Card(J)=1\,\}$ is isomorphic to $\overline{{ \mathfrak{S}}}$. In other words,
\begin{eqnarray}
\{\,J\in { \mathfrak{J}}_{\overline{{ \mathfrak{S}}}} \;\vert\; Card(J)=1\,\}& \cong & \overline{{ \mathfrak{S}}}.
\end{eqnarray}
and for any $\sigma,\sigma'\in {\overline{{ \mathfrak{S}}}}$, we have
\begin{eqnarray}
\{\sigma\}\sqcap_{{}_{{ \mathfrak{J}}_{\overline{{ \mathfrak{S}}}}}} \{\sigma'\} & = & \{\sigma\sqcap_{{}_{{{ \mathfrak{S}}'}}}\sigma'\}.
\end{eqnarray}
\end{lemma}
\begin{proof}
If $\sigma\in \overline{{ \mathfrak{S}}}$ we have obviously $\Theta^{\overline{{ \mathfrak{S}}}}(\sigma)=\{\sigma\}$. 
Conversely, Let us assume that $Card(\Theta^{\overline{{ \mathfrak{S}}}}(\sigma))=1$ and let us suppose that  $\sigma\in {{ \mathfrak{S}}}\smallsetminus \overline{{ \mathfrak{S}}}$. We denote $\sigma'$ the unique element of $\Theta^{\overline{{ \mathfrak{S}}}}(\sigma)$. We then have $\sigma\not=\sigma'$. However, we observe immediately that, for any $\alpha\in \overline{{ \mathfrak{S}}}$, $\epsilon^{{ \mathfrak{S}}}_{{\mathfrak{l}}_{{}_{(\alpha,\alpha^\star)}}}(\sigma)=\epsilon^{{ \mathfrak{S}}}_{{\mathfrak{l}}_{{}_{(\alpha,\alpha^\star)}}}(\sigma')$. Then, using (\ref{realChuseparated}) we deduce $\sigma=\sigma'$ which contradicts the assumption. We then obtain $\sigma\in \overline{{ \mathfrak{S}}}$.\\
The second property is trivial to check.
\end{proof}

\begin{lemma}\label{lemmatricky}
We have necessarily, for any $I$ subset of $\overline{{ \mathfrak{S}}}$
\begin{eqnarray}
I\in { \mathfrak{J}}_{\overline{{ \mathfrak{S}}}}&\Rightarrow & I = Max\{\,\bigsqcup{}^{{}^{{\overline{{ \mathfrak{S}}}}}} J\;\vert\; J\sqsubseteq I\;\textit{\rm and}\;\widehat{\;J\;}{}^{{}^{{\overline{{ \mathfrak{S}}}}}}\}.
\end{eqnarray}
\end{lemma}
\begin{proof} Let us consider $I\in { \mathfrak{J}}_{\overline{{ \mathfrak{S}}}}$.\\
Let $J$ be a subset of ${\overline{{ \mathfrak{S}}}}$ satisfying $J\sqsubseteq I$ and $\widehat{\;J\;}{}^{{}^{{\overline{{ \mathfrak{S}}}}}}$. We have then $J\in {\widehat{{\mathcal{P}}(\overline{{ \mathfrak{S}}})}}$ (we can then apply $cl^{\overline{{ \mathfrak{S}}}}$ on $J$). We have then also $\{\bigsqcup{}^{{}^{{\overline{{ \mathfrak{S}}}}}} J\}=cl^{\overline{{ \mathfrak{S}}}}(J)\sqsubseteq cl^{\overline{{ \mathfrak{S}}}}(I)=I$. Hence, $Max\{\,\bigsqcup{}^{{}^{{\overline{{ \mathfrak{S}}}}}} J\;\vert\; J\sqsubseteq I\;\textit{\rm and}\;\widehat{\;J\;}{}^{{}^{{\overline{{ \mathfrak{S}}}}}}\}\sqsubseteq I$. \\
On another part, for any $\sigma\in I$ we have $\{\sigma\}\sqsubseteq I$ and $\widehat{\;\sigma\;}{}^{{}^{{\overline{{ \mathfrak{S}}}}}}$, and then $\{\sigma\} \sqsubseteq \{\,\bigsqcup{}^{{}^{{\overline{{ \mathfrak{S}}}}}} J\;\vert\; J\sqsubseteq I\;\textit{\rm and}\;\widehat{\;J\;}{}^{{}^{{\overline{{ \mathfrak{S}}}}}}\}$. As a result, we have then $I \sqsubseteq Max\{\,\bigsqcup{}^{{}^{{\overline{{ \mathfrak{S}}}}}} J\;\vert\; J\sqsubseteq I\;\textit{\rm and}\;\widehat{\;J\;}{}^{{}^{{\overline{{ \mathfrak{S}}}}}}\}$. \\ 
As long as we have $\forall \sigma,\sigma'\in Max\{\,\bigsqcup{}^{{}^{{\overline{{ \mathfrak{S}}}}}} J\;\vert\; J\sqsubseteq I\;\textit{\rm and}\;\widehat{\;J\;}{}^{{}^{{\overline{{ \mathfrak{S}}}}}}\}$ the property $\neg \widehat{\;\sigma \sigma'\;}{}^{{}^{{\overline{{ \mathfrak{S}}}}}}$, we can use the antisymmetry of $\sqsubseteq$ and we have then $I = Max\{\,\bigsqcup{}^{{}^{{\overline{{ \mathfrak{S}}}}}} J\;\vert\; J\sqsubseteq I\;\textit{\rm and}\;\widehat{\;J\;}{}^{{}^{{\overline{{ \mathfrak{S}}}}}}\}$.
\end{proof}
Let us then deduce from previous lemma a fundamental formula that will be useful in many parts of our construction.
\begin{lemma}
Let us denote by $U:=\Theta^{\overline{ \mathfrak{S}}}(\xi)\in { \mathfrak{J}}_{\overline{{ \mathfrak{S}}}}\;\;$ for any $\xi\in { \mathfrak{S}}$. We then have
\begin{eqnarray}
\forall \kappa\in {\overline{{ \mathfrak{S}}}}, \exists \varphi \in U & \vert & \forall \chi\in U,\; (\kappa \sqcap_{{}_{{\overline{{ \mathfrak{S}}}}}} \chi) 
\sqsubseteq_{{}_{{\overline{{ \mathfrak{S}}}}}}
(\kappa \sqcap_{{}_{{\overline{{ \mathfrak{S}}}}}} \varphi).\;\;\;\;\;\;\;\;\;\;\;\label{fundamentalformula}
\end{eqnarray}
\end{lemma}
\begin{proof}
Let us first remark that, if we denote $U:=\Theta^{\overline{ \mathfrak{S}}}(\xi)$, we have in Lemma \ref{lemmatricky} that $Max \{\, (\bigsqcup{}^{{}^{{\overline{{ \mathfrak{S}}}}}} V)\;\vert\; V\sqsubseteq U, \;\widehat{\;V\;}{}^{{}^{{\overline{{ \mathfrak{S}}}}}}\,\}=U$.\\
Let us now fix $\kappa\in {\overline{{ \mathfrak{S}}}}$.  Let us denote by $X_U$ the subset $\{\,\kappa \sqcap_{{}_{{\overline{{ \mathfrak{S}}}}}}\omega\;\vert\; \omega\in U\,\}$. We note that $\widehat{\;X_U\;}{}^{{}^{{\overline{{ \mathfrak{S}}}}}}$ because $(\kappa \sqcap_{{}_{{\overline{{ \mathfrak{S}}}}}}\chi) \sqsubseteq_{{}_{{\overline{{ \mathfrak{S}}}}}} \kappa$ for any $\chi\in U$ and we have then
\begin{eqnarray}
\bigsqcup{}^{{}^{{\overline{{ \mathfrak{S}}}}}} X_U & \sqsubseteq_{{}_{{\overline{{ \mathfrak{S}}}}}} & \kappa.\label{eq1pre}
\end{eqnarray}
We also note that
\begin{eqnarray}
\bigsqcup{}^{{}^{{\overline{{ \mathfrak{S}}}}}} X_U=
\bigsqcup{}^{{}^{{\overline{{ \mathfrak{S}}}}}}\{\,\kappa \sqcap_{{}_{{\overline{{ \mathfrak{S}}}}}}\chi\;\vert\; \chi\in U\,\} &\sqsubseteq & Max \{\, (\bigsqcup{}^{{}^{{\overline{{ \mathfrak{S}}}}}} V)\;\vert\; V\sqsubseteq U, \;\widehat{\;V\;}{}^{{}^{{\overline{{ \mathfrak{S}}}}}}\,\}=U.\;\;\;\;\;\;\;\;\;\;\;\;\;
\end{eqnarray}
In other words, 
\begin{eqnarray}
\exists \varphi \in U & \vert & \bigsqcup{}^{{}^{{\overline{{ \mathfrak{S}}}}}}\{\,\kappa \sqcap_{{}_{{\overline{{ \mathfrak{S}}}}}}\chi\;\vert\; \chi\in U\,\} \sqsubseteq_{{}_{{\overline{{ \mathfrak{S}}}}}} \varphi. \label{eq2pre}
\end{eqnarray}
From (\ref{eq1pre}) and (\ref{eq2pre}), we have then obtained
\begin{eqnarray}
\exists \varphi \in U & \vert & \bigsqcup{}^{{}^{{\overline{{ \mathfrak{S}}}}}}\{\,\kappa \sqcap_{{}_{{\overline{{ \mathfrak{S}}}}}}\chi\;\vert\; \chi\in U\,\} \sqsubseteq_{{}_{{\overline{{ \mathfrak{S}}}}}} (\kappa \sqcap_{{}_{{\overline{{ \mathfrak{S}}}}}} \varphi). 
\end{eqnarray}
This concludes the proof.
\end{proof}


\subsection{Reconstruction of hidden states}\label{subsectionreconstruction}

We will now proceed in the inverse direction, and we will show that the set of hidden states can be reconstructed from the set of real states as soon as this set of real states satisfies some minimal requirements.

\begin{definition}\label{definitionrealspaceofstates}
We will define {\em a real space of states} to be a quadruple $({\mathfrak{S}}',{ \mathcal{Q}}({ \mathfrak{S}}'),cl^{{{ \mathfrak{S}}'}},\star)$ such that
\begin{itemize}
\item ${ \mathfrak{S}}'$ is a down complete Inf semi-lattice with bottom element $\bot_{{}_{ \mathfrak{S}'}}$.  We require ${ \mathfrak{S}}'$ to be generated by its maximal elements, i.e. 
 \begin{eqnarray}
\hspace{-1.5cm}&&\forall \sigma \in { \mathfrak{S}}', \;\; \sigma= \bigsqcap{}^{{}^{{ \mathfrak{S}}'}}  \underline{\sigma}_{{}_{{ \mathfrak{S}}'}} \;\;\textit{\rm where}\;\;
\underline{\sigma}_{{}_{ { \mathfrak{S}}'}}:=\{\, \sigma'\in { \mathfrak{S}'}{}^{{}^{pure}}\;\vert\; \sigma'\sqsupseteq_{{}_{{ \mathfrak{S}'}}} \sigma\;\}
\;\;\textit{\rm and}\;\; { \mathfrak{S}'}{}^{{}^{pure}}:=Max({ \mathfrak{S}'})\;\;\;\;\;\;\;\;\;\;\;\;\;\;
\label{RECcompletemeetirreducibleordergeneratingS}
\end{eqnarray}
\item if we denote by $({ \mathcal{P}}({ \mathfrak{S}}'),\sqsubseteq)$ the pre-ordered set defined by
\begin{eqnarray}
\forall I,I'\in { \mathcal{P}}({ \mathfrak{S}}'),&& I\sqsubseteq I' \;\;:\Leftrightarrow \;\; \forall \sigma\in I,\exists \sigma'\in I'\;\vert\; \sigma \sqsubseteq_{{}_{{{ \mathfrak{S}}'}}} \sigma'\label{RECDefinitionpreorderPSbar}
\end{eqnarray}
then ${ \mathcal{Q}}({ \mathfrak{S}}')$ is a given sub-poset of $({ \mathcal{P}}({ \mathfrak{S}}'),\sqsubseteq)$ satisfying
\begin{eqnarray}
&&\forall V\subseteq { \mathfrak{S}}',\;\; \widehat{\;V\;}{}^{{}^{{ \mathfrak{S}}'}}\;\;\Rightarrow\;\;V\in  { \mathcal{Q}}({ \mathfrak{S}}').\label{RECQ1}
\end{eqnarray}
\item $\star$ is a map from ${ \mathfrak{S}}'\smallsetminus \{\bot_{{}_{ \mathfrak{S}'}}\}$ to ${ \mathfrak{S}}'\smallsetminus \{\bot_{{}_{ \mathfrak{S}'}}\}$ satisfying
\begin{eqnarray}
\forall \sigma\in { \mathfrak{S}}'\smallsetminus \{\bot_{{}_{ \mathfrak{S}'}}\},&& (\sigma^{\star})^{\star}=\sigma, \label{RECinvolutive}\\
\forall \sigma_1,\sigma_2\in { \mathfrak{S}}'\smallsetminus \{\bot_{{}_{ \mathfrak{S}'}}\},&& \sigma_1\sqsubseteq_{{}_{{ \mathfrak{S}'}}}\sigma_2\;\;\Rightarrow\;\; \sigma_2^\star \sqsubseteq_{{}_{{ \mathfrak{S}'}}} \sigma_1^\star, \label{RECorderreversing}\\
\forall J\in { \mathcal{Q}}({ \mathfrak{S}}'),\forall \sigma\in { \mathfrak{S}}', & & \{\sigma,\sigma^\star\}\not\sqsubseteq J. \label{RECinconsistent}
\end{eqnarray}
\item $cl^{{{ \mathfrak{S}}'}}$ is a map defined from ${ \mathcal{P}}({ \mathfrak{S}}')$ to ${ \mathcal{P}}({ \mathfrak{S}}')$ and satisfying 
\begin{eqnarray}
\forall J\in {\mathcal{Q}}({ \mathfrak{S}}'),\;&& cl^{{ \mathfrak{S}}'}(J)\in {\mathcal{Q}}({ \mathfrak{S}}'),\label{basicRECQ}\\
\forall J\in { \mathcal{P}}({ \mathfrak{S}}'),&& J\sqsubseteq cl^{{{ \mathfrak{S}}'}}(J),\label{kuratowski1}\\
\forall J,J'\in { \mathcal{P}}({ \mathfrak{S}}'),&& J\sqsubseteq J' \;\;\Rightarrow\;\;cl^{{{ \mathfrak{S}}'}}(J)\sqsubseteq cl^{{{ \mathfrak{S}}'}}(J'),\label{kuratowski2}\\
\forall J\in { \mathcal{P}}({ \mathfrak{S}}'),&& cl^{{{ \mathfrak{S}}'}}\circ cl^{{{ \mathfrak{S}}'}}(J) = cl^{{{ \mathfrak{S}}'}}(J),\label{kuratowski3}\\
\forall \sigma\in { \mathfrak{S}}',&& cl^{{{ \mathfrak{S}}'}}(\{\sigma\})=\{\sigma\},\\
\forall J\in cl^{{{ \mathfrak{S}}'}}({ \mathcal{Q}}({ \mathfrak{S}}')),\forall \sigma,\sigma'\in J,&& \widehat{\sigma\sigma'}{}^{{}^{{{ \mathfrak{S}}'}}}\;\;\Rightarrow\;\; \sigma=\sigma' ,\label{RECdefJSbar2}\\
\forall J,J'\in { \mathcal{P}}({ \mathfrak{S}}'),&&(\,J\sqsubseteq J'\;\;\textit{\rm and}\;\; J'\in cl^{{{ \mathfrak{S}}'}}({ \mathcal{Q}}({ \mathfrak{S}}'))\,)\;\Rightarrow\; J\in { \mathcal{Q}}({ \mathfrak{S}}').\label{RECQ2}\;\;\;\;\;\;\;\;\;\;\;\;
\end{eqnarray}
\end{itemize}
\end{definition}
\noindent From here, we will denote by ${ \mathfrak{J}}_{{{ \mathfrak{S}}'}}$ the set $cl^{{{ \mathfrak{S}}'}}({ \mathcal{Q}}({ \mathfrak{S}}'))$.  The set ${ \mathfrak{J}}_{{{ \mathfrak{S}}'}}$ is naturally equipped with $\sqsubseteq$ which is a partial order on ${ \mathfrak{J}}_{{{ \mathfrak{S}}'}}$ because of the condition (\ref{RECdefJSbar2}).\\

We will adopt the following Inf semi-lattice structure on ${ \mathfrak{J}}_{{{ \mathfrak{S}}'}}$ :
\begin{eqnarray}
\forall J,J'\in { \mathfrak{J}}_{{{ \mathfrak{S}}'}}, \;\;\;\; J\sqcap_{{}_{{ \mathfrak{J}}_{{{ \mathfrak{S}}'}}}} J' &=&J\sqcap J'\\
&=& Max \{\, \sigma\sqcap_{{}_{{\overline{ \mathfrak{S}}}}}\sigma'\;\vert\; \sigma\in J,\sigma'\in J'\,\}.\;\;\;\;\;\;\;\;\;\;\;\;\;\;\;\label{RECdefJsqbarJ}
\end{eqnarray}
Indeed, we first note that $(J\sqcap J')\in { \mathcal{Q}}({ \mathfrak{S}}')$ because  $(J\sqcap J')\sqsubseteq J\in { \mathfrak{J}}_{{{ \mathfrak{S}}'}}$ and because of the property (\ref{RECQ2}). We can then form the entity $cl^{{{ \mathfrak{S}}'}}(J \sqcap J')$. Then, we have (i) $(J\sqcap J')\sqsubseteq J,J'$ implies, using (\ref{kuratowski2}) and the fact that $J$ and $J'$ in ${ \mathfrak{J}}_{{{ \mathfrak{S}}'}}$, $cl^{{{ \mathfrak{S}}'}}(J \sqcap J')\sqsubseteq J,J'$ and then $cl^{{{ \mathfrak{S}}'}}(J \sqcap J')\sqsubseteq J\sqcap J'$, (ii) using (\ref{kuratowski1}) we have also $cl^{{{ \mathfrak{S}}'}}(J \sqcap J')\sqsupseteq (J\sqcap J')$.  Now, we conclude by observing that $\forall \sigma \in cl^{{{ \mathfrak{S}}'}}(J \sqcap J'), \exists \kappa \in J\sqcap J'$ such that $\sigma\sqsubseteq_{{}_{{{ \mathfrak{S}}'}}}\kappa$ and for this $\kappa$ there exists $\omega\in cl^{{{ \mathfrak{S}}'}}(J \sqcap J')$ such that $\kappa\sqsubseteq_{{}_{{{ \mathfrak{S}}'}}}\omega$. But now we use (\ref{RECdefJSbar2}), to deduce that $\sigma=\kappa=\omega$. In other words, $cl^{{{ \mathfrak{S}}'}}(J \sqcap J')= J\sqcap J'$.  Therefore, $J\sqcap_{{}_{{ \mathfrak{J}}_{{{ \mathfrak{S}}'}}}} J'$ being the largest closed element lower-bounding for $\sqsubseteq$ the subsets $J$ and $J'$, we then deduce the first equality.\\
The second equality has already been proved. \\


We note that ${ \mathfrak{S}}'$ can be injected as an Inf semi-lattice in ${ \mathfrak{J}}_{{{ \mathfrak{S}}'}}$ by the trivial map $\sigma\in { \mathfrak{S}}' \;\mapsto\; \{\sigma\}\in { \mathfrak{J}}_{{{ \mathfrak{S}}'}}$ (because of requirement (\ref{RECQ1})).  Explicitly, we have, for any $\sigma,\sigma'\in {{{ \mathfrak{S}}'}}$
\begin{eqnarray}
\{\sigma\}\sqcap_{{}_{{ \mathfrak{J}}_{{{ \mathfrak{S}}'}}}} \{\sigma'\} & = & \{\sigma\sqcap_{{}_{{{ \mathfrak{S}}'}}}\sigma'\}.
\end{eqnarray}
In the following, we will denote shortly by $\sigma$ the element $\{\sigma\}$ of ${ \mathfrak{J}}_{{{ \mathfrak{S}}'}}$ if necessary.\\
The bottom element in ${ \mathfrak{J}}_{{{ \mathfrak{S}}'}}$ is simply given by 
\begin{eqnarray}
\bot_{{}_{{ \mathfrak{J}}_{{{ \mathfrak{S}}'}}}}=\bot_{{}_{{{{ \mathfrak{S}}'}}}}.
\end{eqnarray}

${ \mathfrak{J}}_{{{ \mathfrak{S}}'}}$ being a down complete Inf semi-lattice, the supremum of a collection of elements of ${ \mathfrak{J}}_{{{ \mathfrak{S}}'}}$ can be defined as long as they admit a common upper-bound. 
More explicitly, let $J$ and $J'$ be two elements of ${ \mathfrak{J}}_{{{ \mathfrak{S}}'}}$ such that $\exists J''\in {{{ \mathfrak{J}}_{{{ \mathfrak{S}}'}}}}$ with $J\sqsubseteq J''$ and $J'\sqsubseteq J''$, then the supremum of $J$ and $J'$ in ${{ \mathfrak{J}}_{{{ \mathfrak{S}}'}}}$ (denoted $J\sqcup_{{}_{{ \mathfrak{J}}_{{{ \mathfrak{S}}'}}}} J'$) exists in ${{ \mathfrak{J}}_{{{ \mathfrak{S}}'}}}$.  \\
We note that this supremum is simply given by :
\begin{eqnarray}
J\sqcup_{{}_{{ \mathfrak{J}}_{{{ \mathfrak{S}}'}}}} J' &=& cl^{{{ \mathfrak{S}}'}}(J \cup J').\label{supremaJ}
\end{eqnarray}
Indeed, we note first that, because $J\sqsubseteq J''$ and $J'\sqsubseteq J''$ we have $(J \cup J')\sqsubseteq J''$ and using $J''\in {{{ \mathfrak{J}}_{{{ \mathfrak{S}}'}}}}$ and property (\ref{RECQ2}) we deduce that $(J \cup J')\in { \mathcal{Q}}({ \mathfrak{S}}')$. We can then form the entity $cl^{{{ \mathfrak{S}}'}}(J \cup J')$. Now, we observe that, if $J\sqsubseteq J''$ and $J'\sqsubseteq J''$, then $cl^{{{ \mathfrak{S}}'}}(J \cup J')\sqsubseteq cl^{{{ \mathfrak{S}}'}}(J'')=J''$. This concludes the proof of the formula (\ref{supremaJ}).\\

We then denote by ${\mathfrak{E}}_{{ \mathfrak{J}}_{{{ \mathfrak{S}}'}}}$ the generalized effect space associated to ${ \mathfrak{J}}_{{{ \mathfrak{S}}'}}$ and defined according to Definition \ref{defgeneralizedeffectspace}. The Inf semi-lattice structure on ${\mathfrak{E}}_{{ \mathfrak{J}}_{{{ \mathfrak{S}}'}}}$ is defined according to (\ref{defcapES}).\\

Let us also define $\overline{ \mathfrak{E}}_{{ \mathfrak{J}}_{{{ \mathfrak{S}}'}}}$ to be the sub Inf semi-lattice of ${ \mathfrak{E}}_{{ \mathfrak{J}}_{{{ \mathfrak{S}}'}}}$ formed by the elements of the following set
\begin{eqnarray}
\hspace{-0.5cm}\{\,{ \mathfrak{l}}_{(\sigma,\sigma')}\;\vert\; \sigma,\sigma'\in { \mathfrak{S}}'\smallsetminus \{\bot_{{}_{{ \mathfrak{S}'}}}\}, \sigma'\sqsupseteq_{{}_{{ \mathfrak{S}'}}} \sigma^\star\,\} \cup \{ { \mathfrak{l}}_{{}_{(\sigma,\centerdot)}}\;\vert\; \sigma\in { \mathfrak{S}'}\;\}\cup \{ { \mathfrak{l}}_{{}_{(\centerdot,\sigma)}}\;\vert\; \sigma\in { \mathfrak{S}'}\;\}\cup \{\, { \mathfrak{l}}_{{}_{(\centerdot,\centerdot)}}\,\}\;\;\;\;\;\label{reducedspaceofeffectsprime}
\end{eqnarray}
We note that, for any $\sigma,\sigma'\in { \mathfrak{S}}'\smallsetminus \{\bot_{{}_{ \mathfrak{S}'}}\}$ such that $\sigma'\sqsupseteq_{{}_{{ \mathfrak{S}'}}}\sigma^\star$, the effect ${ \mathfrak{l}}_{{}_{(\sigma,\sigma')}}$ is well defined because of the property $\neg\; \widehat{\sigma \sigma^\star}{}^{{}^{{ \mathfrak{J}}_{ \mathfrak{S}'}}}$ (see condition (\ref{RECinconsistent})).\\

\noindent Endly, we can define an evaluation map denoted $\epsilon^{{ \mathfrak{J}}_{{{ \mathfrak{S}}'}}}$ as in (\ref{defepsilonS}).\\


The homomorphic properties of the evaluation map, i.e.
\begin{eqnarray}
\forall \{\,{ \mathfrak{l}}_i\;\vert\; i\in I\,\}\subseteq \overline{ \mathfrak{E}}_{{ \mathfrak{J}}_{{{ \mathfrak{S}}'}}},\forall J\in { \mathfrak{J}}_{{{ \mathfrak{S}}'}},&& \epsilon^{{ \mathfrak{J}}_{{{ \mathfrak{S}}'}}}_{\bigsqcap{}_{i\in I}^{{\overline{ \mathfrak{E}}_{{ \mathfrak{J}}_{{{ \mathfrak{S}}'}}}}}{ \mathfrak{l}}_i}(J)=\bigwedge{}_{\!\!i\in I}\;\epsilon^{{ \mathfrak{J}}_{{{ \mathfrak{S}}'}}}_{{ \mathfrak{l}}_i}(J),\\
\forall { \mathfrak{l}}\in \overline{ \mathfrak{E}}_{{ \mathfrak{J}}_{{{ \mathfrak{S}}'}}},\forall \{\,J_i\;\vert\;i\in I\,\}\subseteq { \mathfrak{J}}_{{{ \mathfrak{S}}'}},&& \epsilon^{{ \mathfrak{J}}_{{{ \mathfrak{S}}'}}}_{{ \mathfrak{l}}}(\bigsqcap{}^{{}^{{ \mathfrak{J}}_{{{ \mathfrak{S}}'}}}}_{i\in I}J_i)=\bigwedge{}_{\!\!i\in I}\;\epsilon^{{ \mathfrak{J}}_{{{ \mathfrak{S}}'}}}_{{ \mathfrak{l}}}(J_i),
\end{eqnarray}
are trivially true (this is a consequence of the construction of subsection \ref{subsectiongeneralized}).\\

The property (\ref{realChuextensional}), i.e.
\begin{eqnarray}
&& \forall { \mathfrak{l}},{ \mathfrak{l}}'\in \overline{ \mathfrak{E}}_{{ \mathfrak{J}}_{{{ \mathfrak{S}}'}}},\;\;\;\;\;\;\;\;\;\;\;\;(\, \forall S\in { \mathfrak{J}}_{{{ \mathfrak{S}}'}},\; { \epsilon}^{{ \mathfrak{J}}_{{{ \mathfrak{S}}'}}}_{ { \mathfrak{l}}}(S)={ \epsilon}^{{ \mathfrak{J}}_{{{ \mathfrak{S}}'}}}_{ { \mathfrak{l}}'}(S) \,) \Leftrightarrow  (\, { \mathfrak{l}}= { \mathfrak{l}}' \,),\label{RECChuextensional}
\end{eqnarray}
is also trivially true (this is a consequence of the construction of subsection \ref{subsectiongeneralized}).\\

We have to check now the property (\ref{realChuseparated}), i.e.
\begin{eqnarray}
&& \forall J,J'\in { \mathfrak{J}}_{{{ \mathfrak{S}}'}},\;\;\;\;\;\;\;\;\;\;\;\;(\, \forall { \mathfrak{l}}\in \overline{ \mathfrak{E}}_{{ \mathfrak{J}}_{{{ \mathfrak{S}}'}}},\; { \epsilon}^{{ \mathfrak{J}}_{{{ \mathfrak{S}}'}}}_{ { \mathfrak{l}}}(J)={ \epsilon}^{{ \mathfrak{J}}_{{{ \mathfrak{S}}'}}}_{ { \mathfrak{l}}}(J') \,) \Leftrightarrow  (\, J= J' \,).\label{RECChuseparated}
\end{eqnarray}
Let us fix $J,J'\in { \mathfrak{J}}_{{{ \mathfrak{S}}'}}$, we can exploit the property $\forall { \mathfrak{l}}\in \overline{ \mathfrak{E}}_{{ \mathfrak{J}}_{{{ \mathfrak{S}}'}}},\; { \epsilon}^{{ \mathfrak{J}}_{{{ \mathfrak{S}}'}}}_{ { \mathfrak{l}}}(J)={ \epsilon}^{{ \mathfrak{J}}_{{{ \mathfrak{S}}'}}}_{ { \mathfrak{l}}}(J')$ by choosing successively all the effects ${ \mathfrak{l}}_{(\sigma,\sigma^\star)}$ for the different elements $\sigma\in J$.  For any $\sigma\in J$ we have ${ \epsilon}^{{ \mathfrak{J}}_{{{ \mathfrak{S}}'}}}_{ { \mathfrak{l}}_{(\sigma,\sigma^\star)}}(J)=\textit{\bf Y}$ and then ${ \epsilon}^{{{{ \mathfrak{S}}'}}}_{ { \mathfrak{l}}_{(\sigma,\sigma^\star)}}(J')=\textit{\bf Y}$, i.e.  there exists $\sigma'\in J'$ such that $\sigma\sqsubseteq_{{}_{{ \mathfrak{S}}'}}\sigma'$.  As a result, we then obtain $J\sqsubseteq J'$.  By choosing now all the effects ${ \mathfrak{l}}_{(\sigma,\sigma^\star)}$ for the different elements $\sigma\in J'$, we obtain in the same way $J\sqsubseteq J'$. As a conclusion, $J=J'$ (because ${ \mathfrak{J}}_{{{ \mathfrak{S}}'}}$ is a partially ordered set).\\

Endly, we observe that $\neg \widehat{\sigma \sigma^\star}{}^{{}^{{ \mathfrak{J}}_{{{ \mathfrak{S}}'}}}}$ because of condition (\ref{RECinconsistent}).\\

We then conclude on the following theorem establishing the way to reconstruct the set of hidden states from the set of real states.
\begin{theorem}\label{theoremcompletion}
We will assume that $({\mathfrak{S}}',{ \mathcal{Q}}({ \mathfrak{S}}'),cl^{{{ \mathfrak{S}}'}},\star)$ is a real space of states.\\
We define ${ \mathfrak{J}}_{{{ \mathfrak{S}}'}}$ to be the set $cl^{{{ \mathfrak{S}}'}}({ \mathcal{Q}}({ \mathfrak{S}}'))$. The Inf semi-lattice structure being given by (\ref{RECdefJsqbarJ}). We define the generalized space of effects as in Definition \ref{defgeneralizedeffectspace} (it is necessary to use the formula for the suprema established in (\ref{supremaJ}) ). The reduced space of effects is defined as in (\ref{reducedspaceofeffectsprime}). Then, we define an evaluation map ${ \epsilon}^{{ \mathfrak{J}}_{{{ \mathfrak{S}}'}}}$ on the whole ${ \mathfrak{J}}_{{{ \mathfrak{S}}'}}$ as in (\ref{defepsilonS}). 
\\
As a result, ${ \mathfrak{J}}_{{{ \mathfrak{S}}'}}$ is a well defined space of states which admits the real structure $({{{ \mathfrak{S}}'}},\star)$. The elements of ${{{ \mathfrak{S}}'}}$ are the real states, whereas the elements of 
  ${ \mathfrak{J}}_{{{ \mathfrak{S}}'}}\smallsetminus {{{ \mathfrak{S}}'}}$ are the hidden states.
\end{theorem}

Let us now establish a result showing that there exists no hidden states in the completion ${ \mathfrak{J}}_{{{ \mathfrak{S}}'}}$ when ${\mathfrak{S}}'$ is a simplex.

\begin{theorem}\label{completionsimplex}
Let $({\mathfrak{S}}',{ \mathcal{Q}}({ \mathfrak{S}}'),cl^{{{ \mathfrak{S}}'}},\star)$ be a real space of states. We will assume that ${\mathfrak{S}}'$ is a simplex space of states. Then, ${ \mathfrak{J}}_{{{ \mathfrak{S}}'}}={\mathfrak{S}}'$.
\end{theorem}
\begin{proof}
Let us begin by recalling that, as long as ${\mathfrak{S}}'$ is a simplex space of states, the star map $\star$ is necessarily given by (\ref{starsimplex}).\\
Let us consider two elements $\sigma$ and $\sigma'$ in ${\mathfrak{S}}'$ such that $\neg \;\widehat{\sigma \sigma'}{}^{{}^{{\mathfrak{S}}'}}$.  Let us suppose that $\{\sigma,\sigma'\}\in { \mathcal{Q}}({ \mathfrak{S}}')$.\\
As long as ${\mathfrak{S}}'$ is a simplex, we must necessarily have $\sigma^\star \sqsubseteq_{{}_{{\mathfrak{S}}'}} \sigma'$, i.e. $\{\sigma,\sigma^\star\}\sqsubseteq \{\sigma,\sigma'\}$.  As a consequence, we have $\{\sigma,\sigma^\star\}\in { \mathcal{Q}}({ \mathfrak{S}}')$ which contradicts the condition (\ref{RECinconsistent}). \\ 
As a conclusion, if ${\mathfrak{S}}'$ is a simplex space of states, then ${ \mathfrak{J}}_{{{ \mathfrak{S}}'}}\smallsetminus {\mathfrak{S}}'$ is empty.
\end{proof}

\subsection{Ontic completions}\label{subsectioncompleterealspaceofstates}

For the rest of this subsection, we will assume that ${ \mathfrak{S}}'$ satisfies the Finite Rank Condition (Definition \ref{FiniteRankCondition}).\\

For the rest of the present subsection, we will consider $({ \mathfrak{S}}',\star)$ such that
\begin{itemize}
\item ${ \mathfrak{S}}'$ is a down complete Inf semi-lattice with bottom element $\bot_{{}_{ \mathfrak{S}'}}$, such that  ${ \mathfrak{S}}'$ is generated by its maximal elements, i.e. 
 \begin{eqnarray}
\hspace{-1.5cm}&&\forall \sigma \in { \mathfrak{S}}', \;\; \sigma= \bigsqcap{}^{{}^{{ \mathfrak{S}}'}}  \underline{\sigma}_{{}_{{ \mathfrak{S}}'}} \;\;\textit{\rm where}\;\;
\underline{\sigma}_{{}_{ { \mathfrak{S}}'}}:=\{\, \sigma'\in { \mathfrak{S}'}^{{}^{pure}}\;\vert\; \sigma'\sqsupseteq_{{}_{{ \mathfrak{S}'}}} \sigma\;\}
\;\;\textit{\rm and}\;\; { \mathfrak{S}'}^{{}^{pure}}:=Max({ \mathfrak{S}'})\;\;\;\;\;\;\;\;\;\;\;\;\;\;
\end{eqnarray}
\item $\star$ is a map from ${ \mathfrak{S}}'\smallsetminus \{\bot_{{}_{ \mathfrak{S}'}}\}$ to ${ \mathfrak{S}}'\smallsetminus \{\bot_{{}_{ \mathfrak{S}'}}\}$ satisfying
\begin{eqnarray}
\forall \sigma\in { \mathfrak{S}}'\smallsetminus \{\bot_{{}_{ \mathfrak{S}'}}\},&& (\sigma^{\star})^{\star}=\sigma, \label{propstar1}\\
\forall \sigma_1,\sigma_2\in { \mathfrak{S}}'\smallsetminus \{\bot_{{}_{{ \mathfrak{S}'}}}\},&& \sigma_1\sqsubseteq_{{}_{{ \mathfrak{S}'}}}\sigma_2\;\;\Rightarrow\;\; \sigma_2^\star \sqsubseteq_{{}_{{ \mathfrak{S}'}}} \sigma_1^\star,\label{propstar2}\\
\forall \sigma\in { \mathfrak{S}}'\smallsetminus \{\bot_{{}_{ \mathfrak{S}'}}\},&& \neg \;\widehat{\sigma^{\star}\sigma}{}^{{}^{{ \mathfrak{S}'}}}.\label{propstar3}
\end{eqnarray}
\end{itemize}

We then intent to build the missing elements ${\mathcal{Q}}({ \mathfrak{S}}')$ and $cl^{{ \mathfrak{S}}'}$ verifying the other requirements in Definition \ref{definitionrealspaceofstates}.\\

Let us begin with our first basic definition. 

\begin{definition}
We will consider a map from ${\mathcal{P}}({ \mathfrak{S}}')$ to ${\mathcal{P}}({ \mathfrak{S}}')$ denoted ${\mathfrak{c}}^{{ \mathfrak{S}}'}$ and defined by
\begin{eqnarray}
{\mathfrak{c}}^{{ \mathfrak{S}}'}(U)&:=&Max \{\,\bigsqcup{}^{{}^{{ \mathfrak{S}}'}} V\;\vert\; V\sqsubseteq U\;\;\textit{\rm and}\;\; \widehat{\;V\;}{}^{{}^{{ \mathfrak{S}}'}}\,\}\;\;\;\;\;\;\;\;\;\;\;\;\;\;\label{definitioncompleterealspaceofstates2}
\end{eqnarray}
(here, we are directly inspired by the Lemma \ref{lemmatricky}).
\end{definition}

\begin{lemma}
${\mathfrak{c}}^{{ \mathfrak{S}}'}$ is {\em a pre-closure}. In other words,
\begin{eqnarray}
\forall U\in {\mathcal{P}}({ \mathfrak{S}}'),\;&& U\sqsubseteq {\mathfrak{c}}^{{ \mathfrak{S}}'}(U),\label{propB3}\\
\forall U_1,U_2\in {\mathcal{P}}({ \mathfrak{S}}'),\;&& U_1\sqsubseteq U_2\;\;\Rightarrow\;\; {\mathfrak{c}}^{{ \mathfrak{S}}'}(U_1)\sqsubseteq {\mathfrak{c}}^{{ \mathfrak{S}}'}(U_2).\;\;\;\;\;\;\;\;\;\;\;\;\;\;\label{propB4}
\end{eqnarray}
\end{lemma}
\begin{proof}
The property (\ref{propB3}) is easily deduced from the definition (\ref{definitioncompleterealspaceofstates2}). Indeed, we have, for any $\sigma\in U$, $\{\sigma\}\sqsubseteq U$ and $\widehat{\;\{\sigma\}\;}{}^{{}^{{ \mathfrak{S}}'}}$ and $(\bigsqcup{}^{{}^{{ \mathfrak{S}}'}}\{\sigma\})=\sigma$. As a result, we have $\sigma \in \{\,\bigsqcup{}^{{}^{{ \mathfrak{S}}'}} V\;\vert\; V\sqsubseteq U\;\;\textit{\rm and}\;\; \widehat{\;V\;}{}^{{}^{{ \mathfrak{S}}'}}\,\}$ and then $\{\sigma\}\sqsubseteq {\mathfrak{c}}^{{ \mathfrak{S}}'}(U)$. As a final conclusion, we obtain $U\sqsubseteq {\mathfrak{c}}^{{ \mathfrak{S}}'}(U)$. \\
Concerning the property (\ref{propB4}), we have the following analysis. For any $\sigma\in \{\,\bigsqcup{}^{{}^{{ \mathfrak{S}}'}} V\;\vert\; V\sqsubseteq U_1\;\;\textit{\rm and}\;\; \widehat{\;V\;}{}^{{}^{{ \mathfrak{S}}'}}\,\}$ there exists $V\sqsubseteq U_1$ such that $\sigma=\bigsqcup{}^{{}^{{ \mathfrak{S}}'}} V$. But $U_1\sqsubseteq U_2$ implies $V\sqsubseteq U_2$ and then $\sigma\in \{\,\bigsqcup{}^{{}^{{ \mathfrak{S}}'}} V\;\vert\; V\sqsubseteq U_2\;\;\textit{\rm and}\;\; \widehat{\;V\;}{}^{{}^{{ \mathfrak{S}}'}}\,\}$. As a consequence, we have $\{\sigma\}\sqsubseteq {\mathfrak{c}}^{{ \mathfrak{S}}'}(U_2)$, and then $\{\,\bigsqcup{}^{{}^{{ \mathfrak{S}}'}} V\;\vert\; V\sqsubseteq U_1\;\;\textit{\rm and}\;\; \widehat{\;V\;}{}^{{}^{{ \mathfrak{S}}'}}\,\}\sqsubseteq {\mathfrak{c}}^{{ \mathfrak{S}}'}(U_2)$. As a result, we obtain ${\mathfrak{c}}^{{ \mathfrak{S}}'}(U_1)\sqsubseteq {\mathfrak{c}}^{{ \mathfrak{S}}'}(U_2)$. 
This concludes the proof.
\end{proof}
\begin{remark}
We have then proved that ${\mathfrak{c}}^{{ \mathfrak{S}}'}$ is {\em a pre-closure}. However, for some ${ \mathfrak{S}}'$ it appears that ${\mathfrak{c}}^{{ \mathfrak{S}}'}$ is not a closure. In other words,
\begin{eqnarray}
\exists U\in {\mathcal{P}}({ \mathfrak{S}}')&\vert & {\mathfrak{c}}^{{ \mathfrak{S}}'}\circ {\mathfrak{c}}^{{ \mathfrak{S}}'}(U)\not= {\mathfrak{c}}^{{ \mathfrak{S}}'}(U),\label{propB5}
\end{eqnarray}
Let us consider the following case for the Inf semi-lattice ${ \mathfrak{S}}'$ :
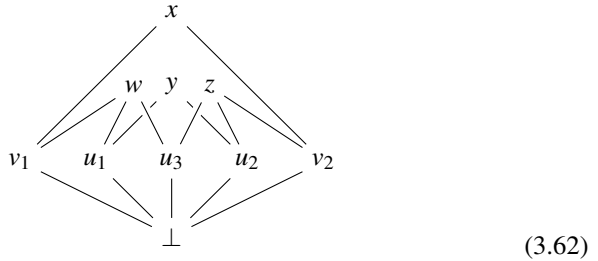
\begin{eqnarray}
&&\begin{tikzpicture}
  \node (a) at (-0.5,1) {$w$};
  \node (b) at (0,2) {$x$};
  \node (c) at (0.5,1) {$z$};
  \node (d) at (-1,0) {$u_{1}$};
  \node (e) at (0,0) {$u_{3}$};
  \node (f) at (1,0) {$u_{2}$};
  \node (g) at (-2,0) {$v_{1}$};
  \node (h) at (2,0) {$v_{2}$};
  \node (i) at (0,1) {$y$};
  \node (min) at (0,-1) {$\bot$};
  \draw (min) -- (d) -- (a) (min) -- (g) (min) -- (h) (d) -- (i) -- (f) 
  (e) -- (min) -- (f) -- (c) (a) -- (g) (c) -- (h)
  (g) -- (b) (h) -- (b);
  \draw[preaction={draw=white, -,line width=4pt}] (a) -- (e) -- (c);
\end{tikzpicture} 
\end{eqnarray}
If we choose $U:=\{u_1,u_2,u_3\}$ we observe that ${\mathfrak{c}}^{{ \mathfrak{S}}'}(U)=\{w,z,y\}$, but we have ${\mathfrak{c}}^{{ \mathfrak{S}}'}\circ {\mathfrak{c}}^{{ \mathfrak{S}}'} (U)=\{w,z,y,x\}$.
\end{remark}

\begin{lemma}
We have
\begin{eqnarray}
\forall U\in {\mathcal{P}}({ \mathfrak{S}}')&\vert & {\mathfrak{c}}^{{ \mathfrak{S}}'}\circ {\mathfrak{c}}^{{ \mathfrak{S}}'}(U)\sqsupseteq {\mathfrak{c}}^{{ \mathfrak{S}}'}(U),\label{propB5bis}
\end{eqnarray}
\end{lemma}
\begin{proof}
Direct consequence of properties (\ref{propB3}) and (\ref{propB4}).
\end{proof}

\begin{definition}
We define the map $cl_c^{{ \mathfrak{S}}'}$ from ${\mathcal{P}}({ \mathfrak{S}}')$ to ${\mathcal{P}}({ \mathfrak{S}}')$ by
\begin{eqnarray}
\forall U\in {\mathcal{P}}({ \mathfrak{S}}'),\;&& cl_c^{{ \mathfrak{S}}'}(U) := \bigsqcup{}_{n\geq 1} ({\mathfrak{c}}^{{ \mathfrak{S}}'})^{\circ n}(U)\label{definitioncompleterealspaceofstates3}
\end{eqnarray}
(we have denoted $({\mathfrak{c}}^{{ \mathfrak{S}}'})^{\circ n}$ the n-th composition by ${\mathfrak{c}}^{{ \mathfrak{S}}'}$).
\end{definition}

\begin{lemma}
$cl_c^{{ \mathfrak{S}}'}$ satisfies
\begin{eqnarray}
\forall U\in {\mathcal{P}}({ \mathfrak{S}}'),\;&& U\sqsubseteq cl_c^{{ \mathfrak{S}}'}(U),\label{propB3c}\\
\forall U_1,U_2\in {\mathcal{P}}({ \mathfrak{S}}'),\;&& U_1\sqsubseteq U_2\;\;\Rightarrow\;\; cl_c^{{ \mathfrak{S}}'}(U_1)\sqsubseteq cl_c^{{ \mathfrak{S}}'}(U_2).\;\;\;\;\;\;\;\;\;\;\;\;\;\;\label{propB4c}
\end{eqnarray}
\end{lemma}
\begin{proof}
The properties (\ref{propB3c}) (\ref{propB4c}) are direct consequences of the properties (\ref{propB3}) (\ref{propB4}).
\end{proof}

\begin{theorem}
Let us now suppose that ${ \mathfrak{S}}'$ satisfies the {\em Finite Rank Condition} (Definition \ref{FiniteRankCondition}).\\
Then, $cl_c^{{ \mathfrak{S}}'}$ satisfies the idempotency property
\begin{eqnarray}
\forall U\subseteq { \mathfrak{S}}',\;&& cl_c^{{ \mathfrak{S}}'}\circ cl_c^{{ \mathfrak{S}}'}(U)= cl_c^{{ \mathfrak{S}}'}(U).\label{propB5c}
\end{eqnarray}
and the following algebraicity property 
\begin{eqnarray}
\hspace{-0.5cm}\forall U\subseteq { \mathfrak{S}}', \forall \{\sigma\}\sqsubseteq cl_c^{{ \mathfrak{S}}'}(U),\;\exists \;\textit{\rm $V$ finite subset of ${ \mathfrak{S}}'$ with $V\sqsubseteq U$ such that}\; \{\sigma\}\sqsubseteq cl_c^{{ \mathfrak{S}}'}(V).\;\;\;\;\;\;\;\label{algebraicclosure}
\end{eqnarray}
\end{theorem}
\begin{proof}
First of all, by choosing $U$ to be a chain in ${ \mathfrak{S}}'$ (explicitly $U:=\{\,u_i\;\vert\; i\geq 1\,\}$ with $u_i\sqsubseteq_{{}_{{ \mathfrak{S}}'}}u_j$ for $i\leq j$) in (\ref{finiterank}), we note that the Finite Rank Condition ensures that the chain $U$ is stabilizing, i.e.  $\exists n\in \mathbb{N}$ with $u_n=u_k$ for any $k\geq n$ (this means that ${ \mathfrak{S}}'$ satisfies the {\em Ascending Chain Condition}).\\
Secondly we note that, if ${ \mathfrak{S}}'$ satisfies the Finite Rank Condition, we then have the following {\em algebraicity} property for ${\mathfrak{c}}^{{ \mathfrak{S}}'}$ :
\begin{eqnarray}
\hspace{-0.5cm}\forall U\subseteq { \mathfrak{S}}', \forall \{\sigma\}\sqsubseteq {\mathfrak{c}}^{{ \mathfrak{S}}'}(U),\;\exists \;\textit{\rm $V$ finite subset of ${ \mathfrak{S}}'$ with $V\sqsubseteq U$ such that}\; \{\sigma\}\sqsubseteq {\mathfrak{c}}^{{ \mathfrak{S}}'}(V).\;\;\;\;\;\;\;\label{algebraicpreclosure}
\end{eqnarray}
We now intent to deduce a similar property for $cl_c^{{ \mathfrak{S}}'}$. \\
Let us consider $\sigma\in { \mathfrak{S}}'$ such that $\{\sigma\}\sqsubseteq cl_c^{{ \mathfrak{S}}'}(U)$. Using the Ascending Chain Condition deduced from the Finite Rank Condition, and the fact that the sequence $({\mathfrak{c}}^{{ \mathfrak{S}}'})^{\circ n}(U)$ forms a nest for $\sqsubseteq$, we deduce that there must exist an integer $n$ such that $\{\sigma\}\sqsubseteq ({\mathfrak{c}}^{{ \mathfrak{S}}'})^{\circ n}(U)$.  From (\ref{algebraicpreclosure}), we then deduce that there exists $V$ finite subset of ${ \mathfrak{S}}'$ with $V\sqsubseteq U$ such that $\{\sigma\}\sqsubseteq ({\mathfrak{c}}^{{ \mathfrak{S}}'})^{\circ n}(V)$. Therefore, we obtain $\{\sigma\}\sqsubseteq \bigsqcup{}_{n\geq 1} ({\mathfrak{c}}^{{ \mathfrak{S}}'})^{\circ n}(V)=cl_c^{{ \mathfrak{S}}'}(V)$. As a conclusion, we have obtained (\ref{algebraicclosure}).\\
Let us now come back to the proof of property (\ref{propB5c}).\\
Let us consider $\sigma\in { \mathfrak{S}}'$ such that $\{\sigma\}\sqsubseteq cl_c^{{ \mathfrak{S}}'}\circ cl_c^{{ \mathfrak{S}}'}(U)$. Using the already proved algebraicity property (\ref{algebraicclosure}), we deduce that there exists a finite subset of ${ \mathfrak{S}}'$ denoted $V$ such that $V\sqsubseteq cl_c^{{ \mathfrak{S}}'}(U)$ and satisfying $\{\sigma\}\sqsubseteq cl_c^{{ \mathfrak{S}}'}(V)$. Then, since $V$ is finite, there must exist an integer $k$ such that $V\sqsubseteq ({\mathfrak{c}}^{{ \mathfrak{S}}'})^{\circ k}(U)$. Therefore, we have $ cl_c^{{ \mathfrak{S}}'}(V) = \bigsqcup{}_{n\geq 1} ({\mathfrak{c}}^{{ \mathfrak{S}}'})^{\circ n}(V) \sqsubseteq \bigsqcup{}_{n\geq 1} ({\mathfrak{c}}^{{ \mathfrak{S}}'})^{\circ n}(({\mathfrak{c}}^{{ \mathfrak{S}}'})^{\circ k}(U))=\bigsqcup{}_{n\geq 1} ({\mathfrak{c}}^{{ \mathfrak{S}}'})^{\circ (n+k)}(U)\sqsubseteq \bigsqcup{}_{n\geq 1} ({\mathfrak{c}}^{{ \mathfrak{S}}'})^{\circ n}(U)$. As a result, we have obtained $\{\sigma\}\sqsubseteq \bigsqcup{}_{n\geq 1} ({\mathfrak{c}}^{{ \mathfrak{S}}'})^{\circ n}(U)=cl_c^{{ \mathfrak{S}}'}(U)$. 
We have then obtained the inequality 
\begin{eqnarray}  
\forall U\subseteq { \mathfrak{S}}',\;&& cl_c^{{ \mathfrak{S}}'}\circ cl_c^{{ \mathfrak{S}}'}(U) \sqsubseteq cl_c^{{ \mathfrak{S}}'}(U).\label{propB5cfirst}
\end{eqnarray}
The inequality
\begin{eqnarray}
\forall U\subseteq { \mathfrak{S}}',\;&& cl_c^{{ \mathfrak{S}}'}\circ cl_c^{{ \mathfrak{S}}'}(U) \sqsupseteq cl_c^{{ \mathfrak{S}}'}(U).\label{propB5csecond}
\end{eqnarray}
is a direct consequence of properties (\ref{propB3c}) and (\ref{propB4c}).
\end{proof}
\begin{remark}
We are convinced that the Finite Rank Condition imposed as a prerequisite for the previous theorem could be alleged. This is not the purpose of the present paper to enter into this detail, and we let this point for a future study. 
\end{remark}

\noindent Let us now introduce our second basic definition. 

\begin{definition}\label{definitioncompleterealspaceofstatesQc}
We define the following subsets of ${ \mathcal{P}}({ \mathfrak{S}}')$ :
\begin{eqnarray}
{\mathcal{K}}({ \mathfrak{S}}')&:=&\{\,U\subseteq { \mathfrak{S}}'\smallsetminus \{\bot_{{}_{ \mathfrak{S}'}}\} \;\vert\; \forall \sigma,\sigma'\in U,\;\sigma^\star \not\sqsubseteq_{{}_{{ \mathfrak{S}}'}}\sigma'
\,\}\;\;\;\;\;\;\;\;\;\;\;\;\;\;\;\label{definitioncompleterealspaceofstates1propk}\\
{\mathcal{Q}}_c({ \mathfrak{S}}')&:=&\{\,U\subseteq { \mathfrak{S}}'\smallsetminus \{\bot_{{}_{ \mathfrak{S}'}}\} \;\vert\; cl_c^{{ \mathfrak{S}}'}(U)\in {\mathcal{K}}({ \mathfrak{S}}')
\,\}\;\cup\; \{\bot_{{}_{ \mathfrak{S}'}}\}
\;\;\;\;\;\;\;\;\;\;\;\;\;\;\;\label{definitioncompleterealspaceofstates1prop}
\end{eqnarray}
and
\begin{eqnarray}
&& { \mathfrak{J}}^c_{{{ \mathfrak{S}}'}}:=cl_c^{{{ \mathfrak{S}}'}}({ \mathcal{Q}}_c({ \mathfrak{S}}')).\label{definitioncompleterealspaceofstates1propJc}
\end{eqnarray} 
\end{definition}

\begin{lemma}
${\mathcal{Q}}_c({ \mathfrak{S}}')$ satisfies the property : 
\begin{eqnarray}
\forall K,K'\subseteq { \mathfrak{S}}'\smallsetminus \{\bot_{{}_{ \mathfrak{S}'}}\}, \;\;\;\;(\,K\sqsubseteq K'\;\;\textit{\rm and}\;\;K'\in {\mathcal{Q}}_c({ \mathfrak{S}}')\,) &\Rightarrow & K\in {\mathcal{Q}}_c({ \mathfrak{S}}').\;\;\;\;\;\;\;\;\;\;\label{Qlowersubset}
\end{eqnarray}
\end{lemma}
\begin{proof}
First of all, we recall the property : 
\begin{eqnarray}
\forall K,K'\subseteq { \mathfrak{S}}'\smallsetminus \{\bot_{{}_{ \mathfrak{S}'}}\}, \;\;\;\;(\,K\sqsubseteq K'\;\;\textit{\rm and}\;\;K'\in {\mathcal{K}}({ \mathfrak{S}}')\,) &\Rightarrow & K\in {\mathcal{K}}({ \mathfrak{S}}').\;\;\;\;\;\;\;\;\;\;\label{Klowersubset}
\end{eqnarray}
Indeed, let us consider $ K,K'\subseteq { \mathfrak{S}}'\smallsetminus \{\bot_{{}_{ \mathfrak{S}'}}\}$ such that $K\sqsubseteq K'$ and $K'\in {\mathcal{K}}({ \mathfrak{S}}')$.
For any $\sigma,\sigma'\in K$ there exists $\kappa,\kappa'\in K'$ such that $\sigma\sqsubseteq_{{}_{{ \mathfrak{S}'}}}\kappa$ and $\sigma'\sqsubseteq_{{}_{{ \mathfrak{S}'}}}\kappa'$. By assumption we have $\kappa^\star \not\sqsubseteq_{{}_{{ \mathfrak{S}}'}}\kappa'$, and we then deduce $\sigma^\star \not\sqsubseteq_{{}_{{ \mathfrak{S}}'}}\sigma'$. This concludes the proof of (\ref{Klowersubset}).\\
Now, let us consider $K,K'\subseteq { \mathfrak{S}}'\smallsetminus \{\bot_{{}_{ \mathfrak{S}'}}\}$, such that $K\sqsubseteq K'$ and $K'\in {\mathcal{Q}}_c({ \mathfrak{S}}')$. We have, using monotonicity of $cl_c^{{ \mathfrak{S}}'}$, the property $cl_c^{{ \mathfrak{S}}'}(K)\sqsubseteq cl_c^{{ \mathfrak{S}}'}(K')$. And $K'\in {\mathcal{Q}}_c({ \mathfrak{S}}')$ means $cl_c^{{ \mathfrak{S}}'}(K')\in {\mathcal{K}}({ \mathfrak{S}}')$. Now, we use (\ref{Klowersubset}) to deduce that $cl_c^{{ \mathfrak{S}}'}(K)\in {\mathcal{K}}({ \mathfrak{S}}')$, i.e.  $K\in {\mathcal{Q}}_c({ \mathfrak{S}}')$. This concludes the proof.
\end{proof}

\begin{lemma}
\begin{eqnarray}
\forall K\in {\mathcal{Q}}_c({ \mathfrak{S}}'),\;&& cl_c^{{ \mathfrak{S}}'}(K)\in {\mathcal{Q}}_c({ \mathfrak{S}}').\label{propB2}
\end{eqnarray}
\end{lemma}
\begin{proof}
Let us consider $K\in {\mathcal{Q}}_c({ \mathfrak{S}}')$.  By definition, we have $cl_c^{{ \mathfrak{S}}'}(K)\in {\mathcal{K}}({ \mathfrak{S}}')$ and then, using the idempotency of $cl_c^{{ \mathfrak{S}}'}$, we have $cl_c^{{ \mathfrak{S}}'}(cl_c^{{ \mathfrak{S}}'}(K))=cl_c^{{ \mathfrak{S}}'}(K)\in {\mathcal{K}}({ \mathfrak{S}}')$. In other words, $cl_c^{{ \mathfrak{S}}'}(K)\in {\mathcal{Q}}_c({ \mathfrak{S}}')$. 
\end{proof}

\begin{lemma}
\begin{eqnarray}
\hspace{-1cm}\forall U\in cl_c^{{{ \mathfrak{S}}'}}({ \mathcal{Q}}_c({ \mathfrak{S}}')),\forall \sigma_1,\sigma_2\in U,&& \widehat{\sigma_1\sigma_2}{}^{{}^{{{ \mathfrak{S}}'}}}\;\;\Rightarrow\;\; \sigma_1=\sigma_2.\label{propB1}
\end{eqnarray}
\end{lemma}
\begin{proof}
It suffices to prove the following simpler property :
\begin{eqnarray}
\hspace{-1cm}\forall U\in {\mathfrak{c}}^{{{ \mathfrak{S}}'}}({ \mathcal{Q}}_c({ \mathfrak{S}}')),\forall \sigma_1,\sigma_2\in U,&& \widehat{\sigma_1\sigma_2}{}^{{}^{{{ \mathfrak{S}}'}}}\;\;\Rightarrow\;\; \sigma_1=\sigma_2.
\end{eqnarray}
Let us consider two elements $\sigma,\sigma'\in {\mathfrak{c}}^{{{ \mathfrak{S}}'}}(U)$. There exists $V,V'\sqsubseteq U$ which are maximal for $\sqsubseteq$ among subsets of $U$ satisfying $\widehat{\;V\;}{}^{{}^{{ \mathfrak{S}}'}}$ and $\widehat{\;V'\;}{}^{{}^{{ \mathfrak{S}}'}}$ and such that $\sigma=\bigsqcup{}^{{}^{{ \mathfrak{S}}'}} V$ and $\sigma'=\bigsqcup{}^{{}^{{ \mathfrak{S}}'}} V'$. If we had $\widehat{\;\sigma\sigma'\;}{}^{{}^{{ \mathfrak{S}}'}}$, we would have $\widehat{\;V\cup V'\;}{}^{{}^{{ \mathfrak{S}}'}}$. This point contradicts the assumption of maximality for $V$ and $V'$. This concludes the proof.
\end{proof}

\begin{lemma}
${\mathcal{Q}}_c({ \mathfrak{S}}')$ satisfies
\begin{eqnarray}
\forall U\subseteq { \mathfrak{S}}',&& \widehat{\;U\;}{}^{{}^{{ \mathfrak{S}'}}}\;\;\Rightarrow\;\; U\in { \mathcal{Q}}_c({ \mathfrak{S}}'),\label{propA1}
\end{eqnarray}
\end{lemma}
\begin{proof}
Let us consider $U\subseteq { \mathfrak{S}}'$ such that $\widehat{\;U\;}{}^{{}^{{ \mathfrak{S}'}}}$. In order to check   the property (\ref{propA1}),  it suffices to note that $cl_c^{{ \mathfrak{S}}'}(U)=\{\bigsqcup{}^{{}^{{ \mathfrak{S}}'}}U\}$ and to remark that any singleton associated to an element of ${ \mathfrak{S}}'$ is necessarily in ${\mathcal{K}}({ \mathfrak{S}}')$.
\end{proof}

\begin{lemma}
${\mathcal{Q}}_c({ \mathfrak{S}}')$ satisfies
\begin{eqnarray}
\forall J\in { \mathcal{Q}}_c({ \mathfrak{S}}'),\forall \sigma\in { \mathfrak{S}}'\smallsetminus \{\bot_{{}_{ \mathfrak{S}'}}\}, & & \{\sigma,\sigma^\star\}\not\sqsubseteq J,\label{propA2}
\end{eqnarray}
In other words, we have
\begin{eqnarray}
{ \mathcal{Q}}_c({ \mathfrak{S}}') & \subseteq & { \mathcal{K}}({ \mathfrak{S}}').
\end{eqnarray}
\end{lemma}
\begin{proof}
 Let us consider $J\in { \mathcal{Q}}_c({ \mathfrak{S}}')$ and let us suppose that there exists $\sigma\in { \mathfrak{S}}'\smallsetminus \{\bot_{{}_{ \mathfrak{S}'}}\}$ such that $\{\sigma,\sigma^\star\}\sqsubseteq J$. Then, let us introduce $J':=cl_c^{{ \mathfrak{S}}'}(J)$. Due to the basic properties of the closure $cl_c^{{ \mathfrak{S}}'}$, we know that $J'$ satisfies $J\sqsubseteq J'$. Moreover, from $J\in { \mathcal{Q}}_c({ \mathfrak{S}}')$, we deduce that $J'\in {\mathcal{K}}({ \mathfrak{S}}')$. Now, we use (\ref{Klowersubset}) to conclude that $J\in {\mathcal{K}}({ \mathfrak{S}}')$. This result contradicts the existence of $\sigma\in { \mathfrak{S}}'\smallsetminus \{\bot_{{}_{ \mathfrak{S}'}}\}$ such that $\{\sigma,\sigma^\star\}\sqsubseteq J$. This concludes the proof of (\ref{propA2}). 
\end{proof}

\begin{lemma}\label{lemmasignifySinQ}
\begin{eqnarray}
  S\in {\mathcal{Q}}_c({ \mathfrak{S}}') &\Leftrightarrow & \forall \alpha,\beta\in cl_c^{{{ \mathfrak{S}}'}}(S),\;\alpha^\star\not\sqsubseteq_{{}_{{ \mathfrak{S}}'}}\beta.
\end{eqnarray}
\end{lemma}
\begin{proof}
Trivial rewriting of the definition of ${\mathcal{Q}}_c({ \mathfrak{S}}')$.
\end{proof}

\begin{theorem}
$({\mathfrak{S}}',{ \mathcal{Q}}_c({ \mathfrak{S}}'),cl_c^{{{ \mathfrak{S}}'}},\star)$ satisfies the axiomatic of a real space of states. 
\end{theorem}
\begin{proof}
This result summarizes the previous lemmas.
\end{proof}

\begin{definition}
The real space of states ${ \mathfrak{J}}^c_{{{ \mathfrak{S}}'}}$ (defined according to Definition \ref{definitioncompleterealspaceofstatesQc}) is called {\em the maximal ontic completion} of $({ \mathfrak{S}}',\star)$. \\
In the following, we will denote $({\mathfrak{S}}',{ \mathcal{Q}}_c({ \mathfrak{S}}'),cl_c^{{{ \mathfrak{S}}'}},\star)$ shortly by $(\!({\mathfrak{S}}',\star)\!)_c$. 
\end{definition}

\begin{definition}
A subset $U$ of ${\mathfrak{S}}'$ is said to be {\em admissible} 
 iff $U$ is an element of ${ \mathcal{Q}}_c({ \mathfrak{S}}')$.
\end{definition}

\begin{definition}
We will say that {\em the ontic completion of $({\mathfrak{S}}',\star)$ 
is complete} iff we have
 \begin{eqnarray}
 {\mathcal{Q}}_c({ \mathfrak{S}}') &\supseteq & \widehat{\mathcal{K}}({ \mathfrak{S}}')
\;\;\;\;\;\;\;\;\;\;\;\;\;\;\;\nonumber\\
 \widehat{\mathcal{K}}({ \mathfrak{S}}') &:= & \{\,U\subseteq { \mathfrak{S}}'\smallsetminus \{\bot_{{}_{ \mathfrak{S}'}}\} \;\vert\; \forall \sigma,\sigma'\in U,\;\sigma^\star \not\sqsubseteq_{{}_{{ \mathfrak{S}}'}}\sigma'
\;\textit{\rm and}\;\sigma\not=\sigma'\;\Rightarrow\;\neg\,\widehat{\sigma\,\sigma'}{}^{{}^{{ \mathfrak{S}}'}}
\,\}\;\;\;\;\;\;\;\;\;\;\;\;\;\;\;
\label{propertyadmitscopletionmaximal}
\end{eqnarray}
\end{definition}

It seems {\it a priori} appealing to require the existence of a complete ontic completion. However, it appears (except in degenerate case) that the ontic completion of $(\overline{ \mathfrak{S}},\star)$ cannot be complete. The following theorem clarifies this point.

\begin{theorem}\label{theoremNOcompleteontic}
We will assume the following basic property for ${\mathfrak{S}}'$ : 
\begin{eqnarray}
\sigma,\sigma'\in {{\mathfrak{S}}'}{}^{{}^{pure}},&&(\sigma\sqcap_{{}_{{{\mathfrak{S}}'}}}\sigma')\sqcoversubset_{{}_{{{\mathfrak{S}}'}}}\sigma,\sigma'.\label{assumedcovering}
\end{eqnarray}
Let us now consider two distinct pure states $ \sigma_1,\sigma_2\in {{\mathfrak{S}}'}{}^{{}^{pure}}$ such that $(\sigma_1\sqcap_{{}_{{{\mathfrak{S}}'}}}\sigma_2)\not= \bot_{{}_{{{\mathfrak{S}}'}}}$ and $\sigma_1^\star \not\sqsubseteq_{{}_{{{\mathfrak{S}}'}}}\sigma_2$, we have then $\{\sigma_1,\sigma_2\}\in \widehat{\mathcal{K}}({{\mathfrak{S}}'})$ but $cl_c^{{{\mathfrak{S}}'}}(\{\sigma_1,\sigma_2\})\notin {\mathcal{K}}({{\mathfrak{S}}'})$, i.e. $\{\sigma_1,\sigma_2\}\notin {\mathcal{Q}}_c({{\mathfrak{S}}'})$. \\
As a conclusion, if ${\mathfrak{S}}'$ satisfies the assumed property (\ref{assumedcovering}) and if there exists $\sigma_1,\sigma_2\in {{\mathfrak{S}}'}{}^{{}^{pure}}$ such that $(\sigma_1\sqcap_{{}_{{{\mathfrak{S}}'}}}\sigma_2)\not= \bot_{{}_{{{\mathfrak{S}}'}}}$, then the ontic completion of $({{\mathfrak{S}}'},\star)$ cannot be complete.  
\end{theorem}
\begin{proof}
Let us fix $\sigma_1,\sigma_2$ in ${{\mathfrak{S}}'}{}^{{}^{pure}}$ such that $\sigma_1^\star \not\sqsubseteq_{{}_{{{\mathfrak{S}}'}}}\sigma_2$.\\
First of all, we note that, as long as $(\sigma_1\sqcap_{{}_{{{\mathfrak{S}}'}}}\sigma_2)\not= \bot_{{}_{{{\mathfrak{S}}'}}}$ we know that there exists $\sigma_3$ in ${{\mathfrak{S}}'}{}^{{}^{pure}}$ such that $\sigma_1^\star \sqsubseteq_{{}_{{{\mathfrak{S}}'}}}\sigma_3$ and $(\sigma_1\sqcap_{{}_{{{\mathfrak{S}}'}}}\sigma_2)\not\sqsubseteq_{{}_{{{\mathfrak{S}}'}}}\sigma_3$. Indeed, if it was not the case we would have $\bot_{{}_{{{\mathfrak{S}}'}}}\not=(\sigma_1\sqcap_{{}_{{{\mathfrak{S}}'}}}\sigma_2)\sqsubseteq_{{}_{{{\mathfrak{S}}'}}}(\bigsqcap_{\sigma_3\in \underline{\sigma_1^\star}_{{}_{{\mathfrak{S}}'}}}\sigma_3)=\sigma_1^\star$ and then $(\sigma_1\sqcap_{{}_{{{\mathfrak{S}}'}}}\sigma_2)=\sigma_1^\star$ which implies $\sigma_1\sqsupseteq_{{}_{{{\mathfrak{S}}'}}}\sigma_1^\star$ which is false. As a conclusion, 
\begin{eqnarray}
\exists \sigma_3\in {{\mathfrak{S}}'}{}^{{}^{pure}}\;\vert\; \sigma_1^\star \sqsubseteq_{{}_{{{\mathfrak{S}}'}}}\sigma_3\;\textit{\rm and}\; (\sigma_1\sqcap_{{}_{{{\mathfrak{S}}'}}}\sigma_2)\not\sqsubseteq_{{}_{{{\mathfrak{S}}'}}}\sigma_3.
\end{eqnarray}
We have obviously the same result if we replace $\sigma_1$ by $\sigma_2$.\\
Let us now use the assumed property satisfied by ${{\mathfrak{S}}'}$.\\
We have $(\sigma_1\sqcap_{{}_{{{\mathfrak{S}}'}}}\sigma_3)\sqcoversubset_{{}_{{{\mathfrak{S}}'}}}\sigma_3,\sigma_1$ and $(\sigma_2\sqcap_{{}_{{{\mathfrak{S}}'}}}\sigma_3)\sqcoversubset_{{}_{{{\mathfrak{S}}'}}}\sigma_3,\sigma_2$. We have then to distinguish the following two options : (1) $(\sigma_1\sqcap_{{}_{{{\mathfrak{S}}'}}}\sigma_3)\sqcup_{{}_{{{\mathfrak{S}}'}}}(\sigma_2\sqcap_{{}_{{{\mathfrak{S}}'}}}\sigma_3)=(\sigma_1\sqcap_{{}_{{{\mathfrak{S}}'}}}\sigma_3)=(\sigma_2\sqcap_{{}_{{{\mathfrak{S}}'}}}\sigma_3)$, or (2) $(\sigma_1\sqcap_{{}_{{{\mathfrak{S}}'}}}\sigma_3)\sqcup_{{}_{{{\mathfrak{S}}'}}}(\sigma_2\sqcap_{{}_{{{\mathfrak{S}}'}}}\sigma_3)=\sigma_3$.  Let us now show that the first option is excluded.  We have necessarily $\sigma_1\sqcap_{{}_{{{\mathfrak{S}}'}}}\sigma_2\sqcap_{{}_{{{\mathfrak{S}}'}}}\sigma_3 \sqsubseteq_{{}_{{{\mathfrak{S}}'}}}\sigma_1\sqcap_{{}_{{{\mathfrak{S}}'}}}\sigma_2\sqsubset_{{}_{{{\mathfrak{S}}'}}}\sigma_1$. Using $(\sigma_1\sqcap_{{}_{{{\mathfrak{S}}'}}}\sigma_3)=(\sigma_2\sqcap_{{}_{{{\mathfrak{S}}'}}}\sigma_3)$ and $\sigma_1\sqcap_{{}_{{{\mathfrak{S}}'}}}\sigma_3\sqcoversubset_{{}_{{{\mathfrak{S}}'}}}\sigma_1$ we then deduce that $\sigma_1\sqcap_{{}_{{{\mathfrak{S}}'}}}\sigma_2=\sigma_1\sqcap_{{}_{{{\mathfrak{S}}'}}}\sigma_3$, and then $\sigma_1\sqcap_{{}_{{{\mathfrak{S}}'}}}\sigma_2\sqsubseteq_{{}_{{{\mathfrak{S}}'}}}\sigma_3$, which contradicts the assumption on $\sigma_3$.\\
We have then obtained the necessary condition : $(\sigma_1\sqcap_{{}_{{{\mathfrak{S}}'}}}\sigma_3)\sqcup_{{}_{{{\mathfrak{S}}'}}}(\sigma_2\sqcap_{{}_{{{\mathfrak{S}}'}}}\sigma_3)=\sigma_3$.  As a consequence, we deduce
\begin{eqnarray}
\{\sigma_1,\sigma_2,\sigma_3\} &\sqsubseteq & Max \{\,\bigsqcup{}^{{}^{{{\mathfrak{S}}'}}} V\;\vert\; V\sqsubseteq \{\sigma_1,\sigma_2\}\;\;\textit{\rm and}\;\; \widehat{\;V\;}{}^{{}^{{{\mathfrak{S}}'}}}\,\}={\mathfrak{c}}^{{{\mathfrak{S}}'}}(\{\sigma_1,\sigma_2\})\;\;\;\;\;\;\;\;\;\;\;\;\;\;\;\;\;\;\\
&\sqsubseteq & cl_c^{{{\mathfrak{S}}'}}(\{\sigma_1,\sigma_2\})
\end{eqnarray}
But now we recall that ($\sigma_1^\star \sqsubseteq_{{}_{{{\mathfrak{S}}'}}}\sigma_3$ or $\sigma_2^\star \sqsubseteq_{{}_{{{\mathfrak{S}}'}}}\sigma_3$) to conclude that $cl_c^{{{\mathfrak{S}}'}}(\{\sigma_1,\sigma_2\})\notin {\mathcal{K}}({{\mathfrak{S}}'})$. Noting that, by assumption, we have $\{\sigma_1,\sigma_2\}\in \widehat{\mathcal{K}}({{\mathfrak{S}}'})$ we deduce that the ontic completion of $({{\mathfrak{S}}'},\star)$ cannot be complete.
\end{proof}

\subsection{Morphisms and ontic completions}\label{subsectionmorphismsonticcompletions}

Let us consider a morphism $f$ from a real states space $({ \mathfrak{S}}'_{A_1},\star)$ to another real states space $({ \mathfrak{S}}'_{A_2},\star)$.  Both of these spaces of states are supposed to be such that the conditions required for the existence of their ontic completions are satisfied.  We denote ${ \mathfrak{S}}_{A_1}:={ \mathfrak{J}}^c_{{{ \mathfrak{S}}'_{A_1}}}$ and $\overline{ \mathfrak{S}}_{A_1}={ \mathfrak{S}}'_{A_1}$, and analogously ${ \mathfrak{S}}_{A_2}:={ \mathfrak{J}}^c_{{{ \mathfrak{S}}'_{A_2}}}$ and $\overline{ \mathfrak{S}}_{A_2}={ \mathfrak{S}}'_{A_2}$. 

We can define a morphism, denoted $F$ here,  from the states space ${ \mathfrak{S}}_{A_1}$ to the states space ${ \mathfrak{S}}_{A_2}$ by
\begin{eqnarray}
\forall \xi \in { \mathfrak{S}}_{A_1},&&\Theta^{{\overline{ \mathfrak{S}}_{A_2}}}(F(\xi))=cl_c^{{\overline{ \mathfrak{S}}_{A_2}}}(\{ f(\omega_{A_1})\;\vert\; \omega_{A_1}\in \Theta^{{\overline{ \mathfrak{S}}_{A_1}}}(\xi)\;\})\;\;\;\;\;\;\;\;
\end{eqnarray}
such that the restriction of $F$ to $\overline{ \mathfrak{S}}_{A_1}$ is equal to $f$.

\section{Bipartite experiments} \label{sectionfirstremarksonbipartite}

After presenting the requirements for the tensor product of real state spaces, we propose a solution for this tensor product and demonstrate that it is at least suitable for the deterministic case.

\subsection{Building principles for the description of bipartite experiments}\label{subsectioncompoundfirstremarks}
 
 During this subsection, we consider the two states/effects Chu Spaces $({ \mathfrak{S}}_A,{ \mathfrak{E}}_{{ \mathfrak{S}_A}},\epsilon^{{ \mathfrak{S}_A}})$ and $({ \mathfrak{S}}_B,{ \mathfrak{E}}_{{ \mathfrak{S}_B}},\epsilon^{{ \mathfrak{S}_B}})$ equipped respectively with their real structure denoted by $(\overline{ \mathfrak{S}}_A,\star)$ and $(\overline{ \mathfrak{S}}_B,\star)$. In order to shorten certain expressions, we will denote ${ \mathfrak{E}}_A:=\overline{ \mathfrak{E}}_{{ \mathfrak{S}_A}}$ and ${ \mathfrak{E}}_B:=\overline{ \mathfrak{E}}_{{ \mathfrak{S}_B}}$.\\
 
We now begin with a basic set of requirements addressed for the description of bipartite experiments. We will denote by ${ \mathfrak{S}}_{AB}={ \mathfrak{S}}_{A}\boxtimes { \mathfrak{S}}_{B}$ the corresponding space of states.\footnote{Throughout this short axiomatic introduction, we adopt the inadequate notation $\boxtimes$ for the tensor product in order to allow for different candidates for this tensor product. These different candidates will be denoted $\otimes$, $\widetilde{\otimes}$,...}\\

First of all, we will assume that the set ${ \mathfrak{S}}_{A}\boxtimes { \mathfrak{S}}_{B}$ admits mixed bipartite states and a completely mixed state
\begin{eqnarray}
&& \forall \{\,\sigma_{i,AB}\;\vert\; i\in I\,\}\subseteq {{ \mathfrak{S}}_{AB}}, \;\;\;\;\;\; \bigsqcap{}^{{}^{{ \mathfrak{S}}_{AB}}}_{i\in I} \sigma_{i,AB}\;\;\textit{\rm exists in}\;\; { \mathfrak{S}}_{AB},
\;\;\;\;\;\;\label{tensorinfimum}\\
&& \exists \bot_{{}_{{{ \mathfrak{S}}_{AB}}}} \;\textit{\rm bottom element of ${{ \mathfrak{S}}_{AB}}$}.\label{tensorbot}
\end{eqnarray}
We moreover assume that 
\begin{eqnarray}
&&\textit{\rm ${ \mathfrak{S}}_{A}\boxtimes { \mathfrak{S}}_{B}$ admits a real structure denoted $(\overline{{ \mathfrak{S}}_{A}\boxtimes { \mathfrak{S}}_{B}},\star)$ (see Definition \ref{definitionstarstructure}).}\;\;\;\;\;\;\;\;\;\;\;\;\;\;\;\;\;\;
\end{eqnarray}

Secondly, for every pair of real states $\sigma_A\in \overline{ \mathfrak{S}}_{A}$ and $\sigma_B\in \overline{ \mathfrak{S}}_{B}$, prepared independently by Alice and Bob, we will assume that there must exist a unique associated bipartite real state in ${{ \mathfrak{S}}_{A}\boxtimes { \mathfrak{S}}_{B}}$.  More explicitly,  we will assume that there are maps
$\iota^{{ \mathfrak{S}}_{A B}}$ which describe the inclusion of 'pure tensors' in ${ \mathfrak{S}}_{A}\boxtimes { \mathfrak{S}}_{B}$ :
\begin{eqnarray}
\begin{array}{rcrcl}
\iota^{{ \mathfrak{S}}_{A B}} &: & \overline{{ \mathfrak{S}}_{A}}\times  \overline{{ \mathfrak{S}}_{B}} &\longrightarrow & {{ \mathfrak{S}}_{A}\boxtimes { \mathfrak{S}}_{B}}\\
& & (\sigma_A,\sigma_B) & \mapsto & \sigma_A \boxtimes \sigma_B.\label{inclusionpuretensors}
\end{array}
\end{eqnarray}

Thirdly,  for every bipartite real states $\sigma_{AB},\sigma'_{AB} \in \overline{{ \mathfrak{S}}}_{AB}$ such that $\sigma_{AB}\not= \sigma'_{AB}$, we will assume that there must exist real effects ${\mathfrak{l}}_A\in { \mathfrak{E}}_{A}$ and ${\mathfrak{l}}_B\in { \mathfrak{E}}_{B}$ such that when Alice and Bob prepare $\sigma_{AB}$ and apply real effects ${\mathfrak{l}}_A$ and ${\mathfrak{l}}_B$ respectively,  the resulting determination (fixed as the product of the separate determinations) is different from the experiment where Alice and Bob prepare $\sigma'_{AB}$ and apply ${\mathfrak{l}}_A$ and ${\mathfrak{l}}_B$ respectively. As a summary, applying real effects locally is sufficient to distinguish all of the states in ${{ \mathfrak{S}}_{AB}}$ (this principle is called "tomographic locality"), i.e.
\begin{eqnarray}
&&\hspace{-1cm}\forall \; \left(\bigsqcap{}^{{}^{{{{{ \mathfrak{S}}_A}\boxtimes {{ \mathfrak{S}}_{B}}}}}}_{i\in I} \sigma_{i,A}\boxtimes \sigma_{i,B}\right) \in \overline{{{ \mathfrak{S}}_A}\boxtimes {{ \mathfrak{S}}_{B}}},\; \forall \; \left(\bigsqcap{}^{{}^{{{{{ \mathfrak{S}}_A}\boxtimes {{ \mathfrak{S}}_{B}}}}}}_{j\in J} \sigma'_{j,A}\boxtimes \sigma'_{j,B}\right) \in \overline{{{ \mathfrak{S}}_A}\boxtimes {{ \mathfrak{S}}_{B}}},\nonumber\\
&&\hspace{-1cm}\left(\forall {\mathfrak{l}}_A\in { \mathfrak{E}}_{A},{\mathfrak{l}}_B\in { \mathfrak{E}}_{B},\;\; \bigwedge{}_{i\in I}{\epsilon}\,{}^{{ \mathfrak{S}}_A}_{{\mathfrak{l}}_A} (\sigma_{i,A})\bullet {\epsilon}\,{}^{{ \mathfrak{S}}_B}_{{\mathfrak{l}}_B} (\sigma_{i,B})=\bigwedge{}_{j\in J}{\epsilon}\,{}^{{ \mathfrak{S}}_A}_{{\mathfrak{l}}_A} (\sigma'_{j,A})\bullet {\epsilon}\,{}^{{ \mathfrak{S}}_B}_{{\mathfrak{l}}_B} (\sigma'_{j,B}) \right)\nonumber\\
&&\hspace{4cm} \;\Rightarrow\; (\,\bigsqcap{}^{{}^{{{{{ \mathfrak{S}}_A}\boxtimes {{ \mathfrak{S}}_{B}}}}}}_{i\in I} \sigma_{i,A}\boxtimes \sigma_{i,B} = \bigsqcap{}^{{}^{{{{{ \mathfrak{S}}_A}\boxtimes {{ \mathfrak{S}}_{B}}}}}}_{j\in J} \sigma'_{j,A}\boxtimes \sigma'_{j,B}\,).\;\;\;\;\;\;\;\;\;\;\;\label{tensorseparated}
\end{eqnarray}

Fourthly, we impose that if Alice (or Bob) prepares a mixture of real states, then this results in a mixture of the respective bipartite real states.  More explicitly, we require, for any $\{\,\sigma_{i,A}\;\vert\; i\in I\,\}\subseteq \overline{{ \mathfrak{S}}_{A}}$, $\{\,\sigma_{j,B}\;\vert\; j\in J\,\}\subseteq \overline{{ \mathfrak{S}}_{B}}$, $\sigma_A\in \overline{{ \mathfrak{S}}_{A}}$ and $\sigma_B\in \overline{{ \mathfrak{S}}_{B}}$, the two bimorphic properties for the tensor product :
\begin{eqnarray}
&& (\bigsqcap{}^{{}^{{{ \mathfrak{S}}_{A}}}}_{i\in I}\,\sigma_{i,A})\boxtimes \sigma_B =  \bigsqcap{}^{{}^{{{{ \mathfrak{S}}_{AB}}}}}_{i\in I} (\sigma_{i,A}\boxtimes \sigma_B),\label{pitensor=tensorpi1pre}\\
&& \sigma_A \boxtimes  (\bigsqcap{}^{{}^{{{ \mathfrak{S}}_{B}}}}_{i\in I}\,\sigma_{i,B}) =  \bigsqcap{}^{{}^{{{{ \mathfrak{S}}_{AB}}}}}_{i\in I} (\sigma_A \boxtimes \sigma_{i,B}).\label{pitensor=tensorpi2pre}
\end{eqnarray}

Endly, we will assume the existence of the two partial traces
\begin{eqnarray}
\left\{
\begin{array}{l}
\textit{\rm there exists an homomorphism denoted $\zeta^{{ \mathfrak{S}}_{A}{ \mathfrak{S}}_{B}}_{(1)}$ from $\overline{{ \mathfrak{S}}_{A}\boxtimes { \mathfrak{S}}_{B}}$ to $\overline{{ \mathfrak{S}}_{A}}$}\\
\textit{\rm there exists an homomorphism denoted $\zeta^{{ \mathfrak{S}}_{A}{ \mathfrak{S}}_{B}}_{(2)}$ from $\overline{{ \mathfrak{S}}_{A}\boxtimes { \mathfrak{S}}_{B}}$ to $\overline{{ \mathfrak{S}}_{B}}$}
\end{array}
\right.\label{requirementpartialtraces}
\end{eqnarray}

\subsection{The minimal tensor product}\label{subsectionbasicnotionstensorproduct}

We now recall some basic elements about the construction of the minimal tensor product of ${ \mathfrak{S}}_{A} $ and ${ \mathfrak{S}}_{B}$ \cite{Buffenoir2023}. The two states/effects Chu Spaces $({ \mathfrak{S}}_A,{ \mathfrak{E}}_{{ \mathfrak{S}_A}},\epsilon^{{ \mathfrak{S}_A}})$ and $({ \mathfrak{S}}_B,{ \mathfrak{E}}_{{ \mathfrak{S}_B}},\epsilon^{{ \mathfrak{S}_B}})$ are equipped respectively with their real structure denoted by $(\overline{ \mathfrak{S}}_A,\star)$ and $(\overline{ \mathfrak{S}}_B,\star)$. Here again, we will denote ${ \mathfrak{E}}_A:=\overline{ \mathfrak{E}}_{{ \mathfrak{S}_A}}$ and ${ \mathfrak{E}}_B:=\overline{ \mathfrak{E}}_{{ \mathfrak{S}_B}}$.

\begin{definition}\label{definnubullet}
The set ${ \mathcal{P}}(\overline{{ \mathfrak{S}}_{A}} \times \overline{{ \mathfrak{S}}_{B}})$ is equipped with the Inf semi-lattice structure $\cup$ and with the following maps defined for any $ {\mathfrak{l}}_{A}\in { \mathfrak{E}}_{A}$ and ${\mathfrak{l}}_{B}\in { \mathfrak{E}}_{B}$,
\begin{eqnarray}
\hspace{-2cm} \begin{array}{rcrcl}
 & & \nu\,{}^{\overline{{ \mathfrak{S}}_{A}}\overline{{ \mathfrak{S}}_{B}}}_{{\mathfrak{l}}_A,{\mathfrak{l}}_B}: { \mathcal{P}}(\overline{{ \mathfrak{S}}_{A}} \times \overline{{ \mathfrak{S}}_{B}}) & \longrightarrow & { \mathfrak{B}}\\
& & \{\,(\sigma_{i,A},\sigma_{i,B})\;\vert\; i\in I\,\} & \mapsto &  \bigwedge{}_{i\in I}\; {\epsilon}\,{}^{{ \mathfrak{S}}_{A}}_{{ \mathfrak{l}}_A}(\sigma_{i,A}) \bullet {\epsilon}\,{}^{{ \mathfrak{S}}_{B}}_{{ \mathfrak{l}}_B}(\sigma_{i,B}).
\end{array}
\end{eqnarray}
\end{definition}

\begin{definition}
${ \mathcal{P}}(\overline{{ \mathfrak{S}}_{A}} \times \overline{{ \mathfrak{S}}_{B}})$ is equipped with a congruence relation defined between any two elements $u_{AB}$ and $u'_{AB}$ of ${ \mathcal{P}}(\overline{{ \mathfrak{S}}_{A}} \times \overline{{ \mathfrak{S}}_{B}})$ by
\begin{eqnarray}
(\,u_{AB} \approx u'_{AB}\,) & :\Leftrightarrow &
(\,\forall {\mathfrak{l}}_A\in { \mathfrak{E}}_{A},\forall {\mathfrak{l}}_B\in { \mathfrak{E}}_{B},\;\;
\nu\,{}^{\overline{{ \mathfrak{S}}_{A}}\overline{{ \mathfrak{S}}_{B}}}_{{\mathfrak{l}}_A,{\mathfrak{l}}_B} (u_{AB})=\nu\,{}^{\overline{{ \mathfrak{S}}_{A}}\overline{{ \mathfrak{S}}_{B}}}_{{\mathfrak{l}}_A,{\mathfrak{l}}_B} (u'_{AB})\,).\;\;\;\;\;\;\;\;\;\;\;\;
\end{eqnarray}
\end{definition}

\begin{definition}
The space ${{\overline{S}}}_{AB}=\overline{{ \mathfrak{S}}_{A}} \widetilde{\otimes} \overline{{ \mathfrak{S}}_{B}}$ is built as the quotient of ${ \mathcal{P}}(\overline{{ \mathfrak{S}}_{A}} \times \overline{{ \mathfrak{S}}_{B}})$ under the congruence relation $\approx$.  
\begin{eqnarray}
\forall \sigma_{AB}\in { \mathcal{P}}(\overline{{ \mathfrak{S}}_{A}} \times \overline{{ \mathfrak{S}}_{B}}), && \widetilde{\sigma_{AB}}:=\{\, u_{AB}\;\vert\; \sigma_{AB}\approx u_{AB}\,\}.
\end{eqnarray}
The map $\nu\,{}^{\overline{{ \mathfrak{S}}_{A}}\overline{{ \mathfrak{S}}_{B}}}_{{\mathfrak{l}}_A,{\mathfrak{l}}_B}$ will be abusively defined as a map from ${{\overline{S}}}_{AB}$ to ${ \mathfrak{B}}$ by $\nu\,{}^{\overline{{ \mathfrak{S}}_{A}}\overline{{ \mathfrak{S}}_{B}}}_{{\mathfrak{l}}_A,{\mathfrak{l}}_B}(\widetilde{\sigma_{AB}}):= \nu\,{}^{\overline{{ \mathfrak{S}}_{A}}\overline{{ \mathfrak{S}}_{B}}}_{{\mathfrak{l}}_A,{\mathfrak{l}}_B}({\sigma_{AB}})$ for any $\sigma_{AB}$ in ${ \mathcal{P}}(\overline{{ \mathfrak{S}}_{A}} \times \overline{{ \mathfrak{S}}_{B}})$.\\
The element $\widetilde{u}\in {{\overline{S}}}_{AB}$ associated to the element $u:=\{\, (\sigma_{i,A},\sigma_{i,B})\;\vert\; i\in I\,\}\in { \mathcal{P}}(\overline{{ \mathfrak{S}}_{A}} \times \overline{{ \mathfrak{S}}_{B}})$ will be denoted $\bigsqcap{}^{{}^{{{\overline{S}}}_{AB}}}_{i\in I} \sigma_{i,A}\widetilde{\otimes}\sigma_{i,B}$.
\end{definition}

\begin{lemma}\label{lemmainclusiontensor}
There exists a map denoted $\iota^{{{\overline{S}}}_{AB}}$ from $\overline{{ \mathfrak{S}}_{A}} {\times} \overline{{ \mathfrak{S}}_{B}}$ to ${{\overline{S}}}_{AB}$. We have explicitly
\begin{eqnarray}
\iota^{{{\overline{S}}}_{AB}}(\sigma_A,\sigma_B) &:=& \widetilde{(\sigma_A,\sigma_B)}.
\end{eqnarray}
We will adopt the basic notation $\sigma_A\widetilde{\otimes}\sigma_B:=\iota^{{{\overline{S}}}_{AB}}(\sigma_A,\sigma_B)$
\end{lemma}

\begin{lemma}\label{lemmaseparatedminimal}
By construction, we have
\begin{eqnarray}
&&\hspace{-1.8cm}\forall \; \left(\bigsqcap{}^{{}^{{{\overline{S}}}_{AB}}}_{i\in I} \sigma_{i,A}\widetilde{\otimes} \sigma_{i,B}\right) \in {{\overline{S}}}_{AB},\; \forall \; \left(\bigsqcap{}^{{}^{{{\overline{S}}}_{AB}}}_{j\in J} \sigma'_{j,A}\widetilde{\otimes} \sigma'_{j,B}\right) \in {{\overline{S}}}_{AB},\nonumber\\
&&\hspace{-1.8cm}\left(\forall {\mathfrak{l}}_A\in { \mathfrak{E}}_{A},{\mathfrak{l}}_B\in { \mathfrak{E}}_{B},\;\; \bigwedge{}_{i\in I}{\epsilon}\,{}^{{ \mathfrak{S}}_A}_{{\mathfrak{l}}_A} (\sigma_{i,A})\bullet {\epsilon}\,{}^{{ \mathfrak{S}}_B}_{{\mathfrak{l}}_B} (\sigma_{i,B})=
\bigwedge{}_{j\in J}{\epsilon}\,{}^{{ \mathfrak{S}}_A}_{{\mathfrak{l}}_A} (\sigma'_{j,A})\bullet {\epsilon}\,{}^{{ \mathfrak{S}}_B}_{{\mathfrak{l}}_B} (\sigma'_{j,B}) \right)\nonumber\\
&&\hspace{5cm} \;\Rightarrow\; (\,\bigsqcap{}^{{}^{{{\overline{S}}}_{AB}}}_{i\in I} \sigma_{i,A}\widetilde{\otimes} \sigma_{i,B} = \bigsqcap{}^{{}^{{{\overline{S}}}_{AB}}}_{j\in J} \sigma'_{j,A}\widetilde{\otimes} \sigma'_{j,B}\,).\;\;\;\;\;\;\;\;\;\;\;
\end{eqnarray}
\end{lemma}

\begin{theorem}
For any $\{\,\sigma_{i,A}\;\vert\; i\in I\,\}\subseteq \overline{{ \mathfrak{S}}_{A}}$, $\{\,\sigma_{j,B}\;\vert\; j\in J\,\}\subseteq \overline{{ \mathfrak{S}}_{B}}$, $\sigma_A\in \overline{{ \mathfrak{S}}_{A}}$ and $\sigma_B\in \overline{{ \mathfrak{S}}_{B}}$, we have the two following fundamental properties
\begin{eqnarray}
&& (\bigsqcap{}^{{}^{\overline{{ \mathfrak{S}}_{A}}}}_{i\in I}\,\sigma_{i,A})\widetilde{\otimes} \sigma_B =  \bigsqcap{}^{{}^{{{{\overline{S} }}_{AB}}}}_{i\in I} (\sigma_{i,A}\widetilde{\otimes} \sigma_B),\label{pitensor=tensorpi1preminimal}\\
&& \sigma_A \widetilde{\otimes}  (\bigsqcap{}^{{}^{\overline{{ \mathfrak{S}}_{B}}}}_{i\in I}\,\sigma_{i,B}) =  \bigsqcap{}^{{}^{{{{\overline{S} }}_{AB}}}}_{i\in I} (\sigma_A \widetilde{\otimes} \sigma_{i,B}).\label{pitensor=tensorpi2preminimal}
\end{eqnarray}
\end{theorem}
\begin{proof}
Using the homomorphic property for ${\epsilon}\,{}^{{ \mathfrak{S}}_A}$ and the distributivity property between $\bullet$ and $\wedge$ in ${\mathfrak{B}}$, we deduce
\begin{eqnarray}
{\epsilon}\,{}^{{ \mathfrak{S}}_A}_{{\mathfrak{l}}_A} (\bigsqcap{}^{{}^{\overline{{ \mathfrak{S}}_{A}}}}_{i\in I}\sigma_{i,A}) \bullet {\epsilon}\,{}^{{ \mathfrak{S}}_B}_{{\mathfrak{l}}_B} (\sigma_B) 
&=&
(\bigwedge{}_{i\in I}\;{\epsilon}\,{}^{{ \mathfrak{S}}_A}_{{\mathfrak{l}}_A} (\sigma_{i,A})) \bullet {\epsilon}\,{}^{{ \mathfrak{S}}_B}_{{\mathfrak{l}}_B} (\sigma_B) \nonumber\\
&=&
\bigwedge{}_{i\in I}\;(\,{\epsilon}\,{}^{{ \mathfrak{S}}_A}_{{\mathfrak{l}}_A} (\sigma_{i,A})\bullet {\epsilon}\,{}^{{ \mathfrak{S}}_B}_{{\mathfrak{l}}_B} (\sigma_B)),\label{demobifilterproperty}
\end{eqnarray}
and then, using Lemma \ref{lemmaseparatedminimal}, we obtain the property (\ref{pitensor=tensorpi1preminimal}). We obtain the property (\ref{pitensor=tensorpi2preminimal}) along the same lines of proof.  
\end{proof}

\begin{definition}\label{defsqsubseteqSAB}
${{\overline{S} }}_{AB}$ is equipped with a partial order defined according to
\begin{eqnarray}
\forall \widetilde{\sigma}_{AB},\widetilde{\sigma}'_{AB} \in {{{\overline{S} }}_{AB}},\;\;\;\; (\,\widetilde{\sigma}_{AB} \sqsubseteq_{{}_{{{\overline{S} }}_{AB}}} \widetilde{\sigma}'_{AB}\,) & :\Leftrightarrow &\nonumber\\
&&\hspace{-2cm}(\,\forall {\mathfrak{l}}_A\in { \mathfrak{E}}_{A},\forall {\mathfrak{l}}_B\in { \mathfrak{E}}_{B},\;\;
{\nu}\,{}^{\overline{{ \mathfrak{S}}_{A}{ \mathfrak{S}}_{B}}}_{{\mathfrak{l}}_A, {\mathfrak{l}}_B} (\widetilde{\sigma}_{AB})\leq {\nu}\,{}^{\overline{{ \mathfrak{S}}_{A}{ \mathfrak{S}}_{B}}}_{{\mathfrak{l}}_A,{\mathfrak{l}}_B} (\widetilde{\sigma}'_{AB})\,).\;\;\;\;\;\;\;\;\;\;\;\;
\end{eqnarray}
\end{definition}

\begin{lemma}\label{Lemmadevelopetildeleqetilde}
Let us consider $u_{AB}:=\{\, (\sigma_{i,A},\sigma_{i,B})\;\vert\; i\in I\,\}$ an element of ${ \mathcal{P}}(\overline{{ \mathfrak{S}}_{A}} \times \overline{{ \mathfrak{S}}_{B}})$. We have explicitly, for any $\sigma_{A}\in { \mathfrak{S}}_{A}$ and $\sigma_{B}\in { \mathfrak{S}}_{B}$, the following equivalence
\begin{eqnarray}
 \left( \widetilde{u_{AB}} \sqsubseteq_{{}_{{{\overline{S} }}_{AB}}} \widetilde{(\sigma_{A}, \sigma_{B})} \right)
&\Leftrightarrow &
\left( (\bigsqcap{}^{{}_{\overline{{ \mathfrak{S}}_{A}}}}_{k\in I}\; \sigma_{k,A}) \;\sqsubseteq_{{}_{\overline{{ \mathfrak{S}}_{A}}}}\sigma_{A}
\;\;\textit{\rm and}\;\;
(\bigsqcap{}^{{}_{\overline{{ \mathfrak{S}}_{B}}}}_{m\in I} \;\sigma_{m,B})\; \sqsubseteq_{{}_{\overline{{ \mathfrak{S}}_{B}}}} \sigma_{B}
\;\;\textit{\rm and}\;\;\right.\nonumber\\
&&\hspace{-2cm}\left.\left(\forall \varnothing \varsubsetneq K  \varsubsetneq I, \;\; (\bigsqcap{}^{{}_{\overline{{ \mathfrak{S}}_{A}}}}_{k\in K}\; \sigma_{k,A}) \;\sqsubseteq_{{}_{\overline{{ \mathfrak{S}}_{A}}}}\sigma_{A}\;\;\;\textit{\rm or}\;\;\; (\bigsqcap{}^{{}_{\overline{{ \mathfrak{S}}_{B}}}}_{m\in I-K} \;\sigma_{m,B})\; \sqsubseteq_{{}_{\overline{{ \mathfrak{S}}_{B}}}} \sigma_{B}\right) \right).\;\;\;\;\;\;\;\;\;\;\;\;\;\;
\label{developmentetildeordersimplify}
\end{eqnarray}
It is recalled that ${ \mathfrak{S}}_{A}$ and ${ \mathfrak{S}}_{B}$ are down-complete Inf semi-lattice and then the infima in this formula are well-defined.
\end{lemma}
\begin{proof}
We intent to expand the inequality $ \widetilde{u_{AB}} \sqsubseteq_{{}_{{{\overline{S} }}_{AB}}} \widetilde{(\sigma_{A}, \sigma_{B})}$. It is equivalent to
\begin{eqnarray}
\forall {\mathfrak{l}}_A\in { \mathfrak{E}}_{A}, \forall {\mathfrak{l}}_B\in { \mathfrak{E}}_{B},&&
\left( \bigwedge{}_{i\in I}\;
{\epsilon}\,{}^{{ \mathfrak{S}}_{A}}_{{\mathfrak{l}}_A} ( \sigma_{i,A})\bullet {\epsilon}\,{}^{{ \mathfrak{S}}_{B}}_{{\mathfrak{l}}_B} ( \sigma_{i,B})\right)
\leq
{\epsilon}\,{}^{{ \mathfrak{S}}_{A}}_{{\mathfrak{l}}_A} (\sigma_{A}) \bullet {\epsilon}\,{}^{{ \mathfrak{S}}_{B}}_{{\mathfrak{l}}_B}(\sigma_{B}).\label{assumptionorder2}\;\;\;\;\;\;\;\;\;\;\;\;
\end{eqnarray}
We intent to choose a pertinent set of effects ${\mathfrak{l}}_A\in { \mathfrak{E}}_{A}$
 and ${\mathfrak{l}}_B\in { \mathfrak{E}}_{B}$ to reformulate this inequality. \\
Let us firstly choose ${\mathfrak{l}}_B={ \mathfrak{Y}}_{{}_{{ \mathfrak{E}}_{B}}}$. Using  (\ref{expressionbullet}), we obtain
\begin{eqnarray}
{\epsilon}\,{}^{{ \mathfrak{S}}_{A}}_{{ \mathfrak{l}}_A}(\bigsqcap{}^{{}^{\overline{{ \mathfrak{S}}_{A}}}}_{i\in I}\; \sigma_{i,A})\leq {\epsilon}\,{}^{{ \mathfrak{S}}_{A}}_{{ \mathfrak{l}}_A}(\sigma_{A})
,\forall { \mathfrak{l}}_A \in { \mathfrak{E}}_A,
\end{eqnarray} 
which leads immediately
\begin{eqnarray}
\bigsqcap{}^{{}^{\overline{{ \mathfrak{S}}_{A}}}}_{i\in I}\; \sigma_{i,A} \;\sqsubseteq_{{}_{\overline{{ \mathfrak{S}}_{A}}}} \sigma_{A}.\label{ineqBY}
\end{eqnarray}
Choosing ${\mathfrak{l}}_A={ \mathfrak{Y}}_{{}_{{ \mathfrak{E}}_{A}}}$, we obtain along the same line 
\begin{eqnarray}
\bigsqcap{}^{{}^{\overline{{ \mathfrak{S}}_{B}}}}_{i\in I} \;\sigma_{i,B}\; \sqsubseteq_{{}_{\overline{{ \mathfrak{S}}_{B}}}} \sigma_{B}.\label{ineqAY}
\end{eqnarray}
Let us now consider $\varnothing \varsubsetneq K  \varsubsetneq I$ and let us choose the real effects ${ \mathfrak{l}}_A$ and ${ \mathfrak{l}}_B$ as follows
\begin{eqnarray}
{ \mathfrak{l}}_A:={ \mathfrak{l}}_{(\centerdot\,,\,\bigsqcap{}^{{}^{\overline{{ \mathfrak{S}}_{A}}}}_{k\in K}\; \sigma_{k,A})} && { \mathfrak{l}}_B:={ \mathfrak{l}}_{(\centerdot\,,\,\bigsqcap{}^{{}^{\overline{{ \mathfrak{S}}_{B}}}}_{m\in I-K}\; \sigma_{m,B})}
\end{eqnarray}
Note that ${ \mathfrak{l}}_A$ and ${ \mathfrak{l}}_B$ are \underline{real} effects because $\bigsqcap{}^{{}^{{ \mathfrak{S}}_{A}}}_{k\in K}\; \sigma_{k,A}\in \overline{{ \mathfrak{S}}_{A}}$ and $\bigsqcap{}^{{}^{{ \mathfrak{S}}_{B}}}_{m\in I-K}\; \sigma_{m,B}\in \overline{{ \mathfrak{S}}_{B}}$.\\
We deduce, from the assumption (\ref{assumptionorder2}), that, for this $\varnothing \varsubsetneq K  \varsubsetneq I$, we have
\begin{eqnarray}
 &&(\bigsqcap{}^{{}^{\overline{{ \mathfrak{S}}_{A}}}}_{k\in K}\; \sigma_{k,A} \;\sqsubseteq_{{}_{\overline{{ \mathfrak{S}}_{A}}}} \sigma_{A}) \;\;\textit{\rm or}\;\; (\bigsqcap{}^{{}^{\overline{{ \mathfrak{S}}_{B}}}}_{m\in I-K} \;\sigma_{m,B}\; \sqsubseteq_{{}_{\overline{{ \mathfrak{S}}_{B}}}} \sigma_{B}).\label{ineqN}
\end{eqnarray}
We let the reader check that we have obtained the whole set of independent inequalities reformulating the property (\ref{assumptionorder2}).
\end{proof}

\begin{theorem} \label{theoremaxiomA1bipartitepre}
${{\overline{S} }}_{AB}$ is a down-complete Inf semi-lattice with
\begin{eqnarray}
\forall \{\,u_i\;\vert\; i\in I\,\}\subseteq { \mathcal{P}}(\overline{{ \mathfrak{S}}_{A}} \times \overline{{ \mathfrak{S}}_{B}}),&&\bigsqcap{}^{{}^{{{\overline{S} }}_{AB}}}_{i\in I} \; \widetilde{u_i}  = \widetilde{\bigcup{}_{i\in I}\; u_i}.\label{tensorinfSAB}
\end{eqnarray}
\end{theorem}

\begin{theorem}\label{bottomcompound}
If ${ \mathfrak{S}}_{A}$ and ${ \mathfrak{S}}_{B}$ admit $\bot_{{}_{{\mathfrak{S}}_A}}$ and $\bot_{{}_{{\mathfrak{S}}_B}}$ respectively as bottom elements, then
${{\overline{S} }}_{AB}$ admits a bottom element explicitly given by $\bot_{{}_{{{\overline{S} }}_{AB}}}=\bot_{{}_{{\mathfrak{S}}_A}}\widetilde{\otimes}\bot_{{}_{{\mathfrak{S}}_B}}$.
\end{theorem}
\begin{proof}
Trivial using the expansion (\ref{developmentetildeordersimplify}).
\end{proof}

\begin{theorem}\label{theorempartialtracesminimal}
We have the two following homomorphisms
\begin{eqnarray}
\hspace{-1cm}\begin{array}{rcrclccrcrcl}
&  &\zeta^{\overline{{ \mathfrak{S}}_{A}}\overline{{ \mathfrak{S}}_{B}}}_{(1)} : \;\;\;\;\;\; {{\overline{S} }}_{AB} & \longrightarrow & \overline{{ \mathfrak{S}}_{A}}  & & &  &  &\zeta^{\overline{{ \mathfrak{S}}_{A}}\overline{{ \mathfrak{S}}_{B}}}_{(2)} : \;\;\;\;\;\; {{\overline{S} }}_{AB} & \longrightarrow & \overline{{ \mathfrak{S}}_{B}}  \\
& & \bigsqcap{}^{{}^{{ \overline{S} }_{AB}}}_{i\in I} \sigma_{i,A} \widetilde{\otimes} \sigma_{i,B} & \mapsto & \bigsqcap{}^{{}^{{\overline{{ \mathfrak{S}}_{A}}}}}_{i\in I} \sigma_{i,A}  & & & & & \bigsqcap{}^{{}^{{ \overline{S} }_{AB}}}_{i\in I} \sigma_{i,A} \widetilde{\otimes} \sigma_{i,B} & \mapsto & \bigsqcap{}^{{}^{{\overline{{ \mathfrak{S}}_{B}}}}}_{i\in I}  \sigma_{i,B}
\end{array}
\end{eqnarray}
\end{theorem}

\vspace{0.2cm}

\subsection{Minimal tensor product vs. canonical tensor product}\label{subsectionremarkstensorclassical}

In the present subsection we recall some basic elements of the canonical construction for tensor product and we analyze its relation with our own construction of tensor product presented in the last subsection.

\begin{definition}{\bf \cite{Fraser1976}}\label{theoremtensorbasic}
The canonical tensor product denoted ${S}_{AB}:=\overline{{ \mathfrak{S}}_{A}} \otimes \overline{{ \mathfrak{S}}_{B}}$ of the two Inf semi-lattices $\overline{{ \mathfrak{S}}_{A}}$ and $\overline{{ \mathfrak{S}}_{B}}$ is obtained as the solution of the following universal problem : there exists a bi-homomorphism, denoted $\iota$ from $\overline{{ \mathfrak{S}}_{A}} \times \overline{{ \mathfrak{S}}_{B}}$ to ${S}_{AB}$, such that, for any Inf semi-lattice ${ \mathfrak{S}}$ and any bi-homomorphism $f$ from $\overline{{ \mathfrak{S}}_{A}} \times \overline{{ \mathfrak{S}}_{B}}$ to ${ \mathfrak{S}}$,  there is a unique homomorphism $g$ from ${S}_{AB}$ to ${ \mathfrak{S}}$ with $f = g \circ \iota$. We denote $\iota(\sigma,\sigma')=\sigma \otimes \sigma'$ for any $\sigma\in \overline{{ \mathfrak{S}}_{A}}$ and $\sigma'\in \overline{{ \mathfrak{S}}_{B}}$.  \\
The tensor product ${S}_{AB}$ exists and is unique up to isomorphism, it is built as the homomorphic image of the free $\sqcap$ semi-lattice generated by the set $\overline{{ \mathfrak{S}}_{A}} \times \overline{{ \mathfrak{S}}_{B}}$ under the congruence relation determined by identifying $(\sigma_1 \sqcap_{{}_{\overline{{ \mathfrak{S}}_{A}}}}\sigma_2, \sigma')$ with $(\sigma_1,  \sigma')\sqcap (\sigma_2, \sigma')$ for all $\sigma_1,\sigma_2\in \overline{{ \mathfrak{S}}_{A}}, \sigma'\in \overline{{ \mathfrak{S}}_{B}}$ and identifying  $(\sigma, \sigma'_1 \sqcap_{{}_{\overline{{ \mathfrak{S}}_{B}}}}\sigma'_2)$ with $(\sigma,  \sigma'_1)\sqcap (\sigma, \sigma'_2)$ for all $\sigma\in \overline{{ \mathfrak{S}}_{A}},\sigma'_1,\sigma'_2\in \overline{{ \mathfrak{S}}_{B}}$. \\
Then, from now on, ${S}_{AB}$ is the Inf semi-lattice (the infimum of $S\subseteq {S}_{AB}$ will be denoted $\bigsqcap{}^{{}^{{S}_{AB}}} S$) generated by the elements $\sigma_A \otimes \sigma_B$ with $\sigma_A \in \overline{{ \mathfrak{S}}_{A}}, \sigma_B \in \overline{{ \mathfrak{S}}_{B}}$ and subject exclusively to the conditions 
\begin{eqnarray}
&&\left\{\begin{array}{l}
(\sigma_A\sqcap_{{}_{\overline{{ \mathfrak{S}}_{A}}}}\sigma'_A)\otimes \sigma_B = (\sigma_A\otimes \sigma_B)\sqcap_{{}_{{S}_{AB}}}(\sigma'_A\otimes \sigma_B),\\
 \sigma_A \otimes (\sigma_B\sqcap_{{}_{\overline{{ \mathfrak{S}}_{B}}}}\sigma'_B) = (\sigma_A\otimes \sigma_B)\sqcap_{{}_{{S}_{AB}}}(\sigma_A\otimes \sigma'_B). 
\end{array}\right.
\end{eqnarray}
The space ${S}_{AB}=\overline{{ \mathfrak{S}}_{A}} \otimes \overline{{ \mathfrak{S}}_{B}}$ is turned into a partially ordered set with the following binary relation
\begin{eqnarray}
\forall \sigma_{AB},\sigma'_{AB} \in {S}_{AB},\;\;\;\; 
(\,\sigma_{AB} \sqsubseteq_{{}_{{S}_{AB}}} \sigma'_{AB}\,) & :\Leftrightarrow &
(\,\sigma_{AB} \sqcap_{{}_{{S}_{AB}}} \sigma'_{AB} = \sigma_{AB}\,).\;\;\;\;\;\;\;\;\;\;\;\;
\end{eqnarray}
\end{definition}

\begin{definition} 
A non-empty subset ${ \mathfrak{R}}$ of $\overline{{ \mathfrak{S}}_{A}} \times \overline{{ \mathfrak{S}}_{B}}$ is called a bi-filter iff 
\begin{eqnarray}
&&\forall \sigma_A,\sigma_{1,A},\sigma_{2,A}\in \overline{{ \mathfrak{S}}_{A}},\forall \sigma_B,\sigma_{1,B},\sigma_{2,B}\in \overline{{ \mathfrak{S}}_{B}},\nonumber\\
&&(\,  (\sigma_{1,A},\sigma_{1,B})\leq (\sigma_{2,A},\sigma_{2,B}) \;\;\textit{\rm and}\;\; (\sigma_{1,A},\sigma_{1,B})\in { \mathfrak{R}} \,)\;\;\Rightarrow\;\; (\sigma_{2,A},\sigma_{2,B})\in { \mathfrak{R}},\;\;\;\;\;\;\;\;\;\;\label{defbifilter1}\\
&&(\sigma_{1,A},\sigma_{B}),(\sigma_{2,A},\sigma_{B})\in { \mathfrak{R}}\;\;\Rightarrow\;\; (\sigma_{1,A}\sqcap_{{}_{\overline{{ \mathfrak{S}}_A}}}\sigma_{2,A},\sigma_{B})\in { \mathfrak{R}},\label{defbifilter2}\\
&&(\sigma_{A},\sigma_{1,B}),(\sigma_{A},\sigma_{2,B})\in { \mathfrak{R}}\;\;\Rightarrow\;\; (\sigma_{A},\sigma_{1,B}\sqcap_{{}_{\overline{{ \mathfrak{S}}_B}}}\sigma_{2,B})\in { \mathfrak{R}}.\label{defbifilter3}
\end{eqnarray}
\end{definition}
\begin{definition}
If $\{(\sigma_{1,A},\sigma_{1,B}),\cdots,(\sigma_{n,A},\sigma_{n,B})\}$ is a non-empty finite subset of $\overline{{ \mathfrak{S}}_{A}} \times \overline{{ \mathfrak{S}}_{B}}$, then the intersection of all bi-filters of $\overline{{ \mathfrak{S}}_{A}} \times \overline{{ \mathfrak{S}}_{B}}$ which contain $(\sigma_{1,A},\sigma_{1,B})$, $\cdots$, $(\sigma_{n,A},\sigma_{n,B})$ is a bi-filter, which we denote by ${ \mathfrak{F}}\{(\sigma_{1,A},\sigma_{1,B}),\cdots,(\sigma_{n,A},\sigma_{n,B})\}$.
\end{definition}

\begin{lemma}\label{FalphaF}
If $F$ is a filter of ${S}_{AB}$ then the set $\alpha(F):=\{\,(\sigma_{A},\sigma_{B})\in \overline{{ \mathfrak{S}}_{A}} \times \overline{{ \mathfrak{S}}_{B}}\;\vert\; \sigma_{A}\otimes \sigma_{B} \in F\;\}$ is a bi-filter of $\overline{{ \mathfrak{S}}_{A}} \times \overline{{ \mathfrak{S}}_{B}}$.
\end{lemma}

\begin{lemma}{\bf \cite[Lemma 1]{Fraser1978}}\label{sigmainFinequality}
Let us choose $\sigma_A,\sigma_{1,A},\cdots,\sigma_{n,A}\in \overline{{ \mathfrak{S}}_{A}}$ and $\sigma_B,\sigma_{1,B},\cdots,\sigma_{n,B}\in \overline{{ \mathfrak{S}}_{B}}$. Then, 
\begin{eqnarray}
\hspace{-1cm}(\sigma_{A},\sigma_{B})\in { \mathfrak{F}}\{(\sigma_{1,A},\sigma_{1,B}),\cdots,(\sigma_{n,A},\sigma_{n,B})\} & \Leftrightarrow & \left( \bigsqcap{}^{{}^{{S}_{AB}}}_{1\leq i\leq n}\; \sigma_{i,A}\otimes \sigma_{i,B}\right) \sqsubseteq_{{}_{{S}_{AB}}} \sigma_{A}\otimes \sigma_{B}.\;\;\;\;\;\;\;\;\;\;\;\;\;\;
\end{eqnarray}
\end{lemma}

\begin{lemma}{\bf \cite[Theorem 1]{Fraser1978}}\label{orderimpliespolynomial}\\
Let us choose $\sigma_A,\sigma_{1,A},\cdots,\sigma_{n,A}\in \overline{{ \mathfrak{S}}_{A}}$ and $\sigma_B,\sigma_{1,B},\cdots,\sigma_{n,B}\in \overline{{ \mathfrak{S}}_{B}}$. Then, 
\begin{eqnarray}
&&\hspace{-2cm} \left( \bigsqcap{}^{{}^{{S}_{AB}}}_{1\leq i\leq n}\; \sigma_{i,A}\otimes \sigma_{i,B}\right) \sqsubseteq_{{}_{{S}_{AB}}} \sigma_{A}\otimes \sigma_{B}  \Leftrightarrow \nonumber\\  
&& (\,\textit{\rm there exists a n$-$ary lattice polynomial $p$}\;\vert\; \sigma_{A}\sqsupseteq_{{}_{\overline{{ \mathfrak{S}}_A}}} p(\sigma_{1,A},\cdots,\sigma_{n,A})\nonumber\\
&&\textit{\rm and}\;\;  \sigma_{B}\sqsupseteq_{{}_{\overline{{ \mathfrak{S}}_B}}} p^\ast(\sigma_{1,B},\cdots,\sigma_{n,B}) \,).
\end{eqnarray}
where $p^\ast$ denotes the lattice polynomial obtained from $p$ by dualizing the lattice operations.
\end{lemma}

\begin{definition}
We denote by ${\overline{S} }{}^{fin}_{AB}$ the sub-poset of ${\overline{S} }_{AB}$ formed by finite infima of pure tensors. 
It is also a sub- Inf semi-lattice of ${\overline{S} }_{AB}$.
\end{definition}

\begin{definition}
We denote by
\begin{eqnarray}
\widetilde{ \mathfrak{F}}\{(\sigma_{i,A},\sigma_{i,B})\;\vert\; i\in I\}&:=&\{\,(\sigma'_A,\sigma'_B)\;\vert\; (\bigsqcap{}^{{}^{{\overline{S} }_{AB}}}_{i\in I} \sigma_{i,A}\widetilde{\otimes} \sigma_{i,B}) \sqsubseteq_{{}_{{\overline{S} }_{AB}}}   \sigma'_{A}\widetilde{\otimes} \sigma'_{B}\}\nonumber\\ 
\end{eqnarray}
\end{definition}

\begin{theorem} \label{theoremsqsubseteqPABgeqSAB} 
We have the following obvious property relating the partial orders of ${{\overline{S} }}{}^{fin}_{AB}$ and ${S}_{AB}$. For any $\{(\sigma_{i,A},\sigma_{i,B})\;\vert\; i\in I\}\subseteq_{fin} S_{AB}$,
\begin{eqnarray}
(\bigsqcap{}^{{}^{{S}_{AB}}}_{i\in I} \sigma_{i,A}\otimes \sigma_{i,B}) \sqsubseteq_{{}_{{S}_{AB}}}  \sigma'_{A}\otimes \sigma'_{B}  & \Rightarrow & (\bigsqcap{}^{{}^{{\overline{S} }_{AB}}}_{i\in I} \sigma_{i,A}\widetilde{\otimes} \sigma_{i,B}) \sqsubseteq_{{}_{{\overline{S} }_{AB}}}   \sigma'_{A}\widetilde{\otimes} \sigma'_{B}.\;\;\;\;\;\;\;\;\;\;\;
\end{eqnarray}
\end{theorem}
\begin{proof}
First of all, it is clear that $\widetilde{ \mathfrak{F}}\{(\sigma_{i,A},\sigma_{i,B})\;\vert\; i\in I\}$ is a bi-filter.\\ 
Secondly, it is easy to check that $(\sigma_{k,A},\sigma_{k,B})\in \widetilde{ \mathfrak{F}}\{(\sigma_{i,A},\sigma_{i,B})\;\vert\; i\in I\}$ for any $k\in I$. Indeed, for any $K\subseteq I$, if $k\in K$ we have $(\bigsqcap{}^{{}_{\overline{{ \mathfrak{S}}_{A}}}}_{l\in K}\; \sigma_{l,A}) \;\sqsubseteq_{{}_{\overline{{ \mathfrak{S}}_{A}}}}\sigma_{k,A} $ and if $k\notin K$ we have $(\bigsqcap{}^{{}_{\overline{{ \mathfrak{S}}_{B}}}}_{m\in I-K} \;\sigma_{m,B})\; \sqsubseteq_{{}_{\overline{{ \mathfrak{S}}_{B}}}} \sigma_{k,B}$. \\
As a conclusion,  and by definition of ${ \mathfrak{F}}\{(\sigma_{i,A},\sigma_{i,B})\;\vert\; i\in I\}$ as the intersection of all bi-filters containing $(\sigma_{i,A},\sigma_{i,B})$ for any $i\in I$, we have then 
$\widetilde{ \mathfrak{F}}\{(\sigma_{i,A},\sigma_{i,B})\;\vert\; i\in I\} \supseteq  { \mathfrak{F}}\{(\sigma_{i,A},\sigma_{i,B})\;\vert\; i\in I\}$.\\
We now use Lemma \ref{sigmainFinequality} to obtain the announced result.
\end{proof}

\begin{theorem} \label{theoremsqsubseteqPAB=SAB}
If $\overline{{ \mathfrak{S}}_{A}}$ \underline{or} $\overline{{ \mathfrak{S}}_{B}}$ are distributive (for example as simplex spaces of states), then ${{\overline{S} }}{}^{fin}_{AB}$ and ${S}_{AB}$ are isomorphic posets.  
\end{theorem}
\begin{proof}
We suppose that $\overline{{ \mathfrak{S}}_{A}}$ or $\overline{{ \mathfrak{S}}_{B}}$ is distributive and we intent to prove that ${ \mathfrak{F}}\{(\sigma_{i,A},\sigma_{i,B})\;\vert\; i\in I\}=\widetilde{ \mathfrak{F}}\{(\sigma_{i,A},\sigma_{i,B})\;\vert\; i\in I\}$ for any $\{(\sigma_{i,A},\sigma_{i,B})\;\vert\; i\in I\}\subseteq_{fin} \overline{{ \mathfrak{S}}_{A}}\times \overline{{ \mathfrak{S}}_{B}}$.\\
Let us prove the following fact : every bi-filter $F$ which contains $(\sigma_{k,A},\sigma_{k,B})$ for any $k\in I$ contains also $\widetilde{ \mathfrak{F}}\{(\sigma_{i,A},\sigma_{i,B})\;\vert\; i\in I\}$.  In fact, we can show that, for any bi-filter $F$ we have
\begin{eqnarray}
&&\hspace{-1cm}(\forall k\in I,\;(\sigma_{k,A},\sigma_{k,B})\in F) \Rightarrow \nonumber \\
&& (\bigsqcup{}^{{}_{\overline{{ \mathfrak{S}}_{A}}}}_{K\in { \mathcal{K}}}\bigsqcap{}^{{}_{\overline{{ \mathfrak{S}}_{A}}}}_{k\in K}\; \sigma_{k,A}, \bigsqcup{}^{{}_{\overline{{ \mathfrak{S}}_{A}}}}_{K'\in { \mathcal{K}}'}\bigsqcap{}^{{}_{\overline{{ \mathfrak{S}}_{B}}}}_{m\in I-K'} \;\sigma_{m,B})\in F, \nonumber\\
&& \forall { \mathcal{K}},{ \mathcal{K}}'\subseteq 2^I,{ \mathcal{K}}\cup { \mathcal{K}}'=2^I, { \mathcal{K}}\cap { \mathcal{K}}'=\varnothing, \{\varnothing\}\in { \mathcal{K}}', I\in { \mathcal{K}}.\;\;\;\;\;\;\;\;\;\;\;\;\;\label{intermediatedistrib}
\end{eqnarray}

The first step towards (\ref{intermediatedistrib}) is obtained by checking that $\forall { \mathcal{K}},{ \mathcal{K}}'\subseteq 2^I,{ \mathcal{K}}\cup { \mathcal{K}}'=2^I, { \mathcal{K}}\cap { \mathcal{K}}'=\varnothing, \{\varnothing\}\in { \mathcal{K}}', I\in { \mathcal{K}}$,
\begin{eqnarray}
 && (\bigsqcup{}^{{}^{{ \mathfrak{S}}}}_{K'\in { \mathcal{K}}'}\bigsqcap{}^{{}^{{ \mathfrak{S}}}}_{m\in I-K'} \;\sigma_{m}) \sqsupseteq_{{}_{ \mathfrak{S}}}
(\bigsqcap{}^{{}^{{ \mathfrak{S}}}}_{K\in { \mathcal{K}}}\bigsqcup{}^{{}^{{ \mathfrak{S}}}}_{k\in K}\; \sigma_{k})\label{firststep}
\end{eqnarray}
for any distributive ${ \mathfrak{S}}$ and any collection of elements of ${ \mathfrak{S}}$  denoted $\sigma_k$ for $k\in I$ for which these two sides of inequality exist.  To check this fact, we have to note that,  using \cite[Lemma 8 p. 50]{Balbes1975}, we have first of all
\begin{eqnarray}
(\bigsqcap{}^{{}^{{ \mathfrak{S}}}}_{K\in { \mathcal{K}}}\bigsqcup{}^{{}^{{ \mathfrak{S}}}}_{k\in K}\; \sigma_{k})= \bigsqcup{}^{{}^{{ \mathfrak{S}}}}\!\!\! \left\{ \bigsqcap{}^{{}^{{ \mathfrak{S}}}}_{K\in { \mathcal{K}}} \pi_K(A) \;\vert\; A\in \prod_{K\in { \mathcal{K}}} K \right\},
\end{eqnarray}
where $\pi_K$ denotes the projection of the component indexed by K in the cardinal product $\prod_{K\in { \mathcal{K}}} K$. Moreover,  for any $A\in \prod_{K\in { \mathcal{K}}} K$,  there exists $L\in { \mathcal{K}}'$ such that $\bigcup \{ \pi_K(A) \;\vert\; K\in { \mathcal{K}}\}\;\supseteq\; (I\smallsetminus L)$ and then $(\bigsqcap{}^{{}^{{ \mathfrak{S}}}}_{K\in { \mathcal{K}}} \pi_K(A)) \sqsubseteq_{{}_{ \mathfrak{S}}} (\bigsqcap{}^{{}^{{ \mathfrak{S}}}}_{m\in I-L} \;\sigma_{m}) \sqsubseteq_{{}_{ \mathfrak{S}}} (\bigsqcup{}^{{}^{{ \mathfrak{S}}}}_{K'\in { \mathcal{K}}'}\bigsqcap{}^{{}^{{ \mathfrak{S}}}}_{m\in I-K'} \;\sigma_{m})$. As a result, we obtain the property (\ref{firststep}).\\

The second step towards (\ref{intermediatedistrib}) consists in showing that 
\begin{eqnarray}
(\forall k\in I,\;(\sigma_{k,A},\sigma_{k,B})\in F) & \Rightarrow & (\bigsqcup{}^{{}_{\overline{{ \mathfrak{S}}_{A}}}}_{K\in { \mathcal{K}}}\bigsqcap{}^{{}_{\overline{{ \mathfrak{S}}_{A}}}}_{k\in K}\; \sigma_{k,A},  \bigsqcap{}^{{}_{\overline{{ \mathfrak{S}}_{B}}}}_{K\in { \mathcal{K}}}\bigsqcup{}^{{}_{\overline{{ \mathfrak{S}}_{B}}}}_{k\in K}\; \sigma_{k,B})\;\in F\;\;\;\;\;\;\;\;\;\;\;\;\;
\end{eqnarray}
for any ${ \mathcal{K}}\subseteq 2^I$. This intermediary result is obtained by induction on the complexity of the polynomial $(\bigsqcup{}^{{}_{\overline{{ \mathfrak{S}}_{A}}}}_{K\in { \mathcal{K}}}\bigsqcap{}^{{}_{\overline{{ \mathfrak{S}}_{A}}}}_{k\in K}\; \sigma_{k,A})$ by using the following elementary result
\begin{eqnarray}
&&\forall \sigma_A,\sigma'_A\in \overline{{ \mathfrak{S}}_A},\sigma_B,\sigma'_B\in \overline{{ \mathfrak{S}}_B},\;\;\;\; \left( (\sigma_A,\sigma_B),(\sigma'_A,\sigma'_B)\in F \right) \Rightarrow \nonumber\\
&&\hspace{2cm} \left\{
\begin{array}{l}
(\sigma_A \sqcup_{{}_{\overline{{ \mathfrak{S}}_A}}} \sigma'_A,\sigma_B \sqcap_{{}_{\overline{{ \mathfrak{S}}_B}}} \sigma'_B)\in F\nonumber\\
(\sigma_A \sqcap_{{}_{\overline{{ \mathfrak{S}}_A}}} \sigma'_A,\sigma_B \sqcup_{{}_{\overline{{ \mathfrak{S}}_B}}} \sigma'_B)\in F
\end{array}\right.
\end{eqnarray}
trivially deduced using the bi-filter character of $F$, i.e. properties (\ref{defbifilter1})(\ref{defbifilter2})(\ref{defbifilter3}). \\

As a final conclusion,  using the explicit definition of ${ \mathfrak{F}}\{(\sigma_{i,A},\sigma_{i,B})\;\vert\; i\in I\}$ as the intersection of all bi-ideals containing $(\sigma_{k,A},\sigma_{k,B})$ for any $k\in I$,  we obtain $\widetilde{ \mathfrak{F}}\{(\sigma_{i,A},\sigma_{i,B})\;\vert\; i\in I\}={ \mathfrak{F}}\{(\sigma_{i,A},\sigma_{i,B})\;\vert\; i\in I\}$.\\

${{\overline{S} }}{}^{fin}_{AB}$ and ${S}_{AB}$ are then isomorphic posets.
\end{proof}

\begin{theorem}\label{SASBditribandcup}
If ${ \mathfrak{S}}_A$ \underline{and} ${{ \mathfrak{S}}_B}$ are  simplex spaces of states, then ${{\overline{S} }}_{AB}$ is also a simplex space of states.  In other words, the bipartite experiments based on a pair ${ \mathfrak{S}}_A,{ \mathfrak{S}}_B$ of deterministic spaces of states are described by a deterministic space of states obtained as the canonical tensor product ${ \mathfrak{S}}_A {\otimes}{ \mathfrak{S}}_B$ (Indeed, using Theorem \ref{theoremsqsubseteqPAB=SAB}, we have also ${{\overline{S} }}{}^{fin}_{AB}={S}_{AB}$ in this situation).
\\
In that case, the explicit expression for the supremum of two elements in ${{\overline{S} }}{}^{fin}_{AB}$ is given by
\begin{eqnarray}
&&\hspace{-1cm}(\bigsqcap{}^{{}^{{{\overline{S} }}_{AB}}}_{i\in I} \sigma_{i,A}\widetilde{\otimes} \sigma_{i,B}) \sqcup{}_{{}_{{{\overline{S} }}_{AB}}} (\bigsqcap{}^{{}^{{{\overline{S} }}_{AB}}}_{j\in J} \sigma'_{j,A}\widetilde{\otimes} \sigma'_{j,B}) = \nonumber\\
&&\hspace{1cm} =\bigsqcap{}^{{}^{{{\overline{S} }}_{AB}}}_{i\in I,\;j\in J}\; (\sigma_{i,A}\sqcup_{{}_{{ \mathfrak{S}}_A}}\sigma'_{j,A}) \widetilde{\otimes} (\sigma_{i,B} \sqcup_{{}_{{ \mathfrak{S}}_B}} \sigma'_{j,B}). \;\;\;\;\;\;\;\;\;\label{formulacupSAB}
\end{eqnarray}
\end{theorem}
\begin{proof}
If  ${{\overline{S} }}_{AB}$ were not a simplex, it would mean that we would have a family  $(\sigma_{i,A}\widetilde{\otimes}\sigma_{i,B})_{i\in I}$ of pure states of ${{\overline{S} }}_{AB}$, and another pure state $\sigma_{A}\widetilde{\otimes}\sigma_{B}$ such that 
\begin{eqnarray}
\bigsqcap{}^{{}^{{{\overline{S} }}_{AB}}}_{i\in I} \sigma_{i,A}\widetilde{\otimes} \sigma_{i,B} &\sqsubseteq_{{{\overline{S} }}_{AB}}
& \sigma_{A}\widetilde{\otimes}\sigma_{B} \label{equationforsimplextensor}
\end{eqnarray}
but with
\begin{eqnarray}
\sigma_{A}\widetilde{\otimes}\sigma_{B} &\notin &
\{\,\sigma_{i,A}\widetilde{\otimes} \sigma_{i,B} \;\vert\; i\in I\,\}.
\end{eqnarray}
However, using the expansion (\ref{developmentetildeordersimplify}) with (\ref{equationforsimplextensor}), we deduce $\sigma_{B}\sqsupseteq_{{}_{{ \mathfrak{S}}_B}} \bigsqcap{}^{{}^{{ \mathfrak{S}}_B}}_{i\in I} \sigma_{i,B}$ and $\sigma_{A}\sqsupseteq_{{}_{{ \mathfrak{S}}_A}} \bigsqcap{}^{{}^{{ \mathfrak{S}}_A}}_{i\in I} \sigma_{i,A}$ and $\forall \varnothing \varsubsetneq K  \varsubsetneq I, \;\; (\bigsqcap{}^{{}_{{ \mathfrak{S}}_{A}}}_{k\in K}\; \sigma_{k,A}) \;\sqsubseteq_{{}_{{ \mathfrak{S}}_{A}}}\sigma_{A}\;\;\;\textit{\rm or}\;\;\; (\bigsqcap{}^{{}_{{ \mathfrak{S}}_{B}}}_{m\in I-K} \;\sigma_{m,B})\; \sqsubseteq_{{}_{{ \mathfrak{S}}_{B}}} \sigma_{B}$. Now, we use the simplexity of ${ \mathfrak{S}}_A$ and ${ \mathfrak{S}}_B$ to deduce that $\sigma_{B}\in \{\,\sigma_{i,B}\;\vert\; i\in I\,\}$ and $\sigma_{A}\in \{\,\sigma_{i,A}\;\vert\; i\in I\,\}$ and $\forall \varnothing \varsubsetneq K  \varsubsetneq I, \;\; \sigma_{A}\in \{\,\sigma_{i,A}\;\vert\; i\in K\,\}\;\;\;\textit{\rm or}\;\;\; \sigma_{B}\in \{\,\sigma_{i,B}\;\vert\; i\in I-K\,\}$. We then obtain immediately the following result : there exists a unique $i\in I$ such that $\sigma_{A}\widetilde{\otimes}\sigma_{B}=\sigma_{i,A}\widetilde{\otimes} \sigma_{i,B}$. We have then obtained a contradiction. As a conclusion, ${{\overline{S} }}_{AB}$ is a simplex.\\ 
Endly, the formula (\ref{formulacupSAB}) is a direct consequence of the distributivity in ${{\overline{S} }}_{AB}$ (see \cite[Theorem 3]{Fraser1978}).
\end{proof}

\subsection{Complementary results about the minimal tensor product}\label{subsectiontensorpreliminary}

During this subsection, we consider the two ontic completions $(\!(\overline{ \mathfrak{S}}_A,\star)\!)_c$ and $(\!(\overline{ \mathfrak{S}}_B,\star)\!)_c$. 

The aim of this introductory subsection is to prove that ${{{ \overline{S}}}}_{AB}:=\overline{ \mathfrak{S}}_A\widetilde{\otimes}\overline{ \mathfrak{S}}_B$ fulfills the basic conditions necessary to consider its "ontic completion".  We recall that the basic conditions necessary for ${{{ \overline{S}}}}_{AB}$ to admit an  ontic completion are
\begin{enumerate}
\item ${{{ \overline{S}}}}_{AB}$ has to be a down complete Inf semi-lattice with bottom element $\bot_{{}_{{{{ \overline{S}}}}_{AB}}}$, such that  ${{{ \overline{S}}}}_{AB}$ is generated by its maximal elements, i.e. 
 \begin{eqnarray}
\hspace{-1.5cm}&&\forall \sigma \in {{{ \overline{S}}}}_{AB}, \;\; \sigma= \bigsqcap{}^{{}^{{{{ \overline{S}}}}_{AB}}}  \underline{\sigma}_{{}_{{{{ \overline{S}}}}_{AB}}} \;\;\textit{\rm where}\;\;
\underline{\sigma}_{{}_{ {{{ \overline{S}}}}_{AB}}}:=\{\, \sigma'\in  {{{ \overline{S}}}}{}_{AB}^{\,{}^{pure}}\;\vert\; \sigma'\sqsupseteq_{{}_{{{{ \overline{S}}}}_{AB}}} \sigma\;\}
\;\;\textit{\rm and}\;\; {{{ \overline{S}}}}{}_{AB}^{\,{}^{pure}}:=Max({{{ \overline{S}}}}_{AB})\;\;\;\;\;\;\;\;\;\;
\end{eqnarray}
\item $\star$ has to be a map from ${{{ \overline{S}}}}_{AB}\smallsetminus \{\bot_{{}_{{{{ \overline{S}}}}_{AB}}}\}$ to ${{{ \overline{S}}}}_{AB}\smallsetminus \{\bot_{{}_{{{{ \overline{S}}}}_{AB}}}\}$ satisfying
\begin{eqnarray}
\forall \sigma\in {{{ \overline{S}}}}_{AB}\smallsetminus \{\bot_{{}_{{{{ \overline{S}}}}_{AB}}}\},&& (\sigma^{\star})^{\star}=\sigma, \\
\forall \sigma_1,\sigma_2\in {{{ \overline{S}}}}_{AB}\smallsetminus \{\bot_{{}_{{{{ \overline{S}}}}_{AB}}}\},&& \sigma_1\sqsubseteq_{{}_{{{{{ \overline{S}}}}_{AB}}}}\sigma_2\;\;\Rightarrow\;\; \sigma_2^\star \sqsubseteq_{{}_{{{{{ \overline{S}}}}_{AB}}}} \sigma_1^\star,\\
\forall \sigma\in {{{ \overline{S}}}}_{AB}\smallsetminus \{\bot_{{}_{{{{ \overline{S}}}}_{AB}}}\},&& \neg \;\widehat{\sigma^{\star}\sigma}{}^{{}^{{{{{ \overline{S}}}}_{AB}}}}.
\end{eqnarray}
\end{enumerate}

\vspace{0.3cm}

We have already proven, in subsection \ref{subsectioncompoundfirstremarks}, that ${{{ \overline{S}}}}_{AB}$ is a down complete Inf semi-lattice and that it admits a bottom element (Theorem \ref{bottomcompound}). 
Let us then prove that ${{{ \overline{S}}}}_{AB}$ is generated by its pure states.\\

We begin first to characterize the maximal states of ${{{ \overline{S}}}}_{AB}=\overline{ \mathfrak{S}}_A\widetilde{\otimes}\overline{ \mathfrak{S}}_B$. 
\begin{theorem}\label{theorempuretilde} 
\begin{eqnarray}
{{{ \overline{S}}}}{}^{\;{}^{pure}}_{AB} & = & \{\, \sigma_A \widetilde{\otimes} \sigma_B\;\vert\; \sigma_A \in \overline{ \mathfrak{S}}^{\;{}^{pure}}_{A}, \sigma_B \in \overline{ \mathfrak{S}}^{\;{}^{pure}}_{B}\,\} = Max({{{ \overline{S}}}}_{AB})
\end{eqnarray}
\end{theorem}
\begin{proof}
First of all, it is a tautological fact that the completely meet-irreducible elements of  ${{{ \overline{S}}}}_{AB}$ are necessarily pure tensors of  ${{{ \overline{S}}}}_{AB}$, i.e.  elements of the form $\sigma_A \widetilde{\otimes} \sigma_B$. \\
Let us then consider $\sigma_A \widetilde{\otimes} \sigma_B$ a completely meet-irreducible element of ${{{ \overline{S}}}}_{AB}$ and let us assume that $\sigma_A  = \bigsqcap{}^{{}^{{ \mathfrak{S}}_{A}}}_{i\in I} \sigma_{i,A}$ for $ \sigma_{i,A}\in \overline{ \mathfrak{S}}_{A}$ for any $i\in I$. We have then $(\sigma_A \widetilde{\otimes} \sigma_B) = ((\bigsqcap{}^{{}^{{ \mathfrak{S}}_{A}}}_{i\in I} \sigma_{i,A})\widetilde{\otimes}  \sigma_B)=\bigsqcap{}^{{}^{{{{ \overline{S}}}}_{AB}}}_{i\in I} ( \sigma_{i,A}\widetilde{\otimes}  \sigma_B)$.  On another part, $\sigma_A \widetilde{\otimes} \sigma_B$ being completely meet-irreducible in ${{{ \overline{S}}}}_{AB}$, there exists $k\in I$ such that $\sigma_A \widetilde{\otimes} \sigma_B=\sigma_{k,A} \widetilde{\otimes} \sigma_B$, i.e, $\sigma_A=\sigma_{k,A}$. As a conclusion, $\sigma_A$ is completely meet-irreducible. In the same way, $\sigma_B$ is completely meet-irreducible.  As a first result, pure states of ${{{ \overline{S}}}}_{AB}$ are necessarily of the form $\sigma_A \widetilde{\otimes} \sigma_B$ with $\sigma_A \in \overline{ \mathfrak{S}}{}^{\;{}^{pure}}_{A}, \sigma_B \in \overline{ \mathfrak{S}}{}^{\;{}^{pure}}_{B}$.\\
Conversely, let us consider $\sigma_A$ a pure state of $\overline{ \mathfrak{S}}_{A}$ and $\sigma_B$ a pure state of $\overline{ \mathfrak{S}}_{B}$, and let us suppose that $ (\bigsqcap{}^{{}^{{{{ \overline{S}}}}_{AB}}}_{i\in I} \sigma_{i,A}\widetilde{\otimes}  \sigma_{i,B})  =  (\sigma_A \widetilde{\otimes} \sigma_B)$ with $\sigma_{i,A}\in \overline{ \mathfrak{S}}_{A}$ and $\sigma_{i,B}\in \overline{ \mathfrak{S}}_{B}$ for any $i\in I$. We now exploit the two conditions $(\bigsqcap{}^{{}_{{ \mathfrak{S}}_{A}}}_{k\in I}\; \sigma_{k,A}) = \sigma_{A}$ and $(\bigsqcap{}^{{}_{{ \mathfrak{S}}_{B}}}_{m\in I} \;\sigma_{m,B}) = \sigma_{B}$ derived from the expansion (\ref{developmentetildeordersimplify}).  From $\sigma_A\in Max(\overline{ \mathfrak{S}}_{A})$ and $\sigma_B\in Max(\overline{ \mathfrak{S}}_{B})$, we deduce that $\sigma_{i,A} = \sigma_A$ and $\sigma_{j,B} = \sigma_B$ for any $i,j\in I$. As a second result, we have then obtained that the state $(\sigma_A \widetilde{\otimes} \sigma_B)$, with $\sigma_A$ a pure state of $\overline{ \mathfrak{S}}_{A}$ and $\sigma_B$ a pure state of $\overline{ \mathfrak{S}}_{B}$, is completely meet-irreducible.\\
From the expansion (\ref{developmentetildeordersimplify}), we deduce also immediately that $(\sigma_A \widetilde{\otimes} \sigma_B)\in Max({{{ \overline{S}}}}_{AB})$ as long as $\sigma_A\in Max(\overline{ \mathfrak{S}}_{A})$ and $\sigma_B\in Max(\overline{ \mathfrak{S}}_{B})$.
\end{proof}

Secondly, we intent to prove that these pure states of ${{{ \overline{S}}}}_{AB}$ form a generating set for ${{{ \overline{S}}}}_{AB}$.

\begin{theorem} \label{theoremA4Ptilde}
\begin{eqnarray}
&&\forall \sigma \in {{{ \overline{S}}}}_{AB}, \;\; \sigma= \bigsqcap{}^{{}^{{{{ \overline{S}}}}_{AB}}}  \underline{\,\sigma\,}_{{}_{ {{{ \overline{S}}}}_{AB}}},\;\;\textit{\rm where}\;\;
\underline{\,\sigma\,}_{{}_{ {{{ \overline{S}}}}_{AB}}}=
({{{ \overline{S}}}}{}_{AB}^{\,{}^{pure}} \cap (\uparrow^{{}^{{{{ \overline{S}}}}_{AB}}}\!\!\!\! \sigma) ).
\end{eqnarray}
\end{theorem}
\begin{proof}
Let us fix $\sigma \in {{{ \overline{S}}}}_{AB}$.\\ 
We note that $\sigma \sqsubseteq_{{}_{{{{ \overline{S}}}}_{AB}}} \sigma'$ for any $\sigma'\in ({{{ \overline{S}}}}{}_{AB}^{\,{}^{pure}} \cap (\uparrow^{{}^{{{{ \overline{S}}}}_{AB}}}\!\!\!\! \sigma)) $ and then $\sigma \sqsubseteq_{{}_{{{{ \overline{S}}}}_{AB}}} \bigsqcap{}^{{}^{{{{ \overline{S}}}}_{AB}}}  \underline{\sigma}_{{}_{ {{{ \overline{S}}}}_{AB}}}$. \\
Secondly,  denoting  
$\sigma:=(\bigsqcap{}^{{}^{{{{ \overline{S}}}}_{AB}}}_{i\in I} \sigma_{i,A}\widetilde{\otimes}  \sigma_{i,B})$, we note immediately that, for any 
$\sigma_{A} \in \overline{ \mathfrak{S}}_{A}^{pure}$ and $\sigma_{B} \in \overline{ \mathfrak{S}}_{B}^{pure}$, if 
$\sigma_{A}\sqsupseteq_{{}_{{ \mathfrak{S}}_{A}}} \sigma_{i,A}$ and $ \sigma_{B}\sqsupseteq_{{}_{{ \mathfrak{S}}_{B}}} \sigma_{i,B}$, then $(\sigma_{A}\widetilde{\otimes}  \sigma_{B}) \sqsupseteq_{{}_{{ \mathfrak{S}}_{AB}}} \sigma$, i.e. $(\sigma_{A}\widetilde{\otimes}  \sigma_{B}) \in \underline{\sigma}_{{}_{ {{{ \overline{S}}}}_{AB}}}$. As a consequence, we have
\begin{eqnarray}
(\bigsqcap{}^{{}^{{{{ \overline{S}}}}_{AB}}}_{i\in I}\bigsqcap{}^{{}^{{{{ \overline{S}}}}_{AB}}}_{\sigma_{A} \in \overline{ \mathfrak{S}}_{A}^{pure}\;\vert\; \sigma_{A}\sqsupseteq_{{}_{\overline{ \mathfrak{S}}_{A}}} \sigma_{i,A}}\bigsqcap{}^{{}^{{{{ \overline{S}}}}_{AB}}}_{\sigma_{B} \in \overline{ \mathfrak{S}}_{B}^{pure}\;\vert\; \sigma_{B}\sqsupseteq_{{}_{\overline{ \mathfrak{S}}_{B}}} \sigma_{i,B}} \sigma_{A}\widetilde{\otimes}  \sigma_{B}) \sqsupseteq_{{}_{{{{ \overline{S}}}}_{AB}}} \bigsqcap{}^{{}^{{{{ \overline{S}}}}_{AB}}}  \underline{\sigma}_{{}_{ {{{ \overline{S}}}}_{AB}}}.\;\;\;\;\;\;\;
\end{eqnarray}
Endly, we have
\begin{eqnarray}
\sigma &=&\bigsqcap{}^{{}^{{{{ \overline{S}}}}_{AB}}}_{i\in I} \sigma_{i,A}\widetilde{\otimes}  \sigma_{i,B} =  
\bigsqcap{}^{{}^{{{{ \overline{S}}}}_{AB}}}_{i\in I} (\bigsqcap{}^{{}^{{{{ \overline{S}}}}_{AB}}}_{\sigma_{A} \in \overline{ \mathfrak{S}}_{A}^{pure}\;\vert\; \sigma_{A}\sqsupseteq_{{}_{\overline{ \mathfrak{S}}_{A}}} \sigma_{i,A}}\sigma_{A})\widetilde{\otimes}  (\bigsqcap{}^{{}^{{{{ \overline{S}}}}_{AB}}}_{\sigma_{B} \in \overline{ \mathfrak{S}}_{B}^{pure}\;\vert\; \sigma_{B}\sqsupseteq_{{}_{\overline{ \mathfrak{S}}_{B}}} \sigma_{i,B}}\sigma_{B})\nonumber\\
&= & 
\bigsqcap{}^{{}^{{{{ \overline{S}}}}_{AB}}}_{i\in I}\bigsqcap{}^{{}^{{{{ \overline{S}}}}_{AB}}}_{\sigma_{A} \in \overline{ \mathfrak{S}}_{A}^{pure}\;\vert\; \sigma_{A}\sqsupseteq_{{}_{\overline{ \mathfrak{S}}_{A}}} \sigma_{i,A}}\bigsqcap{}^{{}^{{{{ \overline{S}}}}_{AB}}}_{\sigma_{B} \in \overline{ \mathfrak{S}}_{B}^{pure}\;\vert\; \sigma_{B}\sqsupseteq_{{}_{\overline{ \mathfrak{S}}_{B}}} \sigma_{i,B}} \sigma_{A}\widetilde{\otimes}  \sigma_{B}.
\end{eqnarray}
As a final conclusion, we obtain
\begin{eqnarray}
\sigma = (\bigsqcap{}^{{}^{{{{ \overline{S}}}}_{AB}}}_{i\in I}\bigsqcap{}^{{}^{{{{ \overline{S}}}}_{AB}}}_{\sigma_{A} \in \overline{ \mathfrak{S}}_{A}^{pure}\;\vert\; \sigma_{A}\sqsupseteq_{{}_{\overline{ \mathfrak{S}}_{A}}} \sigma_{i,A}}\bigsqcap{}^{{}^{{{{ \overline{S}}}}_{AB}}}_{\sigma_{B} \in \overline{ \mathfrak{S}}_{B}^{pure}\;\vert\; \sigma_{B}\sqsupseteq_{{}_{\overline{ \mathfrak{S}}_{B}}} \sigma_{i,B}} \sigma_{A}\widetilde{\otimes}  \sigma_{B}) = \bigsqcap{}^{{}^{{{{ \overline{S}}}}_{AB}}}  \underline{\sigma}_{{}_{ {{{ \overline{S}}}}_{AB}}}.\;\;\;\;\;\;\;\;\;\;\;\;\;\;
\end{eqnarray}
\end{proof}

We have then completed the first group of conditions about ${{{ \overline{S}}}}_{AB}$.\\

By the way, we can exploit the previous result to produce a formula for the suprema of elements of ${{{ \overline{S}}}}_{AB}$.

\begin{theorem}\label{formulacupStilde}
Let $\widetilde{\sigma}_{AB}$ and $\widetilde{\sigma}'_{AB}$ be two elements of ${{{ \overline{S}}}}_{AB}$ having a common upper-bound. Then the supremum of $\{\widetilde{\sigma}_{AB},\widetilde{\sigma}'_{AB}\}$ exists in ${{{ \overline{S}}}}_{AB}$ and its expression is given by
\begin{eqnarray}
\widetilde{\sigma}_{AB} \sqcup_{{}_{{{{ \overline{S}}}}_{AB}}} \widetilde{\sigma}'_{AB} = \bigsqcap{}^{{}^{{{{ \overline{S}}}}_{AB}}}_{\widetilde{\sigma} \in (\underline{\widetilde{\sigma}_{AB}}_{{}_{ {{{ \overline{S}}}}_{AB}}}\!\!\!\!\!\!\!\!\!\cap\; \underline{\widetilde{\sigma}'_{AB}}_{{}_{ {{{ \overline{S}}}}_{AB}}}\!\!\!\!\!\!\!\!\!)} \; \widetilde{\sigma}\label{formulacupPtilde2}
\end{eqnarray}
\end{theorem}
\begin{proof} As long as $\widetilde{\sigma}_{AB}$ and $\widetilde{\sigma}'_{AB}$ have a common upper-bound, $\underline{\widetilde{\sigma}_{AB}}\cap \underline{\widetilde{\sigma}'_{AB}}$ is not empty. Secondly, it is clear that $\widetilde{\sigma}_{AB}= (\bigsqcap{}^{{}^{{{{ \overline{S}}}}_{AB}}}_{\widetilde{\sigma} \in \underline{\widetilde{\sigma}_{AB}}} \; \widetilde{\sigma})\; \sqsubseteq_{{}_{{{{ \overline{S}}}}_{AB}}} \bigsqcap{}^{{}^{{{{ \overline{S}}}}_{AB}}}_{\widetilde{\sigma} \in \underline{\widetilde{\sigma}_{AB}}\cap \underline{\widetilde{\sigma}'_{AB}}} \; \widetilde{\sigma}$ and $\widetilde{\sigma}'_{AB}= (\bigsqcap{}^{{}^{{{{ \overline{S}}}}_{AB}}}_{\widetilde{\sigma} \in \underline{\widetilde{\sigma}'_{AB}}} \; \widetilde{\sigma})\; \sqsubseteq_{{}_{{{{ \overline{S}}}}_{AB}}} \bigsqcap{}^{{}^{{{{ \overline{S}}}}_{AB}}}_{\widetilde{\sigma} \in \underline{\widetilde{\sigma}_{AB}}\cap \underline{\widetilde{\sigma}'_{AB}}} \; \widetilde{\sigma}$. Then, if we suppose there exists $\widetilde{\sigma}''_{AB}$ such that $\widetilde{\sigma}_{AB}, \widetilde{\sigma}'_{AB} \sqsubseteq_{{}_{{{{ \overline{S}}}}_{AB}}} \widetilde{\sigma}''_{AB}$ we can use Theorem \ref{theoremA4Ptilde} to obtain the decomposition $\widetilde{\sigma}''_{AB}=(\bigsqcap{}^{{}^{{{{ \overline{S}}}}_{AB}}}_{\widetilde{\sigma} \in \underline{\widetilde{\sigma}''_{AB}}} \; \widetilde{\sigma})$ with necessarily $\forall \widetilde{\sigma}\in \underline{\widetilde{\sigma}''_{AB}},$ $\widetilde{\sigma}_{AB} \sqsubseteq_{{}_{{{{ \overline{S}}}}_{AB}}}\widetilde{\sigma}$ and $\widetilde{\sigma}'_{AB} \sqsubseteq_{{}_{{{{ \overline{S}}}}_{AB}}}\widetilde{\sigma}$, i.e. $\widetilde{\sigma}\in  \underline{\widetilde{\sigma}_{AB}}\cap \underline{\widetilde{\sigma}'_{AB}}$, and then $(\bigsqcap{}^{{}^{{{{ \overline{S}}}}_{AB}}}_{\widetilde{\sigma} \in \underline{\widetilde{\sigma}_{AB}}\cap \underline{\widetilde{\sigma}'_{AB}}} \; \widetilde{\sigma}) \sqsubseteq_{{}_{{{{ \overline{S}}}}_{AB}}} \widetilde{\sigma}''_{AB}$.  \end{proof}

We now intent to show that ${{{ \overline{S}}}}_{AB}$ fulfills the second group of conditions necessary to build its "ontic completion" (the conditions on the star map).\\

Let us first establish a result about atomicity of ${{{ \overline{S}}}}_{AB}$.  We recall that $\overline{ \mathfrak{S}}_{A}$ and $\overline{ \mathfrak{S}}_{B}$ are both atomic. 

\begin{lemma}
$\overline{ \mathfrak{S}}_{A}$ and $\overline{ \mathfrak{S}}_{B}$ being atomic, then ${{{ \overline{S}}}}_{AB}$ is also atomic, i.e. 
\begin{eqnarray}
\exists { \mathcal{A}}_{{{{ \overline{S}}}}_{AB}}\subseteq {{{ \overline{S}}}}_{AB} & \vert &\forall \alpha_{AB}\in { \mathcal{A}}_{{{{ \overline{S}}}}_{AB}},\;\; (\bot_{{}_{\overline{ \mathfrak{S}}_{A}}}\widetilde{\otimes} \bot_{{}_{\overline{ \mathfrak{S}}_{B}}}) \sqcoversubset_{{}_{{{{ \overline{S}}}}_{AB}}} \alpha_{AB},\\
&&\forall \sigma_{AB} \in {{{ \overline{S}}}}_{AB}, \;\; \exists \alpha_{AB}\in { \mathcal{A}}_{{{{ \overline{S}}}}_{AB}}\;\vert\; \alpha_{AB}\sqsubseteq_{{}_{{{{ \overline{S}}}}_{AB}}} \sigma_{AB}.\;\;\;\;\;\;\;\;\;\;\;\;\;\;
\end{eqnarray}
The set of atoms of ${{{ \overline{S}}}}_{AB}$ is indeed defined by
\begin{eqnarray}
{ \mathcal{A}}_{{{{ \overline{S}}}}_{AB}} & := & \{\, (\alpha_A \widetilde{\otimes}\bot_{{}_{{ \mathfrak{S}}_{B}}})\sqcap_{{}_{{{{ \overline{S}}}}_{AB}}} (\bot_{{}_{{ \mathfrak{S}}_{A}}}\widetilde{\otimes} \alpha_B)\;\vert\; \alpha_A\in { \mathcal{A}}_{\overline{ \mathfrak{S}}_{A}},\; \alpha_B\in { \mathcal{A}}_{\overline{ \mathfrak{S}}_{B}}\,\}.
\end{eqnarray}
\end{lemma}
\begin{proof}
Using the expansion (\ref{developmentetildeordersimplify}), we deduce immediately 
\begin{eqnarray}
\forall \alpha_A\in { \mathcal{A}}_{\overline{ \mathfrak{S}}_{A}},\forall \alpha_B\in { \mathcal{A}}_{\overline{ \mathfrak{S}}_{B}}, && (\alpha_A \widetilde{\otimes}\bot_{{}_{\overline{ \mathfrak{S}}_{B}}})\sqcap_{{}_{{{{ \overline{S}}}}_{AB}}} (\bot_{{}_{\overline{ \mathfrak{S}}_{A}}}\widetilde{\otimes} \alpha_B) \;\not\sqsubseteq_{{}_{{{{ \overline{S}}}}_{AB}}} \; \bot_{{}_{\overline{ \mathfrak{S}}_{A}}} \widetilde{\otimes}\bot_{{}_{\overline{ \mathfrak{S}}_{B}}}.\;\;\;\;\;\;\;\;\;\;\;\;\;\;
\end{eqnarray}
In other words,  $\bot_{{}_{\overline{ \mathfrak{S}}_{A}}} \widetilde{\otimes}\bot_{{}_{\overline{ \mathfrak{S}}_{B}}} \sqsubset_{{}_{{{{ \overline{S}}}}_{AB}}}  
(\alpha_A \widetilde{\otimes}\bot_{{}_{\overline{ \mathfrak{S}}_{B}}})\sqcap_{{}_{{{{ \overline{S}}}}_{AB}}} (\bot_{{}_{\overline{ \mathfrak{S}}_{A}}}\widetilde{\otimes} \alpha_B)$.\\

Secondly, let us show that, for any $\sigma_{AB}:=(\bigsqcap{}^{{}^{{{{ \overline{S}}}}_{AB}}}_{i\in I} \sigma_{i,A}\widetilde{\otimes} \sigma_{i,B})$ distinct from $\bot_{{}_{\overline{ \mathfrak{S}}_{A}}} \widetilde{\otimes}\bot_{{}_{\overline{ \mathfrak{S}}_{B}}}$, there exist $\alpha_A\in { \mathcal{A}}_{\overline{ \mathfrak{S}}_{A}}$ and $\alpha_B\in { \mathcal{A}}_{\overline{ \mathfrak{S}}_{B}}$ such that $((\alpha_A \widetilde{\otimes}\bot_{{}_{\overline{ \mathfrak{S}}_{B}}})\sqcap_{{}_{{{{ \overline{S}}}}_{AB}}} (\bot_{{}_{\overline{ \mathfrak{S}}_{A}}}\widetilde{\otimes} \alpha_B)) \sqsubseteq_{{}_{{{{ \overline{S}}}}_{AB}}}  \sigma_{AB}$. Using once again the expansion (\ref{developmentetildeordersimplify}), we know that $\sigma_{AB} \sqsupset_{{}_{{{{ \overline{S}}}}_{AB}}}  
\bot_{{}_{\overline{ \mathfrak{S}}_{A}}} \widetilde{\otimes}\bot_{{}_{\overline{ \mathfrak{S}}_{B}}}$ (or, in other words, $\sigma_{AB} \not\sqsubseteq_{{}_{{{{ \overline{S}}}}_{AB}}}  
\bot_{{}_{\overline{ \mathfrak{S}}_{A}}} \widetilde{\otimes}\bot_{{}_{\overline{ \mathfrak{S}}_{B}}}$) implies that there exists $\varnothing \subseteq K \subseteq I$ such that $(\bigsqcap{}^{{}_{\overline{ \mathfrak{S}}_{A}}}_{k\in K}\; \sigma_{k,A}) \;\sqsupset_{{}_{\overline{ \mathfrak{S}}_{A}}} \bot_{{}_{\overline{ \mathfrak{S}}_{A}}}$ and $(\bigsqcap{}^{{}_{\overline{ \mathfrak{S}}_{B}}}_{m\in I-K} \;\sigma_{m,B})\; \sqsupset_{{}_{\overline{ \mathfrak{S}}_{B}}} \bot_{{}_{\overline{ \mathfrak{S}}_{B}}} $. Let us fix such a $K$ and let us choose $\alpha_A\in { \mathcal{A}}_{\overline{ \mathfrak{S}}_{A}}$ and $\alpha_B\in { \mathcal{A}}_{\overline{ \mathfrak{S}}_{B}}$ such that $(\bigsqcap{}^{{}_{\overline{ \mathfrak{S}}_{A}}}_{k\in K}\; \sigma_{k,A}) \;\sqsupseteq_{{}_{\overline{ \mathfrak{S}}_{A}}} \alpha_A$ and $(\bigsqcap{}^{{}_{\overline{ \mathfrak{S}}_{B}}}_{m\in I-K} \;\sigma_{m,B})\; \sqsupseteq_{{}_{\overline{ \mathfrak{S}}_{B}}} \alpha_B$.  We obtain $(\bigsqcap{}^{{}^{{{{ \overline{S}}}}_{AB}}}_{i\in K} \sigma_{i,A}\widetilde{\otimes} \sigma_{i,B}) \;\sqsupseteq_{{}_{{{{ \overline{S}}}}_{AB}}} (\alpha_A \widetilde{\otimes} \bot_{{}_{\overline{ \mathfrak{S}}_{B}}})$ and $(\bigsqcap{}^{{}^{{{{ \overline{S}}}}_{AB}}}_{i\in I-K} \sigma_{i,A}\widetilde{\otimes} \sigma_{i,B}) \;\sqsupseteq_{{}_{{{{ \overline{S}}}}_{AB}}} (\bot_{{}_{\overline{ \mathfrak{S}}_{A}}} \widetilde{\otimes} \alpha_B)$. As a first conclusion, we obtain $((\alpha_A \widetilde{\otimes}\bot_{{}_{\overline{ \mathfrak{S}}_{B}}})\sqcap_{{}_{{{{ \overline{S}}}}_{AB}}} (\bot_{{}_{\overline{ \mathfrak{S}}_{A}}}\widetilde{\otimes} \alpha_B)) \sqsubseteq_{{}_{{{{ \overline{S}}}}_{AB}}}  \sigma_{AB}$. \\

Thirdly, let us consider $\sigma_{AB}:=(\bigsqcap{}^{{}^{{{{ \overline{S}}}}_{AB}}}_{i\in I} \sigma_{i,A}\widetilde{\otimes} \sigma_{i,B})$ such that $\sigma_{AB} \sqsubseteq_{{}_{{{{ \overline{S}}}}_{AB}}}  
(\alpha_A \widetilde{\otimes}\bot_{{}_{\overline{ \mathfrak{S}}_{B}}})\sqcap_{{}_{{{{ \overline{S}}}}_{AB}}} (\bot_{{}_{\overline{ \mathfrak{S}}_{A}}}\widetilde{\otimes} \alpha_B)$. As a first case, we may have obviously $\sigma_{AB}=\bot_{{}_{\overline{ \mathfrak{S}}_{A}}} \widetilde{\otimes}\bot_{{}_{\overline{ \mathfrak{S}}_{B}}}$. If however $\sigma_{AB}\not= \bot_{{}_{\overline{ \mathfrak{S}}_{A}}} \widetilde{\otimes}\bot_{{}_{\overline{ \mathfrak{S}}_{B}}}$, the previous result implies that there exist $\alpha'_A\in { \mathcal{A}}_{\overline{ \mathfrak{S}}_{A}}$ and $\alpha'_B\in { \mathcal{A}}_{\overline{ \mathfrak{S}}_{B}}$ such that $((\alpha'_A \widetilde{\otimes}\bot_{{}_{\overline{ \mathfrak{S}}_{B}}})\sqcap_{{}_{{{{ \overline{S}}}}_{AB}}} (\bot_{{}_{\overline{ \mathfrak{S}}_{A}}}\widetilde{\otimes} \alpha'_B)) \sqsubseteq_{{}_{{{{ \overline{S}}}}_{AB}}}  \sigma_{AB}$.  Using once again the expansion (\ref{developmentetildeordersimplify}), we deduce immediately that $\alpha_A=\alpha'_A$ and $\alpha_B=\alpha'_B$. As a result, we obtain
\begin{eqnarray}
&&\hspace{-1cm} \sigma_{AB} \sqsubseteq_{{}_{{{{ \overline{S}}}}_{AB}}}  
(\alpha_A \widetilde{\otimes}\bot_{{}_{\overline{ \mathfrak{S}}_{B}}})\sqcap_{{}_{{{{ \overline{S}}}}_{AB}}} (\bot_{{}_{\overline{ \mathfrak{S}}_{A}}}\widetilde{\otimes} \alpha_B) \Rightarrow \nonumber\\
&&\hspace{1cm} \left( \sigma_{AB}=\bot_{{}_{\overline{ \mathfrak{S}}_{A}}} \widetilde{\otimes}\bot_{{}_{\overline{ \mathfrak{S}}_{B}}} \;\;\textit{\rm or}\;\; \sigma_{AB}=(\alpha_A \widetilde{\otimes}\bot_{{}_{\overline{ \mathfrak{S}}_{B}}})\sqcap_{{}_{{{{ \overline{S}}}}_{AB}}} (\bot_{{}_{\overline{ \mathfrak{S}}_{A}}}\widetilde{\otimes} \alpha_B) \right).\;\;\;\;\;\;\;\;\;\;\;\;\;\;
\end{eqnarray}
As a second conclusion, we then obtain $\bot_{{}_{{{{ \overline{S}}}}_{AB}}}\sqcoversubset_{{}_{{{{ \overline{S}}}}_{AB}}}  
(\alpha_A \widetilde{\otimes}\bot_{{}_{\overline{ \mathfrak{S}}_{B}}})\sqcap_{{}_{{{{ \overline{S}}}}_{AB}}} (\bot_{{}_{\overline{ \mathfrak{S}}_{A}}}\widetilde{\otimes} \alpha_B)$.
\end{proof}

\begin{definition}
We define the map $\star$ from ${{{{ \overline{S}}}}_{AB}}\smallsetminus \{\bot_{{}_{\overline{ \mathfrak{S}}_{A}}}\widetilde{\otimes} \bot_{{}_{\overline{ \mathfrak{S}}_{B}}} \}$ to itself as follows. 
\begin{eqnarray}
\forall \alpha\in \overline{ \mathfrak{S}}{}_{A}^{{}^{pure}},\beta \in \overline{ \mathfrak{S}}{}_{B}^{{}^{pure}} ,&& (\alpha\widetilde{\otimes}\beta)^\star := \alpha^\star\widetilde{\otimes} \bot_{{}_{\overline{ \mathfrak{S}}_{B}}} \sqcap_{{}_{{{{{ \overline{S}}}}_{AB}}}} \bot_{{}_{\overline{ \mathfrak{S}}_{A}}}\widetilde{\otimes} \beta^\star\\
\forall \sigma\in {{{{ \overline{S}}}}_{AB}}\smallsetminus \{\bot_{{}_{\overline{ \mathfrak{S}}_{A}}}\widetilde{\otimes} \bot_{{}_{\overline{ \mathfrak{S}}_{B}}} \},&& \sigma^\star := \bigsqcup{}^{{}^{{{{{ \overline{S}}}}_{AB}}}}_{\omega\;\in \;\underline{\, \sigma \,}_{{}_{{{{{ \overline{S}}}}_{AB}}}}} \omega^\star.
\end{eqnarray}
 Let $\sigma$ be an element of ${{{{ \overline{S}}}}_{AB}}$.
\end{definition}

\begin{theorem}
The map $\star$ is involutive and order-reversing and satisfies $\forall \sigma\in {{{ \overline{S}}}}_{AB}\smallsetminus \{\bot_{{}_{{{{ \overline{S}}}}_{AB}}}\},\;\;\;\; \neg \;\widehat{\sigma^{\star}\sigma}{}^{{}^{{{{{ \overline{S}}}}_{AB}}}}$.
\end{theorem}
\begin{proof}
First of all, the order-reversing property is easy to derive. Indeed, $\sigma\sqsupseteq_{{}_{{{{{ \overline{S}}}}_{AB}}}}\sigma'$ implies $\underline{\, \sigma \,}_{{}_{{{{{ \overline{S}}}}_{AB}}}}\subseteq \underline{\, \sigma' \,}_{{}_{{{{{ \overline{S}}}}_{AB}}}}$, and then implies $\sigma^\star \sqsubseteq_{{}_{{{{{ \overline{S}}}}_{AB}}}}\sigma'{}^\star$.\\
Let us now consider $((\alpha\widetilde{\otimes}\beta)^\star)^\star$ for $\alpha\in \overline{ \mathfrak{S}}{}_{A}^{{}^{pure}},\beta \in \overline{ \mathfrak{S}}{}_{B}^{{}^{pure}}$. Due to the order-reversing property and the fact that $(\alpha\widetilde{\otimes}\beta)^\star$ is an atom, we know that $((\alpha\widetilde{\otimes}\beta)^\star)^\star$ is an element of ${{{ \overline{S}}}}_{AB}^{\;{}^{pure}}$. A direct computation shows that $((\alpha\widetilde{\otimes}\beta)^\star)^\star
\sqsubseteq_{{}_{{{{{ \overline{S}}}}_{AB}}}}
\alpha\widetilde{\otimes}\beta$. Indeed, for any $\omega\in \underline{\alpha^\star}_{{}_{\overline{ \mathfrak{S}}_{A}}}$ and any $\omega'\in \overline{ \mathfrak{S}}{}_{B}^{{}^{pure}}$, we have $\omega^\star \widetilde{\otimes} \bot_{{}_{\overline{ \mathfrak{S}}_{B}}} \sqcap_{{}_{{{{{ \overline{S}}}}_{AB}}}} \bot_{{}_{\overline{ \mathfrak{S}}_{A}}}\widetilde{\otimes} \omega'{}^\star  \sqsubseteq_{{}_{{{{{ \overline{S}}}}_{AB}}}} \alpha \widetilde{\otimes} \beta$. And for any $\omega\in \underline{\beta^\star}_{{}_{\overline{ \mathfrak{S}}_{B}}}$ and any $\omega'\in \overline{ \mathfrak{S}}{}_{A}^{{}^{pure}}$, we have $\omega'{}^\star \widetilde{\otimes} \bot_{{}_{\overline{ \mathfrak{S}}_{B}}} \sqcap_{{}_{{{{{ \overline{S}}}}_{AB}}}} \bot_{{}_{\overline{ \mathfrak{S}}_{A}}}\widetilde{\otimes} \omega{}^\star  \sqsubseteq_{{}_{{{{{ \overline{S}}}}_{AB}}}} \alpha \widetilde{\otimes} \beta$. Using these two results, we obtain $(\alpha^\star\widetilde{\otimes} \bot_{{}_{\overline{ \mathfrak{S}}_{B}}} \sqcap_{{}_{{{{{ \overline{S}}}}_{AB}}}} \bot_{{}_{\overline{ \mathfrak{S}}_{A}}}\widetilde{\otimes} \beta^\star)^\star  \sqsubseteq_{{}_{{{{{ \overline{S}}}}_{AB}}}} \alpha \widetilde{\otimes} \beta$. As a first conclusion, we then obtain $((\alpha\widetilde{\otimes}\beta)^\star)^\star
=\alpha\widetilde{\otimes}\beta$\\
Now, using the order-reversing property, we deduce 
\begin{eqnarray}
(\sigma^\star)^\star := (\bigsqcup{}^{{}^{{{{{ \overline{S}}}}_{AB}}}}_{\omega\;\in \;\underline{\, \sigma \,}_{{}_{{{{{ \overline{S}}}}_{AB}}}}} \omega^\star)^\star=\bigsqcap{}^{{}^{{{{{ \overline{S}}}}_{AB}}}}_{\omega\;\in \;\underline{\, \sigma \,}_{{}_{{{{{ \overline{S}}}}_{AB}}}}} (\omega^\star)^\star=\bigsqcap{}^{{}^{{{{{ \overline{S}}}}_{AB}}}}_{\omega\;\in \;\underline{\, \sigma \,}_{{}_{{{{{ \overline{S}}}}_{AB}}}}} \omega=\sigma
\end{eqnarray}
The property $\forall \sigma\in { \mathfrak{S}}'\smallsetminus \{\bot_{{}_{ \mathfrak{S}'}}\}, \neg \;\widehat{\sigma^{\star}\sigma}{}^{{}^{{ \mathfrak{S}'}}}$ is also required for the map $\star$. This property is easy to check using the formula (\ref{developmentetildeordersimplify}).  Indeed, for any $\{\,(\sigma_{i,A},\sigma_{i,B})\;\vert\;i\in I\}\subseteq \overline{ \mathfrak{S}}_{A}^{{}^{pure}}\times \overline{ \mathfrak{S}}_{B}^{{}^{pure}}$ and any $(\sigma_{A},\sigma_{B})\in \overline{ \mathfrak{S}}{}_{A}^{{}^{pure}}\times \overline{ \mathfrak{S}}{}_{B}^{{}^{pure}}$, we have that
\begin{eqnarray}
\left(\bigsqcap{}^{{}^{{{{ \overline{S}}}}_{AB}}}_{i\in I} \sigma_{i,A}\widetilde{\otimes} \sigma_{i,B}\right)\sqsubseteq_{{}_{{{{ \overline{S}}}}_{AB}}}
\sigma_{A}\widetilde{\otimes}_{{}_{{{{ \overline{S}}}}_{AB}}}\sigma_{B}
\end{eqnarray}
is equivalent to
\begin{eqnarray}
 &&\hspace{-1cm}
\left( (\bigsqcap{}^{{}_{{ \mathfrak{S}}_{A}}}_{k\in I}\; \sigma_{k,A}) \;\sqsubseteq_{{}_{{ \mathfrak{S}}_{A}}}\sigma_{A}
\;\;\textit{\rm and}\;\;
(\bigsqcap{}^{{}_{{ \mathfrak{S}}_{B}}}_{m\in I} \;\sigma_{m,B})\; \sqsubseteq_{{}_{{ \mathfrak{S}}_{B}}} \sigma_{B}
\;\;\textit{\rm and}\;\;\right.\nonumber\\
&&\left.\left(\forall \varnothing \varsubsetneq K  \varsubsetneq I, \;\; (\bigsqcap{}^{{}_{{ \mathfrak{S}}_{A}}}_{k\in K}\; \sigma_{k,A}) \;\sqsubseteq_{{}_{{ \mathfrak{S}}_{A}}}\sigma_{A}\;\;\;\textit{\rm or}\;\;\; (\bigsqcap{}^{{}_{{ \mathfrak{S}}_{B}}}_{m\in I-K} \;\sigma_{m,B})\; \sqsubseteq_{{}_{{ \mathfrak{S}}_{B}}} \sigma_{B}\right) \right).\;\;\;\;\;\;\;\;\;\;\;\;\;\;
\end{eqnarray}
On another part, 
\begin{eqnarray}
\left(\bigsqcap{}^{{}^{{{{ \overline{S}}}}_{AB}}}_{i\in I} \sigma_{i,A}\widetilde{\otimes} \sigma_{i,B}\right)^\star=\left(\bigsqcup{}^{{}^{{{{ \overline{S}}}}_{AB}}}_{i\in I}(\sigma_{i,A}^\star\widetilde{\otimes} \bot_{{}_{\overline{ \mathfrak{S}}_{B}}} \sqcap_{{}_{{{{{ \overline{S}}}}_{AB}}}} \bot_{{}_{\overline{ \mathfrak{S}}_{A}}}\widetilde{\otimes} \sigma_{i,B}^\star)\right) 
\sqsubseteq_{{}_{{{{ \overline{S}}}}_{AB}}}
\sigma_{A}\widetilde{\otimes}_{{}_{{{{ \overline{S}}}}_{AB}}}\sigma_{B}\;\;\;\;\;\;\;\;\;
\end{eqnarray}
is equivalent to
\begin{eqnarray}
 &&\hspace{-1cm}
\left( (\bigsqcup{}^{{}_{{ \mathfrak{S}}_{A}}}_{k\in I}\; \sigma_{k,A}^\star) \;\sqsubseteq_{{}_{{ \mathfrak{S}}_{A}}}\sigma_{A}
\;\;\textit{\rm or}\;\;
(\bigsqcup{}^{{}_{{ \mathfrak{S}}_{B}}}_{m\in I} \;\sigma_{m,B}^\star)\; \sqsubseteq_{{}_{{ \mathfrak{S}}_{B}}} \sigma_{B}
\;\;\textit{\rm or}\;\;\right.\nonumber\\
&&\left.\left(\exists \varnothing \varsubsetneq K  \varsubsetneq I, \;\; (\bigsqcup{}^{{}_{{ \mathfrak{S}}_{A}}}_{k\in K}\; \sigma_{k,A}^\star) \;\sqsubseteq_{{}_{{ \mathfrak{S}}_{A}}}\sigma_{A}\;\;\;\textit{\rm and}\;\;\; (\bigsqcup{}^{{}_{{ \mathfrak{S}}_{B}}}_{m\in I-K} \;\sigma_{m,B}^\star)\; \sqsubseteq_{{}_{{ \mathfrak{S}}_{B}}} \sigma_{B}\right) \right).\;\;\;\;\;\;\;\;\;\;\;\;\;\;
\end{eqnarray}
As a conclusion, we obtain that either there exists $\varnothing \varsubsetneq K  \subseteq I$ such that $(\bigsqcap{}^{{}_{{ \mathfrak{S}}_{A}}}_{k\in K}\; \sigma_{k,A})$ and $(\bigsqcup{}^{{}_{{ \mathfrak{S}}_{A}}}_{k\in K}\; \sigma_{k,A}^\star)$ admit $\sigma_A$ as their common upper bound, or there exists $\varnothing \varsubsetneq K  \subseteq I$ such that $(\bigsqcap{}^{{}_{{ \mathfrak{S}}_{B}}}_{k\in K}\; \sigma_{k,B})$ and $(\bigsqcup{}^{{}_{{ \mathfrak{S}}_{B}}}_{k\in K}\; \sigma_{k,B}^\star)$ admit $\sigma_B$ as their common upper bound. We have then obtained a contradiction.  Hence, $\left(\bigsqcap{}^{{}^{{{{ \overline{S}}}}_{AB}}}_{i\in I} \sigma_{i,A}\widetilde{\otimes} \sigma_{i,B}\right)$ and $\left(\bigsqcap{}^{{}^{{{{ \overline{S}}}}_{AB}}}_{i\in I} \sigma_{i,A}\widetilde{\otimes} \sigma_{i,B}\right)^\star$ do not admit any common upper bound in ${{{ \overline{S}}}}{}_{AB}^{\,{}^{pure}}$.
\end{proof}

We have then completed the check of the second condition necessary to build the "ontic completion" of ${{{ \overline{S}}}}_{AB}$.\\


We take the occasion of this introductory subsection to prove two complementary results.


\begin{theorem}\label{firstcoveringlemma}
Let us suppose that $\overline{ \mathfrak{S}}_A$ and $\overline{ \mathfrak{S}}_B$ both satisfy the following basic property
\begin{eqnarray}
\hspace{-0.5cm}\forall \alpha,\beta\in \overline{ \mathfrak{S}}{}^{{}^{pure}},&& \beta\not=\alpha \;\;\Rightarrow \;\; (\alpha\sqcap_{{}_{\overline{ \mathfrak{S}}}}\beta) \sqcoversubset_{{}_{\overline{ \mathfrak{S}}}}\alpha,\beta.\;\;\;\;\;\;\;\;\;\;\;\;\;\;\;\;\;\;\label{coveringpropertySbarpre}
\end{eqnarray}
Then ${{{ \overline{S}}}}_{AB}$ satisfies the same property. \\
In other words, for any $\sigma,\omega\in {{{ \overline{S}}}}{}_{AB}^{{}^{pure}}$ such that $\omega\not=\sigma$ we have then $(\omega\sqcap_{{}_{{{{ \overline{S}}}}_{AB}}}\sigma)\sqcoversubset_{{}_{{{{ \overline{S}}}}_{AB}}}\sigma,\omega$. 
\end{theorem}
\begin{proof}
Trivial consequence of the expansion formula (\ref{developmentetildeordersimplify}).
\end{proof}

\begin{theorem}\label{secondcoveringlemma}
Let us suppose that $\overline{ \mathfrak{S}}_A$ and $\overline{ \mathfrak{S}}_B$ both satisfy the following basic property
\begin{eqnarray}
\hspace{-0.5cm}\forall \lambda,\alpha,\beta,\gamma,\delta\in \overline{ \mathfrak{S}}{}^{{}^{pure}} \;\textit{\rm distinct},&& (\,(\alpha\sqcap_{{}_{\overline{\mathfrak{S}}}}\beta)\not=(\gamma\sqcap_{{}_{\overline{\mathfrak{S}}}}\delta)\;\;\;\textit{\rm and}\;\;\;(\alpha\sqcap_{{}_{\overline{\mathfrak{S}}}}\beta),(\gamma\sqcap_{{}_{\overline{\mathfrak{S}}}}\delta)\sqcoversubset_{{}_{\mathfrak{S}}}\lambda \,) \;\;\Rightarrow \nonumber\\
&&\hspace{1cm} (\alpha\sqcap_{{}_{\overline{ \mathfrak{S}}}}\beta\sqcap_{{}_{\overline{ \mathfrak{S}}}}\gamma\sqcap_{{}_{\overline{ \mathfrak{S}}}}\delta) \sqcoversubset_{{}_{\overline{ \mathfrak{S}}}}(\alpha\sqcap_{{}_{\overline{\mathfrak{S}}}}\beta),(\gamma\sqcap_{{}_{\overline{\mathfrak{S}}}}\delta).\;\;\;\;\;\;\;\;\;\;\;\;\;\;\;\;\;\;\label{secondcoveringpropertySbarpre}
\end{eqnarray}
Then ${{{ \overline{S}}}}_{AB}$ satisfies the same property.  \\
In other words, if we consider $\lambda,\alpha,\beta,\gamma,\delta\in {{{ \overline{S}}}}{}_{AB}^{{}^{pure}}$ distinct, such that $(\,(\alpha\sqcap_{{}_{{{{ \overline{S}}}}_{AB}}}\beta)\not=(\gamma\sqcap_{{}_{{{{ \overline{S}}}}_{AB}}}\delta)$ and $(\alpha\sqcap_{{}_{{{{ \overline{S}}}}_{AB}}}\beta),(\gamma\sqcap_{{}_{{{{ \overline{S}}}}_{AB}}}\delta)\sqcoversubset_{{}_{{{{ \overline{S}}}}_{AB}}}\lambda \,)$, we have then $(\alpha\sqcap_{{}_{{{{ \overline{S}}}}_{AB}}}\beta\sqcap_{{}_{{{{ \overline{S}}}}_{AB}}}\gamma\sqcap_{{}_{{{{ \overline{S}}}}_{AB}}}\delta) \sqcoversubset_{{}_{{{{ \overline{S}}}}_{AB}}}(\alpha\sqcap_{{}_{{{{ \overline{S}}}}_{AB}}}\beta),(\gamma\sqcap_{{}_{{{{ \overline{S}}}}_{AB}}}\delta)$.
\end{theorem}
\begin{proof}
Direct consequence of the expansion formula (\ref{developmentetildeordersimplify}). Indeed, $\lambda,\alpha,\beta,\gamma,\delta\in {{{ \overline{S}}}}{}_{AB}^{{}^{pure}}$ all distinct with $(\,(\alpha\sqcap_{{}_{{{{ \overline{S}}}}_{AB}}}\beta)\not=(\gamma\sqcap_{{}_{{{{ \overline{S}}}}_{AB}}}\delta)$ and $(\alpha\sqcap_{{}_{{{{ \overline{S}}}}_{AB}}}\beta),(\gamma\sqcap_{{}_{{{{ \overline{S}}}}_{AB}}}\delta)\sqcoversubset_{{}_{{{{{{{ \overline{S}}}}_{AB}}}}}}\lambda \,)$, imply that we are in one of the following cases (we denote $\lambda:=\sigma_\lambda\widetilde{\otimes}\kappa_\lambda$, $\alpha:=\sigma_\alpha\widetilde{\otimes}\kappa_\alpha$, $\beta:=\sigma_\beta\widetilde{\otimes}\kappa_\beta$,$\gamma:=\sigma_\gamma\widetilde{\otimes}\kappa_\gamma$, $\delta:=\sigma_\delta\widetilde{\otimes}\kappa_\delta$ with $\sigma_\lambda,\sigma_\alpha,\sigma_\beta,\sigma_\gamma,\sigma_\delta\in {{{ \overline{\mathfrak{S}}}}}{}_{A}^{{}^{pure}}$, $\kappa_\lambda,\sigma_\alpha,\sigma_\beta,\sigma_\gamma,\sigma_\delta\in {{{ \overline{\mathfrak{S}}}}}{}_{B}^{{}^{pure}}$): \\
(1)  $\kappa_\alpha=\kappa_\beta=\kappa_\lambda$ and $\kappa_\gamma=\kappa_\delta=\kappa_\lambda$ and $(\sigma_\alpha\sqcap_{{}_{\overline{\mathfrak{S}_A}}}\sigma_\beta),(\sigma_\gamma\sqcap_{{}_{\overline{\mathfrak{S}}_A}}\sigma_\delta)\sqcoversubset_{{}_{\overline{\mathfrak{S}_A}}}\sigma_\lambda$ and $(\sigma_\alpha\sqcap_{{}_{\overline{\mathfrak{S}}_A}}\sigma_\beta)\not=(\sigma_\gamma\sqcap_{{}_{\overline{\mathfrak{S}}_A}}\sigma_\delta)$.\\
(2) $\sigma_\alpha=\sigma_\beta=\sigma_\lambda$ and $\sigma_\gamma=\sigma_\delta=\sigma_\lambda$ and $(\kappa_\alpha\sqcap_{{}_{\overline{\mathfrak{S}_B}}}\kappa_\beta),(\kappa_\gamma\sqcap_{{}_{\overline{\mathfrak{S}}_B}}\kappa_\delta)\sqcoversubset_{{}_{\overline{\mathfrak{S}}_B}}\kappa_\lambda$ and $(\kappa_\alpha\sqcap_{{}_{\overline{\mathfrak{S}}_B}}\kappa_\beta)\not=(\kappa_\gamma\sqcap_{{}_{\overline{\mathfrak{S}}_B}}\kappa_\delta)$.\\
(3) $\sigma_\alpha=\sigma_\beta=\sigma_\lambda$ and $\kappa_\gamma=\kappa_\delta=\kappa_\lambda$ and $(\kappa_\alpha\sqcap_{{}_{\overline{\mathfrak{S}_B}}}\kappa_\beta)\sqcoversubset_{{}_{\overline{\mathfrak{S}}_B}}\kappa_\lambda$ and $(\sigma_\gamma\sqcap_{{}_{\overline{\mathfrak{S}_A}}}\sigma_\delta)\sqcoversubset_{{}_{\overline{\mathfrak{S}}_A}}\sigma_\lambda$.\\
(4) $\kappa_\alpha=\kappa_\beta=\kappa_\lambda$ and $\sigma_\gamma=\sigma_\delta=\sigma_\lambda$ and $(\sigma_\alpha\sqcap_{{}_{\overline{\mathfrak{S}_A}}}\sigma_\beta)\sqcoversubset_{{}_{\overline{\mathfrak{S}}_A}}\sigma_\lambda$ and $(\kappa_\gamma\sqcap_{{}_{\overline{\mathfrak{S}_B}}}\kappa_\delta)\sqcoversubset_{{}_{\overline{\mathfrak{S}}_B}}\kappa_\lambda$.\\
In the case (1) we then apply the property (\ref{secondcoveringpropertySbarpre}) satisfied by ${{{ \overline{\mathfrak{S}}}}}{}_{A}$ to conclude\\
In the case (2) we then apply the property (\ref{secondcoveringpropertySbarpre}) satisfied by ${{{ \overline{\mathfrak{S}}}}}{}_{B}$ to conclude.\\
In the case (3) we just observe that $(\alpha\sqcap_{{}_{{{{ \overline{S}}}}_{AB}}}\beta\sqcap_{{}_{{{{ \overline{S}}}}_{AB}}}\gamma\sqcap_{{}_{{{{ \overline{S}}}}_{AB}}}\delta)=(\sigma_\gamma\sqcap_{{}_{\overline{\mathfrak{S}_A}}}\sigma_\delta)\widetilde{\otimes}\kappa_\lambda \sqcap_{{}_{\overline{S}_{AB}}} \sigma_\lambda \widetilde{\otimes}(\kappa_\alpha\sqcap_{{}_{\overline{\mathfrak{S}_B}}}\kappa_\beta)$ and we recall that 
$(\kappa_\alpha\sqcap_{{}_{\overline{\mathfrak{S}}_A}}\kappa_\beta)\sqcoversubset_{{}_{\overline{\mathfrak{S}}_B}}\kappa_\lambda$ and $(\sigma_\gamma\sqcap_{{}_{\overline{\mathfrak{S}}_A}}\sigma_\delta)\sqcoversubset_{{}_{\overline{\mathfrak{S}}_A}}\sigma_\lambda$ to conclude.\\
The case (4) is similar to the case (3).
\end{proof}
  
\subsection{Bipartite deterministic experiments}\label{subsectionbipartiteclassical}

During this subsection, we intent to show that the minimal tensor product is adequate for the description of bipartite deterministic  experiments.
Here, ${ \mathfrak{S}}_{A}$ and ${ \mathfrak{S}}_{B}$ will be two deterministic spaces of states (i.e. they are simplex spaces of states). They are naturally equipped with their real structures described in (\ref{starsimplex}).

\begin{theorem}
The minimal tensor product is adequate for the description of bipartite deterministic experiments.
\end{theorem}
\begin{proof}
By assumption, ${ \mathfrak{S}}_{A}$ and ${ \mathfrak{S}}_{B}$ are deterministic spaces of states and then ${ \mathfrak{S}}_{A}$ and ${ \mathfrak{S}}_{B}$ admit trivial real structures (see  property (\ref{starsimplex}) ) with ${ \mathfrak{S}}_{A}=\overline{{ \mathfrak{S}}_{A}}$ and ${ \mathfrak{S}}_{B}=\overline{{ \mathfrak{S}}_{B}}$.  We then consider ${ \mathfrak{S}}_{A}\boxtimes { \mathfrak{S}}_{B}=\overline{{ \mathfrak{S}}_{A}}\widetilde{\otimes}\overline{{ \mathfrak{S}}_{B}}=\overline{{ \mathfrak{S}}_{A}}{\otimes}\overline{{ \mathfrak{S}}_{B}}$.\\
First of all, we observe that the requirements (\ref{tensorinfimum}) and (\ref{tensorbot}) are direct consequences of Theorem \ref{theoremaxiomA1bipartitepre} and Theorem \ref{bottomcompound}. \\
Then, because of the Theorem \ref{SASBditribandcup}, we know that ${ \mathfrak{S}}_{A} {\boxtimes} { \mathfrak{S}}_{B}$ is also a deterministic space of states, and then ${ \mathfrak{S}}_{A} {\boxtimes} { \mathfrak{S}}_{B}$ admits a real structure according to (\ref{starsimplex}). This real structure $(\overline{{ \mathfrak{S}}_{A} \boxtimes { \mathfrak{S}}_{B}},\star)$ is simply given by $\overline{{ \mathfrak{S}}_{A} \boxtimes { \mathfrak{S}}_{B}}=\overline{{ \mathfrak{S}}_{A}} \widetilde{\otimes} \overline{{ \mathfrak{S}}_{B}}$ and the star $\star$ is the involution defined in (\ref{starsimplex}).\\
Secondly, the requirement (\ref{inclusionpuretensors}) is guarantied by Lemma \ref{lemmainclusiontensor} and the basic properties ${ \mathfrak{S}}_{A}=\overline{{ \mathfrak{S}}_{A}}$ and ${ \mathfrak{S}}_{B}=\overline{{ \mathfrak{S}}_{B}}$.\\
Thirdly, the requirements (\ref{pitensor=tensorpi1pre}) and (\ref{pitensor=tensorpi2pre}) are guarantied by the properties (\ref{pitensor=tensorpi1preminimal}) and (\ref{pitensor=tensorpi2preminimal}) using ${ \mathfrak{S}}_{A}=\overline{{ \mathfrak{S}}_{A}}$ and ${ \mathfrak{S}}_{B}=\overline{{ \mathfrak{S}}_{B}}$.\\
Endly, the requirement (\ref{requirementpartialtraces}) is guarantied by Theorem \ref{theorempartialtracesminimal}. 
\end{proof}

\subsection{Morphisms for the bipartite deterministic experiments}\label{subsectionsymmetriesbipartitedeterministic}

Let us consider ${ \mathfrak{S}}_{A_1}$, ${ \mathfrak{S}}_{A_2}$,  ${ \mathfrak{S}}_{B_1}$ and ${ \mathfrak{S}}_{B_2}$ four simplex spaces of states. Let us consider a morphism $f$ (resp. $g$) from the states space ${ \mathfrak{S}}_{A_1}$ (resp.  ${ \mathfrak{S}}_{B_1}$) to the states space ${ \mathfrak{S}}_{A_2}$ (resp. ${ \mathfrak{S}}_{B_2}$). \\
We can define the morphism $(f\widetilde{\otimes} g)$ from the simplex states space ${{S} }_{A_1B_1}={ \mathfrak{S}}_{A_1}\widetilde{\otimes}{ \mathfrak{S}}_{B_1}$ to the simplex states space ${{S}}_{A_2B_2}={ \mathfrak{S}}_{A_2}\widetilde{\otimes}{ \mathfrak{S}}_{B_2}$ by
\begin{eqnarray}
(f\widetilde{\otimes} g)(\bigsqcap{}^{{}^{{ {S} }_{A_1B_1}}}_{i\in I} \sigma_{i,A_1}\widetilde{\otimes} \sigma_{i,B_1}) & := & \bigsqcap{}^{{}^{{{S} }_{A_2B_2}}}_{i\in I} f(\sigma_{i,A_1})\widetilde{\otimes} g(\sigma_{i,B_1}).
\end{eqnarray}

\subsection{The fundamental description of bipartite indeterministic experiments}\label{subsectiontensorcompleterealspace}

During this subsection, we consider two ontic completions $(\!( \overline{ \mathfrak{S}}_{A},\star)\!)_c$ and $(\!( \overline{ \mathfrak{S}}_{B},\star)\!)_c$. As usual we denote
${ \mathfrak{S}}_A={ \mathfrak{J}}^c_{{\overline{ \mathfrak{S}}_A}}$ and ${ \mathfrak{S}}_B={ \mathfrak{J}}^c_{{\overline{ \mathfrak{S}}_B}}$.\\
We now intent to expose a first natural proposal for the description of the bipartite experiments associated to completely indeterministic compound systems (we will consider these compound systems as being built from two systems which separate descriptions are ensured by the ontic completions $(\!( \overline{ \mathfrak{S}}_{A},\star)\!)_c$ and $(\!( \overline{ \mathfrak{S}}_{B},\star)\!)_c$).

\begin{definition}\label{definitiongenerictensorproduct}
We define the tensor product of ${ \mathfrak{S}}_A={ \mathfrak{J}}^c_{{\overline{ \mathfrak{S}}_A}}$ and ${ \mathfrak{S}}_B={ \mathfrak{J}}^c_{{\overline{ \mathfrak{S}}_B}}$ denoted ${ \mathfrak{S}}_A\widehat{\otimes}{ \mathfrak{S}}_B$ by
\begin{eqnarray}
{ \mathfrak{S}}_A\;\widehat{\otimes}\; { \mathfrak{S}}_B &:=& { \mathfrak{J}}^c_{{\overline{ \mathfrak{S}}_A}\widetilde{\otimes}{\overline{ \mathfrak{S}}_B}}
\end{eqnarray}
We note that the ontic completion of $(\overline{ \mathfrak{S}}_{A}\widetilde{\otimes} \overline{ \mathfrak{S}}_{B},\star)$ can be defined because the basic conditions are fulfilled.
\end{definition}

Our aim is now to check all the basic requirements addressed in subsection \ref{subsectioncompoundfirstremarks}.\\

\begin{lemma}\label{generictensordowncomplete}
${{{ \mathfrak{S}}_A}\widehat{\otimes} {{ \mathfrak{S}}_{B}}}$ is a down complete Inf semi-lattice. 
The lowest bound is given by $\bot_{{}_{{{{ \mathfrak{S}}_A}\widehat{\otimes} {{ \mathfrak{S}}_{B}}}}}=\bot_{{}_{{{\overline{ \mathfrak{S}}_A}\widetilde{\otimes} {\overline{ \mathfrak{S}}_{B}}}}}$.
\end{lemma}
\begin{proof}
Direct consequence of the construction of the ontic completion $(\!( \overline{ \mathfrak{S}}_{A}\widetilde{\otimes} \overline{ \mathfrak{S}}_{B},\star)\!)_c$.
\end{proof}

\begin{lemma}\label{generictensorstarstructure}
${{{ \mathfrak{S}}_A}\widehat{\otimes} {{ \mathfrak{S}}_{B}}}$ admits a real structure given by $({\overline{ \mathfrak{S}}_A}\widetilde{\otimes} {\overline{ \mathfrak{S}}_B},\star)$. 
\end{lemma}
\begin{proof}
Direct consequence of Theorem \ref{theoremcompletion} for ${{{ \mathfrak{S}}_A}\widehat{\otimes} {{ \mathfrak{S}}_{B}}}$ defined as the ontic completion of $({\overline{ \mathfrak{S}}_A}\widetilde{\otimes} {\overline{ \mathfrak{S}}_B},\star)$.
\end{proof}

\begin{lemma}\label{generictensorpuretensors}
There exists a map $\iota^{{{ \mathfrak{S}}_A}\widehat{\otimes} {{ \mathfrak{S}}_{B}}}$ from $\overline{{ \mathfrak{S}}_A} {\times} \overline{{ \mathfrak{S}}_{B}}$ to ${{{ \mathfrak{S}}_A}\widehat{\otimes} {{ \mathfrak{S}}_{B}}}$. 
\end{lemma}
\begin{proof}
Trivial
\end{proof}


\begin{lemma}\label{generictensorlocaltomography}
\begin{eqnarray}
&&\hspace{-1cm}\forall \; \left(\bigsqcap{}^{{}^{{{{{ \mathfrak{S}}_A}\widehat{\otimes} {{ \mathfrak{S}}_{B}}}}}}_{i\in I} \sigma_{i,A}\widehat{\otimes} \sigma_{i,B}\right) \in \overline{{{ \mathfrak{S}}_A}\widehat{\otimes} {{ \mathfrak{S}}_{B}}},\; \forall \; \left(\bigsqcap{}^{{}^{{{{{ \mathfrak{S}}_A}\widehat{\otimes} {{ \mathfrak{S}}_{B}}}}}}_{j\in J} \sigma'_{j,A}\widehat{\otimes} \sigma'_{j,B}\right) \in \overline{{{ \mathfrak{S}}_A}\widehat{\otimes} {{ \mathfrak{S}}_{B}}},\nonumber\\
&&\hspace{-1cm}\left(\forall {\mathfrak{l}}_A\in { \mathfrak{E}}_{A},{\mathfrak{l}}_B\in { \mathfrak{E}}_{B},\;\; \bigwedge{}_{i\in I}{\epsilon}\,{}^{{ \mathfrak{S}}_A}_{{\mathfrak{l}}_A} (\sigma_{i,A})\bullet {\epsilon}\,{}^{{ \mathfrak{S}}_B}_{{\mathfrak{l}}_B} (\sigma_{i,B})=\bigwedge{}_{j\in J}{\epsilon}\,{}^{{ \mathfrak{S}}_A}_{{\mathfrak{l}}_A} (\sigma'_{j,A})\bullet {\epsilon}\,{}^{{ \mathfrak{S}}_B}_{{\mathfrak{l}}_B} (\sigma'_{j,B}) \right)\nonumber\\
&&\hspace{4cm} \;\Rightarrow\; (\,\bigsqcap{}^{{}^{{{{{ \mathfrak{S}}_A}\widehat{\otimes} {{ \mathfrak{S}}_{B}}}}}}_{i\in I} \sigma_{i,A}\widehat{\otimes} \sigma_{i,B} = \bigsqcap{}^{{}^{{{{{ \mathfrak{S}}_A}\widehat{\otimes} {{ \mathfrak{S}}_{B}}}}}}_{j\in J} \sigma'_{j,A}\widehat{\otimes} \sigma'_{j,B}\,).\;\;\;\;\;\;\;\;\;\;\;
\end{eqnarray}
\end{lemma}
\begin{proof}
Direct consequence of the definition of the tensor product.
\end{proof}

\begin{lemma}\label{generictensorbimorphic}
We have the two required bimorphic properties for the tensor product. Explicitly, for any $\{\,\sigma_{i,A}\;\vert\; i\in I\,\}\subseteq \overline{{ \mathfrak{S}}_{A}}$, $\{\,\sigma_{j,B}\;\vert\; j\in J\,\}\subseteq \overline{{ \mathfrak{S}}_{B}}$, $\sigma_A\in \overline{{ \mathfrak{S}}_{A}}$ and $\sigma_B\in \overline{{ \mathfrak{S}}_{B}}$, we have
\begin{eqnarray}
&& (\bigsqcap{}^{{}^{{{ \mathfrak{S}}_{A}}}}_{i\in I}\,\sigma_{i,A})\widehat{\otimes} \sigma_B =  \bigsqcap{}^{{}^{{{{{ \mathfrak{S}}_A \widehat{\otimes} { \mathfrak{S}}_{B}}}}}}_{i\in I} (\sigma_{i,A}\widehat{\otimes} \sigma_B),\label{pitensor=tensorpi1pregeneric}\\
&& \sigma_A \widehat{\otimes}  (\bigsqcap{}^{{}^{{{ \mathfrak{S}}_{B}}}}_{i\in I}\,\sigma_{i,B}) =  \bigsqcap{}^{{}^{{{{{ \mathfrak{S}}_A \widehat{\otimes} { \mathfrak{S}}_{B}}}}}}_{i\in I} (\sigma_A \widehat{\otimes} \sigma_{i,B}).\label{pitensor=tensorpi2pregeneric}
\end{eqnarray}
\end{lemma}
\begin{proof}
Direct consequence of Lemma \ref{generictensorlocaltomography}.
\end{proof}

\begin{lemma}\label{generictensorpartialtraces}
We have the two following partial traces
\begin{eqnarray}
\hspace{-2cm}\begin{array}{rcrclccrcrcl}
&  &\zeta^{{{ \mathfrak{S}}_{A}}{{ \mathfrak{S}}_{B}}}_{(1)} : \;\;\;\;\;\; \overline{{ \mathfrak{S}}_A \widehat{\otimes} { \mathfrak{S}}_{B}} & \longrightarrow & \overline{{ \mathfrak{S}}_{A}}  & & &  &  &\zeta^{{{ \mathfrak{S}}_{A}}{{ \mathfrak{S}}_{B}}}_{(2)} : \;\;\;\;\;\; \overline{{ \mathfrak{S}}_A \widehat{\otimes} { \mathfrak{S}}_{B}} & \longrightarrow & \overline{{ \mathfrak{S}}_{B}}  \\
& & \bigsqcap{}^{{}^{\overline{{ \mathfrak{S}}_A} \widetilde{\otimes} \overline{{ \mathfrak{S}}_{B}}}}_{i\in I} \sigma_{i,A} \widetilde{\otimes} \sigma_{i,B} & \mapsto & \bigsqcap{}^{{}^{{\overline{{ \mathfrak{S}}_{A}}}}}_{i\in I} \sigma_{i,A}  & & & & & \bigsqcap{}^{{}^{\overline{{ \mathfrak{S}}_A} \widetilde{\otimes} \overline{{ \mathfrak{S}}_{B}}}}_{i\in I} \sigma_{i,A} \widetilde{\otimes} \sigma_{i,B} & \mapsto & \bigsqcap{}^{{}^{{\overline{{ \mathfrak{S}}_{B}}}}}_{i\in I}  \sigma_{i,B}
\end{array}
\end{eqnarray}
\end{lemma}
\begin{proof}
Direct consequence of Theorem \ref{theorempartialtracesminimal}.
\end{proof}

\begin{theorem}
The tensor product ${ \mathfrak{S}}_{A}\widehat{\otimes} { \mathfrak{S}}_{B}$ defined in Definition \ref{definitiongenerictensorproduct} satisfies the requirements addressed in subsection \ref{subsectioncompoundfirstremarks} for the description of generic bipartite experiments.
\end{theorem}
\begin{proof}
According to subsection \ref{subsectioncompoundfirstremarks}, this is a direct consequence of Lemma \ref{generictensordowncomplete}, Lemma \ref{generictensorstarstructure}, Lemma \ref{generictensorpuretensors}, Lemma \ref{generictensorbimorphic}, Lemma \ref{generictensorlocaltomography}, Lemma \ref{generictensorpartialtraces}.
\end{proof}

We note that the minimal tensor product had already proven to be adequate for the description of bipartite deterministic experiments, especially because of the property saying that the minimal tensor product of simplex spaces of states is itself a simplex space of states (Theorem \ref{SASBditribandcup}). \\

Now, the  tensor product ${ \mathfrak{S}}_{A}\widehat{\otimes} { \mathfrak{S}}_{B}$ defined in Definition \ref{definitiongenerictensorproduct} reveals to be adequate for the description of bipartite completely indeterministic experiments on a compound system build from two completely indeterministic individual systems described respectively by the ontic completions $(\!( \overline{ \mathfrak{S}}_{A},\star)\!)_c$ and $(\!( \overline{ \mathfrak{S}}_{B},\star)\!)_c$. This fact is exhibited in the following theorem.

\begin{theorem}\label{theoremirreducibilitytensor}
If $(\overline{{ \mathfrak{S}}_{A}},\star)$ and $(\overline{{ \mathfrak{S}}_{B}},\star)$ are completely indeterministic spaces of states,  then $(\overline{{ \mathfrak{S}}_{A}}\widetilde{\otimes} \overline{{ \mathfrak{S}}_{B}},\star)$ is also a completely indeterministic space of states.
\end{theorem}
\begin{proof}
Let us consider $(\overline{{ \mathfrak{S}}_{A}},\star)$ and $(\overline{{ \mathfrak{S}}_{B}},\star)$ two completely indeterministic spaces of states. Let us consider $\sigma:=\sigma_A\widetilde{\otimes}\sigma_B, \lambda:=\lambda_A\widetilde{\otimes}\lambda_B\in (\overline{{ \mathfrak{S}}_{A}}\widetilde{\otimes}\overline{{ \mathfrak{S}}_{B}})^{{}^{pure}}$ such that $\sigma\underline{\perp}\lambda$. We recall that $\sigma\underline{\perp}\lambda$ means  $\lambda_A\widetilde{\otimes}\lambda_B \sqsupseteq_{{}_{{ \mathfrak{S}}_{A}\widetilde{\otimes} { \mathfrak{S}}_{B}}} (\sigma_A^\star\widetilde{\otimes}\bot_{{}_{\overline{{ \mathfrak{S}}_{B}}}} \sqcap_{{}_{{ \mathfrak{S}}_{A}\widetilde{\otimes} { \mathfrak{S}}_{B}}} \bot_{{}_{\overline{{ \mathfrak{S}}_{A}}}}\widetilde{\otimes} \sigma_B^\star)$ which means $\lambda_A \sqsupseteq_{{}_{{ \mathfrak{S}}_{A}}}\sigma_A^\star$ or $\lambda_B \sqsupseteq_{{}_{{ \mathfrak{S}}_{B}}}\sigma_B^\star$, i.e. $\sigma_A\underline{\perp}\lambda_A$ or $\sigma_B\underline{\perp}\lambda_B$. Let us distinguish the three different cases.\\
(1) $\sigma_A\underline{\perp}\lambda_A$ and $\sigma_B\underline{\perp}\lambda_B$. $(\overline{{ \mathfrak{S}}_{A}},\star)$ and $(\overline{{ \mathfrak{S}}_{B}},\star)$ being two completely indeterministic spaces of states, we deduce that there exists $\kappa_A\in \overline{{ \mathfrak{S}}_{A}}{}^{{}^{pure}}$ and $\kappa_B\in \overline{{ \mathfrak{S}}_{B}}{}^{{}^{pure}}$ such that $\kappa_A\not\!\!\!\underline{\perp} \lambda_A$, $\kappa_A\not\!\!\!\underline{\perp} \sigma_A$, and $\kappa_B\not\!\!\!\underline{\perp} \lambda_B$, $\kappa_B\not\!\!\!\underline{\perp} \sigma_B$, i.e. $\kappa_A \not\sqsupseteq_{{}_{{ \mathfrak{S}}_{A}}}\sigma_A^\star$ and $\kappa_A \not\sqsupseteq_{{}_{{ \mathfrak{S}}_{A}}}\lambda_A^\star$ and $\kappa_B \not\sqsupseteq_{{}_{{ \mathfrak{S}}_{B}}}\lambda_B^\star$ and $\kappa_B \not\sqsupseteq_{{}_{{ \mathfrak{S}}_{B}}}\sigma_B^\star$.  Then, we have $(\kappa_A\widetilde{\otimes}\kappa_B) \not\sqsupseteq_{{}_{{ \mathfrak{S}}_{A}\widetilde{\otimes} { \mathfrak{S}}_{B}}} (\sigma_A^\star\widetilde{\otimes}\bot_{{}_{\overline{{ \mathfrak{S}}_{B}}}} \sqcap_{{}_{{ \mathfrak{S}}_{A}\widetilde{\otimes} { \mathfrak{S}}_{B}}} \bot_{{}_{\overline{{ \mathfrak{S}}_{A}}}}\widetilde{\otimes} \sigma_B^\star)$ and $(\kappa_A\widetilde{\otimes}\kappa_B) \not\sqsupseteq_{{}_{{ \mathfrak{S}}_{A}\widetilde{\otimes} { \mathfrak{S}}_{B}}} (\lambda_A^\star\widetilde{\otimes}\bot_{{}_{\overline{{ \mathfrak{S}}_{B}}}} \sqcap_{{}_{{ \mathfrak{S}}_{A}\widetilde{\otimes} { \mathfrak{S}}_{B}}} \bot_{{}_{\overline{{ \mathfrak{S}}_{A}}}}\widetilde{\otimes} \lambda_B^\star)$, i.e. $(\kappa_A\widetilde{\otimes}\kappa_B) \not\!\!\!\underline{\perp} (\sigma_A\widetilde{\otimes}\sigma_B)$ and $(\kappa_A\widetilde{\otimes}\kappa_B) \not\!\!\!\underline{\perp} (\lambda_A\widetilde{\otimes}\lambda_B)$.\\
(2) $\sigma_A\underline{\perp}\lambda_A$ and $\sigma_B\not\!\!\!\underline{\perp}\lambda_B$. $(\overline{{ \mathfrak{S}}_{A}},\star)$ being a completely indeterministic spaces of states, we deduce that there exists $\kappa_A\in \overline{{ \mathfrak{S}}_{A}}{}^{{}^{pure}}$ such that $\kappa_A\not\!\!\!\underline{\perp} \lambda_A$, $\kappa_A\not\!\!\!\underline{\perp} \sigma_A$, i.e. $\kappa_A \not\sqsupseteq_{{}_{{ \mathfrak{S}}_{A}}}\sigma_A^\star$ and $\kappa_A \not\sqsupseteq_{{}_{{ \mathfrak{S}}_{A}}}\lambda_A^\star$. Let us fix $\kappa_B:=\lambda_B$. We then note that  $\kappa_B \not\sqsupseteq_{{}_{{ \mathfrak{S}}_{B}}}\lambda_B^\star$ and $\kappa_B \not\sqsupseteq_{{}_{{ \mathfrak{S}}_{B}}}\sigma_B^\star$. We then have $(\kappa_A\widetilde{\otimes}\kappa_B) \not\sqsupseteq_{{}_{{ \mathfrak{S}}_{A}\widetilde{\otimes} { \mathfrak{S}}_{B}}} (\sigma_A^\star\widetilde{\otimes}\bot_{{}_{\overline{{ \mathfrak{S}}_{B}}}} \sqcap_{{}_{{ \mathfrak{S}}_{A}\widetilde{\otimes} { \mathfrak{S}}_{B}}} \bot_{{}_{\overline{{ \mathfrak{S}}_{A}}}}\widetilde{\otimes} \sigma_B^\star)$ and $(\kappa_A\widetilde{\otimes}\kappa_B) \not\sqsupseteq_{{}_{{ \mathfrak{S}}_{A}\widetilde{\otimes} { \mathfrak{S}}_{B}}} (\lambda_A^\star\widetilde{\otimes}\bot_{{}_{\overline{{ \mathfrak{S}}_{B}}}} \sqcap_{{}_{{ \mathfrak{S}}_{A}\widetilde{\otimes} { \mathfrak{S}}_{B}}} \bot_{{}_{\overline{{ \mathfrak{S}}_{A}}}}\widetilde{\otimes} \lambda_B^\star)$, i.e.  $(\kappa_A\widetilde{\otimes}\kappa_B) \not\!\!\!\underline{\perp} (\sigma_A\widetilde{\otimes}\sigma_B)$ and $(\kappa_A\widetilde{\otimes}\kappa_B) \not\!\!\!\underline{\perp} (\lambda_A\widetilde{\otimes}\lambda_B)$.\\
(3) idem.\\
This concludes the proof.
\end{proof}


\subsection{Morphisms for the bipartite indeterministic experiments}\label{subsectionsymmetriesbipartiteindeterministic}

Let us consider a morphism $f$ (resp. $g$) from a states space ${ \mathfrak{S}}_{A_1}:={ \mathfrak{J}}^c_{{\overline{ \mathfrak{S}}_{A_1}}}$ (resp.  ${ \mathfrak{S}}_{B_1}:={ \mathfrak{J}}^c_{{\overline{ \mathfrak{S}}_{B_1}}}$) to another states space ${ \mathfrak{S}}_{A_2}:={ \mathfrak{J}}^c_{{\overline{ \mathfrak{S}}_{A_2}}}$ (resp. ${ \mathfrak{S}}_{B_2}:={ \mathfrak{J}}^c_{{\overline{ \mathfrak{S}}_{B_2}}}$). 

We can define the morphism $(f\widetilde{\otimes} g)$ from the states space ${\hat{S}}_{A_1B_1}:={ \mathfrak{J}}^c_{{ \overline{S} }_{A_1B_1}}$ with ${ \overline{S} }_{A_1B_1}={{\overline{ \mathfrak{S}}_{A_1}}\widetilde{\otimes}{\overline{ \mathfrak{S}}_{B_1}}}$ to the states space ${\hat{S} }_{A_2B_2}:={ \mathfrak{J}}^c_{{ \overline{S} }_{A_2B_2}}$ with ${ \overline{S} }_{A_2B_2}=
{{\overline{ \mathfrak{S}}_{A_2}}\widetilde{\otimes}{\overline{ \mathfrak{S}}_{B_2}}}$ by
\begin{eqnarray}
&&\hspace{-1cm}\Theta^{{\overline{ \mathfrak{S}}_{A_2}}\widetilde{\otimes}{\overline{ \mathfrak{S}}_{B_2}}}((f\widetilde{\otimes} g)(\xi))=\label{morphismtensorindeterminism}\\
&&cl_c^{{\overline{ \mathfrak{S}}_{A_2}}\widetilde{\otimes}{\overline{ \mathfrak{S}}_{B_2}}}(\{\bigsqcap{}^{{}^{{ \overline{S} }_{A_2B_2}}}_{i\in I} f(\omega_{i,A_1})\widetilde{\otimes}g(\omega_{i,B_1})\;\vert\; (\bigsqcap{}^{{}^{{ \overline{S} }_{A_1B_1}}}_{i\in I}\omega_{i,A_1}\widetilde{\otimes}\omega_{i,B_1})\in \Theta^{{\overline{ \mathfrak{S}}_{A_1}}\widetilde{\otimes}{\overline{ \mathfrak{S}}_{B_1}}}(\xi)\;\})\;\;\;\;\;\;\;\;\nonumber
\end{eqnarray}
for any $\xi\in {\hat{S}}_{A_1B_1}$.

\section{Ontic completions and contextuality}\label{sectionontic}

In this section, we will clarify the concept of compatibility between real measurement operators and construct the concept of operational description, which will lead us to establish the concept of contextuality. Throughout our demonstration, we will observe how generalized spaces of states with hidden states are not only obtained naturally by a completion procedure from real indeterministic spaces of states (see Section \ref{sectioncharacterizationhiddenstates}), but are also linked to \underline{contextual} empirical models associated with the operational descriptions of these real spaces of states. The notions relative to contextuality presented here are directly inspired by the notions developed in \cite{Abramsky2011}.

\subsection{Notions relative to contextuality}\label{subsectioncontextuality}

During this subsection, the states/effects Chu Space $({ \mathfrak{S}},{ \mathfrak{E}}_{{ \mathfrak{S}}},\epsilon^{{ \mathfrak{S}}})$ will be equipped with a real structure denoted by $(\overline{ \mathfrak{S}},\star)$.  As before, $\overline{ \mathfrak{E}}_{{ \mathfrak{S}}}$ is defined as the sub Inf semi-lattice of ${ \mathfrak{E}}_{ \mathfrak{S}}$ formed by the elements of 
\begin{eqnarray*}
\{\,{ \mathfrak{l}}_{(\sigma,\sigma')}\;\vert\; \sigma,\sigma'\in \overline{ \mathfrak{S}}\smallsetminus \{\bot_{{}_{{ \mathfrak{S}}}}\}, \sigma'\sqsupseteq_{{}_{\overline{ \mathfrak{S}}}} \sigma^\star\,\} \cup \{ { \mathfrak{l}}_{{}_{(\sigma,\centerdot)}}\;\vert\; \sigma\in \overline{ \mathfrak{S}}\;\}\cup \{ { \mathfrak{l}}_{{}_{(\centerdot,\sigma)}}\;\vert\; \sigma\in \overline{ \mathfrak{S}}\;\}\cup \{\, { \mathfrak{l}}_{{}_{(\centerdot,\centerdot)}}\,\}
\end{eqnarray*}

We recall that, ${\mathfrak{B}}$ being a simplex, the $N-$th tensor product ${\mathfrak{B}}^{\widetilde{\otimes}N}$ is also a simplex (cf Lemma \ref{SASBditribandcup}). The space ${\mathfrak{B}}^{\widetilde{\otimes}N}$ is then naturally equipped with a real structure (cf formula (\ref{starsimplex})).  Hence, the reduced space of effects $\overline{ \mathfrak{E}}_{{{\mathfrak{B}}^{\widetilde{\otimes}N}}}$ is naturally defined.

\begin{definition}\label{definitionjointmorphism}
$N$ real measurement maps ${\phi}_1,\cdots,{\phi}_N\in \overline{ \mathfrak{M}}_{ \mathfrak{S}}$ are said to be {\em jointly compatible} iff there exists a map $\Psi_{{}_{({\phi}_1,\cdots,{\phi}_N)}}$, called {\em joint morphism}, defined from ${ \mathfrak{S}}$ to ${\mathfrak{B}}^{\widetilde{\otimes}N}$ and satisfying
\begin{eqnarray}
\forall \sigma_1,\sigma_2\in { \mathfrak{S}},&& \Psi_{{}_{({\phi}_1,\cdots,{\phi}_N)}}(\sigma_1\sqcap_{{}_{{ \mathfrak{S}}}}\sigma_2)=\Psi_{{}_{({\phi}_1,\cdots,{\phi}_N)}}(\sigma_1)\sqcap_{{}_{{\mathfrak{B}}^{\widetilde{\otimes}N}}}\Psi_{{}_{({\phi}_1,\cdots,{\phi}_N)}}(\sigma_2),\label{relationsjointchannelmorphism}\;\;\;\;\;\;\;\;\;\;\;\;\\
&& \Psi_{{}_{({\phi}_1,\cdots,{\phi}_N)}}^\ast (\overline{ \mathfrak{E}}_{{{\mathfrak{B}}^{\widetilde{\otimes}N}}})\subseteq \overline{ \mathfrak{E}}_{ \mathfrak{S}},\label{relationsjointchannelreal}\\
\forall i=1,\cdots,N, && {\phi}_i = \zeta^{{}^{{\mathfrak{B}}\cdots{\mathfrak{B}}}}_{(i)}  \circ \Psi_{{}_{({\phi}_1,\cdots,{\phi}_N)}}.\label{relationsjointchannel}
\end{eqnarray}
We recall that $\zeta^{{}^{{\mathfrak{B}}\cdots{\mathfrak{B}}}}_{(i)} $ is the partial trace projecting the tensor product ${\mathfrak{B}}^{\widetilde{\otimes}N}$ on its $i-$th component. Here, we have denoted by $\Psi_{{}_{({\phi}_1,\cdots,{\phi}_N)}}^\ast$ the right component of the Chu morphism $(\Psi_{{}_{({\phi}_1,\cdots,{\phi}_N)}},\Psi_{{}_{({\phi}_1,\cdots,{\phi}_N)}}^\ast)$.  
\end{definition}

\begin{definition}
A {\em compatibility context} $U$ is defined to be a subset of $\overline{ \mathfrak{E}}_{ \mathfrak{S}}$ such that the elements of $\{\,{ \mathfrak{m}}_{ \mathfrak{l}}\;\vert\; { \mathfrak{l}}\in U\,\}$ form a family of jointly compatible real measurements. A {\em compatibility cover} ${ \mathcal{C}}$ is defined as a family of {\em maximal compatibility contexts} covering the whole set of real measurements. In other words,
\begin{eqnarray}\left\{
\begin{array}{l}
(\bigcup_{{}_{C\in { \mathcal{C}}}} C) = \overline{ \mathfrak{E}}_{ \mathfrak{S}}\\
\forall C,C’ \in { \mathcal{C}},\;\;\; C \subseteq C’ \;\Rightarrow \; C =C’\\
\forall C\in { \mathcal{C}},\; \textit{\rm the elements of}\;\{\,{ \mathfrak{m}}_{ \mathfrak{l}}\;\vert\; { \mathfrak{l}}\in C\,\} \;\textit{\rm are jointly compatible}.
\end{array}
\right.
\end{eqnarray}
\end{definition}

\noindent Let us briefly clarify the situation for ${ \mathfrak{S}}$ being a deterministic space of states.

\begin{theorem}\label{theoremcompatiblemeasurements}
Let ${ \mathfrak{S}}$ be a simplex. \\ 
We note that ${ \mathfrak{S}}=\overline{ \mathfrak{S}}$, $\overline{ \mathfrak{E}}_{ \mathfrak{S}}={ \mathfrak{E}}_{ \mathfrak{S}}$ and all measurements are real.\\
Let us consider any two measurement morphisms ${ \mathfrak{m}}_{{ \mathfrak{l}}_{(\sigma_1,\sigma'_1)}}$ and ${ \mathfrak{m}}_{{ \mathfrak{l}}_{(\sigma_2,\sigma'_2)}}$.  These morphisms are necessarily compatible.
\end{theorem}
\begin{proof}
We define the map $\Psi$ as a map from ${ \mathfrak{S}}$ to ${ \mathfrak{B}}\widetilde{\otimes}{ \mathfrak{B}}$ which satisfies, for any $\sigma\in { \mathfrak{S}}$, 
\begin{eqnarray}
\Psi (\sigma) &:=& \bigsqcap{}^{{}^{{ \mathfrak{B}}\widetilde{\otimes}{ \mathfrak{B}}}}_{\sigma'\in \underline{\sigma}_{{}_{{ \mathfrak{S}}}}} \Psi (\sigma')\label{equalitiypsisigma}
\end{eqnarray}
and for any $\sigma\in { \mathfrak{S}}^{{}^{pure}}=Max({ \mathfrak{S}})$
\begin{eqnarray}
\hspace{-1cm} && \left\{
\begin{array}{ll}
\Psi(\sigma)=\textit{\bf Y}\widetilde{\otimes}\textit{\bf Y} 
&\textit{\rm if} \;\; \epsilon^{{ \mathfrak{S}}}_{{ \mathfrak{l}}_{(\sigma_1,\sigma'_1)}}(\sigma)=\textit{\bf Y}\;\; \textit{\rm and} \;\; \epsilon^{{ \mathfrak{S}}}_{{ \mathfrak{l}}_{(\sigma_2,\sigma'_2)}}(\sigma)=\textit{\bf Y},\\
\Psi(\sigma)=\textit{\bf Y}\widetilde{\otimes}\textit{\bf Y}\sqcap_{{}_{{ \mathfrak{B}}\widetilde{\otimes}{ \mathfrak{B}}}}  \textit{\bf Y}\widetilde{\otimes}\textit{\bf N} &\textit{\rm if} \;\; \epsilon^{{ \mathfrak{S}}}_{{ \mathfrak{l}}_{(\sigma_1,\sigma'_1)}}(\sigma)=\textit{\bf Y} \;\textit{\rm and}\;  \epsilon^{{ \mathfrak{S}}}_{{ \mathfrak{l}}_{(\sigma_2,\sigma'_2)}}(\sigma)=\bot, \\
\Psi(\sigma)=\textit{\bf Y}\widetilde{\otimes}\textit{\bf Y}\sqcap_{{}_{{ \mathfrak{B}}\widetilde{\otimes}{ \mathfrak{B}}}}  \textit{\bf N}\widetilde{\otimes}\textit{\bf Y} &\textit{\rm if} \;\; \epsilon^{{ \mathfrak{S}}}_{{ \mathfrak{l}}_{(\sigma_1,\sigma'_1)}}(\sigma)=\bot \;\textit{\rm and}\;  \epsilon^{{ \mathfrak{S}}}_{{ \mathfrak{l}}_{(\sigma_2,\sigma'_2)}}(\sigma)=\textit{\bf Y}, \\
\Psi(\sigma)=\textit{\bf N}\widetilde{\otimes}\textit{\bf N} &
\textit{\rm if} \;\; \epsilon^{{ \mathfrak{S}}}_{{ \mathfrak{l}}_{(\sigma_1,\sigma'_1)}}(\sigma)=\textit{\bf N}\;\; \textit{\rm and} \;\; \epsilon^{{ \mathfrak{S}}}_{{ \mathfrak{l}}_{(\sigma_2,\sigma'_2)}}(\sigma)=\textit{\bf N},\\
\Psi(\sigma)=\textit{\bf N}\widetilde{\otimes}\textit{\bf N}\sqcap_{{}_{{ \mathfrak{B}}\widetilde{\otimes}{ \mathfrak{B}}}}  \textit{\bf N}\widetilde{\otimes}\textit{\bf Y} &\textit{\rm if} \;\; \epsilon^{{ \mathfrak{S}}}_{{ \mathfrak{l}}_{(\sigma_1,\sigma'_1)}}(\sigma)=\textit{\bf N} \;\textit{\rm and}\;  \epsilon^{{ \mathfrak{S}}}_{{ \mathfrak{l}}_{(\sigma_2,\sigma'_2)}}(\sigma)=\bot, \\
\Psi(\sigma)=\textit{\bf N}\widetilde{\otimes}\textit{\bf N}\sqcap_{{}_{{ \mathfrak{B}}\widetilde{\otimes}{ \mathfrak{B}}}}  \textit{\bf Y}\widetilde{\otimes}\textit{\bf N} &\textit{\rm if} \;\; \epsilon^{{ \mathfrak{S}}}_{{ \mathfrak{l}}_{(\sigma_1,\sigma'_1)}}(\sigma)=\bot\;\textit{\rm and}\;  \epsilon^{{ \mathfrak{S}}}_{{ \mathfrak{l}}_{(\sigma_2,\sigma'_2)}}(\sigma)=\textit{\bf N} , \\
\Psi(\sigma)=\textit{\bf Y}\widetilde{\otimes}\textit{\bf N} &
\textit{\rm if} \;\; \epsilon^{{ \mathfrak{S}}}_{{ \mathfrak{l}}_{(\sigma_1,\sigma'_1)}}(\sigma)=\textit{\bf Y}\;\; \textit{\rm and} \;\; \epsilon^{{ \mathfrak{S}}}_{{ \mathfrak{l}}_{(\sigma_2,\sigma'_2)}}(\sigma)=\textit{\bf N},\\
\Psi(\sigma)=\textit{\bf N}\widetilde{\otimes}\textit{\bf Y} &
\textit{\rm if} \;\; \epsilon^{{ \mathfrak{S}}}_{{ \mathfrak{l}}_{(\sigma_1,\sigma'_1)}}(\sigma)=\textit{\bf N}\;\; \textit{\rm and} \;\; \epsilon^{{ \mathfrak{S}}}_{{ \mathfrak{l}}_{(\sigma_2,\sigma'_2)}}(\sigma)=\textit{\bf Y},\\
\Psi(\sigma)=\bot \widetilde{\otimes}\bot &
\textit{\rm if} \;\; \epsilon^{{ \mathfrak{S}}}_{{ \mathfrak{l}}_{(\sigma_1,\sigma'_1)}}(\sigma)=\bot\;\; \textit{\rm and} \;\; \epsilon^{{ \mathfrak{S}}}_{{ \mathfrak{l}}_{(\sigma_2,\sigma'_2)}}(\sigma)=\bot.
\end{array}\right. \;\;\;\;\;\;\;\;\label{equalitiypsisigmapure}
\end{eqnarray}
We observe that $\Psi_{({ \mathfrak{m}}_{{ \mathfrak{l}}_{(\sigma_1,\sigma'_1)}},{ \mathfrak{m}}_{{ \mathfrak{l}}_{(\sigma_2,\sigma'_2)}})}$ is {unambiguously} defined by (\ref{equalitiypsisigma}) because ${ \mathfrak{S}}$ is a simplex and then satisfies (\ref{simplexdecompunique}).\\
Secondly, $\Psi_{({ \mathfrak{m}}_{{ \mathfrak{l}}_{(\sigma_1,\sigma'_1)}},{ \mathfrak{m}}_{{ \mathfrak{l}}_{(\sigma_2,\sigma'_2)}})}$ is a morphism by construction.\\
Endly, we can check easily for any $\sigma\in { \mathfrak{S}}$ (the proof begins on each expressions of (\ref{equalitiypsisigmapure}) by a simple computation and is extended to the whole set of elements of ${\mathfrak{S}}$ using the expression  (\ref{equalitiypsisigma}) and the homomorphic properties of ${ \mathfrak{m}}_{{ \mathfrak{l}}_{(\sigma_1,\sigma'_1)}}$ and ${ \mathfrak{m}}_{{ \mathfrak{l}}_{(\sigma_2,\sigma'_2)}}$)
\begin{eqnarray}
&&(\zeta^{{}^{{\mathfrak{B}}{\mathfrak{B}}}}_{(1)} \circ \Psi_{({ \mathfrak{m}}_{{ \mathfrak{l}}_{(\sigma_1,\sigma'_1)}},{ \mathfrak{m}}_{{ \mathfrak{l}}_{(\sigma_2,\sigma'_2)}})})(\sigma)={ \mathfrak{m}}_{{ \mathfrak{l}}_{(\sigma_1,\sigma'_1)}}(\sigma)\\
&&(\zeta^{{}^{{\mathfrak{B}}{\mathfrak{B}}}}_{(2)} \circ \Psi_{({ \mathfrak{m}}_{{ \mathfrak{l}}_{(\sigma_1,\sigma'_1)}},{ \mathfrak{m}}_{{ \mathfrak{l}}_{(\sigma_2,\sigma'_2)}})})(\sigma)={ \mathfrak{m}}_{{ \mathfrak{l}}_{(\sigma_2,\sigma'_2)}}(\sigma).
\end{eqnarray}
$\Psi$ is then the joint morphism $\Psi_{({ \mathfrak{m}}_{{ \mathfrak{l}}_{(\sigma_1,\sigma'_1)}},{ \mathfrak{m}}_{{ \mathfrak{l}}_{(\sigma_2,\sigma'_2)}})}$.
\end{proof}

\begin{theorem}\label{theoremNcompatiblemeasurements}
If ${ \mathfrak{S}}$ is a simplex then the whole set of elements of $\{\,{ \mathfrak{m}}_{{ \mathfrak{l}}} \;\vert\; {{ \mathfrak{l}}}\in { \mathfrak{E}}_{ \mathfrak{S}}\,\}$ are jointly compatible measurement maps.  In other words, the compatibility cover ${ \mathcal{C}}$ is then reduced to the singleton $\{\overline{ \mathfrak{E}}_{ \mathfrak{S}}\}$.
\end{theorem}
\begin{proof}
Straightforward using the generalization of the construction used in previous Theorem for the joint-morphism.
\end{proof}

\noindent Let us now come back to the generic case for ${\mathfrak{S}}$. We intent to clarify the structure of contexts.

\begin{lemma}\label{genmeasurementselfcompatible}
If the $N-$uple of real measurement maps $({ \mathfrak{m}}_{{ \mathfrak{l}}_1},\cdots, { \mathfrak{m}}_{{ \mathfrak{l}}_N})$ are jointly compatible, then the $(N+1)-$uple of real measurement maps $({ \mathfrak{m}}_{{ \mathfrak{l}}_1},{ \mathfrak{m}}_{{ \mathfrak{l}}_1},\cdots, { \mathfrak{m}}_{{ \mathfrak{l}}_N})$ are jointly-compatible, for any ${ \mathfrak{l}}_1,\cdots,{ \mathfrak{l}}_N\in { \mathfrak{E}}$. 
\end{lemma}
\begin{proof}
Let us consider that the $N-$uple of real measurement maps $({ \mathfrak{m}}_{{ \mathfrak{l}}_1},\cdots, { \mathfrak{m}}_{{ \mathfrak{l}}_N})$ are jointly compatible. There exists a joint morphism $\Psi_{{}_{({ \mathfrak{m}}_{{ \mathfrak{l}}_1},\cdots, { \mathfrak{m}}_{{ \mathfrak{l}}_N})}}$ satisfying the required conditions mentioned in Definition \ref{definitionjointmorphism}.  Using the map $\varphi$ defined as a morphism from ${\mathfrak{B}}$ to ${\mathfrak{B}}^{\widetilde{\otimes}2}$ by 
\begin{eqnarray}
\varphi (\textit{\bf Y}):=\textit{\bf Y}\widetilde{\otimes}\textit{\bf Y} && \varphi (\textit{\bf N}):=\textit{\bf N}\widetilde{\otimes}\textit{\bf N}
\end{eqnarray}
we then build a map denoted $\Phi$ from ${ \mathfrak{S}}$ to ${\mathfrak{B}}^{\widetilde{\otimes}(N+1)}$ by
\begin{eqnarray}
\Phi &:= & (\varphi\widetilde{\otimes}id^{\widetilde{\otimes}(N-1)}) \circ \Psi_{{}_{({ \mathfrak{m}}_{{ \mathfrak{l}}_1},\cdots, { \mathfrak{m}}_{{ \mathfrak{l}}_N})}}
\end{eqnarray}
It is easy to check that $\Phi$ is a joint morphism for the $(N+1)-$uple of measurement maps $({ \mathfrak{m}}_{{ \mathfrak{l}}_1},{ \mathfrak{m}}_{{ \mathfrak{l}}_1},\cdots, { \mathfrak{m}}_{{ \mathfrak{l}}_N})$.  In other words,
\begin{eqnarray}
\Phi &=& \Psi_{{}_{({ \mathfrak{m}}_{{ \mathfrak{l}}_1},{ \mathfrak{m}}_{{ \mathfrak{l}}_1},\cdots, { \mathfrak{m}}_{{ \mathfrak{l}}_N})}}.
\end{eqnarray}
\end{proof}

\begin{lemma}\label{genmeasurementconjugatecompatible} 
If the $N-$uple of real measurement maps $({ \mathfrak{m}}_{{ \mathfrak{l}}_1},\cdots, { \mathfrak{m}}_{{ \mathfrak{l}}_{N}})$ are jointly compatible, then the $(N+1)-$uple of real measurement maps $({ \mathfrak{m}}_{\overline{{ \mathfrak{l}}_{1}}},{ \mathfrak{m}}_{{ \mathfrak{l}}_{1}},\cdots, { \mathfrak{m}}_{{ \mathfrak{l}}_{N}})$ are also jointly-compatible. 
\end{lemma}
\begin{proof}
We first note that
\begin{eqnarray}
&&{ \mathfrak{m}}_{\overline{\, \mathfrak{l}\,}} = \omega \circ { \mathfrak{m}}_{{ \mathfrak{l}}}\\
&\textit{\rm with}& \forall \alpha\in { \mathfrak{B}},\;\omega(\alpha):=\overline{\,\alpha\,}.
\end{eqnarray}
On another part, using the Lemma \ref{genmeasurementselfcompatible}, we know that there exists a joint morphism for for the $(N+1)-$uple of measurement maps $({ \mathfrak{m}}_{{ \mathfrak{l}}_1},{ \mathfrak{m}}_{{ \mathfrak{l}}_1},\cdots, { \mathfrak{m}}_{{ \mathfrak{l}}_N})$, it is denoted $\Psi_{{}_{({ \mathfrak{m}}_{{ \mathfrak{l}}_1},{ \mathfrak{m}}_{{ \mathfrak{l}}_1},\cdots, { \mathfrak{m}}_{{ \mathfrak{l}}_N})}}$. \\
We then  define the map denoted $\Phi'$ from ${ \mathfrak{S}}$ to ${\mathfrak{B}}^{\widetilde{\otimes}(N+1)}$ by
\begin{eqnarray}
\Phi &:= & (\omega \widetilde{\otimes}id^{\widetilde{\otimes}N}) \circ \Psi_{{}_{({ \mathfrak{m}}_{{ \mathfrak{l}}_1},{ \mathfrak{m}}_{{ \mathfrak{l}}_1},\cdots, { \mathfrak{m}}_{{ \mathfrak{l}}_N})}}.
\end{eqnarray}
We check easily that $\Phi$ is a joint morphism for the $(N+1)-$uple of real measurement maps $({ \mathfrak{m}}_{\overline{{ \mathfrak{l}}_{1}}},{ \mathfrak{m}}_{{ \mathfrak{l}}_{1}},\cdots, { \mathfrak{m}}_{{ \mathfrak{l}}_{N}})$. In other words,
\begin{eqnarray}
\Phi' &=& \Psi_{{}_{({ \mathfrak{m}}_{\overline{{ \mathfrak{l}}_1}},{ \mathfrak{m}}_{{ \mathfrak{l}}_1},\cdots, { \mathfrak{m}}_{{ \mathfrak{l}}_N})}}.
\end{eqnarray}
\end{proof}

\begin{lemma}\label{genmeasurementconstantcompatible} 
For any $\eta\in { \mathfrak{B}}$, we define the morphism $\tau_{\eta}$ as the constant map from ${ \mathfrak{S}}$ to ${ \mathfrak{B}}$ sending any state to $\eta$.\\
If the $N-$uple of real measurement maps $({ \mathfrak{m}}_{{ \mathfrak{l}}_{1}},\cdots, { \mathfrak{m}}_{{ \mathfrak{l}}_{N}})$ are jointly compatible, then the $(N+1)-$uple of real measurement maps $(\tau_{\eta},{ \mathfrak{m}}_{{ \mathfrak{l}}_{1}},\cdots, { \mathfrak{m}}_{{ \mathfrak{l}}_{N}})$ are jointly compatible.
\end{lemma}
\begin{proof}
Direct consequence of Lemma \ref{genmeasurementselfcompatible} using the property 
\begin{eqnarray}
\tau_{\eta} &=& \tau_{\eta} \circ { \mathfrak{m}}_{{ \mathfrak{l}}}.
\end{eqnarray}
\end{proof}

\begin{lemma}\label{genmeasurementwedgecompatible} For any $\eta\in { \mathfrak{B}}$, if the $N-$uple of real measurement maps $({ \mathfrak{m}}_{{ \mathfrak{l}}_{1}},\cdots, { \mathfrak{m}}_{{ \mathfrak{l}}_{N}})$ are jointly compatible ($N\geq 2$), then the $(N+1)-$ uple of real measurement maps $({ \mathfrak{m}}_{({ \mathfrak{l}}_{1} 
\sqcap_{{}_{{ \mathfrak{E}}_{ \mathfrak{S}}}}
{ \mathfrak{l}}_{2})}$,${ \mathfrak{m}}_{{ \mathfrak{l}}_{1}},\cdots, { \mathfrak{m}}_{{ \mathfrak{l}}_{N}})$ are jointly compatible.
\end{lemma}
\begin{proof}
Let us fix $N=2$ for simplicity.\\
First of all, we remark that, if we denote by $\Psi_{({ \mathfrak{m}}_{{ \mathfrak{l}}_{1}},{ \mathfrak{m}}_{{ \mathfrak{l}}_{2}})}$ the joint morphism associated to the compatible real measurement maps $({ \mathfrak{m}}_{{ \mathfrak{l}}_{1}},{ \mathfrak{m}}_{{ \mathfrak{l}}_{2}})$, the property (\ref{relationsjointchannelreal}), i.e. $\Psi_{{}_{({ \mathfrak{m}}_{{ \mathfrak{l}}_{1}},{ \mathfrak{m}}_{{ \mathfrak{l}}_{2}})}}^\ast (\overline{ \mathfrak{E}}_{{{\mathfrak{B}}^{\widetilde{\otimes}2}}})\subseteq \overline{ \mathfrak{E}}_{ \mathfrak{S}}$ implies immediately 
\begin{eqnarray}
({ \mathfrak{l}}_{1} \sqcap_{{}_{{ \mathfrak{E}}_{ \mathfrak{S}}}} { \mathfrak{l}}_{2}) &\in & \overline{ \mathfrak{E}}_{ \mathfrak{S}}.\label{precisionrealmorphism}
\end{eqnarray}
Let us then consider the morphism denoted $\Omega$ and defined from ${ \mathfrak{B}}\widetilde{\otimes} { \mathfrak{B}}$ to ${ \mathfrak{B}}\widetilde{\otimes} { \mathfrak{B}}\widetilde{\otimes} { \mathfrak{B}}$ as follows. For any $\sigma \in { \mathfrak{B}}\widetilde{\otimes} { \mathfrak{B}}$, we adopt the writing of $\sigma$ in terms of pure states  
$\sigma=(\bigsqcap{}^{{}^{{ \mathfrak{B}}\widetilde{\otimes} { \mathfrak{B}}}}_{i\in I}\alpha_i\widetilde{\otimes} \beta_i)$ with $\alpha_i,\beta_i\in \{\,\textit{\bf Y}\,,\, \textit{\bf N}\,\}$, then $\Omega(\sigma)$ is defined by
\begin{eqnarray}
&& \Omega (\bigsqcap{}^{{}^{{ \mathfrak{B}}\widetilde{\otimes} { \mathfrak{B}}}}_{i\in I}\alpha_i\widetilde{\otimes} \beta_i) := \bigsqcap{}^{{}^{{ \mathfrak{B}}\widetilde{\otimes} { \mathfrak{B}}\widetilde{\otimes} { \mathfrak{B}}}}_{i\in I}(\alpha_i\widetilde{\otimes}\alpha_i\widetilde{\otimes} \beta_i \sqcap_{{}_{{ \mathfrak{B}}\widetilde{\otimes} { \mathfrak{B}}\widetilde{\otimes} { \mathfrak{B}}}} \beta_i\widetilde{\otimes}\alpha_i\widetilde{\otimes} \beta_i)
\end{eqnarray}
The map $\Omega$ is build as a morphism Moreover, we have immediately
\begin{eqnarray}
&&\zeta^{{}^{{ \mathfrak{B}} { \mathfrak{B}}{ \mathfrak{B}}}}_{(2)} \circ \Omega \circ \Psi_{({ \mathfrak{m}}_{{ \mathfrak{l}}_{1}},{ \mathfrak{m}}_{{ \mathfrak{l}}_{2}})}= \zeta^{{}^{{ \mathfrak{B}}{ \mathfrak{B}}}}_{(1)}\circ \Psi_{({ \mathfrak{m}}_{{ \mathfrak{l}}_{1}},{ \mathfrak{m}}_{{ \mathfrak{l}}_{2}})}={ \mathfrak{m}}_{{ \mathfrak{l}}_{1}},\\
&& \zeta^{{}^{{ \mathfrak{B}} { \mathfrak{B}}{ \mathfrak{B}}}}_{(3)} \circ \Omega \circ \Psi_{({ \mathfrak{m}}_{{ \mathfrak{l}}_{1}},{ \mathfrak{m}}_{{ \mathfrak{l}}_{2}})} = \zeta^{{}^{{ \mathfrak{B}}{ \mathfrak{B}}}}_{2}\circ \Psi_{({ \mathfrak{m}}_{{ \mathfrak{l}}_{1}},{ \mathfrak{m}}_{{ \mathfrak{l}}_{2}})}={ \mathfrak{m}}_{{ \mathfrak{l}}_{2}},\\
&& \zeta^{{}^{{ \mathfrak{B}} { \mathfrak{B}}{ \mathfrak{B}}}}_{(1)} \circ \Omega \circ \Psi_{({ \mathfrak{m}}_{{ \mathfrak{l}}_{1}},{ \mathfrak{m}}_{{ \mathfrak{l}}_{2}})} ={ \mathfrak{m}}_{({ \mathfrak{l}}_{1} 
\sqcap_{{}_{{ \mathfrak{E}}_{ \mathfrak{S}}}}
{ \mathfrak{l}}_{2})}.
\end{eqnarray}
Then, $\Omega \circ \Psi_{({ \mathfrak{m}}_{{ \mathfrak{l}}_{1}},{ \mathfrak{m}}_{{ \mathfrak{l}}_{2}})}$ is a joint morphism for the triple $({ \mathfrak{m}}_{({ \mathfrak{l}}_{1} 
\sqcap_{{}_{{ \mathfrak{E}}_{ \mathfrak{S}}}}
{ \mathfrak{l}}_{2})},{ \mathfrak{m}}_{{ \mathfrak{l}}_{1}}$, ${ \mathfrak{m}}_{{ \mathfrak{l}}_{2}})$.\\
As a result, ${ \mathfrak{m}}_{{ \mathfrak{l}}_{1}}$, ${ \mathfrak{m}}_{{ \mathfrak{l}}_{2}}$ and ${ \mathfrak{m}}_{({ \mathfrak{l}}_{1} 
\sqcap_{{}_{{ \mathfrak{E}}_{ \mathfrak{S}}}}
{ \mathfrak{l}}_{2})}$ are jointly compatible.
\end{proof}

We can then summarize Lemma \ref{genmeasurementwedgecompatible}, Lemma \ref{genmeasurementconjugatecompatible} and Lemma \ref{genmeasurementconstantcompatible} in terms of the following theorem on contexts :
\begin{theorem}  \label{synthesiscontexts}
Let $C \in { \mathcal{C}}$ designate a maximal compatibility context, we have 
\begin{eqnarray}
&&{ \mathfrak{Y}}_{{ \mathfrak{E}}_{ \mathfrak{S}}}\in C \;\;\textit{\rm and}\;\; \bot_{{ \mathfrak{E}}_{ \mathfrak{S}}}\in C  \;\;\textit{\rm and}\;\;  \overline{{ \mathfrak{Y}}_{{ \mathfrak{E}}_{ \mathfrak{S}}}}\in C
,\label{contexteta}\\
&& \forall { \mathfrak{l}}\in C,\;\;\overline{\; \mathfrak{l}\;}\in C,\label{contextbar}\\
&&\forall { \mathfrak{l}}_1,{ \mathfrak{l}}_2\in C,\;\;({ \mathfrak{l}}_1\sqcap_{{}_{{{ \mathfrak{E}}_{ \mathfrak{S}}}}}{ \mathfrak{l}}_2)\in C.\label{contextwedge}
\end{eqnarray}
We note the following fact, implicit in (\ref{contextwedge}) and clarified in Lemma \ref{genmeasurementwedgecompatible} :
\begin{eqnarray}
\forall { \mathfrak{l}}_1,{ \mathfrak{l}}_2\in \overline{ \mathfrak{E}}_{ \mathfrak{S}},&& ({ \mathfrak{l}}_1\sqcap_{{}_{{{ \mathfrak{E}}_{ \mathfrak{S}}}}}{ \mathfrak{l}}_2)\in { \mathfrak{E}}_{ \mathfrak{S}}\smallsetminus \overline{ \mathfrak{E}}_{ \mathfrak{S}}\;\;\Rightarrow \;\; \left(\nexists C\in {\mathcal{C}}\;\;\vert\;\;{ \mathfrak{l}}_1,{ \mathfrak{l}}_2\in C\right).\;\;\;\;\;\;\;\;\;
\end{eqnarray}
\end{theorem}

\noindent We now intent to introduce our central ontic notion.

\begin{definition}
Once is known the compatibility cover ${ \mathcal{C}}$, we define an {\em operational description} ${ \mathfrak{D}}$ to be a family of maps $({ \mathfrak{d}}^{C})_{C \in { \mathcal{C}}}$,  where each map ${ \mathfrak{d}}^{C}$ is an element of ${ \mathfrak{B}}^C$ (i.e.  a counterfactual statement in ${ \mathfrak{B}}$ is attributed to any of the jointly compatible measurements ${ \mathfrak{m}}_{ \mathfrak{l}}$ for any ${ \mathfrak{l}}\in C$), these maps being assigned to satisfy the following conditions
\begin{eqnarray}
&& \forall C\in { \mathcal{C}},\;\; { \mathfrak{d}}^{C}({ \mathfrak{Y}}_{\overline{ \mathfrak{E}}_{ \mathfrak{S}}})=\textit{\bf Y}
\label{sheafdistributioneta}\\
&& \forall C\in { \mathcal{C}},\forall { \mathfrak{l}}\in C,\;\; { \mathfrak{d}}^{C}(\overline{\; \mathfrak{l}\;})= \overline{{ \mathfrak{d}}^{C}({ \mathfrak{l}})},\label{sheafdistributionbar}\\
&& \forall C\in { \mathcal{C}},\forall \{\,{ \mathfrak{l}}_i\;\vert\; i\in I\,\}\subseteq C,\;\; { \mathfrak{d}}^{C}(\bigsqcap{}^{{}^{{ \mathfrak{E}}_{ \mathfrak{S}}}}_{i\in I}{ \mathfrak{l}}_i)= \bigwedge{}_{i\in I} \;{ \mathfrak{d}}^{C}({ \mathfrak{l}}_i).\label{sheafdistributionwedge}
\end{eqnarray}
\end{definition}

\begin{definition}
If the operational description ${ \mathfrak{D}}:=({ \mathfrak{d}}^{C})_{C \in { \mathcal{C}}}$  satisfies the following supplementary condition 
\begin{eqnarray}
\textit{ (sheaf condition)}&& \forall C,C'\in { \mathcal{C}},\forall { \mathfrak{l}}\in C\cap C', \;\; { \mathfrak{d}}^{C}({ \mathfrak{l}})={ \mathfrak{d}}^{C'}({ \mathfrak{l}}),\;\;\;\;\;\;\;\;\;\;\;\;\;\;\;\;\;\;\;\;\label{sheafdistributionrecollement}
\end{eqnarray}
it will be designated as a {\em coherent operational description}. \\
The set of coherent  operational descriptions associated to the compatibility cover ${ \mathcal{C}}$ will be denoted ${ \mathcal{D}}$ and called {\em the empirical model} for the space of states ${ \mathfrak{S}}$.
\end{definition}

\begin{definition}
The empirical model ${ \mathcal{D}}$ associated to the space of states ${ \mathfrak{S}}$ is an Inf semi-lattice by defining its infima as follows :
\begin{eqnarray}
\forall ({ \mathfrak{d}}_1^{C})_{C \in { \mathcal{C}}},({ \mathfrak{d}}_2^{C})_{C \in { \mathcal{C}}}\in { \mathcal{D}},\forall C\in { \mathcal{C}},\forall { \mathfrak{l}}\in C,&& ({ \mathfrak{d}}_1^{C}\sqcap_{{}_{{ \mathcal{D}}}} { \mathfrak{d}}_2^{C})({ \mathfrak{l}}):=  { \mathfrak{d}}_1^{C}({ \mathfrak{l}})\wedge { \mathfrak{d}}_2^{C}({ \mathfrak{l}}).\;\;\;\;\;\;\;\;\;\;\;\;\;\;\;\;\;\;\;
\end{eqnarray}
\end{definition}

\begin{definition}
If an operational description ${ \mathfrak{D}}:=({ \mathfrak{d}}^C)_{C\in { \mathcal{C}}}$ is such that there exists a real state $\sigma\in \overline{ \mathfrak{S}}$ satisfying  
\begin{eqnarray}
\forall C\in { \mathcal{C}},\forall { \mathfrak{l}}\in C&& { \mathfrak{d}}^C({ \mathfrak{l}}) = \epsilon^{{ \mathfrak{S}}}_{ \mathfrak{l}}(\sigma),\label{definglobally}
\end{eqnarray}
we will say that ${ \mathfrak{D}}$ is {\em defined globally}. It will then be denoted $({ \mathfrak{d}}^C_{(\sigma)})_{C\in { \mathcal{C}}}$. Note that if ${ \mathfrak{D}}$ is defined globally, then it is trivially a coherent  operational description. 
\end{definition}

\begin{definition}
The empirical model ${ \mathcal{D}}$ is said to be {\em non-contextual} iff every ${ \mathfrak{D}}\in { \mathcal{D}}$ is defined globally, otherwise it is said to be {\em contextual}.
\end{definition}

\begin{theorem}
If ${ \mathfrak{S}}$ is a simplex then every coherent  operational description ${ \mathfrak{d}}$ are defined globally. The empirical model associated to a simplex space of states is then non-contextual.
\end{theorem}
\begin{proof}
Direct consequence of Theorem \ref{theoremNcompatiblemeasurements} and Theorem \ref{blepsilonsigma}.
\end{proof}

It appears clearly that an operational description ${ \mathfrak{D}}$ defines a collection of preparation procedures. More explicitly, each complete context $U$ leads to the preparation procedure selecting the samples which return the correct result (i.e. the correct counterfactual statement) to the family of jointly-compatible measurements $({ \mathfrak{m}}_{ \mathfrak{l}})_{{ \mathfrak{l}}\in U}$. \\ 
In an empirical model,  the aim of a preparation procedure is to select an ontic state for the system (or at least a mixture of individual ontic states) compatible with all the observed data. The knowledge of this ontic state is necessary and exactly sufficient to fix the result returned by any of the measurements  that can be realized simultaneously (i.e. in a jointly-compatible way) on the system.  

\begin{definition}
The {\em ontic state} associated to the operational description ${ \mathfrak{D}}:=({ \mathfrak{d}}^{C})_{C \in { \mathcal{C}}} \in { \mathcal{D}}$ is a family $(\Sigma^C_{ \mathfrak{D}})_{C \in { \mathcal{C}}}\in \overline{ \mathfrak{S}}^{{ \mathcal{C}}}$ satisfying the following conditions
\begin{eqnarray}
&&\hspace{-1.5cm} \forall C \in { \mathcal{C}}, \forall { \mathfrak{l}}\in C,\;\; { \mathfrak{m}}_{ \mathfrak{l}}(\Sigma^C_{ \mathfrak{D}})={ \mathfrak{d}}^C({{ \mathfrak{l}}}),\label{defonticstate1}\\
&&\hspace{-1.5cm} \forall C\in { \mathcal{C}},\Sigma^C_{ \mathfrak{D}}=\bigsqcap{}^{{}^{ \mathfrak{S}}}\{\,\sigma^C\in \overline{ \mathfrak{S}}\;\vert\; \forall { \mathfrak{l}}\in C,\;{ \mathfrak{m}}_{ \mathfrak{l}}(\sigma^C)\geq { \mathfrak{d}}^C({{ \mathfrak{l}}})\,\}.\label{defonticstate2}\;\;\;\;\;\;\;\;\;\;\;\;\;\;\;\;\;\;
\end{eqnarray}
\end{definition}

\begin{theorem}\label{existenceonticstate}
Let us fix a coherent  operational description ${ \mathfrak{D}}:=({ \mathfrak{d}}^{C})_{C \in { \mathcal{C}}} \in {{ \mathcal{D}}}$ associated to the compatibility cover ${ \mathcal{C}}$.  Let us also fix a compatibility context $C$.  Then, 
\begin{eqnarray}
\exists \sigma \in \overline{ \mathfrak{S}} &\vert & (\,\epsilon^{ \mathfrak{S}}_{{ \mathfrak{l}}}(\sigma)={ \mathfrak{d}}^C({{ \mathfrak{l}}}),\;\forall { \mathfrak{l}}\in C\,).\label{descriptionstate}
\end{eqnarray} 
More explicitly, we will fix explicitly the solution $\sigma^C_{ \mathfrak{D}}$ of the equation (\ref{descriptionstate}) defined by
\begin{eqnarray}
{ \mathfrak{l}}^C_{ \mathfrak{D}}&:=&\bigsqcap{}^{{}^{\overline{ \mathfrak{E}}_{ \mathfrak{S}}}}\{\, { \mathfrak{l}}\in C \;\vert\; { \mathfrak{d}}^C({{ \mathfrak{l}}})=\textit{\bf Y}\,\}\label{onticstate1}\\
\sigma^C_{ \mathfrak{D}}&:=&\bigsqcap{}^{{}^{{ \mathfrak{S}}}}(\epsilon^{ \mathfrak{S}}_{{ \mathfrak{l}}^C_{ \mathfrak{D}}})^{-1}(\textit{\bf Y}) \;\in \overline{ \mathfrak{S}} \label{onticstate2}
\end{eqnarray}
The family $(\sigma^C_{ \mathfrak{D}})_{C \in { \mathcal{C}}}\in \overline{ \mathfrak{S}}{}^{\, \mathcal{C}}$ is the ontic state associated to ${ \mathfrak{D}}$.
\end{theorem}
\begin{proof}
Let us consider the effect ${ \mathfrak{l}}^C_{ \mathfrak{D}}:=\bigsqcap{}^{{}^{\overline{ \mathfrak{E}}_{ \mathfrak{S}}}}\{\, { \mathfrak{l}}\in C \;\vert\; { \mathfrak{d}}^C({{ \mathfrak{l}}})=\textit{\bf Y}\,\}$.  Note that the infimum in the definition of ${ \mathfrak{l}}^C_{ \mathfrak{D}}$ exists because $\overline{ \mathfrak{E}}_{ \mathfrak{S}}$ is a down-complete Inf semi-lattice and because ${ \mathfrak{Y}}_{\overline{ \mathfrak{E}}_{ \mathfrak{S}}}\in \{\, { \mathfrak{l}}\in C \;\vert\; { \mathfrak{d}}^C({{ \mathfrak{l}}})=\textit{\bf Y}\,\}\not= \varnothing$.  Moreover, due to the property (\ref{contextwedge}), 
 we deduce that ${ \mathfrak{l}}^C_{ \mathfrak{D}}$ is in $C$. \\
Secondly, ${ \mathfrak{d}}^{C}({ \mathfrak{l}}^C_{ \mathfrak{D}})=\textit{\bf Y}$ because of the property (\ref{sheafdistributionwedge}). \\
Thirdly, we note that, for any ${ \mathfrak{l}}\in C$, we have ${ \mathfrak{l}}\sqsupseteq_{{}_{{ \mathfrak{E}}}}{ \mathfrak{l}}^C_{ \mathfrak{D}}$ implies ${ \mathfrak{d}}^{C}({ \mathfrak{l}})=\textit{\bf Y}$ because ${ \mathfrak{d}}^{C}({ \mathfrak{l}})\wedge \textit{\bf Y}={ \mathfrak{d}}^{C}({ \mathfrak{l}})\wedge { \mathfrak{d}}^{C}({ \mathfrak{l}}^C_{ \mathfrak{D}})={ \mathfrak{d}}^{C}({ \mathfrak{l}}\sqcap_{{}_{\overline{ \mathfrak{E}}_{ \mathfrak{S}}}}
{ \mathfrak{l}}^C_{ \mathfrak{D}})={ \mathfrak{d}}^{C}({ \mathfrak{l}}^C_{ \mathfrak{D}})=\textit{\bf Y}$.  Conversely, the property ${ \mathfrak{d}}^{C}({ \mathfrak{l}})=\textit{\bf Y}$ for ${ \mathfrak{l}}\in C$ implies ${ \mathfrak{l}}\sqsupseteq_{{}_{\overline{ \mathfrak{E}}_{ \mathfrak{S}}}}{ \mathfrak{l}}^C_{ \mathfrak{D}}$ from the explicit definition of ${ \mathfrak{l}}^C_{ \mathfrak{D}}$.\\
As a result, we have $\{\, { \mathfrak{l}}\in C \;\vert\; { \mathfrak{d}}^C({{ \mathfrak{l}}})=\textit{\bf Y}\,\}=(\uparrow^{{}^{C}}\!\!{ \mathfrak{l}}^C_{ \mathfrak{D}})$.\\
We also note that, due to properties (\ref{sheafdistributioneta}) and the definition of ${ \mathfrak{l}}^C_{ \mathfrak{D}}$, we have ${ \mathfrak{l}}^C_{ \mathfrak{D}}={ \mathfrak{l}}^C_{ \mathfrak{D}} \sqcap_{{}_{{ \mathfrak{E}}}} {{ \mathfrak{Y}}}_{{ \mathfrak{E}}}$. Hence, there must exist a state $\sigma^C_{ \mathfrak{D}}\in { \mathfrak{S}}$ such that ${ \mathfrak{l}}^C_{ \mathfrak{D}}={ \mathfrak{l}}_{(\sigma^C_{ \mathfrak{D}},\,\cdot\,)}$. In other words,  $\sigma^C_{ \mathfrak{D}}:=\bigsqcap^{{}^{{ \mathfrak{S}}}}(\epsilon^{ \mathfrak{S}}_{{ \mathfrak{l}}^C_{ \mathfrak{D}}})^{-1}(\textit{\bf Y})$.\\
We can summarize these results by the following equivalence : 
\begin{eqnarray}
\forall C \in { \mathcal{C}},\forall { \mathfrak{l}}\in C, &&{ \mathfrak{d}}^C({{ \mathfrak{l}}})=\textit{\bf Y}\;\Leftrightarrow\; \epsilon^{ \mathfrak{S}}_{{ \mathfrak{l}}}(\sigma^C_{ \mathfrak{D}})=\textit{\bf Y}.
\end{eqnarray}
Using (\ref{etbar}) and (\ref{sheafdistributionbar}), we deduce also : $\epsilon^{ \mathfrak{S}}_{{ \mathfrak{l}}}(\sigma^C_{ \mathfrak{D}})=\textit{\bf N} \;\Leftrightarrow\; \epsilon^{ \mathfrak{S}}_{\overline{\, \mathfrak{l}\,}}(\sigma^C_{ \mathfrak{D}})=\textit{\bf Y} \;\Leftrightarrow\; { \mathfrak{d}}^C(\overline{\,{ \mathfrak{l}}\,})=\textit{\bf Y} \;\Leftrightarrow\; { \mathfrak{d}}^C({ \mathfrak{l}})=\textit{\bf N}$.  
As a final conclusion, we have for any ${ \mathfrak{l}}\in C$ the equality $\epsilon^{ \mathfrak{S}}_{{ \mathfrak{l}}}(\sigma^C_{ \mathfrak{D}})={ \mathfrak{d}}^C({{ \mathfrak{l}}})$.\\
Note then that the defining property (\ref{defonticstate1}) is identical to the property (\ref{descriptionstate}) satisfied by $\sigma^C_{ \mathfrak{D}}$ for any $C \in { \mathcal{C}}$. \\
Let us now consider a family $(\rho^C)_{C \in { \mathcal{C}}}\in { \mathfrak{S}}^{{ \mathcal{C}}}$ satisfying $\forall C \in { \mathcal{C}}, \forall { \mathfrak{l}}\in C, \epsilon^{ \mathfrak{S}}_{{ \mathfrak{l}}}(\rho^C)\geq{ \mathfrak{d}}^C({{ \mathfrak{l}}})$. \\
We have then in particular $\epsilon^{ \mathfrak{S}}_{{ \mathfrak{l}}^C_{ \mathfrak{D}}}(\rho^C)\geq 
{ \mathfrak{d}}^C({ \mathfrak{l}}^C_{ \mathfrak{D}})=\textit{\bf Y}$. Then $\rho^C \sqsupseteq_{{}_{{ \mathfrak{S}}}} \bigsqcap{}^{{}^{{ \mathfrak{S}}}}(\epsilon^{ \mathfrak{S}}_{{ \mathfrak{l}}^C_{ \mathfrak{D}}})^{-1}(\textit{\bf Y})=\sigma^C_{ \mathfrak{D}}$.  This is the defining property (\ref{defonticstate2}).\\
As a conclusion, $(\sigma^C_{ \mathfrak{D}})_{C \in { \mathcal{C}}}$ is the ontic state associated to ${ \mathfrak{D}}$.
\end{proof}

\subsection{Hidden states and contextuality}\label{hiddenstatesandcontextual}  

During this subsection, $(\!({\mathfrak{S}}',\star)\!)_c:=({\mathfrak{S}}',{ \mathcal{Q}}_c({ \mathfrak{S}}'),cl_c^{{{ \mathfrak{S}}'}},\star)$ will define the ontic completion of $({\mathfrak{S}}',\star)$. We will denote equivalently by ${\mathfrak{S}}$ the space of states ${ \mathfrak{J}}^c_{{{ \mathfrak{S}}'}}$, and we will then denote by $\overline{ \mathfrak{S}}$ the space of states ${\mathfrak{S}}'$. As shown before, the states/effects Chu Space $({ \mathfrak{S}},{ \mathfrak{E}}_{{ \mathfrak{S}}},\epsilon^{{ \mathfrak{S}}})$ is equipped with a real structure given simply by $(\overline{ \mathfrak{S}},\star)$.  \\
As before, $\overline{ \mathfrak{E}}_{{ \mathfrak{S}}}$ is defined as the sub Inf semi-lattice of ${ \mathfrak{E}}_{ \mathfrak{S}}$ generated by the elements of 
\begin{eqnarray*}
\{\,{ \mathfrak{l}}_{(\sigma,\sigma')}\;\vert\; \sigma,\sigma'\in \overline{ \mathfrak{S}}\smallsetminus \{\bot_{{}_{{ \mathfrak{S}}}}\}, \sigma'\sqsupseteq_{{}_{\overline{ \mathfrak{S}}}} \sigma^\star\,\} \cup \{ { \mathfrak{l}}_{{}_{(\sigma,\centerdot)}}\;\vert\; \sigma\in \overline{ \mathfrak{S}}\;\}\cup \{ { \mathfrak{l}}_{{}_{(\centerdot,\sigma)}}\;\vert\; \sigma\in \overline{ \mathfrak{S}}\;\}\cup \{\, { \mathfrak{l}}_{{}_{(\centerdot,\centerdot)}}\,\}
\end{eqnarray*}

\begin{definition}
In the following, we will adopt the following notation :
\begin{eqnarray}
{ \mathfrak{U}}_{\overline{ \mathfrak{S}}} & :=& \left( \{ { \mathfrak{l}}_{{}_{(\sigma,\centerdot)}}\;\vert\; \sigma\in \overline{ \mathfrak{S}}\;\}\cup \{ { \mathfrak{l}}_{{}_{(\centerdot,\sigma)}}\;\vert\; \sigma\in \overline{ \mathfrak{S}}\;\}\cup \{\, { \mathfrak{l}}_{{}_{(\centerdot,\centerdot)}}\,\}\right).
\end{eqnarray}
Let us consider a maximal compatibility context $C$ in ${ \mathcal{C}}$. We are necessarily in one of the two following cases :
\begin{itemize} 
\item Maximal compatibility context of Type 1 :
\begin{eqnarray}
C &\subseteq & { \mathfrak{U}}_{\overline{ \mathfrak{S}}}
\end{eqnarray}
\item Maximal compatibility context of Type 2 : 
\begin{eqnarray}
\exists \sigma,\sigma'\in \overline{ \mathfrak{S}}\smallsetminus \{\bot_{{}_{{ \mathfrak{S}}}}\},\sigma'\sqsupseteq_{{}_{\overline{ \mathfrak{S}}}}\sigma^\star  & \vert &{ \mathfrak{l}}_{(\sigma,\sigma')}\in C.
\end{eqnarray}
\end{itemize}
\end{definition}

\begin{lemma} \label{lemmacontext}
Let us consider a compatibility context denoted $C$ such that $C\subseteq { \mathfrak{U}}_{\overline{ \mathfrak{S}}}$. We have necessarily
\begin{eqnarray}
&&\exists\, \omega\in \overline{ \mathfrak{S}} \;\;\vert \;\; C\subseteq \left( \{\, { \mathfrak{l}}_{(\sigma,\centerdot)}\;\vert\; \sigma\sqsubseteq_{{}_{{ \mathfrak{S}}}} \omega\,\}\cup \{\, { \mathfrak{l}}_{(\centerdot,\sigma)}\;\vert\; \sigma\sqsubseteq_{{}_{{ \mathfrak{S}}}} \omega\,\}\cup \{\, { \mathfrak{l}}_{(\centerdot,\centerdot)}\,\}\right),\label{lemmacontext1}\;\;\;\;\;\;\;\;\;\;\;\;\;\;\;\;\;\;\;\;\\
&&\textit{\rm or} \nonumber\\
&&\exists\, \gamma,\gamma',\omega,\omega'\in \overline{ \mathfrak{S}} \;\;\vert \;\;\gamma' \sqsupseteq_{{}_{\overline{ \mathfrak{S}}}} \gamma^\star,\omega \sqsupseteq_{{}_{\overline{ \mathfrak{S}}}} \gamma, \omega' \sqsupseteq_{{}_{\overline{ \mathfrak{S}}}} \gamma'\nonumber\\
&&C\subseteq \left( \{\, { \mathfrak{l}}_{(\sigma,\centerdot)}\;\vert\; \gamma \sqsubseteq_{{}_{{ \mathfrak{S}}}} \sigma\sqsubseteq_{{}_{{ \mathfrak{S}}}} \omega\,\}\cup \{\, { \mathfrak{l}}_{(\centerdot,\sigma)}\;\vert\; \gamma \sqsubseteq_{{}_{{ \mathfrak{S}}}}\sigma\sqsubseteq_{{}_{{ \mathfrak{S}}}} \omega\,\}\cup \{\, { \mathfrak{l}}_{(\centerdot,\centerdot)}\,\}\cup \right.\nonumber\\
&&\hspace{2cm}\left.  \{\, { \mathfrak{l}}_{(\sigma,\centerdot)}\;\vert\; \gamma' \sqsubseteq_{{}_{{ \mathfrak{S}}}} \sigma\sqsubseteq_{{}_{{ \mathfrak{S}}}} \omega'\,\}\cup \{\, { \mathfrak{l}}_{(\centerdot,\sigma)}\;\vert\; \gamma' \sqsubseteq_{{}_{{ \mathfrak{S}}}}\sigma\sqsubseteq_{{}_{{ \mathfrak{S}}}} \omega'\,\}\right)\label{lemmacontext2}\;\;\;\;\;\;\;\;\;\;\;\;\;\;\;\;\;\;\;\;
\end{eqnarray}
\end{lemma}
\begin{proof}
Let us consider a compatibility context denoted $C$ such that $C\subseteq { \mathfrak{U}}_{\overline{ \mathfrak{S}}}$ and let us consider the following subset of ${ \mathcal{P}}(\overline{ \mathfrak{S}})$
\begin{eqnarray}
{\mathfrak{Q}}(\overline{ \mathfrak{S}})&:=&\{\,U\subseteq C \;\vert\; \forall \sigma,\sigma'\in \overline{ \mathfrak{S}}\;\;\textit{\rm s.t.}\;\; { \mathfrak{l}}_{(\sigma,\centerdot)},{ \mathfrak{l}}_{(\sigma',\centerdot)}\in U\;\;\textit{\rm we have}\;\;\sigma'\not\sqsupseteq_{{}_{{ \mathfrak{S}}'}}\sigma^\star\,\}\;\;\;\;\;\;\;\;\;\;\;\;   
\end{eqnarray} 
and let us consider an element $V$ which is maximal for inclusion in ${\mathfrak{Q}}(\overline{ \mathfrak{S}})$.  As a trivial consequence of the requirement (\ref{relationsjointchannelreal}) (more explicitly, of the property (\ref{precisionrealmorphism})), we know that $\omega:=\bigsqcup{}^{{}^{{ \mathfrak{S}}}}\{\sigma\;\vert\; { \mathfrak{l}}_{(\sigma,\centerdot)}\in V\}$ exists in $\overline{ \mathfrak{S}}$.\\
If we denote $C_\omega :=\left( \{\, { \mathfrak{l}}_{(\sigma,\centerdot)}\;\vert\; \sigma\sqsubseteq_{{}_{{ \mathfrak{S}}}} \omega\,\}\cup \{\, { \mathfrak{l}}_{(\centerdot,\sigma)}\;\vert\; \sigma\sqsubseteq_{{}_{{ \mathfrak{S}}}} \omega\,\}\cup \{\, { \mathfrak{l}}_{(\centerdot,\centerdot)}\,\}\right)$, we note that $C_\omega$ is a compatibility context (the joint morphism is built straightforwardly) and $V\subseteq C_\omega$ by construction.\\
We will now distinguish two cases : (1) $C\smallsetminus V=\varnothing$, and (2) $C\smallsetminus V\not= \varnothing$. \\
The first case is treated immediately, we obtain (\ref{lemmacontext1}). \\
Let us then focus on the second case.  For any $\alpha\in \overline{ \mathfrak{S}}\smallsetminus \{\bot_{{}_{\overline{ \mathfrak{S}}}}\}$ such that ${ \mathfrak{l}}_{(\alpha,\centerdot)} \in C\smallsetminus V$, we have necessarily, due to the requirement (\ref{relationsjointchannelreal}) (more explicitly, of the property (\ref{precisionrealmorphism})), for any $\sigma\in  \overline{ \mathfrak{S}}\smallsetminus \{\bot_{{}_{\overline{ \mathfrak{S}}}}\}$ such that ${ \mathfrak{l}}_{(\sigma,\centerdot)} \in V$ the property $\alpha^\star \sqsubseteq_{{}_{\overline{ \mathfrak{S}}}}\sigma$, i.e. $\alpha^\star \sqsubseteq_{{}_{\overline{ \mathfrak{S}}}} \bigsqcap{}^{{}^{{ \mathfrak{S}}}}\{\, \sigma\in  \overline{ \mathfrak{S}}\smallsetminus \{\bot_{{}_{\overline{ \mathfrak{S}}}}\}\;\vert\;{ \mathfrak{l}}_{(\sigma,\centerdot)} \in V\,\}$. If we denote by $\gamma :=\bigsqcap{}^{{}^{{ \mathfrak{S}}}}\{\, \sigma\in  \overline{ \mathfrak{S}}\smallsetminus \{\bot_{{}_{\overline{ \mathfrak{S}}}}\}\;\vert\;{ \mathfrak{l}}_{(\sigma,\centerdot)} \in V\,\}$ and $\gamma':=\bigsqcap{}^{{}^{{ \mathfrak{S}}}}\{\, \sigma\in  \overline{ \mathfrak{S}}\smallsetminus \{\bot_{{}_{\overline{ \mathfrak{S}}}}\}\;\vert\;{ \mathfrak{l}}_{(\sigma,\centerdot)} \in C\smallsetminus V\,\}$, we have $\gamma' \sqsupseteq_{{}_{\overline{ \mathfrak{S}}}} \gamma^\star$. On another part, and for the same reasons as before,  if we denote $\omega':=\bigsqcup{}^{{}^{{ \mathfrak{S}}}}\{\, \sigma\in  \overline{ \mathfrak{S}}\smallsetminus \{\bot_{{}_{\overline{ \mathfrak{S}}}}\}\;\vert\;{ \mathfrak{l}}_{(\sigma,\centerdot)} \in C\smallsetminus V\,\}$, then $C_{\omega'}$ is a compatibility context (the joint morphism is built straightforwardly) and $(C\smallsetminus V)\subseteq C_{\omega'}$ by construction.
\\
This concludes the proof.
\end{proof}

\begin{lemma}\label{lemmacontexttype1}
In the following, we will adopt the following notation
\begin{eqnarray}
&&\forall \omega\in \overline{ \mathfrak{S}}^{{}^{pure}} \;\;\vert \;\; C^{(1)}_\omega := \left( \{\, { \mathfrak{l}}_{(\sigma,\centerdot)}\;\vert\; \sigma\sqsubseteq_{{}_{{ \mathfrak{S}}}} \omega\,\}\cup \{\, { \mathfrak{l}}_{(\centerdot,\sigma)}\;\vert\; \sigma\sqsubseteq_{{}_{{ \mathfrak{S}}}} \omega\,\}\cup \{\, { \mathfrak{l}}_{(\centerdot,\centerdot)}\,\}\right).\;\;\;\;\;\;\;\;\;\;\;\;\;\;
\end{eqnarray}
Let $C$ be a maximal compatibility context of type 1. Then, 
\begin{eqnarray}
&&\existunique\, \omega\in \overline{ \mathfrak{S}}^{{}^{pure}} \;\;\vert \;\; C= C^{(1)}_\omega.
\end{eqnarray}
\end{lemma}
\begin{proof}
Trivial consequence of Lemma \ref{lemmacontext}.
\end{proof}

\begin{lemma}\label{lemmacontexttype2a}
Let $C$ be a maximal compatibility context of type 2 and let $\gamma,\gamma'\in \overline{ \mathfrak{S}} \smallsetminus \{\bot_{{}_{\overline{ \mathfrak{S}}}}\}$ with $\gamma'\sqsupseteq_{{}_{\overline{ \mathfrak{S}}}}\gamma^\star$ and such that ${ \mathfrak{l}}_{(\gamma,\gamma')}\in C$. Then, 
\begin{eqnarray}
&&\exists\, \lambda,\lambda',\omega,\omega'\in \overline{ \mathfrak{S}} \smallsetminus \{\bot_{{}_{\overline{ \mathfrak{S}}}}\} \;\;\vert \;\;\lambda' \sqsupseteq_{{}_{\overline{ \mathfrak{S}}}} \lambda^\star,\omega \sqsupseteq_{{}_{\overline{ \mathfrak{S}}}} \gamma \sqsupseteq_{{}_{\overline{ \mathfrak{S}}}} \lambda, \omega' \sqsupseteq_{{}_{\overline{ \mathfrak{S}}}} \gamma' \sqsupseteq_{{}_{\overline{ \mathfrak{S}}}} \lambda'\nonumber\\
&&(C\cap { \mathfrak{U}}_{\overline{ \mathfrak{S}}})\subseteq \left( \{\, { \mathfrak{l}}_{(\sigma,\centerdot)}\;\vert\; \lambda \sqsubseteq_{{}_{{ \mathfrak{S}}}} \sigma\sqsubseteq_{{}_{{ \mathfrak{S}}}} \omega\,\}\cup \{\, { \mathfrak{l}}_{(\centerdot,\sigma)}\;\vert\; \lambda \sqsubseteq_{{}_{{ \mathfrak{S}}}}\sigma\sqsubseteq_{{}_{{ \mathfrak{S}}}} \omega\,\}\cup \{\, { \mathfrak{l}}_{(\centerdot,\centerdot)}\,\}\cup \right.\nonumber\\
&&\hspace{1cm}\left.  \{\, { \mathfrak{l}}_{(\sigma,\centerdot)}\;\vert\; \lambda' \sqsubseteq_{{}_{{ \mathfrak{S}}}} \sigma\sqsubseteq_{{}_{{ \mathfrak{S}}}} \omega'\,\}\cup \{\, { \mathfrak{l}}_{(\centerdot,\sigma)}\;\vert\; \lambda' \sqsubseteq_{{}_{{ \mathfrak{S}}}}\sigma\sqsubseteq_{{}_{{ \mathfrak{S}}}} \omega'\,\}\right)
\end{eqnarray}
\end{lemma}
\begin{proof}
Let $C$ be a maximal compatibility context of type 2 and let $\gamma,\gamma'\in \overline{ \mathfrak{S}} \smallsetminus \{\bot_{{}_{\overline{ \mathfrak{S}}}}\}$ with $\gamma'\sqsupseteq_{{}_{\overline{ \mathfrak{S}}}}\gamma^\star$ and such that ${ \mathfrak{l}}_{(\gamma,\gamma')}\in C$. \\
Let us first note that $(C\cap { \mathfrak{U}}_{\overline{ \mathfrak{S}}})$ is a compatibility context included in ${ \mathfrak{U}}_{\overline{ \mathfrak{S}}}$. \\
Secondly, we remark that ${ \mathfrak{l}}_{(\gamma,\centerdot)}, { \mathfrak{l}}_{(\gamma',\centerdot)}\in (C\cap { \mathfrak{U}}_{\overline{ \mathfrak{S}}})$.  Indeed, we recall that ${ \mathfrak{l}}_{(\gamma,\centerdot)}={ \mathfrak{l}}_{(\gamma,\gamma')} \sqcap_{{}_{\overline{ \mathfrak{S}}}} {\mathfrak{Y}}_{{\mathfrak{E}}_{\mathfrak{S}}}$ and ${ \mathfrak{l}}_{(\gamma',\centerdot)}=\overline{ { \mathfrak{l}}_{(\gamma,\gamma')} \sqcap_{{}_{\overline{ \mathfrak{S}}}} \overline{{\mathfrak{Y}}_{{\mathfrak{E}}_{\mathfrak{S}}}}}$ and we have Theorem \ref{synthesiscontexts} for the maximal compatibility context $C$. Moreover, we have $\neg\widehat{\gamma\gamma'}{}^{{}^{\overline{ \mathfrak{S}}}}$. We are then in the case (\ref{lemmacontext2}) and not in the case (\ref{lemmacontext1}). Because of Lemma \ref{lemmacontext}, we know that 
\begin{eqnarray}
&&\exists\, \lambda,\lambda',\omega,\omega'\in \overline{ \mathfrak{S}} \smallsetminus \{\bot_{{}_{\overline{ \mathfrak{S}}}}\} \;\;\vert \;\;\lambda' \sqsupseteq_{{}_{\overline{ \mathfrak{S}}}} \lambda^\star,\omega \sqsupseteq_{{}_{\overline{ \mathfrak{S}}}} \lambda, \omega' \sqsupseteq_{{}_{\overline{ \mathfrak{S}}}} \lambda'\nonumber\\
&&(C\cap { \mathfrak{U}}_{\overline{ \mathfrak{S}}})\subseteq \left( \{\, { \mathfrak{l}}_{(\sigma,\centerdot)}\;\vert\; \lambda \sqsubseteq_{{}_{{ \mathfrak{S}}}} \sigma\sqsubseteq_{{}_{{ \mathfrak{S}}}} \omega\,\}\cup \{\, { \mathfrak{l}}_{(\centerdot,\sigma)}\;\vert\; \lambda \sqsubseteq_{{}_{{ \mathfrak{S}}}}\sigma\sqsubseteq_{{}_{{ \mathfrak{S}}}} \omega\,\}\cup \{\, { \mathfrak{l}}_{(\centerdot,\centerdot)}\,\}\cup \right.\nonumber\\
&&\hspace{1cm}\left.  \{\, { \mathfrak{l}}_{(\sigma,\centerdot)}\;\vert\; \lambda' \sqsubseteq_{{}_{{ \mathfrak{S}}}} \sigma\sqsubseteq_{{}_{{ \mathfrak{S}}}} \omega'\,\}\cup \{\, { \mathfrak{l}}_{(\centerdot,\sigma)}\;\vert\; \lambda' \sqsubseteq_{{}_{{ \mathfrak{S}}}}\sigma\sqsubseteq_{{}_{{ \mathfrak{S}}}} \omega'\,\}\right)
\end{eqnarray}
and we must obviously have $\lambda \sqsubseteq_{{}_{\overline{ \mathfrak{S}}}} \gamma$ and $\lambda' \sqsubseteq_{{}_{\overline{ \mathfrak{S}}}} \gamma'$, because ${ \mathfrak{l}}_{(\gamma,\centerdot)}, { \mathfrak{l}}_{(\gamma',\centerdot)} \in C$.
\end{proof}

\begin{lemma}\label{lemmacontexttype2b}
Let $C$ be a maximal compatibility context of type 2 such that there exists $\gamma\in \overline{ \mathfrak{S}} \smallsetminus \{\bot_{{}_{\overline{ \mathfrak{S}}}}\}$ with ${ \mathfrak{l}}_{(\gamma,\gamma^\star)}\in C$. Then, 
\begin{eqnarray}
&&\exists\, \omega,\omega'\in \overline{ \mathfrak{S}} \smallsetminus \{\bot_{{}_{\overline{ \mathfrak{S}}}}\} \;\;\vert \;\;\omega \sqsupseteq_{{}_{\overline{ \mathfrak{S}}}} \gamma, \omega' \sqsupseteq_{{}_{\overline{ \mathfrak{S}}}} \gamma^\star\nonumber\\
&&(C\cap { \mathfrak{U}}_{\overline{ \mathfrak{S}}})\subseteq \left( \{\, { \mathfrak{l}}_{(\sigma,\centerdot)}\;\vert\; \gamma \sqsubseteq_{{}_{{ \mathfrak{S}}}} \sigma\sqsubseteq_{{}_{{ \mathfrak{S}}}} \omega\,\}\cup \{\, { \mathfrak{l}}_{(\centerdot,\sigma)}\;\vert\; \gamma \sqsubseteq_{{}_{{ \mathfrak{S}}}}\sigma\sqsubseteq_{{}_{{ \mathfrak{S}}}} \omega\,\}\cup \{\, { \mathfrak{l}}_{(\centerdot,\centerdot)}\,\}\cup \right.\nonumber\\
&&\hspace{2cm}\left.  \{\, { \mathfrak{l}}_{(\sigma,\centerdot)}\;\vert\; \gamma^\star \sqsubseteq_{{}_{{ \mathfrak{S}}}} \sigma\sqsubseteq_{{}_{{ \mathfrak{S}}}} \omega'\,\}\cup \{\, { \mathfrak{l}}_{(\centerdot,\sigma)}\;\vert\; \gamma^\star \sqsubseteq_{{}_{{ \mathfrak{S}}}}\sigma\sqsubseteq_{{}_{{ \mathfrak{S}}}} \omega'\,\}\right)
\end{eqnarray}
\end{lemma}
\begin{proof}
We obviously take Lemma \ref{lemmacontexttype2a} as a departure point.\\
If $\gamma'=\gamma^\star$ we have then $\gamma^\star=\gamma' \sqsupseteq_{{}_{\overline{ \mathfrak{S}}}}\lambda' \sqsupseteq_{{}_{\overline{ \mathfrak{S}}}}\lambda^\star \sqsupseteq_{{}_{\overline{ \mathfrak{S}}}} \gamma^\star$.
This concludes the proof.
\end{proof}

In the following, ${ \mathcal{C}}$ will designate the compatibility cover and we will consider a coherent operational description denoted ${ \mathfrak{D}}:=({ \mathfrak{d}}^{C})_{C \in { \mathcal{C}}}$.  We will denote by $(\Sigma^C_{ \mathfrak{D}})_{C \in { \mathcal{C}}}$ the corresponding ontic state. 

\begin{lemma}
For any pair $C^{(1)}_{\omega},C^{(1)}_{\omega'}$ of maximal compatibility contexts of type 1 associated to $\omega,\omega'\in \overline{ \mathfrak{S}}^{{}^{pure}}$ (see Lemma \ref{lemmacontexttype1}), we have
\begin{eqnarray}
\left\{\begin{array}{rcl}
\omega \sqcap_{{}_{\overline{ \mathfrak{S}}}} \Sigma^{C^{(1)}_{\omega'}}_{ \mathfrak{D}} & \sqsubseteq_{{}_{\overline{ \mathfrak{S}}}} &  \Sigma^{C^{(1)}_{\omega}}_{ \mathfrak{D}}\\
\omega' \sqcap_{{}_{\overline{ \mathfrak{S}}}} \Sigma^{C^{(1)}_{\omega}}_{ \mathfrak{D}} & \sqsubseteq_{{}_{\overline{ \mathfrak{S}}}} &  \Sigma^{C^{(1)}_{\omega'}}_{ \mathfrak{D}}
\end{array}
\right.\label{consitencyontic}
\end{eqnarray}
\end{lemma}
\begin{proof}
By definition of the ontic state, we have
\begin{eqnarray}
\forall { \mathfrak{l}}\in C^{(1)}_{\omega},&& { \mathfrak{d}}^{C^{(1)}_{\omega}}({ \mathfrak{l}})=\epsilon^{ \mathfrak{S}}_{ \mathfrak{l}}(\Sigma^{C^{(1)}_{\omega}}_{ \mathfrak{D}})
\end{eqnarray}
The consistency of ${ \mathfrak{D}}$ implies
\begin{eqnarray}
\forall { \mathfrak{l}}\in C^{(1)}_{\omega}\cap C^{(1)}_{\omega'},&& \epsilon^{ \mathfrak{S}}_{ \mathfrak{l}}(\Sigma^{C^{(1)}_{\omega}}_{ \mathfrak{D}}) =\epsilon^{ \mathfrak{S}}_{ \mathfrak{l}}(\Sigma^{C^{(1)}_{\omega'}}_{ \mathfrak{D}})
\end{eqnarray}
But, 
\begin{eqnarray}
C^{(1)}_{\omega}\cap C^{(1)}_{\omega'} &=& \left( \{\, { \mathfrak{l}}_{(\sigma,\centerdot)}\;\vert\; \sigma\sqsubseteq_{{}_{{ \mathfrak{S}}}} \omega \sqcap_{{}_{{ \mathfrak{S}}}}\omega'\,\}\cup \{\, { \mathfrak{l}}_{(\centerdot,\sigma)}\;\vert\; \sigma\sqsubseteq_{{}_{{ \mathfrak{S}}}} \omega \sqcap_{{}_{{ \mathfrak{S}}}}\omega'\,\}\cup \{\, { \mathfrak{l}}_{(\centerdot,\centerdot)}\,\}\right).\;\;\;\;\;\;\;\;\;\;\;\;\;\;\;\;\;\;\;\;\;\;
\end{eqnarray}
Hence, 
\begin{eqnarray}
\hspace{-1.3cm}\left( \sigma \sqsubseteq_{{}_{{ \mathfrak{S}}}} \omega \sqcap_{{}_{{ \mathfrak{S}}}}\omega' \;\;\textit{\rm and}\;\; \sigma \sqsubseteq_{{}_{{ \mathfrak{S}}}} \Sigma^{C^{(1)}_{\omega}}_{ \mathfrak{D}}\right) \;\;\Rightarrow \;\; \left(\epsilon^{ \mathfrak{S}}_{{ \mathfrak{l}}_{(\sigma,\centerdot)}}(\Sigma^{C^{(1)}_{\omega}}_{ \mathfrak{D}})=\textit{\bf Y}=\epsilon^{ \mathfrak{S}}_{{ \mathfrak{l}}_{(\sigma,\centerdot)}}(\Sigma^{C^{(1)}_{\omega'}}_{ \mathfrak{D}})\right)
\;\;\Rightarrow \;\;  \sigma \sqsubseteq_{{}_{{ \mathfrak{S}}}} \Sigma^{C^{(1)}_{\omega'}}_{ \mathfrak{D}}.\;\;\;\;\;
\end{eqnarray}
In other words,  noting that $\Sigma^{C^{(1)}_{\omega}}_{ \mathfrak{D}} \sqsubseteq_{{}_{\overline{ \mathfrak{S}}}} \omega$, we obtain $\omega' \sqcap_{{}_{\overline{ \mathfrak{S}}}} \Sigma^{C^{(1)}_{\omega}}_{ \mathfrak{D}}  \sqsubseteq_{{}_{\overline{ \mathfrak{S}}}}   \Sigma^{C^{(1)}_{\omega'}}_{ \mathfrak{D}}$. The other inequality is obtained in the same way.
\end{proof}

\begin{lemma}\label{lemmanexistsomega0}
If there exists $\omega,\omega'\in \overline{ \mathfrak{S}}^{{}^{pure}}$ such that $\Sigma^{C^{(1)}_{\omega}}_{ \mathfrak{D}}$ and $\Sigma^{C^{(1)}_{\omega'}}_{ \mathfrak{D}}$ admit a common upper-bound in $\overline{ \mathfrak{S}}$, then for any $\omega''\in \underline{\Sigma^{C^{(1)}_{\omega}}_{ \mathfrak{D}} \sqcup_{{}_{\overline{ \mathfrak{S}}}} \Sigma^{C^{(1)}_{\omega'}}_{ \mathfrak{D}}}_{{}_{\overline{ \mathfrak{S}}}}$ we have $\Sigma^{C^{(1)}_{\omega''}}_{ \mathfrak{D}} \sqsupseteq_{{}_{\overline{ \mathfrak{S}}}} (\Sigma^{C^{(1)}_{\omega}}_{ \mathfrak{D}} \sqcup_{{}_{\overline{ \mathfrak{S}}}} \Sigma^{C^{(1)}_{\omega'}}_{ \mathfrak{D}})$
\end{lemma}
\begin{proof}
Straightforward.
\end{proof}

\begin{lemma}\label{lemmacU=U}
For any  coherent  operational description ${ \mathfrak{D}}$ the subset $U_{ \mathfrak{D}}$ of $\overline{ \mathfrak{S}}$ defined by $U_{ \mathfrak{D}}:=\{\, \Sigma^{C^{(1)}_{\omega}}_{ \mathfrak{D}}\;\vert\; \omega\in \overline{ \mathfrak{S}}{}^{{}^{pure}}\,\}$ satisfies
\begin{eqnarray}
cl_c^{\overline{ \mathfrak{S}}}(U_{ \mathfrak{D}})=U_{ \mathfrak{D}}
\end{eqnarray}
\end{lemma}
\begin{proof}
As a strict rewriting of property (\ref{consitencyontic}), we note that we have 
\begin{eqnarray}
\forall \omega\in \overline{ \mathfrak{S}}{}^{{}^{pure}},\exists \alpha\in U_{ \mathfrak{D}} &\vert & \alpha \sqsubseteq_{{}_{\overline{ \mathfrak{S}}}} \omega\;\textit{\rm and}\; \forall \beta\in U_{ \mathfrak{D}},(\omega \sqcap_{{}_{\overline{ \mathfrak{S}}}}\beta)  \sqsubseteq_{{}_{{ \mathfrak{S}}}}  \alpha.
\end{eqnarray}
and then
 \begin{eqnarray}
Max\{\,\bigsqcup{}^{{}^{\overline{ \mathfrak{S}}}}_{\kappa\in U_{ \mathfrak{D}}} (\omega\sqcap_{{}_{\overline{ \mathfrak{S}}}}\kappa)\;\vert\; \omega\in \overline{ \mathfrak{S}}{}^{{}^{pure}} \,\} \sqsubseteq U_{ \mathfrak{D}}.\;\;\;\;\;\;\;\;\;
\end{eqnarray}
Now, we note that
 \begin{eqnarray}
(V\sqsubseteq U_{ \mathfrak{D}}\;\textit{\rm and}\;\widehat{V}{}^{{}^{\overline{ \mathfrak{S}}}}) &\Rightarrow & \exists \omega\in \overline{ \mathfrak{S}}{}^{{}^{pure}}\;\vert\; V\sqsubseteq \{\, (\omega\sqcap_{{}_{\overline{ \mathfrak{S}}}}\kappa)\;\vert\; \kappa\in U_{ \mathfrak{D}} \,\}
\end{eqnarray}
and then
 \begin{eqnarray}
Max\{\,\bigsqcup{}^{{}^{\overline{ \mathfrak{S}}}}V\;\vert\; V\sqsubseteq U_{ \mathfrak{D}}\;\textit{\rm and}\;\widehat{V}{}^{{}^{\overline{ \mathfrak{S}}}}\} \sqsubseteq
Max\{\,\bigsqcup{}^{{}^{\overline{ \mathfrak{S}}}}_{\kappa\in U_{ \mathfrak{D}}} (\omega\sqcap_{{}_{\overline{ \mathfrak{S}}}}\kappa)\;\vert\; \omega\in \overline{ \mathfrak{S}}{}^{{}^{pure}} \,\}.\;\;\;\;\;\;\;\;\;
\end{eqnarray}
In other words
\begin{eqnarray}
{\mathfrak{c}}^{\overline{ \mathfrak{S}}}(U_{ \mathfrak{D}})\sqsubseteq U_{ \mathfrak{D}}
\end{eqnarray}
But we have also (see property (\ref{propB3}))
\begin{eqnarray}
{\mathfrak{c}}^{\overline{ \mathfrak{S}}}(U_{ \mathfrak{D}})\sqsupseteq U_{ \mathfrak{D}}
\end{eqnarray}
Then, by iteration, we deduce the announced result.
\end{proof}

\begin{lemma}\label{lemmanexistsomega1}
The subset of $\overline{ \mathfrak{S}}$ defined by $\{\, \Sigma^{C^{(1)}_{\omega}}_{ \mathfrak{D}}\;\vert\; \omega\in \overline{ \mathfrak{S}}^{{}^{pure}}\,\}$ satisfies 
\begin{eqnarray}
\nexists \gamma\in \overline{ \mathfrak{S}} &\vert & \{\,\gamma,\gamma^\star\,\} \sqsubseteq \{\, \Sigma^{C^{(1)}_{\omega}}_{ \mathfrak{D}}\;\vert\; \omega\in \overline{ \mathfrak{S}}^{{}^{pure}}\,\}.
\end{eqnarray}
\end{lemma}
\begin{proof}
Let us suppose that there exists $\gamma\in \overline{ \mathfrak{S}}$ such that $\{\,\gamma,\gamma^\star\,\} \sqsubseteq \{\, \Sigma^{C^{(1)}_{\sigma}}_{ \mathfrak{D}}\;\vert\; \sigma\in \overline{ \mathfrak{S}}^{{}^{pure}}\,\}$, and let us exhibit a contradiction. We denote by $\omega$ and $\omega'$ the elements of $\overline{ \mathfrak{S}}^{{}^{pure}}$ such that $\gamma \sqsubseteq_{{}_{\overline{ \mathfrak{S}}}} \Sigma^{C^{(1)}_{\omega}}_{ \mathfrak{D}}$ and $\gamma^\star \sqsubseteq_{{}_{\overline{ \mathfrak{S}}}}\Sigma^{C^{(1)}_{\omega'}}_{ \mathfrak{D}}$. There exists a maximal compatibility context of type 2 denoted $C$ such that ${ \mathfrak{l}}_{(\gamma,\gamma^\star)}\in C$.  Applying Theorem \ref{synthesiscontexts} to the maximal context $C$ we deduce easily that ${ \mathfrak{l}}_{(\gamma,\centerdot)}\in C$ and ${ \mathfrak{l}}_{(\centerdot,\gamma^\star)}\in C$. Due to the sheaf condition satisfied by ${ \mathfrak{D}}$, we know that
\begin{eqnarray}
&&{ \mathfrak{d}}^{C}({ \mathfrak{l}}_{(\gamma,\centerdot)}) = { \mathfrak{d}}^{C^{(1)}_{\omega}}({ \mathfrak{l}}_{(\gamma,\centerdot)}) = \epsilon^{ \mathfrak{S}}_{{ \mathfrak{l}}_{(\gamma,\centerdot)}}(\Sigma^{C^{(1)}_{\omega}}_{ \mathfrak{D}})=\textit{\bf Y}\label{absurd1a}\\
&&{ \mathfrak{d}}^{C}({ \mathfrak{l}}_{(\centerdot,\gamma^\star)}) = { \mathfrak{d}}^{C^{(1)}_{\omega'}}({ \mathfrak{l}}_{(\centerdot, \gamma^\star)}) = \epsilon^{ \mathfrak{S}}_{{ \mathfrak{l}}_{(\centerdot, \gamma^\star)}}(\Sigma^{C^{(1)}_{\omega'}}_{ \mathfrak{D}})=\textit{\bf N}\label{absurd1b}
\end{eqnarray}
but we have also
\begin{eqnarray}
&&{ \mathfrak{d}}^{C}({ \mathfrak{l}}_{(\gamma,\centerdot)})={ \mathfrak{d}}^{C}({ \mathfrak{Y}}_{{ \mathfrak{E}}_{ \mathfrak{S}}} \sqcap_{{}_{{ \mathfrak{E}}_{ \mathfrak{S}}}} { \mathfrak{l}}_{(\gamma,\gamma^\star)})=\textit{\bf Y} \wedge { \mathfrak{d}}^{C}({ \mathfrak{l}}_{(\gamma,\gamma^\star)})\label{absurd2a}
\\
&&{ \mathfrak{d}}^{C}({ \mathfrak{l}}_{(\centerdot,\gamma^\star)}) = { \mathfrak{d}}^{C}(\overline{{ \mathfrak{Y}}_{{ \mathfrak{E}}_{ \mathfrak{S}}}} \sqcap_{{}_{{ \mathfrak{E}}_{ \mathfrak{S}}}} { \mathfrak{l}}_{(\gamma,\gamma^\star)})=\textit{\bf N} \wedge { \mathfrak{d}}^{C}({ \mathfrak{l}}_{(\gamma,\gamma^\star)}).\label{absurd2b}
\end{eqnarray}
From (\ref{absurd1a}) and (\ref{absurd2a}) we deduce that necessarily ${ \mathfrak{d}}^{C}({ \mathfrak{l}}_{(\gamma,\gamma^\star)})=\textit{\bf Y}$, but using (\ref{absurd2b}) we deduce ${ \mathfrak{d}}^{C}({ \mathfrak{l}}_{(\centerdot,\gamma^\star)})=\textit{\bf N} \wedge \textit{\bf Y}=\bot$ which contradicts (\ref{absurd1b}). We have then obtained the announced contradiction. This concludes the proof.
\end{proof}

\begin{definition}\label{definitionoperationaldescriptionadmissible}
We will say that the operational description ${ \mathfrak{D}}$ is {\em admissible} iff $\{\, \Sigma^{C^{(1)}_{\omega}}_{ \mathfrak{D}}\;\vert\; \omega\in \overline{ \mathfrak{S}}^{{}^{pure}}\,\}$ is admissible, i.e. iff $\{\, \Sigma^{C^{(1)}_{\omega}}_{ \mathfrak{D}}\;\vert\; \omega\in \overline{ \mathfrak{S}}^{{}^{pure}}\,\}\in { \mathcal{Q}}_c(\overline{ \mathfrak{S}})$. 
\end{definition}

\begin{theorem}\label{theoremoperationaldescriptionadmissible}
Every operational description ${ \mathfrak{D}}$ are admissible.
\end{theorem}
\begin{proof}
Direct consequence of Lemma \ref{lemmanexistsomega1} and Lemma \ref{lemmacU=U}.
\end{proof}

\begin{lemma}\label{lemmasolutioncontext0}
As long as the operational description ${ \mathfrak{D}}$ is admissible (see Theorem \ref{theoremoperationaldescriptionadmissible}), the supremum
\begin{eqnarray}
\Sigma_{ \mathfrak{D}} &:=& \bigsqcup{}^{{}^{{ \mathfrak{S}}}}_{\omega\in \overline{ \mathfrak{S}}^{{}^{pure}}}\; \Sigma^{C^{(1)}_{\omega}}_{ \mathfrak{D}}\label{solutionontic1}
\end{eqnarray}
 is well defined as an element of ${ \mathfrak{S}}$.  
Then,we have
\begin{eqnarray}
\forall \omega \in \overline{ \mathfrak{S}}^{{}^{pure}},&&\Sigma^{C^{(1)}_{\omega}}_{ \mathfrak{D}} = \omega \sqcap_{{}_{{ \mathfrak{S}}}} \Sigma_{ \mathfrak{D}}\label{solutionontic2}
\end{eqnarray}
\end{lemma}
\begin{proof}
From Theorem \ref{theoremoperationaldescriptionadmissible},  we know that 
$\{\, \Sigma^{C^{(1)}_{\omega}}_{ \mathfrak{D}}\;\vert\; \omega\in \overline{ \mathfrak{S}}^{{}^{pure}}\,\}$ is in ${ \mathcal{Q}}_c(\overline{ \mathfrak{S}})$ and, then, 
\begin{eqnarray}
\Sigma_{ \mathfrak{D}} &:=& cl_c^{{{ \mathfrak{S}}'}}(\{\, \Sigma^{C^{(1)}_{\omega}}_{ \mathfrak{D}}\;\vert\; \omega\in \overline{ \mathfrak{S}}^{{}^{pure}}\,\})\\
&=&\bigsqcup{}^{{}^{{ \mathfrak{S}}}}_{\omega\in \overline{ \mathfrak{S}}^{{}^{pure}}}\; \Sigma^{C^{(1)}_{\omega}}_{ \mathfrak{D}}
\end{eqnarray}
 is well defined as an element of ${ \mathfrak{S}}$.  \\
Then, $\forall \omega' \in \overline{ \mathfrak{S}}^{{}^{pure}}, \; \omega \sqcap_{{}_{{ \mathfrak{S}}}} \Sigma^{C^{(1)}_{\omega'}}_{ \mathfrak{D}}  \sqsubseteq_{{}_{{ \mathfrak{S}}}}   \Sigma^{C^{(1)}_{\omega}}_{ \mathfrak{D}}$ implies $\omega \sqcap_{{}_{{ \mathfrak{S}}}} \Sigma_{ \mathfrak{D}}  \sqsubseteq_{{}_{{ \mathfrak{S}}}}   \Sigma^{C^{(1)}_{\omega}}_{ \mathfrak{D}}$. \\
On another part, $\Sigma^{C^{(1)}_{\omega}}_{ \mathfrak{D}}  \sqsubseteq_{{}_{{ \mathfrak{S}}}}   \Sigma_{ \mathfrak{D}}$ implies $\omega \sqcap_{{}_{{ \mathfrak{S}}}} \Sigma^{C^{(1)}_{\omega}}_{ \mathfrak{D}}  \sqsubseteq_{{}_{{ \mathfrak{S}}}} \omega \sqcap_{{}_{{ \mathfrak{S}}}} \Sigma_{ \mathfrak{D}}$.\\
Endly, we recall that $\omega \sqcap_{{}_{{ \mathfrak{S}}}} \Sigma^{C^{(1)}_{\omega}}_{ \mathfrak{D}} = \Sigma^{C^{(1)}_{\omega}}_{ \mathfrak{D}}$.\\
As a conclusion, we obtain $\Sigma^{C^{(1)}_{\omega}}_{ \mathfrak{D}} = \omega \sqcap_{{}_{{ \mathfrak{S}}}} \Sigma_{ \mathfrak{D}}$.\\
Conversely, we easily check that the expression (\ref{solutionontic2}) satisfies the constraints (\ref{consitencyontic}).
\end{proof}

\begin{lemma}\label{lemmasolutioncontext1}
Let us consider a maximal compatibility context of type 1 denoted $C^{(1)}_{\omega}$.  We have
\begin{eqnarray}
\forall { \mathfrak{l}}\in C^{(1)}_{\omega},&& { \mathfrak{d}}^{C^{(1)}_{\omega}}({ \mathfrak{l}})=\epsilon^{ \mathfrak{S}}_{ \mathfrak{l}}(\Sigma_{ \mathfrak{D}}).
\end{eqnarray}
\end{lemma}
\begin{proof}
Let us compute ${ \mathfrak{d}}^{C^{(1)}_{\omega}}({ \mathfrak{l}})$ for any ${ \mathfrak{l}}\in C^{(1)}_{\omega}$.  First of all, we note that, for any $\alpha \sqsubseteq_{{}_{ \mathfrak{S}}}\omega$, we have
\begin{eqnarray}
&&\epsilon^{ \mathfrak{S}}_{{ \mathfrak{l}}_{(\alpha,\centerdot)}}(\Sigma_{ \mathfrak{D}})=\epsilon^{ \mathfrak{S}}_{{ \mathfrak{Y}}_{{ \mathfrak{E}}_{ \mathfrak{S}}} \sqcap_{{}_{{ \mathfrak{E}}_{ \mathfrak{S}}}}
{ \mathfrak{l}}_{(\alpha,\centerdot)}}(\Sigma_{ \mathfrak{D}})=\textit{\bf Y}\wedge \epsilon^{ \mathfrak{S}}_{{ \mathfrak{l}}_{(\alpha,\centerdot)}}(\Sigma_{ \mathfrak{D}})
\end{eqnarray}
and then 
\begin{eqnarray}
\epsilon^{ \mathfrak{S}}_{{ \mathfrak{l}}_{(\alpha,\centerdot)}}(\Sigma_{ \mathfrak{D}}) & \in & \{\textit{\bf Y},\bot\}.
\end{eqnarray}
As a consequence, we have
\begin{eqnarray}
{ \mathfrak{d}}^{C^{(1)}_{\omega}}({ \mathfrak{l}}_{(\alpha,\centerdot)})=\epsilon^{ \mathfrak{S}}_{{ \mathfrak{l}}_{(\alpha,\centerdot)}}(\Sigma^{C^{(1)}_{\omega}}_{ \mathfrak{D}})=\epsilon^{ \mathfrak{S}}_{{ \mathfrak{l}}_{(\alpha,\centerdot)}}(\omega \sqcap_{{}_{ \mathfrak{S}}}
\Sigma_{ \mathfrak{D}})=\textit{\bf Y}\wedge
\epsilon^{ \mathfrak{S}}_{{ \mathfrak{l}}_{(\alpha,\centerdot)}}(\Sigma_{ \mathfrak{D}})=\epsilon^{ \mathfrak{S}}_{{ \mathfrak{l}}_{(\alpha,\centerdot)}}(\Sigma_{ \mathfrak{D}}). \;\;\;\;\;\;\;\;\;\;
\end{eqnarray}
We can check analogously 
\begin{eqnarray}
{ \mathfrak{d}}^{C^{(1)}_{\omega}}({ \mathfrak{l}}_{(\centerdot,\alpha)})=\epsilon^{ \mathfrak{S}}_{{ \mathfrak{l}}_{(\centerdot,\alpha)}}(\Sigma_{ \mathfrak{D}}). 
\end{eqnarray}
\end{proof}

\begin{lemma}\label{lemmasolutioncontext2}
Let us consider a maximal compatibility context of type 2 denoted by $C$. We have
\begin{eqnarray}
\forall { \mathfrak{l}}\in C, && { \mathfrak{d}}^{C}({ \mathfrak{l}}) = \epsilon^{ \mathfrak{S}}_{{ \mathfrak{l}}}(\Sigma_{ \mathfrak{D}}).
\end{eqnarray}
\end{lemma}
\begin{proof}
Let us consider a maximal compatibility context of type 2 denoted by $C$ and let us choose $\lambda,\lambda'\in \overline{ \mathfrak{S}} \smallsetminus \{\bot_{{}_{\overline{ \mathfrak{S}}}}\}$ such that $\lambda' \sqsupseteq_{{}_{\overline{ \mathfrak{S}}}} \lambda^\star$ and $\omega,\omega'\in  {\overline{ \mathfrak{S}}}^{{}^{pure}}$ such that $\omega  \sqsupseteq_{{}_{\overline{ \mathfrak{S}}}} \lambda, \omega'  \sqsupseteq_{{}_{\overline{ \mathfrak{S}}}} \lambda'$ with 
\begin{eqnarray}
&& (C\cap { \mathfrak{U}}_{\overline{ \mathfrak{S}}})\subseteq \left( \{\, { \mathfrak{l}}_{(\sigma,\centerdot)}\;\vert\; \lambda \sqsubseteq_{{}_{{ \mathfrak{S}}}} \sigma\sqsubseteq_{{}_{{ \mathfrak{S}}}} \omega\,\}\cup \{\, { \mathfrak{l}}_{(\centerdot,\sigma)}\;\vert\; \lambda \sqsubseteq_{{}_{{ \mathfrak{S}}}}\sigma\sqsubseteq_{{}_{{ \mathfrak{S}}}} \omega\,\}\cup \{\, { \mathfrak{l}}_{(\centerdot,\centerdot)}\,\}\cup \right.\nonumber\\
&&\hspace{1cm}\left.  \{\, { \mathfrak{l}}_{(\sigma,\centerdot)}\;\vert\; \lambda' \sqsubseteq_{{}_{{ \mathfrak{S}}}} \sigma\sqsubseteq_{{}_{{ \mathfrak{S}}}} \omega'\,\}\cup \{\, { \mathfrak{l}}_{(\centerdot,\sigma)}\;\vert\; \lambda' \sqsubseteq_{{}_{{ \mathfrak{S}}}}\sigma\sqsubseteq_{{}_{{ \mathfrak{S}}}} \omega'\,\}\right)
\end{eqnarray}
We note that
\begin{eqnarray}
\forall { \mathfrak{l}}_{(\alpha,\alpha')}\in C, && \left( \omega \sqsupseteq_{{}_{\overline{ \mathfrak{S}}}} \alpha \sqsupseteq_{{}_{\overline{ \mathfrak{S}}}} \lambda \;\;\textit{\rm and}\;\; \omega' \sqsupseteq_{{}_{\overline{ \mathfrak{S}}}} \alpha' \sqsupseteq_{{}_{\overline{ \mathfrak{S}}}} \lambda'\right).
\end{eqnarray}
We also note that, for any ${ \mathfrak{l}}_{(\alpha,\alpha')}\in C$, we have
\begin{eqnarray}
\textit{\bf Y}\wedge { \mathfrak{d}}^{C}({ \mathfrak{l}}_{(\alpha,\alpha')}) &=&{ \mathfrak{d}}^{C}({ \mathfrak{Y}}_{{ \mathfrak{E}}_{ \mathfrak{S}}} \sqcap_{{}_{{ \mathfrak{E}}_{ \mathfrak{S}}}}
{ \mathfrak{l}}_{(\alpha,\alpha')})\nonumber\\
&=& { \mathfrak{d}}^{C}({ \mathfrak{l}}_{(\alpha,\centerdot)})\nonumber\\
&=& { \mathfrak{d}}^{C^{(1)}_\omega}({ \mathfrak{l}}_{(\alpha,\centerdot)})\nonumber\\
&=& \epsilon^{ \mathfrak{S}}_{{ \mathfrak{l}}_{(\alpha,\centerdot)}}(\Sigma_{ \mathfrak{D}})\nonumber\\
&=& \epsilon^{ \mathfrak{S}}_{{ \mathfrak{Y}}_{{ \mathfrak{E}}_{ \mathfrak{S}}} \sqcap_{{}_{{ \mathfrak{E}}_{ \mathfrak{S}}}} { \mathfrak{l}}_{(\alpha,\alpha')}}(\Sigma_{ \mathfrak{D}})\nonumber\\
&=& \textit{\bf Y}\wedge \epsilon^{ \mathfrak{S}}_{{ \mathfrak{l}}_{(\alpha,\alpha')}}(\Sigma_{ \mathfrak{D}})
\end{eqnarray}
and 
\begin{eqnarray}
\textit{\bf N}\wedge { \mathfrak{d}}^{C}({ \mathfrak{l}}_{(\alpha,\alpha')}) &=&{ \mathfrak{d}}^{C}(\overline{{ \mathfrak{Y}}_{{ \mathfrak{E}}_{ \mathfrak{S}}}} \sqcap_{{}_{{ \mathfrak{E}}_{ \mathfrak{S}}}}
{ \mathfrak{l}}_{(\alpha,\alpha')})\nonumber\\
&=& { \mathfrak{d}}^{C}({ \mathfrak{l}}_{(\centerdot,\alpha')})\nonumber\\
&=& { \mathfrak{d}}^{C^{(1)}_{\omega'}}({ \mathfrak{l}}_{(\centerdot,\alpha')})\nonumber\\
&=& \epsilon^{ \mathfrak{S}}_{{ \mathfrak{l}}_{(\centerdot,\alpha')}}(\Sigma_{ \mathfrak{D}})\nonumber\\
&=& \epsilon^{ \mathfrak{S}}_{\overline{{ \mathfrak{Y}}_{{ \mathfrak{E}}_{ \mathfrak{S}}}} \sqcap_{{}_{{ \mathfrak{E}}_{ \mathfrak{S}}}} { \mathfrak{l}}_{(\alpha,\alpha')}}(\Sigma_{ \mathfrak{D}})\nonumber\\
&=& \textit{\bf N}\wedge \epsilon^{ \mathfrak{S}}_{{ \mathfrak{l}}_{(\alpha,\alpha')}}(\Sigma_{ \mathfrak{D}}).
\end{eqnarray}
From these two relations, we deduce immediately
\begin{eqnarray}
{ \mathfrak{d}}^{C}({ \mathfrak{l}}_{(\alpha,\alpha')}) &=& \epsilon^{ \mathfrak{S}}_{{ \mathfrak{l}}_{(\alpha,\alpha')}}(\Sigma_{ \mathfrak{D}}).
\end{eqnarray}
\end{proof}

\begin{theorem} \label{DSigmaepsilon}
Let us consider an operational description ${ \mathfrak{D}}$. We have then necessarily the following property :
\begin{eqnarray}
\exists \Sigma_{ \mathfrak{D}}:=  (\bigsqcup{}^{{}^{{ \mathfrak{S}}}}_{\omega\in \overline{ \mathfrak{S}}^{{}^{pure}}}\; \Sigma^{C^{(1)}_{\omega}}_{ \mathfrak{D}}) &\vert & \forall C\in {\mathcal{C}},\forall { \mathfrak{l}}\in C,\;\; { \mathfrak{d}}^{C}({ \mathfrak{l}})=\epsilon^{ \mathfrak{S}}_{ \mathfrak{l}}(\Sigma_{ \mathfrak{D}}).\;\;\;\;\;\;\;\;\;\;
\end{eqnarray}
\end{theorem}
\begin{proof}
Direct consequence of Lemma \ref{lemmasolutioncontext0}, Lemma \ref{lemmasolutioncontext1},Lemma \ref{lemmasolutioncontext2}.
\end{proof}

\begin{theorem}\label{DSigmaisomorphism}
The map from ${ \mathcal{D}}$ to ${ \mathfrak{S}}$, sending ${ \mathfrak{D}}$ to $\Sigma_{ \mathfrak{D}}:=  (\bigsqcup{}^{{}^{{ \mathfrak{S}}}}_{\omega\in \overline{ \mathfrak{S}}^{{}^{pure}}}\; \Sigma^{C^{(1)}_{\omega}}_{ \mathfrak{D}})$ is an isomorphism. 
\end{theorem}
\begin{proof}
First of all, we note that we can effectively associate to any ${ \mathfrak{D}}\in { \mathcal{D}}$ an element $\Sigma_{ \mathfrak{D}}$ satisfying
\begin{eqnarray}
\forall C\in {\mathcal{C}},\forall { \mathfrak{l}}\in C, && { \mathfrak{d}}^{C}({ \mathfrak{l}})=\epsilon^{ \mathfrak{S}}_{ \mathfrak{l}}(\Sigma_{ \mathfrak{D}}).
\end{eqnarray}
This point is guarantied by Lemma \ref{lemmasolutioncontext0}, Lemma \ref{lemmasolutioncontext1} and Lemma \ref{lemmasolutioncontext2}.  More explicitly, we have $\Sigma_{ \mathfrak{D}}:=  (\bigsqcup{}^{{}^{{ \mathfrak{S}}}}_{\omega\in \overline{ \mathfrak{S}}^{{}^{pure}}}\; \Sigma^{C^{(1)}_{\omega}}_{ \mathfrak{D}})\in { \mathfrak{S}}$ where the state $\Sigma^{C^{(1)}_{\omega}}_{ \mathfrak{D}}$ is defined in reference to (\ref{onticstate1}) (\ref{onticstate2}) and Lemma \ref{lemmacontexttype1}.\\
Note that the supremum defining $\Sigma_{ \mathfrak{D}}$ exists in ${ \mathfrak{S}}$ because ${ \mathfrak{D}}$ is admissible.\\
Secondly, this map is clearly injective.  Indeed, let us consider a pair ${ \mathfrak{D}},{ \mathfrak{D}}'$ of coherent  operational descriptions such that $\Sigma_{ \mathfrak{D}}=\Sigma_{ \mathfrak{D}'}$.  We have obviously $\forall C\in {\mathcal{C}},\forall { \mathfrak{l}}\in C,\;\; { \mathfrak{d}}^{C}({ \mathfrak{l}})=\epsilon^{ \mathfrak{S}}_{ \mathfrak{l}}(\Sigma_{ \mathfrak{D}})=\epsilon^{ \mathfrak{S}}_{ \mathfrak{l}}(\Sigma_{ \mathfrak{D}'})={ \mathfrak{d}'}^{C}({ \mathfrak{l}})$.\\
Endly, this map is also clearly a surjection. Indeed, we can associate, to any $\alpha\in { \mathfrak{S}}$,  a coherent  operational description ${ \mathfrak{D}}_\alpha:=({ \mathfrak{d}}_{(\alpha)}^{C})_{C\in {\mathcal{C}}}$ by $\forall C\in {\mathcal{C}},\forall { \mathfrak{l}}\in C,{ \mathfrak{d}}_{(\alpha)}^{C}({ \mathfrak{l}}):=\epsilon^{ \mathfrak{S}}_{ \mathfrak{l}}(\alpha)$.  We have then, for any $\omega\in \overline{ \mathfrak{S}}^{{}^{pure}}$, 
\begin{eqnarray}
\Sigma^{C^{(1)}_{\omega}}_{ \mathfrak{D}_\alpha} &=& \alpha \sqcap_{{}_{{ \mathfrak{S}}}} \omega.
\end{eqnarray}
We then build as usual $\Sigma_{ \mathfrak{D}_\alpha}:=  (\bigsqcup{}^{{}^{{ \mathfrak{S}}}}_{\omega\in \overline{ \mathfrak{S}}^{{}^{pure}}}\; \Sigma^{C^{(1)}_{\omega}}_{ \mathfrak{D}_\alpha})$. We have by construction $\forall C\in {\mathcal{C}},\forall { \mathfrak{l}}\in C,\epsilon^{ \mathfrak{S}}_{ \mathfrak{l}}(\Sigma_{ {\mathfrak{D}}_{(\alpha)}})=\epsilon^{ \mathfrak{S}}_{ \mathfrak{l}}(\alpha)$. Using (\ref{realChuseparated}), we then conclude that $\Sigma_{ {\mathfrak{D}}_{(\alpha)}} =\alpha$.\\
Endly, we derive easily the homomorphic property from the properties of the evaluation map $\epsilon^{ \mathfrak{S}}$ :
\begin{eqnarray}
\forall C\in {\mathcal{C}},\forall { \mathfrak{l}}\in C,&& ({ \mathfrak{d}}^{C}\sqcap_{{}_{\mathcal{D}}} { \mathfrak{d}'}^{C})({ \mathfrak{l}})=\epsilon^{ \mathfrak{S}}_{ \mathfrak{l}}(\Sigma_{ \mathfrak{D}})\wedge \epsilon^{ \mathfrak{S}}_{ \mathfrak{l}}(\Sigma_{ \mathfrak{D}'})
=\epsilon^{ \mathfrak{S}}_{ \mathfrak{l}}(\Sigma_{ \mathfrak{D}}\sqcap_{{}_{{ \mathfrak{S}}}} \Sigma_{ \mathfrak{D}'}).\;\;\;\;\;\;\;\;\;\;\;\;\;\;\;\;\;\;\;\;
\end{eqnarray}
\end{proof}

\section{The multidimensional linear indeterministic space of states  (the multi-quantum-bit system)}\label{sectionmultidimensional}

The previous section established a general proposal for the tensor product of indeterministic spaces of states. In the next section, we will see that this construction is sufficient to achieve our main objective (i.e. to exhibit the main typically quantum properties in our quantum logic framework). The present section is devoted to a detailed analysis of the space of states obtained as the iterated tensor product of the elementary indeterministic space of states associated with the quantum bit. Not only can this space of states be simply described using the ontic completion tools developed above, but we will also demonstrate how close this space is to an irreducible Hilbert geometry.

\subsection{Preliminary remarks}\label{subsectionpreliminaryNdim}

To begin this subsection, we consider the two copies of one-dimensional indeterministic spaces of states ${{{ \overline{\mathfrak{S}}}}}_{A}={\mathfrak{Z}}'_{N_A}$ and ${{{ \overline{\mathfrak{S}}}}}_{B}={\mathfrak{Z}}'_{N_B}$ and we consider the space of states ${{{ \overline{S}}}}_{AB}:={{{ \overline{\mathfrak{S}}}}}_{A}\widetilde{\otimes}{{{ \overline{\mathfrak{S}}}}}_{B}$. We denote ${ \mathfrak{S}}_A={ \mathfrak{J}}^c_{{\overline{ \mathfrak{S}}_A}}$ and ${ \mathfrak{S}}_B={ \mathfrak{J}}^c_{{\overline{ \mathfrak{S}}_B}}$. We intent to describe the ontic completion ${{{ \hat{S}}}}_{AB}:={ \mathfrak{S}}_{A}\widehat{\otimes} { \mathfrak{S}}_{B}$ defined according to the subsection \ref{subsectiontensorcompleterealspace}. 
If necessary, in order to shorten our formulas, we will eventually replace the notation ${{{ \overline{S}}}}{}_{AB}$ by $\overline{\mathfrak{S}}$ and the notation ${{{ \hat{S}}}}{}_{AB}$ by ${\mathfrak{S}}$.

\begin{theorem}
Let us consider $\mu,\nu,\phi\in \overline{ \mathfrak{S}}{}^{{}^{pure}}$ and $\gamma\in \overline{ \mathfrak{S}}$ such that $\gamma=(\phi\sqcap_{{}_{\overline{\mathfrak{S}}}}\nu)\sqcoversubset_{{}_{\overline{ \mathfrak{S}}}} \phi,\nu$ and $\mu^\star\not\sqsubseteq_{{}_{\overline{ \mathfrak{S}}}}\gamma$, we have
\begin{eqnarray}
\chi:=(\mu^\star\sqcup_{{}_{{ \mathfrak{S}}}}\gamma)\;\;\textit{\rm exists in}\;{ \mathfrak{S}}.
\label{thirdcoveringpropertySbari}
\end{eqnarray}
We have then either $\chi\in \overline{ \mathfrak{S}}{}^{{}^{pure}}$ or $\chi\in { \mathfrak{S}}\smallsetminus \overline{ \mathfrak{S}}$ and then, in the last case,  denoting as usual $\Theta^{\overline{ \mathfrak{S}}}(\chi):=Max((\downarrow_{{}_{{ \mathfrak{S}}}}\chi)\cap \overline{ \mathfrak{S}})$, we have
\begin{eqnarray}
&&\forall \alpha\in \Theta^{\overline{ \mathfrak{S}}}(\chi),\;\;  (\exists \kappa\in \overline{ \mathfrak{S}}{}^{{}^{pure}}\;\vert\; \alpha\sqcoversubset_{{}_{\overline{ \mathfrak{S}}}}\kappa)\;\textit{\rm and}\;
\alpha\sqcoversubset_{{}_{{ \mathfrak{S}}}} \chi 
\label{thirdcoveringpropertySbarii}\;\;\;\;\;\;\;\;\;\;\;\;\;\;\;\;\;\;\\
&& \forall \alpha,\beta\in \Theta^{\overline{ \mathfrak{S}}}(\chi),\;\alpha\not=\beta,\;\;\;\;(\alpha\sqcap_{{}_{\overline{ \mathfrak{S}}}}\beta)\sqcoversubset_{{}_{\overline{ \mathfrak{S}}}} \alpha,\beta,\label{thirdcoveringpropertySbariii}\\
&&\forall \alpha,\beta,\delta\in \Theta^{\overline{ \mathfrak{S}}}(\chi)\;\textit{\rm distinct}, \;\;\;\;(\alpha\sqcap_{{}_{\overline{ \mathfrak{S}}}}\beta)\not=(\alpha\sqcap_{{}_{\overline{ \mathfrak{S}}}}\delta),\;\;\;\;\;\;\;\;\;\;\;\;\;\;\;\;\;\;\label{thirdcoveringpropertySbariv}
\end{eqnarray}

Let us once again consider $\mu\in \overline{ \mathfrak{S}}{}^{{}^{pure}}$ and $\gamma\in \overline{ \mathfrak{S}}$ such that $\mu^\star\not\sqsubseteq_{{}_{\overline{ \mathfrak{S}}}}\gamma$.
Let us assume that there exists $\nu\in \overline{ \mathfrak{S}}{}^{{}^{pure}}$ such that $\gamma=(\mu\sqcap_{{}_{\overline{\mathfrak{S}}}}\nu)\sqcoversubset_{{}_{\overline{ \mathfrak{S}}}} \mu,\nu$ then let us here again denote $\chi=(\mu^\star\sqcup_{{}_{{ \mathfrak{S}}}}\gamma)$.\\ 
If we have $\chi\in \overline{ \mathfrak{S}}$, then necessarily $\chi\in \overline{ \mathfrak{S}}{}^{{}^{pure}}$, and if we have $\chi\in { \mathfrak{S}}\smallsetminus \overline{ \mathfrak{S}}$ then
\begin{eqnarray}
&&\exists (\varphi_{\epsilon})_{\epsilon},(\psi_{\epsilon})_{\epsilon} \in (\overline{ \mathfrak{S}}{}^{{}^{pure}})^{\Theta^{\overline{ \mathfrak{S}}}(\chi)}\;\vert\; \nonumber\\
&&\hspace{0.3cm}\varphi_{\gamma}:=\mu,\psi_{\gamma}:=\nu, \;\;\forall \alpha\in \Theta^{\overline{ \mathfrak{S}}}(\chi)\smallsetminus \{\gamma\},\;\;\; \varphi_\gamma^\star\sqsubseteq_{{}_{\overline{ \mathfrak{S}}}} \varphi_\alpha,\;\;\varphi_\gamma^\star \sqsubseteq_{{}_{\overline{ \mathfrak{S}}}} \psi_\alpha,\label{thirdcoveringpropertySbar0}\\
&&\hspace{0.3cm}\forall \epsilon\in \Theta^{\overline{ \mathfrak{S}}}(\chi)\smallsetminus\{\gamma\},\;\;\;
\psi_{\epsilon},\varphi_{\epsilon}\sqcoversupset_{{}_{\overline{ \mathfrak{S}}}}\epsilon,\;\;
\psi_{\epsilon}\not=\varphi_{\epsilon}\label{thirdcoveringpropertySbar0bis}\\
&&\hspace{0.3cm}\forall \epsilon\in \Theta^{\overline{ \mathfrak{S}}}(\chi)\smallsetminus \{\gamma\},\;\;\;\varphi_{\epsilon}^\star \sqsubseteq_{{}_{{ \mathfrak{S}}}}\psi_{\epsilon}, 
\;\;\;\;\;\;\;\;\;\;\;\;\label{thirdcoveringpropertySbarviic}\\
&&\hspace{0.3cm}\forall \alpha,\beta\in \Theta^{\overline{ \mathfrak{S}}}(\chi)\smallsetminus \{\gamma\},\alpha\not=\beta,\;\;\; \varphi_{\alpha}^\star \sqsubseteq_{{}_{\overline{ \mathfrak{S}}}}\varphi_{\beta},\;\;\varphi_{\alpha}^\star \not\sqsubseteq_{{}_{\overline{ \mathfrak{S}}}}\psi_{\beta},\label{thirdcoveringpropertySbarviib}\;\;\;\;\;\;\;\;\;\;\;\;\;\;\;\\
&&\hspace{0.3cm}\forall \alpha,\beta\in \Theta^{\overline{ \mathfrak{S}}}(\chi)\smallsetminus \{\gamma\},\alpha\not=\beta,\;\exists \lambda_{\alpha\beta}\in \overline{ \mathfrak{S}}{}^{{}^{pure}}\;\vert\; 
(\varphi_{\beta}\sqcap_{{}_{\overline{ \mathfrak{S}}}}\psi_{\alpha}), (\psi_{\beta}\sqcap_{{}_{\overline{ \mathfrak{S}}}}\varphi_{\alpha})\sqcoversubset_{{}_{\overline{ \mathfrak{S}}}}\lambda_{\alpha\beta},
\;\;\;\;\;\;\;\;\;\;\;\;\nonumber\\
&&\hspace{6.5cm}\textit{\rm and}\;\lambda_{\alpha\beta}^\star \sqsubseteq_{{}_{\overline{ \mathfrak{S}}}}\varphi_{\alpha},\;\;\lambda_{\alpha\beta}^\star \sqsubseteq_{{}_{\overline{ \mathfrak{S}}}}\varphi_{\beta},\label{thirdcoveringpropertySbarviia}\\
&&\hspace{0.3cm}\forall \alpha\in \Theta^{\overline{ \mathfrak{S}}}(\chi)\smallsetminus \{\gamma\},\exists \mu_{\alpha\gamma}\in \overline{ \mathfrak{S}}{}^{{}^{pure}}\;\vert\; 
(\varphi_{\gamma}\sqcap_{{}_{\overline{ \mathfrak{S}}}}\varphi_{\alpha}), (\psi_{\gamma}\sqcap_{{}_{\overline{ \mathfrak{S}}}}\psi_{\alpha})\sqcoversubset_{{}_{\overline{ \mathfrak{S}}}}\mu_{\alpha\gamma}
\;\;\;\;\;\;\;\;\;\;\;\;\nonumber\\
&&\hspace{6.5cm}\textit{\rm and}\;\mu_{\alpha\gamma}^\star \sqsubseteq_{{}_{\overline{ \mathfrak{S}}}}\psi_{\alpha},\label{thirdcoveringpropertySbarviiabis}
\end{eqnarray}
The collections $(\varphi_{\epsilon})_{\epsilon},(\psi_{\epsilon})_{\epsilon} \in (\overline{ \mathfrak{S}}{}^{{}^{pure}})^{\Theta^{\overline{ \mathfrak{S}}}(\chi)}$ satisfying the properties (\ref{thirdcoveringpropertySbar0})(\ref{thirdcoveringpropertySbarviic})(\ref{thirdcoveringpropertySbarviib}) are {unique} up to a permutation on the elements of $\Theta^{\overline{ \mathfrak{S}}}(\chi)\smallsetminus \{\gamma\}$. \\

Let us once again consider $\mu\in \overline{ \mathfrak{S}}{}^{{}^{pure}}$ and $\gamma\in \overline{ \mathfrak{S}}$ such that $\mu^\star\not\sqsubseteq_{{}_{\overline{ \mathfrak{S}}}}\gamma$.
If there exists $\nu\in \overline{ \mathfrak{S}}{}^{{}^{pure}}$ such that $\gamma=(\mu\sqcap_{{}_{\overline{\mathfrak{S}}}}\nu)\sqcoversubset_{{}_{\overline{ \mathfrak{S}}}} \mu,\nu$ then let us here again denote $\chi=(\mu^\star\sqcup_{{}_{{ \mathfrak{S}}}}\gamma)$. If we have $\chi\in { \mathfrak{S}}\smallsetminus \overline{ \mathfrak{S}}$ then
\begin{eqnarray}
&&\forall \beta\in \Theta^{\overline{ \mathfrak{S}}}(\chi)\smallsetminus \{\gamma\},\exists \omega,\kappa\in \overline{ \mathfrak{S}}{}^{{}^{pure}}\;\vert\; \nonumber\\
&&\hspace{3cm}(\omega\sqcap_{{}_{\overline{ \mathfrak{S}}}}\kappa)\sqcoversubset_{{}_{\overline{ \mathfrak{S}}}}\omega,\kappa,\;\textit{\rm and}\;
\omega^\star\sqcup_{{}_{{ \mathfrak{S}}}}(\omega\sqcap_{{}_{\overline{ \mathfrak{S}}}}\kappa)=\beta\sqcup_{{}_{{ \mathfrak{S}}}}\gamma^\star
\;\;\;\;\;\;\;\;\;\;\;\;
\label{thirdcoveringpropertySbarvi}
\end{eqnarray}
\end{theorem}
\begin{proof}
We will consider $\alpha,\alpha',\alpha''\in \overline{ \mathfrak{S}}{}_{A}^{{}^{pure}}$ and $\beta,\beta',\beta''\in \overline{ \mathfrak{S}}{}_{B}^{{}^{pure}}$ and we intent to compute 
$$cl_c^{{{{ \overline{S}}}}_{AB}}\left(
\left\{(\alpha''\widetilde{\otimes}\beta''\sqcap_{{}_{{{{ \overline{S}}}}_{AB}}}\!\!\!\alpha'\widetilde{\otimes}\beta')
\,,\,(\alpha\widetilde{\otimes}\beta)^\star\right\}\right)$$ with $\alpha'\not=\alpha''$ or $\beta'\not=\beta''$, and with $(\alpha\widetilde{\otimes}\beta)^\star\not\sqsubseteq_{{}_{{{{ \overline{\mathfrak{S}}}}}}}(\alpha''\widetilde{\otimes}\beta''\sqcap_{{}_{{{{ \overline{\mathfrak{S}}}}}}}\!\!\!\alpha'\widetilde{\otimes}\beta')$.\\

We will distinguish different cases.\\

(1) $\alpha'\not=\alpha''$  and $\beta''=\beta'$.  We have $(\alpha''\widetilde{\otimes}\beta''\sqcap_{{}_{{{{ \overline{\mathfrak{S}}}}}}}\!\!\!\alpha'\widetilde{\otimes}\beta')=\bot_{{}_{\overline{ \mathfrak{S}}_{A}}}\widetilde{\otimes}\beta'$ and then $\alpha^\star\widetilde{\otimes}\beta'$ is the supremum of $(\alpha''\widetilde{\otimes}\beta''\sqcap_{{}_{{{{ \overline{\mathfrak{S}}}}}}}\!\!\!\alpha'\widetilde{\otimes}\beta')$ and $(\alpha^\star \widetilde{\otimes}\bot_{{}_{\overline{ \mathfrak{S}}_{B}}})\sqcap_{{}_{{{{ \overline{\mathfrak{S}}}}}}} (\bot_{{}_{\overline{ \mathfrak{S}}_{A}}}\widetilde{\otimes} \beta^\star)$.  Hence, $$cl_c^{{{{ \overline{\mathfrak{S}}}}}}\left(
\left\{(\alpha''\widetilde{\otimes}\beta''\sqcap_{{}_{{{{ \overline{\mathfrak{S}}}}}}}\!\!\!\alpha'\widetilde{\otimes}\beta')
\,,\,(\alpha\widetilde{\otimes}\beta)^\star\right\}\right)=\{\alpha^\star\widetilde{\otimes}\beta'\}\;\in \; {\mathcal{K}}({{{ \overline{\mathfrak{S}}}}})$$

(2) $\alpha'=\alpha''$  and $\beta''\not=\beta'$.  We have $(\alpha''\widetilde{\otimes}\beta''\sqcap_{{}_{{{{ \overline{\mathfrak{S}}}}}}}\!\!\!\alpha'\widetilde{\otimes}\beta')=\alpha'\widetilde{\otimes}\bot_{{}_{\overline{ \mathfrak{S}}_{A}}}$ and then $\alpha'\widetilde{\otimes}\beta^\star$ is the supremum of $(\alpha''\widetilde{\otimes}\beta''\sqcap_{{}_{{{{ \overline{\mathfrak{S}}}}}}}\!\!\!\alpha'\widetilde{\otimes}\beta')$ and $(\alpha^\star \widetilde{\otimes}\bot_{{}_{\overline{ \mathfrak{S}}_{B}}})\sqcap_{{}_{{{{ \overline{\mathfrak{S}}}}}}} (\bot_{{}_{\overline{ \mathfrak{S}}_{A}}}\widetilde{\otimes} \beta^\star)$.  Hence, $$cl_c^{{{{ \overline{\mathfrak{S}}}}}}\left(
\left\{(\alpha''\widetilde{\otimes}\beta''\sqcap_{{}_{{{{ \overline{\mathfrak{S}}}}}}}\!\!\!\alpha'\widetilde{\otimes}\beta')
\,,\,(\alpha\widetilde{\otimes}\beta)^\star\right\}\right)=\{\alpha'\widetilde{\otimes}\beta^\star\}\;\in \; {\mathcal{K}}({{{ \overline{\mathfrak{S}}}}})$$

(3) $\alpha''\not=\alpha'$ or $\beta''\not=\beta'$.  We have to distinguish different sub-cases :\\
(3A) $\alpha'\sqsupseteq_{{}_{\overline{ \mathfrak{S}}_{A}}} \alpha^\star$ or $\beta'\sqsupseteq_{{}_{\overline{ \mathfrak{S}}_{A}}} \beta^\star$. Then, $\alpha'\widetilde{\otimes}\beta'$ is the supremum in ${{{ \overline{\mathfrak{S}}}}}$ of $(\alpha''\widetilde{\otimes}\beta''\sqcap_{{}_{{{{ \overline{\mathfrak{S}}}}}}}\!\!\!\alpha'\widetilde{\otimes}\beta')$ and $(\alpha^\star \widetilde{\otimes}\bot_{{}_{\overline{ \mathfrak{S}}_{B}}})\sqcap_{{}_{{{{ \overline{\mathfrak{S}}}}}}} (\bot_{{}_{\overline{ \mathfrak{S}}_{A}}}\widetilde{\otimes} \beta^\star)$. And then, $$cl_c^{{{{ \overline{\mathfrak{S}}}}}}\left(
\left\{(\alpha''\widetilde{\otimes}\beta''\sqcap_{{}_{{{{ \overline{\mathfrak{S}}}}}}}\!\!\!\alpha'\widetilde{\otimes}\beta')
\,,\,(\alpha\widetilde{\otimes}\beta)^\star\right\}\right)=\{\alpha'\widetilde{\otimes}\beta'\}\;\in \; {\mathcal{K}}({{{ \overline{\mathfrak{S}}}}})$$
(3B) $\alpha''\sqsupseteq_{{}_{\overline{ \mathfrak{S}}_{A}}} \alpha^\star$ or $\beta''\sqsupseteq_{{}_{\overline{ \mathfrak{S}}_{A}}} \beta^\star$. Then, $\alpha''\widetilde{\otimes}\beta''$ is the supremum in ${{{ \overline{\mathfrak{S}}}}}$ of $(\alpha''\widetilde{\otimes}\beta''\sqcap_{{}_{{{{ \overline{\mathfrak{S}}}}}}}\!\!\!\alpha'\widetilde{\otimes}\beta')$ and $(\alpha^\star \widetilde{\otimes}\bot_{{}_{\overline{ \mathfrak{S}}_{B}}})\sqcap_{{}_{{{{ \overline{\mathfrak{S}}}}}}} (\bot_{{}_{\overline{ \mathfrak{S}}_{A}}}\widetilde{\otimes} \beta^\star)$. And then, $$cl_c^{{{{ \overline{\mathfrak{S}}}}}}\left(
\left\{(\alpha''\widetilde{\otimes}\beta''\sqcap_{{}_{{{{ \overline{\mathfrak{S}}}}}}}\!\!\!\alpha'\widetilde{\otimes}\beta')
\,,\,(\alpha\widetilde{\otimes}\beta)^\star\right\}\right)=\{\alpha''\widetilde{\otimes}\beta''\}\;\in \; {\mathcal{K}}({{{ \overline{\mathfrak{S}}}}})$$
(3C) $\alpha'\not\sqsupseteq_{{}_{\overline{ \mathfrak{S}}_{A}}} \alpha^\star$ and $\beta'\not\sqsupseteq_{{}_{\overline{ \mathfrak{S}}_{A}}} \beta^\star$ and $\alpha''\sqsupseteq_{{}_{\overline{ \mathfrak{S}}_{A}}} \alpha^\star$ or $\beta''\sqsupseteq_{{}_{\overline{ \mathfrak{S}}_{A}}} \beta^\star$. Then, we note that $(\alpha''\widetilde{\otimes}\beta'')\sqcap_{{}_{{{{ \overline{\mathfrak{S}}}}}}}\!\!\!(\alpha'\widetilde{\otimes}\beta')$ and $(\alpha^\star \widetilde{\otimes}\bot_{{}_{\overline{ \mathfrak{S}}_{B}}})\sqcap_{{}_{{{{ \overline{\mathfrak{S}}}}}}} (\bot_{{}_{\overline{ \mathfrak{S}}_{A}}}\widetilde{\otimes} \beta^\star)$ have no common upper-bound in ${{{ \overline{\mathfrak{S}}}}}$. We have then to note that $(\alpha'' \widetilde{\otimes}\bot_{{}_{\overline{ \mathfrak{S}}_{B}}})\sqcap_{{}_{{{{ \overline{\mathfrak{S}}}}}}} (\bot_{{}_{\overline{ \mathfrak{S}}_{A}}}\widetilde{\otimes} \beta')$ and $(\alpha^\star \widetilde{\otimes}\bot_{{}_{\overline{ \mathfrak{S}}_{B}}})\sqcap_{{}_{{{{ \overline{S}}}}_{AB}}} (\bot_{{}_{\overline{ \mathfrak{S}}_{A}}}\widetilde{\otimes} \beta^\star)$ admit a supremum in ${{{ \overline{\mathfrak{S}}}}}$ which is $(\alpha'' \widetilde{\otimes}\beta^\star)\sqcap_{{}_{{{{ \overline{\mathfrak{S}}}}}}} (\alpha^\star\widetilde{\otimes} \beta')$. We note also that $(\alpha' \widetilde{\otimes}\bot_{{}_{\overline{ \mathfrak{S}}_{B}}})\sqcap_{{}_{{{{ \overline{\mathfrak{S}}}}}}} (\bot_{{}_{\overline{ \mathfrak{S}}_{A}}}\widetilde{\otimes} \beta'')$ and $(\alpha^\star \widetilde{\otimes}\bot_{{}_{\overline{ \mathfrak{S}}_{B}}})\sqcap_{{}_{{{{ \overline{\mathfrak{S}}}}}}} (\bot_{{}_{\overline{ \mathfrak{S}}_{A}}}\widetilde{\otimes} \beta^\star)$ admit a supremum in ${{{ \overline{\mathfrak{S}}}}}$ which is $(\alpha' \widetilde{\otimes}\beta^\star)\sqcap_{{}_{{{{ \overline{\mathfrak{S}}}}}}} (\alpha^\star\widetilde{\otimes} \beta'')$. Hence, we check directly that 
\begin{eqnarray}
&&{ \mathfrak{c}}^{{{{ \overline{\mathfrak{S}}}}}}\left(
\left\{(\alpha''\widetilde{\otimes}\beta'')\sqcap_{{}_{{{{ \overline{\mathfrak{S}}}}}}}\!\!\!(\alpha'\widetilde{\otimes}\beta')
\,,\,(\alpha\widetilde{\otimes}\beta)^\star\right\}\right) =:U \nonumber\\
&&U= \left\{ (\alpha''\widetilde{\otimes}\beta'')\sqcap_{{}_{{{{ \overline{\mathfrak{S}}}}}}}\!\!\!(\alpha'\widetilde{\otimes}\beta')
\,,\,
(\alpha'' \widetilde{\otimes}\beta^\star)\sqcap_{{}_{{{{ \overline{\mathfrak{S}}}}}}} (\alpha^\star\widetilde{\otimes} \beta')\,,\,
(\alpha' \widetilde{\otimes}\beta^\star)\sqcap_{{}_{{{{ \overline{\mathfrak{S}}}}}}} (\alpha^\star\widetilde{\otimes} \beta'')
\right\}\;\;\;\;\;\;\;\;\;\;\;\;\;\;\;
\end{eqnarray}
It is then rather easy to check that ${ \mathfrak{c}}^{{{{ \overline{\mathfrak{S}}}}}}(U)=U$. Hence, $cl_c^{{{{ \overline{\mathfrak{S}}}}}}(U)=U$. 
Endly, we check directly that $U\in {\mathcal{K}}({{{ \overline{\mathfrak{S}}}}})$ (note that this check uses explicitly the properties  $\alpha'\not\sqsupseteq_{{}_{\overline{ \mathfrak{S}}_{A}}} \alpha^\star$ and $\beta'\not\sqsupseteq_{{}_{\overline{ \mathfrak{S}}_{A}}} \beta^\star$  and $\alpha''\sqsupseteq_{{}_{\overline{ \mathfrak{S}}_{A}}} \alpha^\star$ or $\beta''\sqsupseteq_{{}_{\overline{ \mathfrak{S}}_{A}}} \beta^\star$ assumed in case 3C). As a result
\begin{eqnarray}
&&cl_c^{{{{ \overline{\mathfrak{S}}}}}}\left(
\left\{(\alpha''\widetilde{\otimes}\beta'')\sqcap_{{}_{{{{ \overline{S}}}}_{AB}}}\!\!\!(\alpha'\widetilde{\otimes}\beta')
\,,\,(\alpha\widetilde{\otimes}\beta)^\star\right\}\right) =:U \nonumber\\
&&U= \left\{ (\alpha''\widetilde{\otimes}\beta'')\sqcap_{{}_{{{{ \overline{\mathfrak{S}}}}}}}\!\!\!(\alpha'\widetilde{\otimes}\beta')
\,,\,
(\alpha'' \widetilde{\otimes}\beta^\star)\sqcap_{{}_{{{{ \overline{\mathfrak{S}}}}}}} (\alpha^\star\widetilde{\otimes} \beta')\,,\,
(\alpha' \widetilde{\otimes}\beta^\star)\sqcap_{{}_{{{{ \overline{\mathfrak{S}}}}}}} (\alpha^\star\widetilde{\otimes} \beta'')
\right\}\;\;\;\;\;\;\;\;\;\;\;\;\;\;\;\label{computationexample}
\end{eqnarray}

As a conclusion of this case analysis we have proved the property (\ref{thirdcoveringpropertySbari}). By the way we can check directly on the expressions of $cl_c^{{{{ \overline{\mathfrak{S}}}}}}\left(
\left\{(\alpha''\widetilde{\otimes}\beta'')\sqcap_{{}_{{{{ \overline{\mathfrak{S}}}}}}}\!\!\!(\alpha'\widetilde{\otimes}\beta')
\,,\,(\alpha\widetilde{\otimes}\beta)^\star\right\}\right)$ established along this case analysis the properties (\ref{thirdcoveringpropertySbarii})(\ref{thirdcoveringpropertySbariii})(\ref{thirdcoveringpropertySbariv}) and also the properties (\ref{thirdcoveringpropertySbarviic})(\ref{thirdcoveringpropertySbarviib}) if we restrict ourselves to the case $\alpha=\alpha''$ and $\beta''=\beta$.\\

Let us summarize. If $\chi:=(\alpha\widetilde{\otimes}\beta)^\star \sqcup_{{}_{{{{ \widetilde{\mathfrak{S}}}}}}}
(\alpha\widetilde{\otimes}\beta\sqcap_{{}_{{{{ \overline{\mathfrak{S}}}}}}}\!\!\!\alpha'\widetilde{\otimes}\beta')$ is in ${ \mathfrak{S}}\smallsetminus \overline{ \mathfrak{S}}$, we have
\begin{eqnarray}
\Theta^{{{{ \overline{\mathfrak{S}}}}}}\left(\chi\right)&=&\{\gamma,\delta, \epsilon\}\\
\gamma=(\alpha\widetilde{\otimes}\beta\sqcap_{{}_{{{{ \overline{\mathfrak{S}}}}}}}\!\!\!\alpha'\widetilde{\otimes}\beta'),&& \varphi_\gamma=\alpha\widetilde{\otimes}\beta,\;\;\;\;\psi_\gamma=\alpha'\widetilde{\otimes}\beta',\nonumber\\
\delta=(\alpha\widetilde{\otimes}\beta^\star\sqcap_{{}_{{{{ \overline{\mathfrak{S}}}}}}}\!\!\!\alpha^\star\widetilde{\otimes}\beta'),&& \varphi_\delta=\alpha\widetilde{\otimes}\beta^\star,\;\;\;\;\psi_\delta=\alpha^\star\widetilde{\otimes}\beta',\nonumber\\
\epsilon=(\alpha'\widetilde{\otimes}\beta^\star\sqcap_{{}_{{{{ \overline{\mathfrak{S}}}}}}}\!\!\!\alpha^\star\widetilde{\otimes}\beta),&& \varphi_\epsilon=\alpha^\star\widetilde{\otimes}\beta,\;\;\;\;\psi_\epsilon=\alpha'\widetilde{\otimes}\beta^\star,\nonumber\\
\lambda_{\delta\epsilon}=\alpha^\star\widetilde{\otimes}\beta^\star,&& \mu_{\gamma\delta}=\alpha\widetilde{\otimes}\beta',\;\;\;\;\mu_{\gamma\epsilon}=\alpha'\widetilde{\otimes}\beta.\nonumber
\end{eqnarray}
The collections $(\varphi_{\epsilon})_{\epsilon},(\psi_{\epsilon})_{\epsilon} \in (\overline{ \mathfrak{S}}{}^{{}^{pure}})^{\Theta^{\overline{ \mathfrak{S}}}(\chi)}$ given above satisfy the properties (\ref{thirdcoveringpropertySbar0})(\ref{thirdcoveringpropertySbarviic})(\ref{thirdcoveringpropertySbarviib}). They are {unique} up to a permutation on the set $\{\delta,\epsilon\}$ because of the following basic property
\begin{eqnarray}
\forall \alpha\in \{\gamma,\delta, \epsilon\},&& \underline{\,\alpha\,}_{{{{{ \overline{\mathfrak{S}}}}}}}=\{\varphi_\alpha,\psi_\alpha\}
\end{eqnarray}
inherited from (\ref{developmentetildeordersimplify}). \\ 
The properties (\ref{thirdcoveringpropertySbarviia})(\ref{thirdcoveringpropertySbarviiabis}) can be checked directly.\\

Concerning the property (\ref{thirdcoveringpropertySbarvi}), it suffice to compute
\begin{eqnarray}
&&cl_c^{{{{ \overline{\mathfrak{S}}}}}}\left(
\left\{\left((\alpha\widetilde{\otimes}\beta)\sqcap_{{}_{{{{ \overline{\mathfrak{S}}}}}}}\!\!\!(\alpha'\widetilde{\otimes}\beta')\right)^\star
\,,\, 
(\alpha \widetilde{\otimes}\beta^\star)\sqcap_{{}_{{{{ \overline{\mathfrak{S}}}}}}} (\alpha^\star\widetilde{\otimes} \beta')
\right\}\right) = \nonumber\\
&&= \left\{ (\alpha^\star\widetilde{\otimes}\beta'{}^\star)\sqcap_{{}_{{{{ \overline{\mathfrak{S}}}}}}}\!\!\!(\alpha'{}^\star\widetilde{\otimes}\beta^\star)
\,,\,
(\alpha \widetilde{\otimes}\beta^\star)\sqcap_{{}_{{{{ \overline{\mathfrak{S}}}}}}} (\alpha^\star\widetilde{\otimes} \beta')\,,\,
(\alpha \widetilde{\otimes}\beta'{}^\star)\sqcap_{{}_{{{{ \overline{\mathfrak{S}}}}}}} (\alpha'{}^\star\widetilde{\otimes} \beta')
\right\}\;\;\;\;\;\;\;\;\;\;\;\;\;\;\;
\end{eqnarray}
and to observe then that
\begin{eqnarray}
&&cl_c^{{{{ \overline{\mathfrak{S}}}}}}\left(
\left\{\left((\alpha\widetilde{\otimes}\beta)\sqcap_{{}_{{{{ \overline{\mathfrak{S}}}}}}}\!\!\!(\alpha'\widetilde{\otimes}\beta')\right)^\star
\,,\, 
(\alpha \widetilde{\otimes}\beta^\star)\sqcap_{{}_{{{{ \overline{\mathfrak{S}}}}}}} (\alpha^\star\widetilde{\otimes} \beta')
\right\}\right) = \nonumber\\
&&cl_c^{{{{ \overline{\mathfrak{S}}}}}}\left(\left\{ 
(\alpha^\star\widetilde{\otimes}\beta'{}^\star)^\star
\,,\, 
(\alpha^\star\widetilde{\otimes}\beta'{}^\star)\sqcap_{{}_{{{{ \overline{\mathfrak{S}}}}}}}\!\!\!(\alpha'{}^\star\widetilde{\otimes}\beta^\star)
\right\}\right)
\end{eqnarray}

\end{proof}

Now, we consider $N$ copies of one-dimensional indeterministic spaces of states ${{{ \overline{\mathfrak{S}}}}}_{A_i}={\mathfrak{Z}}'_{M_{A_i}}$ for $i=1,\cdots,N$ ($M_{A_i}\in \mathbb{N}$) and we now consider the real structure defined by ${{{ \overline{S}}}}_{A_1\cdots A_N}:={{{ \overline{\mathfrak{S}}}}}_{A_1}\widetilde{\otimes}\cdots \widetilde{\otimes}{{{ \overline{\mathfrak{S}}}}}_{A_N}$. We denote ${ \mathfrak{S}}_{A_i}={ \mathfrak{J}}^c_{{\overline{ \mathfrak{S}}_{A_i}}}$ for $i=1,\cdots,N$. And we introduce, as before, the tensor product ${{{ \hat{S}}}}_{A_1\cdots A_N}:={ \mathfrak{J}}^c_{\overline{{ \mathfrak{S}}}_{A_1}\widehat{\otimes} \cdots \widehat{\otimes} \overline{{ \mathfrak{S}}}_{A_N}}$. This space of states will be said to be the $N-$dimensional indeterministic space of states. If necessary, in order to shorten our formulas, we will denote ${{{ \overline{S}}}}{}_{A_1\cdots A_N}$ simply by $\overline{\mathfrak{S}}$ and ${{{\hat{S}}}}{}_{A_1\cdots A_N}$ simply by ${\mathfrak{S}}$.\\

For simplicity, we will restrict ourselves to the case $N=3$ in order to clarify a central aspect of the iteration of the tensor product.

We will then consider $\sigma_{A_1},\sigma_{A_1}',\sigma_{A_1}''\in {{{ \overline{\mathfrak{S}}}}}_{A_1}$,  $\sigma_{A_2},\sigma_{A_2}',\sigma_{A_2}''\in {{{ \overline{\mathfrak{S}}}}}_{A_2}$, $\sigma_{A_3},\sigma_{A_3}',\sigma_{A_3}''\in {{{ \overline{\mathfrak{S}}}}}_{A_3}$. We will suppose that $\sigma_{A_1}'\not=\sigma_{A_1}''$ and $\sigma_{A_1}^\star\not=\sigma'_{A_1},\sigma''_{A_1}$, and $\sigma_{A_2}'\not=\sigma_{A_2}''$ and $\sigma_{A_2}^\star\not=\sigma'_{A_2},\sigma_{A_2}''$,  and $\sigma_{A_3}'\not=\sigma_{A_3}''$ and $\sigma_{A_3}^\star\not=\sigma'_{A_3},\sigma''_{A_3}$. 

From the following properties
\begin{eqnarray}
&&\hspace{-1.5cm}(\sigma_{A_1}\widetilde{\otimes}\sigma_{A_2}\widetilde{\otimes}\sigma_{A_3})^\star \sqcup_{{}_{{{{ \overline{S}}}}_{A_1A_2A_3}}} (\bot_{{}_{{{{ \overline{\mathfrak{S}}}}}_{A_1}}}\widetilde{\otimes}\sigma_{A_2}''\widetilde{\otimes}\sigma_{A_3}'' \sqcap_{{}_{{{{ \overline{S}}}}_{A_1A_2A_3}}} \sigma'_{A_1}\widetilde{\otimes}\sigma'_{A_2}\widetilde{\otimes}\sigma'_{A_3})=\sigma_{A_1}^\star\widetilde{\otimes}\sigma_{A_2}''\widetilde{\otimes}\sigma_{A_3}''\;\;\;\;\;\;\;\;\;\;\;\;\;\;\;\;\;\;\;\;\\
&&\hspace{-1.5cm}(\sigma_{A_1}\widetilde{\otimes}\sigma_{A_2}\widetilde{\otimes}\sigma_{A_3})^\star \sqcup_{{}_{{{{ \overline{S}}}}_{A_1A_2A_3}}} (\sigma_{A_1}''\widetilde{\otimes}\sigma_{A_2}''\widetilde{\otimes}\sigma_{A_3}'' \sqcap_{{}_{{{{ \overline{S}}}}_{A_1A_2A_3}}} \bot_{{}_{{{{ \overline{\mathfrak{S}}}}}_{A_1}}}\widetilde{\otimes}\sigma'_{A_2}\widetilde{\otimes}\sigma'_{A_3})=\sigma_{A_1}^\star\widetilde{\otimes}\sigma'_{A_2}\widetilde{\otimes}\sigma'_{A_3}\;\;\;\;\;\;\;\;\;\;\;\;\;\;\;\;\;\;\;\;
\end{eqnarray}
we deduce
\begin{eqnarray}
&&\hspace{-1.2cm}\left\{ \sigma_{A_1}^\star\widetilde{\otimes}\sigma_{A_2}''\widetilde{\otimes}\sigma_{A_3}''\,,\, \sigma_{A_1}^\star\widetilde{\otimes}\sigma'_{A_2}\widetilde{\otimes}\sigma'_{A_3}\right\} \sqsubseteq \nonumber\\
&&{\mathfrak{c}}^{{{{ \overline{S}}}}_{A_1A_2A_3}}(\left\{(\sigma_{A_1}\widetilde{\otimes}\sigma_{A_2}\widetilde{\otimes}\sigma_{A_3})^\star\,,\, (\sigma_{A_1}''\widetilde{\otimes}\sigma_{A_2}''\widetilde{\otimes}\sigma_{A_3}'' \sqcap_{{}_{{{{ \overline{S}}}}_{A_1A_2A_3}}} \sigma'_{A_1}\widetilde{\otimes}\sigma'_{A_2}\widetilde{\otimes}\sigma'_{A_3})\right\})\;\;\;\;\;\;\;\;\;\;\;\;\;\;
\end{eqnarray}
Now we note that 
$\alpha:=\sigma_{A_1}^\star\widetilde{\otimes}\sigma_{A_2}''\widetilde{\otimes}\sigma_{A_3}'',$ and $\beta:=\sigma_{A_1}^\star\widetilde{\otimes}\sigma'_{A_2}\widetilde{\otimes}\sigma'_{A_3}$ are two distinct elements of ${{{ \overline{S}}}}{}_{A_1A_2A_3}^{{}^{pure}}$ and $(\alpha\sqcap_{{}_{{{{ \overline{S}}}}{}_{A_1A_2A_3}}}\beta)\not= \bot_{{}_{{{{{ \overline{S}}}}{}_{A_1A_2A_3}}}}$ and that $\forall \gamma,\gamma'\in {{{ \overline{S}}}}{}_{A_1A_2A_3}^{{}^{pure}}$ we have $(\gamma\sqcap_{{}_{{{{{{ \overline{S}}}}{}_{A_1A_2A_3}}}}}\gamma')\sqcoversubset_{{}_{{{{{{ \overline{S}}}}{}_{A_1A_2A_3}}}}}\gamma,\gamma'$. In other words, the conditions of Theorem \ref{theoremNOcompleteontic} are satisfied. Then, 
$cl^{{{{{ \overline{S}}}}{}_{A_1A_2A_3}}}(\{\alpha,\beta\})\notin {\mathcal{K}}({{{{ \overline{S}}}}{}_{A_1A_2A_3}})$.\\
As a result, we have obtained 
\begin{eqnarray}
\hspace{-1cm}cl_c^{{{{ \overline{S}}}}_{A_1A_2A_3}}(\left\{(\sigma_{A_1}\widetilde{\otimes}\sigma_{A_2}\widetilde{\otimes}\sigma_{A_3})^\star\,,\, (\sigma_{A_1}''\widetilde{\otimes}\sigma_{A_2}''\widetilde{\otimes}\sigma_{A_3}'' \sqcap_{{}_{{{{ \overline{S}}}}_{A_1A_2A_3}}} \sigma'_{A_1}\widetilde{\otimes}\sigma'_{A_2}\widetilde{\otimes}\sigma'_{A_3})\right\})&\notin & {\mathcal{K}}({{{{ \overline{S}}}}{}_{A_1A_2A_3}})\;\;\;\;\;\;\;\;\;\;\;\;\;\;
\end{eqnarray}
In other words,  as long as $\sigma_{A_1}'\not=\sigma_{A_1}''$ and $\sigma_{A_1}^\star\not=\sigma'_{A_1},\sigma''_{A_1}$,  and $\sigma_{A_2}'\not=\sigma_{A_2}''$ and $\sigma_{A_2}^\star\not=\sigma'_{A_2},\sigma_{A_2}''$,  and $\sigma_{A_3}'\not=\sigma_{A_3}''$  and $\sigma_{A_3}^\star\not=\sigma'_{A_3},\sigma''_{A_3}$, we have
\begin{eqnarray}
\hspace{-1.5cm}\nexists  \alpha\in {{{{ \widetilde{S}}}}{}_{A_1A_2A_3}} &\vert &(\,(\sigma_{A_1}''\widetilde{\otimes}\sigma_{A_2}''\widetilde{\otimes}\sigma_{A_3}'' \sqcap_{{}_{{{{ \overline{S}}}}_{A_1A_2A_3}}} \sigma'_{A_1}\widetilde{\otimes}\sigma'_{A_2}\widetilde{\otimes}\sigma'_{A_3}) \sqsubseteq_{{}_{{{{ \widetilde{S}}}}_{A_1A_2A_3}}}\!\! \alpha\;\;\textit{\rm and}\;\; \sigma_{A_1}\widetilde{\otimes}\sigma_{A_2}\widetilde{\otimes}\sigma_{A_3} \;\underline{\perp}\; \alpha\,).\;\;\;\;\;\;\;\;\;\;\;\;
\end{eqnarray}

If, on the contrary, the condition ($\sigma_{A_1}'\not=\sigma_{A_1}''$, $\sigma_{A_2}'\not=\sigma_{A_2}''$, $\sigma_{A_3}'\not=\sigma_{A_3}''$) is not satisfied anymore (we maintain nevertheless the conditions $\sigma_{A_1}^\star\not=\sigma'_{A_1},\sigma''_{A_1}$, $\sigma_{A_2}^\star\not=\sigma'_{A_2},\sigma_{A_2}''$, $\sigma_{A_3}^\star\not=\sigma'_{A_3},\sigma''_{A_3}$ because, if this condition is not satisfied, the result is easy to compute), we have then first of all
\begin{eqnarray}
\hspace{-1.5cm}\exists  \alpha\in {{{{ \widetilde{S}}}}{}_{A_1A_2A_3}} &\vert &(\,(\sigma_{A_1}'\widetilde{\otimes}\sigma_{A_2}''\widetilde{\otimes}\sigma_{A_3}'' \sqcap_{{}_{{{{ \overline{S}}}}_{A_1A_2A_3}}} \sigma'_{A_1}\widetilde{\otimes}\sigma'_{A_2}\widetilde{\otimes}\sigma'_{A_3}) \sqsubseteq_{{}_{{{{ \widetilde{S}}}}_{A_1A_2A_3}}}\!\! \alpha\;\;\textit{\rm and}\;\; \sigma_{A_1}\widetilde{\otimes}\sigma_{A_2}\widetilde{\otimes}\sigma_{A_3} \;\underline{\perp}\; \alpha\,).\;\;\;\;\;\;\;\;\;\;\;\;\nonumber\\
&&\alpha=\sigma'_{A_1} \widetilde{\otimes} \left((\sigma_{A_2}''\widetilde{\otimes}\sigma_{A_3}''\sqcap_{{}_{{{{ \overline{S}}}}_{A_2A_3}}} \sigma'_{A_2}\widetilde{\otimes}\sigma'_{A_3})\sqcup_{{}_{{{{ \hat{S}}}}_{A_2A_3}}}(\sigma_{A_2}\widetilde{\otimes}\sigma_{A_3})^\star \right)
\end{eqnarray}
i.e.
\begin{eqnarray}
&&\hspace{-1.8cm}\Theta^{{{{ \overline{S}}}}_{A_1A_2A_3}}(\alpha)=\left\{ (\sigma_{A_1}'\widetilde{\otimes}\sigma_{A_2}''\widetilde{\otimes}\sigma_{A_3}'' \sqcap_{{}_{{{{ \overline{S}}}}_{A_1A_2A_3}}} \sigma_{A_1}'\widetilde{\otimes}\sigma'_{A_2}\widetilde{\otimes}\sigma'_{A_3})\,,\, 
(\sigma_{A_1}'\widetilde{\otimes}\sigma_{A_2}^\star\widetilde{\otimes}\sigma_{A_3}'' \sqcap_{{}_{{{{ \overline{S}}}}_{A_1A_2A_3}}} \sigma_{A_1}'\widetilde{\otimes}\sigma'_{A_2}\widetilde{\otimes}\sigma_{A_3}^\star)\,,\,
\right.\nonumber\\
&&\hspace{3cm}\left. (\sigma_{A_1}'\widetilde{\otimes}\sigma_{A_2}''\widetilde{\otimes}\sigma_{A_3}^\star \sqcap_{{}_{{{{ \overline{S}}}}_{A_1A_2A_3}}} \sigma_{A_1}'\widetilde{\otimes}\sigma_{A_2}^\star\widetilde{\otimes}\sigma'_{A_3})\right\}\label{expressionchiN=3}
\end{eqnarray}
and we have also
\begin{eqnarray}
\hspace{-1.5cm}\exists  \alpha\in {{{{ \widetilde{S}}}}{}_{A_1A_2A_3}} &\vert &(\,(\sigma_{A_1}'\widetilde{\otimes}\sigma_{A_2}'\widetilde{\otimes}\sigma_{A_3}'' \sqcap_{{}_{{{{ \overline{S}}}}_{A_1A_2A_3}}} \sigma'_{A_1}\widetilde{\otimes}\sigma'_{A_2}\widetilde{\otimes}\sigma'_{A_3}) \sqsubseteq_{{}_{{{{ \widetilde{S}}}}_{A_1A_2A_3}}}\!\! \alpha\;\;\textit{\rm and}\;\; \sigma_{A_1}\widetilde{\otimes}\sigma_{A_2}\widetilde{\otimes}\sigma_{A_3} \;\underline{\perp}\; \alpha\,).\;\;\;\;\;\;\;\;\;\;\;\;\nonumber\\
&&\alpha=\sigma'_{A_1} \widetilde{\otimes} \sigma_{A_2}'\widetilde{\otimes} \left((\sigma_{A_3}''\sqcap_{{}_{{{{ \overline{S}}}}_{A_3}}} \sigma'_{A_3})\sqcup_{{}_{{{{ \hat{S}}}}_{A_3}}}\sigma_{A_3}^\star \right)
\end{eqnarray}
i.e.
\begin{eqnarray}
&&
\hspace{-1.8cm}\Theta^{{{{ \overline{S}}}}_{A_1A_2A_3}}(\alpha)=\{\sigma_{A_1}' \widetilde{\otimes} \sigma_{A_2}'\widetilde{\otimes}\sigma_{A_3}^\star\}.\;\;\;\;\;\;\;\;\;\;\;\;\;\;\;\;
\end{eqnarray}
The results obtained by flipping the tensors and exchanging the corresponding indices are obvious.\\

Let us summarize the previous result, but before that, let us introduce some notations.

\begin{definition}\label{definwr}
Let us consider $\nu$  and $\phi$ two  elements of ${{{ \overline{\mathfrak{S}}}}}{}^{{}^{pure}}$. We will define the following reflexive and symmetric binary relation
\begin{eqnarray}
\hspace{-1cm}\nu \wr \phi\;\;&:\Leftrightarrow & (\,\exists I\subseteq\{1,\cdots,N\}, Card(I)\geq N-2\;\vert\; \forall i\in I, \zeta^{{}^{{\mathfrak{S}}_{A_1}\cdots{\mathfrak{S}}_{A_N}}}_{(i)}(\nu) = \zeta^{{}^{{\mathfrak{S}}_{A_1}\cdots{\mathfrak{S}}_{A_N}}}_{(i)}(\phi)\,).\;\;\;\;\;\;\;\;\;\;\;\;
\end{eqnarray}
\end{definition}

\begin{definition}\label{definUpsilonDelta}
Let us fix $\nu\in {{{ \overline{\mathfrak{S}}}}}{}^{{}^{pure}}$ and $j,k\in \{1,\cdots,N\}$ with $j<k$. 
Let us denote, for any $\kappa\in {{{ \overline{\mathfrak{S}}}}}{}^{{}^{pure}}$, 
\begin{eqnarray}
\Upsilon^{(j,k)}_{\kappa}:=\zeta^{{}^{{\mathfrak{S}}_{A_1}\cdots{\mathfrak{S}}_{A_N}}}_{(j)}(\kappa)\widetilde{\otimes}\zeta^{{}^{{\mathfrak{S}}_{A_1}\cdots{\mathfrak{S}}_{A_N}}}_{(k)}(\kappa)\in \overline{\mathfrak{S}}{}_{A_j}^{{}^{pure}}\widetilde{\otimes}\overline{\mathfrak{S}}{}_{A_k}^{{}^{pure}}
\end{eqnarray}
and, for any $(\alpha\widetilde{\otimes}\beta)\in \overline{\mathfrak{S}}{}_{A_j}^{{}^{pure}}\widetilde{\otimes}\overline{\mathfrak{S}}{}_{A_k}^{{}^{pure}}$,
\begin{eqnarray}
\Delta^{(j,k);\nu}_{\alpha\widetilde{\otimes}\beta}:=\delta_1\widetilde{\otimes}\cdots\widetilde{\otimes}\delta_{j-1}\widetilde{\otimes}\alpha\widetilde{\otimes}\delta_{j+1}\widetilde{\otimes}\cdots\widetilde{\otimes}\delta_{k-1}\widetilde{\otimes}\beta \widetilde{\otimes} \delta_{k+1}\widetilde{\otimes}\cdots\widetilde{\otimes}\delta_{N}
\end{eqnarray}
with $\delta_m:=\zeta^{{}^{{\mathfrak{S}}_{A_1}\cdots{\mathfrak{S}}_{A_N}}}_{(m)}(\nu)$ for any $m\in \{1,\cdots,N\}$. 
\end{definition}

\begin{theorem}\label{lambdawr}
Let us consider $\nu$   
and $\phi$ 
 two distinct elements of ${{{ \overline{\mathfrak{S}}}}}{}^{{}^{pure}}$ and $\mu$ 
 another element of ${{{ \overline{\mathfrak{S}}}}}{}^{{}^{pure}}$.  Let us assume that $\mu \not\!\!\!\underline{\perp}\nu,\phi$. 
 We have
\begin{eqnarray}
&&\hspace{-2cm}(\,\exists \chi\in ({{{ \widetilde{\mathfrak{S}}}}}\smallsetminus \overline{\mathfrak{S}}) \;\vert\; (\,(\nu\sqcap_{{}_{{{{ \overline{\mathfrak{S}}}}}}}\phi) \sqcoversubset_{{}_{{{{ \widetilde{\mathfrak{S}}}}}}}\chi \;\;\;\;\textit{\rm and}\;\;\;\; \chi\,\underline{\perp}\,\mu\,)\,)\;\;\Leftrightarrow\\
&&\hspace{-2cm}(\,\exists I\subseteq\{1,\cdots,N\}, Card(I)= N-2\;\vert\; \nonumber\\
&&\hspace{-1.5cm}\forall i\in I, \forall \sigma,\kappa\in \overline{\mathfrak{S}}{}^{{}^{pure}},(\,(\sigma \sqcap_{{}_{{{{ \overline{\mathfrak{S}}}}}}}\chi) \sqcoversubset_{{}_{{{{ \widetilde{\mathfrak{S}}}}}}}\sigma,\chi,\;\;(\kappa \sqcap_{{}_{{{{ \overline{\mathfrak{S}}}}}}}\chi) \sqcoversubset_{{}_{{{{ \widetilde{\mathfrak{S}}}}}}}\kappa,\chi\,)\;\;\Rightarrow\;\; \zeta^{{}^{{\mathfrak{S}}_{A_1}\cdots{\mathfrak{S}}_{A_N}}}_{(i)}(\sigma) = \zeta^{{}^{{\mathfrak{S}}_{A_1}\cdots{\mathfrak{S}}_{A_N}}}_{(i)}(\kappa)\,). \nonumber
\end{eqnarray}
In particular, $\exists \chi\in {{{ \widetilde{\mathfrak{S}}}}}\;\vert\; (\,(\nu\sqcap_{{}_{{{{ \overline{\mathfrak{S}}}}}}}\phi) \sqcoversubset_{{}_{{{{ \widetilde{\mathfrak{S}}}}}}}\chi \;\;\;\;\textit{\rm and}\;\;\;\; \chi\,\underline{\perp}\,\mu\,)\,)$ implies $\forall i\in I,\;\; \zeta^{{}^{{\mathfrak{S}}_{A_1}\cdots{\mathfrak{S}}_{A_N}}}_{(i)}(\nu) = \zeta^{{}^{{\mathfrak{S}}_{A_1}\cdots{\mathfrak{S}}_{A_N}}}_{(i)}(\phi)$.\\
 
Let us denote $j,k\in \{1,\cdots,N\}$ with $j<k$ such that for any $i\in \{1,\cdots,N\}\smallsetminus \{j,k\}$ we have 
$\zeta^{{}^{{\mathfrak{S}}_{A_1}\cdots{\mathfrak{S}}_{A_N}}}_{(i)}(\nu) = \zeta^{{}^{{\mathfrak{S}}_{A_1}\cdots{\mathfrak{S}}_{A_N}}}_{(i)}(\phi)$.  Endly, let us denote $\alpha:=\zeta^{{}^{{\mathfrak{S}}_{A_1}\cdots{\mathfrak{S}}_{A_N}}}_{(j)}(\mu)$, $\beta:=\zeta^{{}^{{\mathfrak{S}}_{A_1}\cdots{\mathfrak{S}}_{A_N}}}_{(k)}(\mu)$, 
$\alpha':=\zeta^{{}^{{\mathfrak{S}}_{A_1}\cdots{\mathfrak{S}}_{A_N}}}_{(j)}(\nu)$, $\beta':=\zeta^{{}^{{\mathfrak{S}}_{A_1}\cdots{\mathfrak{S}}_{A_N}}}_{(k)}(\nu)$, 
$\alpha'':=\zeta^{{}^{{\mathfrak{S}}_{A_1}\cdots{\mathfrak{S}}_{A_N}}}_{(j)}(\phi)$, $\beta'':=\zeta^{{}^{{\mathfrak{S}}_{A_1}\cdots{\mathfrak{S}}_{A_N}}}_{(k)}(\phi)$.\\
As long as $\lambda:=\mu\sqcup_{{}_{\mathfrak{S}}}(\nu\sqcap_{{}_{\overline{\mathfrak{S}}}}\phi)$ exists in ${{\mathfrak{S}}}\smallsetminus {\overline{\mathfrak{S}}}$, we have (see (\ref{computationexample}))
\begin{eqnarray}
\Theta^{\overline{\mathfrak{S}}}(\lambda) &=& \left\{ \Delta^{(j,k);\nu}_{\alpha''\widetilde{\otimes}\beta''}\sqcap_{{}_{{{{ \overline{\mathfrak{S}}}}}}}\!\!\!\Delta^{(j,k);\nu}_{\alpha'\widetilde{\otimes}\beta'}
\,,\,
\Delta^{(j,k);\nu}_{\alpha'' \widetilde{\otimes}\beta^\star}\sqcap_{{}_{{{{ \overline{\mathfrak{S}}}}}}} \Delta^{(j,k);\nu}_{\alpha^\star\widetilde{\otimes} \beta'}\,,\,
\Delta^{(j,k);\nu}_{\alpha' \widetilde{\otimes}\beta^\star}\sqcap_{{}_{{{{ \overline{\mathfrak{S}}}}}}} \Delta^{(j,k);\nu}_{\alpha^\star\widetilde{\otimes} \beta''}
\right\}\;\;\;\;\;\;\;\;\;\;\;\;\;\;\;
\end{eqnarray}
\end{theorem}
\begin{proof}
Once again we will fix $N=3$.
It suffices to remark, using the expression (\ref{expressionchiN=3}), the expression of the infimum (\ref{infimumJprop}), the property (\ref{thirdcoveringpropertySbarii}) and the expansion formula (\ref{developmentetildeordersimplify}), that the relations $(\sigma \sqcap_{{}_{{{{ \overline{\mathfrak{S}}}}}}}\chi) \sqcoversubset_{{}_{{{{ \widetilde{\mathfrak{S}}}}}}}\sigma,\chi$ and $(\kappa \sqcap_{{}_{{{{ \overline{\mathfrak{S}}}}}}}\chi) \sqcoversubset_{{}_{{{{ \widetilde{\mathfrak{S}}}}}}}\kappa,\chi$ impose necessarily $\sigma$ and $\kappa$ to be both elements of the following subset of pure states
\begin{eqnarray}
&&\hspace{-1cm}\{ (\sigma_{A_1}'\widetilde{\otimes}\sigma_{A_2}''\widetilde{\otimes}\sigma_{A_3}'')\,, \,(\sigma_{A_1}'\widetilde{\otimes}\sigma'_{A_2}\widetilde{\otimes}\sigma'_{A_3})\,,\,(\sigma_{A_1}'\widetilde{\otimes}\sigma_{A_2}^\star\widetilde{\otimes}\sigma_{A_3}'')\,,\, \nonumber\\
&&\hspace{2cm}(\sigma_{A_1}'\widetilde{\otimes}\sigma'_{A_2}\widetilde{\otimes}\sigma_{A_3}^\star)\,,\, (\sigma_{A_1}'\widetilde{\otimes}\sigma_{A_2}''\widetilde{\otimes}\sigma_{A_3}^\star)\,,\, (\sigma_{A_1}'\widetilde{\otimes}\sigma_{A_2}^\star\widetilde{\otimes}\sigma'_{A_3})\}\;\;\;\;\;\;\;\;\;\;\;
\end{eqnarray}
This concludes the proof, which can be straightforwardly generalized to generic $N$.
\end{proof}

We end this subsection by giving some basic properties that will be fruitful in the sequel.

\begin{theorem}
\begin{eqnarray}
\hspace{-0.5cm}\forall \alpha,\beta\in {{{ \overline{\mathfrak{S}}}}}{}^{{}^{pure}},&& \beta\not=\alpha \;\;\Rightarrow \;\; (\alpha\sqcap_{{}_{{{{ \overline{\mathfrak{S}}}}}}}\beta) \sqcoversubset_{{}_{{{{ \overline{\mathfrak{S}}}}}}}\alpha,\beta.\;\;\;\;\;\;\;\;\;\;\;\;\;\;\;\;\;\;\label{coveringpropertySbar}
\end{eqnarray}
\end{theorem}
\begin{proof}
The property (\ref{coveringpropertySbar}) is obtained by a trivial recursion proof (on the number $N$) using Theorem \ref{firstcoveringlemma}.
\end{proof}

\begin{theorem}
We have the following basic property
\begin{eqnarray}
&&\hspace{-0.5cm}\forall \lambda,\alpha,\beta,\gamma,\delta\in {{{ \overline{\mathfrak{S}}}}}{}^{{}^{pure}} \;\textit{\rm distinct},\nonumber\\
&& (\,(\alpha\sqcap_{{}_{{{{ \overline{\mathfrak{S}}}}}}}\beta)\not=(\gamma\sqcap_{{}_{{{{ \overline{\mathfrak{S}}}}}}}\delta)\;\;\;\textit{\rm and}\;\;\;(\alpha\sqcap_{{}_{{{{ \overline{\mathfrak{S}}}}}}}\beta),(\gamma\sqcap_{{}_{{{{ \overline{\mathfrak{S}}}}}}}\delta)\sqcoversubset_{{}_{\overline{\mathfrak{S}}}}\lambda \,) \;\;\Rightarrow \nonumber\\
&&\hspace{1cm} (\alpha\sqcap_{{}_{{{{ \overline{\mathfrak{S}}}}}}}\beta\sqcap_{{}_{{{{ \overline{\mathfrak{S}}}}}}}\gamma\sqcap_{{}_{{{{ \overline{\mathfrak{S}}}}}}}\delta) \sqcoversubset_{{}_{{{{ \overline{\mathfrak{S}}}}}}}(\alpha\sqcap_{{}_{{{{ \overline{\mathfrak{S}}}}}}}\beta),(\gamma\sqcap_{{}_{{{{ \overline{\mathfrak{S}}}}}}}\delta).\;\;\;\;\;\;\;\;\;\;\;\;\;\;\;\;\;\;\label{secondcoveringpropertySbar}
\end{eqnarray}
\end{theorem}
\begin{proof}
The property (\ref{secondcoveringpropertySbar}) is obtained by a trivial recursion proof (on the number $N$) using Theorem \ref{secondcoveringlemma}. 
\end{proof}

\begin{theorem}\label{zetalambda}
\begin{eqnarray}
&&\hspace{-2cm}\forall \lambda,\alpha,\beta,\gamma,\delta\in {{{ \overline{\mathfrak{S}}}}}{}^{{}^{pure}} \;\textit{\rm distinct},
(\,(\alpha\sqcap_{{}_{{{{ \overline{\mathfrak{S}}}}}}}\beta)\not=(\gamma\sqcap_{{}_{{{{ \overline{\mathfrak{S}}}}}}}\delta)\;\;\;\textit{\rm and}\;\;\;(\alpha\sqcap_{{}_{{{{ \overline{\mathfrak{S}}}}}}}\beta),(\gamma\sqcap_{{}_{{{{ \overline{\mathfrak{S}}}}}}}\delta)\sqcoversubset_{{}_{\overline{\mathfrak{S}}}}\lambda \,) \;\;\Rightarrow \label{secondcoveringpropertySbarbis}\\
&&\hspace{-2cm}\exists I\subseteq\{1,\cdots,N\}, Card(I)= N-2\;\vert\; \forall i\in I, \forall \sigma,\kappa\in \{\lambda,\alpha,\beta,\gamma,\delta\},\;\;\zeta^{{}^{{\mathfrak{S}}_{A_1}\cdots{\mathfrak{S}}_{A_N}}}_{(i)}(\sigma) = \zeta^{{}^{{\mathfrak{S}}_{A_1}\cdots{\mathfrak{S}}_{A_N}}}_{(i)}(\kappa).\;\;\;\;\;\;\nonumber
\end{eqnarray}
\end{theorem}
\begin{proof}
Direct consequence of the expansion formula (\ref{developmentetildeordersimplify}).  Let us consider $\lambda:=\sigma^{\lambda}_{A_1}\widetilde{\otimes}\cdots\widetilde{\otimes}\sigma^{\lambda}_{A_N}$, $\alpha:=\sigma^{\alpha}_{A_1}\widetilde{\otimes}\cdots\widetilde{\otimes}\sigma^{\alpha}_{A_N}$, $\beta:=\sigma^{\beta}_{A_1}\widetilde{\otimes}\cdots\widetilde{\otimes}\sigma^{\beta}_{A_N}$. The property $(\alpha\sqcap_{{}_{{{{ \overline{\mathfrak{S}}}}}}}\beta)\sqcoversubset_{{}_{\overline{\mathfrak{S}}}}\lambda$ imposes 
\begin{eqnarray}
\hspace{-1cm}\exists i\in \{1,\cdots,N\}, \;\vert\; (\,\forall j\in \{1,\cdots,N\}\smallsetminus\{i\}, \;\;\sigma^{\alpha}_{A_j}=\sigma^{\beta}_{A_j}=\sigma^{\lambda}_{A_j}
\;\;\textit{\rm and}\;\;(\sigma^{\alpha}_{A_i}\sqcap_{{}_{{{{ \overline{\mathfrak{S}}_{A_i}}}}}}\!\!\sigma^{\beta}_{A_i})\sqcoversubset_{{}_{{{{ \overline{\mathfrak{S}}_{A_i}}}}}}\sigma^{\lambda}_{A_i}\,).\;\;\;\;\;\;
\end{eqnarray}
This leads directly to the announced result.
\end{proof}

\vspace{0.3cm}

\subsection{Towards Hilbert geometry}\label{subsectionhilbertgeometry}

Let us now focus on projectivity and orthogonality properties relative to the multidimensional linear indeterministic space of states.\\

\begin{definition}
We will define the following subset of  ${ \mathfrak{S}}$ :
\begin{eqnarray}
\hspace{-1.5cm}{\check{\mathfrak{S}}} &:= & \overline{ \mathfrak{S}}\cup \left\{\chi\in { \mathfrak{S}}\smallsetminus \overline{ \mathfrak{S}}\;\vert\; \exists \mu,\nu,\phi\in \overline{ \mathfrak{S}}{}^{{}^{pure}},\;\;\;\;\gamma:=
\nu\sqcap_{{}_{\overline{ \mathfrak{S}}}}\phi
\sqcoversubset_{{}_{\overline{ \mathfrak{S}}}}\nu,\phi,\;\;\;\;\gamma\not\sqsupseteq_{{}_{\overline{ \mathfrak{S}}}}\mu^\star,\;\;\;\; \chi=\mu^\star\sqcup_{{}_{{ \mathfrak{S}}}}\gamma \right\}\label{defSbarhat}\;\;\;\;\;\;\;\;\;\;\;\;\;\\
\hspace{-1.5cm}{\widecheck{\mathfrak{S}}} &:= &  \overline{ \mathfrak{S}}\cup \left\{\chi\in { \mathfrak{S}}\smallsetminus \overline{ \mathfrak{S}}\;\vert\; \exists \mu,\phi\in \overline{ \mathfrak{S}}{}^{{}^{pure}},\;\;\mu\sqcap_{{}_{\overline{ \mathfrak{S}}}}\phi
\sqcoversubset_{{}_{\overline{ \mathfrak{S}}}}\mu,\phi,\;\;\;\;\; \chi=\mu^\star\sqcup_{{}_{{ \mathfrak{S}}}}(\mu\sqcap_{{}_{\overline{ \mathfrak{S}}}}\phi) \right\}\label{defSwidecheck}\;\;\;\;\;\;\;\;\;\;\;\;\;
\end{eqnarray}
and 
\begin{eqnarray}
{\mathfrak{G}}^{\check{\mathfrak{S}}} &:=&\overline{ \mathfrak{S}}{}^{{}^{pure}} \cup ({\check{\mathfrak{S}}}\smallsetminus \overline{ \mathfrak{S}}),\\
{\mathfrak{G}}^{\widecheck{\mathfrak{S}}} &:=&\overline{ \mathfrak{S}}{}^{{}^{pure}} \cup ({\widecheck{\mathfrak{S}}}\smallsetminus \overline{ \mathfrak{S}}).
\end{eqnarray}
We note obviously that ${\check{\mathfrak{S}}}\supseteq {\widecheck{\mathfrak{S}}}$ and ${\mathfrak{G}}^{\check{\mathfrak{S}}}\supseteq {\mathfrak{G}}^{\widecheck{\mathfrak{S}}}$.
\end{definition}

\begin{definition}
On ${\mathfrak{G}}^{\check{\mathfrak{S}}}$, we will define the {\em consistency relation}, denoted $\asymp$
, as follows (see Definition \ref{definwr} for a definition of $\wr$):
\begin{eqnarray}
&&\hspace{-1cm}\forall \chi,\chi'\in {\check{\mathfrak{S}}}\smallsetminus \overline{ \mathfrak{S}},\;\;\;\;\chi \asymp \chi'\;\;:\Leftrightarrow\;\; (\,\chi=\chi'\;\;\textit{\rm or}\;\;\exists \epsilon\in \Theta^{\overline{ \mathfrak{S}}}(\chi)\cap \Theta^{\overline{ \mathfrak{S}}}(\chi'),\;\;\epsilon=(\chi\sqcap_{{}_{{\mathfrak{S}}}}\chi')\,)\;\;\;\;\;\;\;\;\;\;\;\;\\
&&\hspace{-1cm}\forall \chi\in {\check{\mathfrak{S}}}\smallsetminus \overline{ \mathfrak{S}},\forall \sigma\in \overline{ \mathfrak{S}}{}^{{}^{pure}},\;\;\;\;\chi \asymp \sigma\;\;\Leftrightarrow\;\; \sigma \asymp \chi \;\;:\Leftrightarrow\;\; (\,\exists \epsilon\in\Theta^{\overline{ \mathfrak{S}}}(\chi)\;\vert\; \epsilon \sqcoversubset_{{}_{\overline{ \mathfrak{S}}}}\sigma\,)\\
&&\hspace{-1cm}\forall \sigma,\sigma'\in \overline{ \mathfrak{S}}{}^{{}^{pure}},\;\;\;\;\sigma \asymp \sigma'\;\;:\Leftrightarrow\;\; \sigma \wr \sigma'
\end{eqnarray}
Note that, due to the property (\ref{coveringpropertySbar}) we have $\forall \sigma,\sigma'\in \overline{ \mathfrak{S}}{}^{{}^{pure}},(\sigma\sqcap_{{}_{\overline{\mathfrak{S}}}}\sigma') \sqcoversubset_{{}_{\overline{  \mathfrak{S}}}}\sigma,\sigma'\;\textit{\rm or }\; \sigma=\sigma'$.
\end{definition}

\begin{remark}
Note that, using (\ref{thirdcoveringpropertySbarii}), we obtain easily 
\begin{eqnarray}
\hspace{-1.2cm}\forall \chi,\chi'\in {\check{\mathfrak{S}}}\smallsetminus \overline{ \mathfrak{S}},&&(\,\chi \asymp \chi'\;\textit{\rm and}\;\chi \not= \chi'\,)\;\;\Rightarrow\;\; Card(\Theta^{\overline{ \mathfrak{S}}}(\chi)\cap \Theta^{\overline{ \mathfrak{S}}}(\chi'))=1.\label{Rek1}
\end{eqnarray}
\end{remark}

\begin{remark}
Note that, using (\ref{thirdcoveringpropertySbarii}), we obtain easily 
\begin{eqnarray}
\hspace{-1.2cm}\forall \chi\in {\check{\mathfrak{S}}}\smallsetminus \overline{ \mathfrak{S}},\forall \sigma\in \overline{\mathfrak{S}}{}^{{}^{pure}},&& 
(\,\chi \asymp \sigma\,)\;\;\Rightarrow\;\; Card(\Theta^{\overline{ \mathfrak{S}}}(\chi)\cap (\downarrow_{{}_{\overline{\mathfrak{S}}}}\sigma))=1.\label{Rek2}
\end{eqnarray}
\end{remark}

\begin{definition}
A {\em consistent subset of states} is defined to be a subset of ${\mathfrak{G}}^{\check{\mathfrak{S}}}$ such that all its elements are pairwise consistent. 
\end{definition}

\begin{definition}
The {\em consistency cover} of ${\mathfrak{G}}^{\check{\mathfrak{S}}}$, denoted ${ \mathfrak{E}}^{\check{\mathfrak{S}}}$, is defined as a family of maximal consistent subsets covering the whole set of elements of ${\mathfrak{G}}^{\check{\mathfrak{S}}}$ :
\begin{eqnarray}\left\{
\begin{array}{l}
(\bigcup_{{}_{U\in { \mathfrak{E}}^{\check{\mathfrak{S}}}}} U) = {\mathfrak{G}}^{\check{\mathfrak{S}}}\\
\forall U,U’ \in { \mathfrak{E}}^{\check{\mathfrak{S}}},\;\;\; U \subseteq U’ \;\Rightarrow \; U =U’\\
\forall U\in { \mathfrak{E}}^{\check{\mathfrak{S}}},\;\forall \sigma,\sigma'\in U,\;\sigma\asymp \sigma'.
\end{array}
\right.
\end{eqnarray}
\end{definition}

\begin{definition}
Let us consider any consistent subset $U\in { \mathfrak{E}}^{\check{\mathfrak{S}}}$. \\
We define a ternary relation ${\mathfrak{r}}^U$ on $U$ (it will be called {\em colinearity}) as follows
\begin{eqnarray}
\forall \sigma_1,\sigma_2,\sigma_3\in U,&&{\mathfrak{r}}^U(\sigma_1,\sigma_2,\sigma_3)\;\;:\Leftrightarrow\;\; (\, \sigma_2=\sigma_3\;\textit{\rm or}\; (\sigma_2 \sqcap_{{}_{{\mathfrak{S}}}}\sigma_3)\sqcoversubset_{{}_{{\mathfrak{S}}}}\sigma_1\,).\;\;\;\;\;\;\;\;\;\;\;
\end{eqnarray}
\end{definition}

\begin{lemma}\label{lemmaproj0}
We note the following sheaf condition
\begin{eqnarray}
\hspace{-1cm}\forall \sigma_1,\sigma_2,\sigma_3\in { \mathfrak{G}}^{\check{\mathfrak{S}}},\forall U,U'\in { \mathfrak{E}}^{\check{\mathfrak{S}}},&& \sigma_1,\sigma_2,\sigma_3\in U\cap U'\;\Rightarrow \; {\mathfrak{r}}^U(\sigma_1,\sigma_2,\sigma_3)={\mathfrak{r}}^{U'}(\sigma_1,\sigma_2,\sigma_3).\label{VeblenYoung0}\;\;\;\;\;\;\;\;\;\;\;\;\;\;
\end{eqnarray}
\end{lemma}

Let us begin by checking a basic symmetry result of the colinearity relation.

\begin{lemma}
The ternary relation ${\mathfrak{r}}^U$ defined on the consistent subset $U$ is symmetric under any permutation of its arguments.
\end{lemma}
\begin{proof}
It suffices to prove that it is symmetric under the exchange of the two first arguments (because it is already symmetric under the exchange of its two last arguments).  This is a simple computation. We reproduce only one of its cases. \\
Let us then consider $\sigma_1,\sigma_2,\sigma_3\in {\mathfrak{G}}^{\check{\mathfrak{S}}}$ such that $\sigma_1,\sigma_2,\sigma_3$ are pairwise compatible.  We will choose for example $\sigma_2,\sigma_3\in {\check{\mathfrak{S}}}\smallsetminus \overline{ \mathfrak{S}}$. \\ 
Let us first consider that $\sigma_2\not=\sigma_3$.  Then, ${\mathfrak{r}}^U(\sigma_1,\sigma_2,\sigma_3)$ is equivalent to $\sigma_1\sqcoversupset_{{}_{{\mathfrak{S}}}}(\sigma_2 \sqcap_{{}_{{\mathfrak{S}}}}\sigma_3)$, which means that there exists $\epsilon \sqcoversubset_{{}_{{\mathfrak{S}}}}\sigma_1$ where $\{\epsilon\}= \Theta^{\overline{ \mathfrak{S}}}(\sigma_2)\cap \Theta^{\overline{ \mathfrak{S}}}(\sigma_3)$ (here we use property (\ref{Rek1}) or (\ref{Rek2})). We then deduce that $(\sigma_1 \sqcap_{{}_{{\mathfrak{S}}}}\sigma_3)=\epsilon$ and using $\epsilon\in \Theta^{\overline{ \mathfrak{S}}}(\sigma_2)$ we obtain $\sigma_2\sqcoversupset_{{}_{{\mathfrak{S}}}}(\sigma_1 \sqcap_{{}_{{\mathfrak{S}}}}\sigma_3)$ from (\ref{thirdcoveringpropertySbarii}). This proves ${\mathfrak{r}}^U(\sigma_2,\sigma_1,\sigma_3)$. \\
Let us now consider that $\sigma_2=\sigma_3$. Then, we have immediately ${\mathfrak{r}}^U(\sigma_1,\sigma_2,\sigma_3)$.  We have also trivially $\sigma_2\sqcoversupset_{{}_{{\mathfrak{S}}}}(\sigma_1 \sqcap_{{}_{{\mathfrak{S}}}}\sigma_2)$ which means that we have also ${\mathfrak{r}}^U(\sigma_2,\sigma_1,\sigma_3)$. 
\end{proof}

\begin{lemma}{\bf [Veblen-Young's first geometric property]}\label{lemmaproj1}
\begin{eqnarray}
\forall U\in { \mathfrak{E}}^{\check{\mathfrak{S}}}, \forall \sigma_1,\sigma_2\in U,&&{\mathfrak{r}}^U(\sigma_1,\sigma_2,\sigma_2).\label{VeblenYoung3}
\end{eqnarray}
\end{lemma}
\begin{proof}
Trivial.
\end{proof}


Before studying the different properties of the colinearity relation, we begin to prove an important lemma (Lemma \ref{nondegenerateallsigmapure}) describing the structure of the consistent subsets. But before that, we will check a small technical result.

\begin{lemma}\label{lemmazetacup}
Let us consider $\chi\in ({\check{\mathfrak{S}}}\smallsetminus \overline{\mathfrak{S}})$. We have
\begin{eqnarray}
&&\hspace{0.5cm} \forall \sigma_1,\sigma_2\in \Theta^{\overline{ \mathfrak{S}}}(\chi),\sigma_1\not=\sigma_2,\;\;\;\; \sigma_1\sqcup_{{}_{{ \mathfrak{S}}}}\sigma_2=\chi.\;\;\;\;\;\;\;\;\;\;\;\;\;\;\;\;\;\;
\label{thirdcoveringpropertySbarv}
\end{eqnarray}
\end{lemma}
\begin{proof}
From property (\ref{thirdcoveringpropertySbarii}) we know that
$\forall \sigma\in \Theta^{\overline{ \mathfrak{S}}}(\chi),\;\;  
\sigma\sqcoversubset_{{}_{{ \mathfrak{S}}}} \chi$. Then, if we consider $\sigma_1,\sigma_2\in \Theta^{\overline{ \mathfrak{S}}}(\chi)$ with $\sigma_1\not=\sigma_2$, we have $\sigma_1 \sqsubset_{{}_{{ \mathfrak{S}}}} (\sigma_1\sqcup_{{}_{{ \mathfrak{S}}}}\sigma_2) \sqsubseteq_{{}_{{ \mathfrak{S}}}} \chi$ and $\sigma_1\sqcoversubset_{{}_{{ \mathfrak{S}}}} \chi$. Then, we conclude that $(\sigma_1\sqcup_{{}_{{ \mathfrak{S}}}}\sigma_2)= \chi$.
\end{proof}

\begin{lemma}\label{nondegenerateallsigmapure}
Let $U\subseteq {\mathfrak{G}}^{\check{\mathfrak{S}}}$ be a consistent subset of states, such that there exists $\sigma_1,\sigma_2,\sigma_3,\sigma_4,\lambda\in U$ satisfying
\begin{eqnarray}
\forall i,j\in \{1,2,3,4\}, i\not=j,&& \sigma_i\not=\sigma_j\\
\forall i\in \{1,2,3,4\}, && \sigma_i\not=\lambda\\
\forall i,j,k,l\in \{1,2,3,4\}, \{i,j\}\not=\{k,l\},&& (\sigma_i\sqcap_{{}_{{\mathfrak{S}}}}\sigma_j)\not=(\sigma_k\sqcap_{{}_{{\mathfrak{S}}}}\sigma_l),\\
\lambda \sqcoversupset_{{}_{{\mathfrak{S}}}} (\sigma_1 \sqcap_{{}_{{\mathfrak{S}}}}\sigma_2),&&\lambda \sqcoversupset_{{}_{{\mathfrak{S}}}} (\sigma_3 \sqcap_{{}_{{\mathfrak{S}}}}\sigma_4),
\end{eqnarray}
then we have necessarily 
\begin{eqnarray}
&&\forall i\in \{1,2,3,4\},\; \sigma_i\in \overline{\mathfrak{S}}{}^{{}^{pure}}.
\end{eqnarray}
\end{lemma}
\begin{proof}

The proof is in the Appendix.
\end{proof}

\begin{lemma}\label{typoU}
Let us consider $U\in { \mathfrak{E}}^{\check{\mathfrak{S}}}$ such that there exists $\chi\in ({\check{\mathfrak{S}}}\smallsetminus \overline{ \mathfrak{S}})$ with $\chi\in U$.  The subset $U$ is necessarily of one of the two following types
\begin{itemize}
\item (Consistent Type 1) $\existunique \alpha\in \Theta^{\overline{ \mathfrak{S}}}(\chi)$ such that $\forall \sigma\in U$ we have $(\sigma\sqcap_{{}_{\overline{ \mathfrak{S}}}}\chi)=\alpha$.
\item (Consistent Type 2) $\exists \alpha,\beta\in \Theta^{\overline{ \mathfrak{S}}}(\chi),\alpha\not=\beta$ and $\exists \sigma,\sigma'\in U\cap \overline{ \mathfrak{S}}{}^{{}^{pure}}$ such that $(\sigma\sqcap_{{}_{\overline{ \mathfrak{S}}}}\chi)=\alpha$ and $(\sigma'\sqcap_{{}_{\overline{ \mathfrak{S}}}}\chi)=\beta$.
\end{itemize}
\end{lemma}
\begin{proof}
Direct consequence of Lemma \ref{nondegenerateallsigmapure} and the assumption of maximality on $U$. 
\end{proof}

\begin{lemma}\label{lemmaconsistenttype2}
Let us consider a maximal consistent subset of states $U\in { \mathfrak{E}}^{\check{\mathfrak{S}}}$ and let us assume that there exists $\lambda\in ({\check{\mathfrak{S}}}\smallsetminus \overline{ \mathfrak{S}})$ and $\sigma_1,\sigma_2,\sigma_3,\sigma_4\in \overline{ \mathfrak{S}}{}^{{}^{pure}}$ such that $\lambda,\sigma_1,\sigma_2,\sigma_3,\sigma_4\in U$  and such that ${\mathfrak{r}}^U(\lambda,\sigma_1,\sigma_2)$ and ${\mathfrak{r}}^U(\lambda,\sigma_3,\sigma_4)$ and $\forall \mu,\nu,\rho\in \{\sigma_1,\sigma_2,\sigma_3,\sigma_4\}, \;\neg\; {\mathfrak{r}}^U(\mu,\nu,\rho)$.  Then, the consistent subset $U$ is necessarily of {\em consistent type 2}
. 
\end{lemma}
\begin{proof}
As long as $\lambda\in ({\check{\mathfrak{S}}}\smallsetminus \overline{ \mathfrak{S}})$, $\Theta^{\overline{\mathfrak{S}}}(\lambda) $ is given as in Theorem \ref{lambdawr}.\\
The conditions ${\mathfrak{r}}^U(\lambda,\sigma_1,\sigma_2)$ and ${\mathfrak{r}}^U(\lambda,\sigma_3,\sigma_4)$ and $\forall \mu,\nu,\rho\in \{\sigma_1,\sigma_2,\sigma_3,\sigma_4\}, \;\neg\; {\mathfrak{r}}^U(\mu,\nu,\rho)$ impose the existence of $\epsilon,\delta\in \Theta^{\overline{\mathfrak{S}}}(\lambda)$ such that $\epsilon\not=\delta$ and $(\sigma_1\sqcap_{{}_{\overline{\mathfrak{S}}}}\sigma_2)=\epsilon$ and $(\sigma_3\sqcap_{{}_{\overline{\mathfrak{S}}}}\sigma_4)=\delta$.  We note that these conditions impose also the distinctness of the elements $\sigma_i$, i.e.  $\alpha'\not=\alpha''$ and $ \beta'\not=\beta''$.
\end{proof}

Let us now come back to the properties of the colinearity relation.

\begin{lemma}{\bf [Veblen-Young's second geometric property]}\label{lemmaproj2}
\begin{eqnarray}
\hspace{-1cm}\forall U\in { \mathfrak{E}}^{\check{\mathfrak{S}}}, \forall \sigma_1,\sigma_2,\sigma_3,\sigma_4\in U,&& (\,{\mathfrak{r}}^U(\sigma_1,\sigma_3,\sigma_4)\;\textit{\rm and}\;{\mathfrak{r}}^U(\sigma_2,\sigma_3,\sigma_4)\,)\;\;\Rightarrow\;\; {\mathfrak{r}}^U(\sigma_1,\sigma_2,\sigma_3).\label{VeblenYoung2}\;\;\;\;\;\;\;\;\;\;\;\;\;\;\;\;
\end{eqnarray}
\end{lemma}
\begin{proof}
This is a straightforward check.  We reproduce one of its cases. Let us then consider for example 
 $\sigma_3,\sigma_4\in ({\check{\mathfrak{S}}}\smallsetminus \overline{\mathfrak{S}})$ with $\sigma_3\not=\sigma_4$.  \\
Then, ${\mathfrak{r}}^U(\sigma_1,\sigma_3,\sigma_4)$ is equivalent to $\sigma_1\sqcoversupset_{{}_{{\mathfrak{S}}}}(\sigma_3 \sqcap_{{}_{{\mathfrak{S}}}}\sigma_4)$ or, using property (\ref{Rek1}),  $\epsilon \sqcoversubset_{{}_{\overline{\mathfrak{S}}}} \sigma_1$ with $\{\epsilon\}=\Theta^{\overline{ \mathfrak{S}}}(\sigma_3)\cap \Theta^{\overline{ \mathfrak{S}}}(\sigma_4)$. Analogously,  ${\mathfrak{r}}^U(\sigma_2,\sigma_3,\sigma_4)$ is equivalent to $\sigma_2\sqcoversupset_{{}_{{\mathfrak{S}}}}(\sigma_3 \sqcap_{{}_{{\mathfrak{S}}}}\sigma_4)$ or, using property (\ref{Rek1})
,  $\mu \sqcoversubset_{{}_{\overline{\mathfrak{S}}}} \sigma_2$ with $\{\mu\}=\Theta^{\overline{ \mathfrak{S}}}(\sigma_3)\cap \Theta^{\overline{ \mathfrak{S}}}(\sigma_4)$. Using property (\ref{Rek1})
, we also know that $\mu=\epsilon$.  We then note that $\mu=(\sigma_2 \sqcap_{{}_{{\mathfrak{S}}}}\sigma_3)$ and then $\sigma_1\sqcoversupset_{{}_{{\mathfrak{S}}}}(\sigma_2 \sqcap_{{}_{{\mathfrak{S}}}}\sigma_3)$, i.e. ${\mathfrak{r}}^U(\sigma_1,\sigma_2,\sigma_3)$.
\end{proof}

\begin{lemma}{\bf [Veblen-Young's third geometric property]}\label{lemmaproj3}\\
We will consider $\sigma_1,\sigma_2,\sigma_3,\sigma_4\in { \mathfrak{E}}^{\check{\mathfrak{S}}}$. We have the following result
 \begin{eqnarray}
&&\hspace{-1.5cm}(\,\exists U\in { \mathfrak{E}}^{\check{\mathfrak{S}}},\sigma_1,\sigma_2,\sigma_3,\sigma_4\in U,\exists \lambda\in U\;\vert\; {\mathfrak{r}}^U(\lambda,\sigma_1,\sigma_2)\;\textit{\rm and}\;{\mathfrak{r}}^U(\lambda,\sigma_3,\sigma_4)
\,)\;\;\Rightarrow\;\;\;\;\;\;\;\;\;\;\;\nonumber\\
&&\hspace{-1.5cm} (\,(\exists U',U''\in { \mathfrak{E}}^{\check{\mathfrak{S}}}\;\vert\;\sigma_1,\sigma_2,\sigma_3,\sigma_4\in U',U'')\;\textit{\rm and}\;\exists\mu\in U'\;\vert\; {\mathfrak{r}}^{U'}(\mu,\sigma_1,\sigma_3)\;\textit{\rm and}\;{\mathfrak{r}}^{U'}(\mu,\sigma_2,\sigma_4) \nonumber\\
&&\hspace{3.5cm} \textit{\rm and}\;\exists\chi\in U''\;\vert\; {\mathfrak{r}}^{U''}(\chi,\sigma_2,\sigma_3)\;\textit{\rm and}\;{\mathfrak{r}}^{U''}(\chi,\sigma_1,\sigma_4)\,).\;\;\;\;\;\;\;\;\;\;\;\;\;\;\;\;\label{VeblenYoung1}
\end{eqnarray}
\end{lemma}
\begin{proof}
Let us first note that if $(\sigma_1\sqcap_{{}_{\overline{ \mathfrak{S}}}}\sigma_2)=(\sigma_3\sqcap_{{}_{\overline{ \mathfrak{S}}}}\sigma_4)$ then we have also $(\sigma_1\sqcap_{{}_{\overline{ \mathfrak{S}}}}\sigma_3)=(\sigma_2\sqcap_{{}_{\overline{ \mathfrak{S}}}}\sigma_4)$ and also $(\sigma_1\sqcap_{{}_{\overline{ \mathfrak{S}}}}\sigma_4)=(\sigma_2\sqcap_{{}_{\overline{ \mathfrak{S}}}}\sigma_3)$ and then the proof of (\ref{VeblenYoung1}) is trivial. We will then assume $(\sigma_1\sqcap_{{}_{\overline{ \mathfrak{S}}}}\sigma_2)\not=(\sigma_3\sqcap_{{}_{\overline{ \mathfrak{S}}}}\sigma_4)$ and $(\sigma_1\sqcap_{{}_{\overline{ \mathfrak{S}}}}\sigma_3)\not=(\sigma_2\sqcap_{{}_{\overline{ \mathfrak{S}}}}\sigma_4)$ and $(\sigma_1\sqcap_{{}_{\overline{ \mathfrak{S}}}}\sigma_4)\not=(\sigma_2\sqcap_{{}_{\overline{ \mathfrak{S}}}}\sigma_3)$.\\
We will also assume $\lambda\not=\sigma_1,\lambda\not=\sigma_2,\lambda\not=\sigma_3,\lambda\not=\sigma_4$. Indeed, if we have for example $\lambda=\sigma_1$ then it suffices to take $\mu:=\sigma_4$ and $\chi:=\sigma_3$ to satisfy the announced properties.\\
We will then assume the existence of $\lambda\in U\in { \mathfrak{E}}^{\check{\mathfrak{S}}}$ such that ${\mathfrak{r}}^U(\lambda,\sigma_1,\sigma_2)\;\textit{\rm and}\;{\mathfrak{r}}^U(\lambda,\sigma_3,\sigma_4)$. \\
The assumption ${\mathfrak{r}}^U(\lambda,\sigma_1,\sigma_2)\;\textit{\rm and}\;{\mathfrak{r}}^U(\lambda,\sigma_3,\sigma_4)$ is equivalent to 
\begin{eqnarray}
(\sigma_1\sqcap_{{}_{\overline{ \mathfrak{S}}}}\sigma_2)\sqcoversubset_{{}_{{ \mathfrak{S}}}}\lambda, && (\sigma_3\sqcap_{{}_{\overline{ \mathfrak{S}}}}\sigma_4)\sqcoversubset_{{}_{{ \mathfrak{S}}}}\lambda.\label{reformassumption}
\end{eqnarray}
We have then, using (\ref{thirdcoveringpropertySbariii}) (if $\lambda\in  ({\check{\mathfrak{S}}}\smallsetminus \overline{\mathfrak{S}})$) or using (\ref{secondcoveringpropertySbar}) (if $\lambda\in \overline{\mathfrak{S}}{}^{{}^{pure}}$), the properties 
\begin{eqnarray}
(\sigma_1\sqcap_{{}_{\overline{ \mathfrak{S}}}}\sigma_2\sqcap_{{}_{\overline{ \mathfrak{S}}}}\sigma_3\sqcap_{{}_{\overline{ \mathfrak{S}}}}\sigma_4)\sqcoversubset_{{}_{\overline{ \mathfrak{S}}}} (\sigma_1\sqcap_{{}_{\overline{ \mathfrak{S}}}}\sigma_2),(\sigma_3\sqcap_{{}_{\overline{ \mathfrak{S}}}}\sigma_4)\label{4=3}
\end{eqnarray}
We know from Theorem \ref{zetalambda} and Theorem \ref{lambdawr} that
\begin{eqnarray}
&&\hspace{-2cm}\exists I\subseteq\{1,\cdots,N\}, Card(I)= N-2\;\vert\; \forall i\in I, \forall \alpha,\beta\in \{\sigma_1,\sigma_2,\sigma_3,\sigma_4\},\;\;\zeta^{{}^{{\mathfrak{S}}_{A_1}\cdots{\mathfrak{S}}_{A_N}}}_{(i)}(\alpha) = \zeta^{{}^{{\mathfrak{S}}_{A_1}\cdots{\mathfrak{S}}_{A_N}}}_{(i)}(\beta).\nonumber\\
&&\label{conditionzeta}
\end{eqnarray}
Having noted the property $(\sigma_1\sqcap_{{}_{\overline{ \mathfrak{S}}}}\sigma_3)\not=(\sigma_2\sqcap_{{}_{\overline{ \mathfrak{S}}}}\sigma_4)$,  we will consider any $\xi\in \overline{\mathfrak{S}}{}^{{}^{pure}}$ such that 
\begin{eqnarray}
\hspace{-1cm}\xi\sqsupseteq_{{}_{\overline{ \mathfrak{S}}}} (\sigma_1\sqcap_{{}_{\overline{ \mathfrak{S}}}}\sigma_3)^\star, && \xi\not\sqsupseteq_{{}_{\overline{ \mathfrak{S}}}} (\sigma_2\sqcap_{{}_{\overline{ \mathfrak{S}}}}\sigma_4)^\star,\;\;\;\;\;\;\forall i\in I, \zeta^{{}^{{\mathfrak{S}}_{A_1}\cdots{\mathfrak{S}}_{A_N}}}_{(i)}(\xi) = \zeta^{{}^{{\mathfrak{S}}_{A_1}\cdots{\mathfrak{S}}_{A_N}}}_{(i)}(\sigma_1)\;\;\;\;\;\;\;\;\;\;\;\;\;\;\;\;\;\;
\label{assumptionxi}
\end{eqnarray}
(we note that $\sigma_1\sqcap_{{}_{\overline{ \mathfrak{S}}}}\sigma_3$ and $\sigma_2\sqcap_{{}_{\overline{ \mathfrak{S}}}}\sigma_4$ are different from $\bot_{{}_{\overline{ \mathfrak{S}}}}$).\\
From (\ref{thirdcoveringpropertySbari}), and using $\xi\not\sqsupseteq_{{}_{\overline{ \mathfrak{S}}}} (\sigma_2\sqcap_{{}_{\overline{ \mathfrak{S}}}}\sigma_4)^\star$ and property (\ref{conditionzeta}), we know that there exists $\upsilon\in {\mathfrak{G}}^{\check{\mathfrak{S}}}$ such that
\begin{eqnarray}
\upsilon := \xi^\star \sqcup_{{}_{{ \mathfrak{S}}}}(\sigma_2\sqcap_{{}_{\overline{ \mathfrak{S}}}}\sigma_4)
\end{eqnarray}
From (\ref{thirdcoveringpropertySbarii}), we know that $\upsilon\sqcoversupset_{{}_{{ \mathfrak{S}}}} (\sigma_2\sqcap_{{}_{\overline{ \mathfrak{S}}}}\sigma_4)$ and then ${\mathfrak{r}}^U(\upsilon,\sigma_2,\sigma_4)$. 
We intent to show that we have also ${\mathfrak{r}}^U(\upsilon,\sigma_1,\sigma_3)$. \\
We then consider any $\omega\in \overline{\mathfrak{S}}{}^{{}^{pure}}$ such that 
\begin{eqnarray}
\omega\sqsupseteq_{{}_{\overline{ \mathfrak{S}}}} (\sigma_1\sqcap_{{}_{\overline{ \mathfrak{S}}}}\sigma_3)^\star, && \omega\not=\xi,\;\;\;\;\;\;\forall i\in I, \zeta^{{}^{{\mathfrak{S}}_{A_1}\cdots{\mathfrak{S}}_{A_N}}}_{(i)}(\omega) = \zeta^{{}^{{\mathfrak{S}}_{A_1}\cdots{\mathfrak{S}}_{A_N}}}_{(i)}(\xi)\;\;\;\;\;\;\;\;\;\;\;\;\;\;\;\;\;\;\label{omegazeta}
\end{eqnarray}
and we intent to show that $\omega^\star\sqsubseteq_{{}_{{ \mathfrak{S}}}}\upsilon$ (we note that  
$(\sigma_1\sqcap_{{}_{\overline{ \mathfrak{S}}}}\sigma_3)^\star$ is not in $\overline{\mathfrak{S}}{}^{{}^{pure}}$ and then, such an $\omega$ different from $\xi$ always exists).\\
We first note that $(\xi\sqcap_{{}_{\overline{ \mathfrak{S}}}}\omega)^\star\not\sqsubseteq_{{}_{\overline{ \mathfrak{S}}}}\sigma_2$. Indeed, if we had $(\xi\sqcap_{{}_{\overline{ \mathfrak{S}}}}\omega)^\star\sqsubseteq_{{}_{\overline{ \mathfrak{S}}}}\sigma_2$, and using the fact that $(\xi\sqcap_{{}_{\overline{ \mathfrak{S}}}}\omega)^\star\sqsubseteq_{{}_{\overline{ \mathfrak{S}}}}(\sigma_1\sqcap_{{}_{\overline{ \mathfrak{S}}}}\sigma_3)$, we would have $(\xi\sqcap_{{}_{\overline{ \mathfrak{S}}}}\omega)^\star\sqsubseteq_{{}_{\overline{ \mathfrak{S}}}}(\sigma_1\sqcap_{{}_{\overline{ \mathfrak{S}}}}\sigma_3\sqcap_{{}_{\overline{ \mathfrak{S}}}}\sigma_2)$. But now,  using (\ref{4=3}), we have in particular $\sigma_1\sqcap_{{}_{\overline{ \mathfrak{S}}}}\sigma_3\sqcap_{{}_{\overline{ \mathfrak{S}}}}\sigma_2=(\sigma_1\sqcap_{{}_{\overline{ \mathfrak{S}}}}\sigma_2\sqcap_{{}_{\overline{ \mathfrak{S}}}}\sigma_3\sqcap_{{}_{\overline{ \mathfrak{S}}}}\sigma_4)$. As a consequence, we would have $(\xi\sqcap_{{}_{\overline{ \mathfrak{S}}}}\omega)^\star\sqsubseteq_{{}_{\overline{ \mathfrak{S}}}}
(\sigma_1\sqcap_{{}_{\overline{ \mathfrak{S}}}}\sigma_2\sqcap_{{}_{\overline{ \mathfrak{S}}}}\sigma_3\sqcap_{{}_{\overline{ \mathfrak{S}}}}\sigma_4) \sqsubseteq_{{}_{\overline{ \mathfrak{S}}}} (\sigma_2\sqcap_{{}_{\overline{ \mathfrak{S}}}}\sigma_4)$, and then $\xi^\star\sqsubseteq_{{}_{\overline{ \mathfrak{S}}}}
 (\sigma_2\sqcap_{{}_{\overline{ \mathfrak{S}}}}\sigma_4)$ which is false by assumption (\ref{assumptionxi}). \\
Then, from (\ref{thirdcoveringpropertySbari}), and using the fact that $\sigma_2$ is in $\overline{ \mathfrak{S}}{}^{{}^{pure}}$ (consequence of Lemma \ref{nondegenerateallsigmapure}) and the fact that $(\xi\sqcap_{{}_{\overline{ \mathfrak{S}}}}\omega)^\star\not\sqsubseteq_{{}_{\overline{ \mathfrak{S}}}}\sigma_2$ and also the property (\ref{omegazeta}), we know that there exists $\nu\in {\mathfrak{G}}^{\check{\mathfrak{S}}}$ such that 
\begin{eqnarray}
\nu := \sigma_2^\star \sqcup_{{}_{{ \mathfrak{S}}}}(\xi\sqcap_{{}_{\overline{ \mathfrak{S}}}}\omega).\label{defnu}
\end{eqnarray}
Using $(\xi\sqcap_{{}_{\overline{ \mathfrak{S}}}}\omega)\sqsupseteq_{{}_{\overline{ \mathfrak{S}}}}(\sigma_1\sqcap_{{}_{\overline{ \mathfrak{S}}}}\sigma_3)^\star$, we deduce as before
\begin{eqnarray}
\nu \sqsupseteq_{{}_{\overline{ \mathfrak{S}}}} (\sigma_2\sqcap_{{}_{\overline{ \mathfrak{S}}}}\sigma_1\sqcap_{{}_{\overline{ \mathfrak{S}}}}\sigma_3)^\star=(\sigma_1\sqcap_{{}_{\overline{ \mathfrak{S}}}}\sigma_2\sqcap_{{}_{\overline{ \mathfrak{S}}}}\sigma_3\sqcap_{{}_{\overline{ \mathfrak{S}}}}\sigma_4)^\star
\end{eqnarray}
(we note using, for example,  the property (\ref{thirdcoveringpropertySbariii}) and (\ref{reformassumption}) that $(\sigma_1\sqcap_{{}_{\overline{ \mathfrak{S}}}}\sigma_2\sqcap_{{}_{\overline{ \mathfrak{S}}}}\sigma_3\sqcap_{{}_{\overline{ \mathfrak{S}}}}\sigma_4)$ is different from $\bot_{{}_{\overline{\mathfrak{S}}}}$).\\
We then deduce
\begin{eqnarray}
\nu \sqsupseteq_{{}_{\overline{ \mathfrak{S}}}} (\sigma_2\sqcap_{{}_{\overline{ \mathfrak{S}}}}\sigma_4)^\star.\label{propnu}
\end{eqnarray}
From this inequality we deduce that
\begin{eqnarray}
\exists \alpha\in \Theta^{\overline{ \mathfrak{S}}}(\nu) &\vert & \alpha \sqsupseteq_{{}_{\overline{ \mathfrak{S}}}} (\sigma_2\sqcap_{{}_{\overline{ \mathfrak{S}}}}\sigma_4)^\star.
\end{eqnarray}
Let us fix one of these $\alpha$.\\
We recall that, as long as $\xi\in \overline{\mathfrak{S}}{}^{{}^{pure}}$ we have $(\xi \sqcap_{{}_{{ \mathfrak{S}}}}\nu)=Max(\{\,\xi\sqcap_{{}_{\overline{ \mathfrak{S}}}} \kappa\;\vert\; \kappa\in  \Theta^{\overline{ \mathfrak{S}}}(\nu)\,\})$. Moreover, we have $(\xi \sqcap_{{}_{\overline{ \mathfrak{S}}}}\omega)\in \Theta^{\overline{ \mathfrak{S}}}(\nu)$. Then, using (\ref{thirdcoveringpropertySbarii}) 
we deduce as well 
\begin{eqnarray}
\xi \sqcap_{{}_{\overline{ \mathfrak{S}}}}\omega &=& \xi \sqcap_{{}_{{ \mathfrak{S}}}}\nu.
\end{eqnarray}
As a result, we obtain for the chosen $\alpha$
\begin{eqnarray}
\xi \sqcap_{{}_{\overline{ \mathfrak{S}}}}\omega &\sqsupseteq_{{}_{\overline{ \mathfrak{S}}}}& \xi \sqcap_{{}_{{ \mathfrak{S}}}}\alpha.\label{eqinterm}
\end{eqnarray}
Now, since $\upsilon \sqsupseteq_{{}_{{ \mathfrak{S}}}}\xi^\star$ and $\upsilon \sqsupseteq_{{}_{{ \mathfrak{S}}}}(\sigma_2\sqcap_{{}_{\overline{ \mathfrak{S}}}}\sigma_4)$ and $ \alpha \sqsupseteq_{{}_{\overline{ \mathfrak{S}}}} (\sigma_2\sqcap_{{}_{\overline{ \mathfrak{S}}}}\sigma_4)^\star$, we deduce that 
\begin{eqnarray}
&& \upsilon \sqsupseteq_{{}_{{ \mathfrak{S}}}} (\xi^\star \sqcup_{{}_{{ \mathfrak{S}}}} \alpha^\star).
\end{eqnarray}
and then using (\ref{eqinterm})
\begin{eqnarray}
&& \upsilon \sqsupseteq_{{}_{{ \mathfrak{S}}}} (\xi^\star\sqcup_{{}_{{ \mathfrak{S}}}} \omega^\star)\sqsupseteq_{{}_{{ \mathfrak{S}}}} \omega^\star
\end{eqnarray}
We have then obtained the announced intermediary result : for any $\omega\in \overline{\mathfrak{S}}{}^{{}^{pure}}$ such that the conditions (\ref{omegazeta}) are satisfied, we have $\omega^\star\sqsubseteq_{{}_{{ \mathfrak{S}}}}\upsilon$.  We note that we have also by assumption $\xi\sqsupseteq_{{}_{\overline{ \mathfrak{S}}}} (\sigma_1\sqcap_{{}_{\overline{ \mathfrak{S}}}}\sigma_3)^\star$ and $\xi^\star\sqsubseteq_{{}_{{ \mathfrak{S}}}}\upsilon$.  \\
Hence, we can conclude 
\begin{eqnarray}
&& \upsilon \sqsupseteq_{{}_{{ \mathfrak{S}}}}(\sigma_1\sqcap_{{}_{\overline{ \mathfrak{S}}}}\sigma_3)\label{conclusionupsilon}
\end{eqnarray}
i.e. explicitly ${\mathfrak{r}}^{U'}(\upsilon,\sigma_1,\sigma_3)$.  Obviously, we have $\upsilon\asymp \sigma_1,\upsilon\asymp \sigma_2,\upsilon\asymp \sigma_3,\upsilon\asymp \sigma_4$ and then the existence of the consistent subset $U'$ does not pose any problem.\\
As a final conclusion, we obtain the existence of $\mu:=\upsilon\in {\mathfrak{G}}^{\check{\mathfrak{S}}}$ such that ${\mathfrak{r}}^{U'}(\mu,\sigma_1,\sigma_3)$ and ${\mathfrak{r}}^{U'}(\mu,\sigma_2,\sigma_4)$. \\
By exchanging the roles of $\sigma_3$ and $\sigma_4$ in our construction, we obtain also the existence of $\chi\in {\mathfrak{G}}^{\check{\mathfrak{S}}}$ such that ${\mathfrak{r}}^{U''}(\chi,\sigma_1,\sigma_4)$ and ${\mathfrak{r}}^{U''}(\chi,\sigma_2,\sigma_3)$. 
\end{proof}

Let us now summarize what we have found so far.

\begin{theorem}
${ \mathfrak{E}}^{\check{\mathfrak{S}}}$ is equipped with a structure of {\em projective sheaf}.
In other words, apart from the sheaf condition (\ref{VeblenYoung0}), we have proved the Veblen-Young's geometric axioms of projective spaces (\ref{VeblenYoung3})(\ref{VeblenYoung2})(\ref{VeblenYoung1}).
\end{theorem}
\begin{proof}
This result summarizes Lemma \ref{lemmaproj0}, Lemma \ref{lemmaproj1}, Lemma \ref{lemmaproj2}, Lemma \ref{lemmaproj3}.
\end{proof}

\begin{definition}
For any pair $\sigma_1,\sigma_2\in { \mathfrak{G}}^{\check{\mathfrak{S}}}$, we define the line $\ell_{(\sigma_1\sigma_2)}$ to be the set of elements $\sigma_3$ in ${ \mathfrak{G}}^{\check{\mathfrak{S}}}$ such that $\sigma_3=\sigma_1$ or $\sigma_3=\sigma_2$ or $\sigma_3\sqcoversupset_{{}_{{\mathfrak{S}}}}(\sigma_1\sqcap_{{}_{{\mathfrak{S}}}}\sigma_2)$.
\end{definition}
\begin{definition}
For any triple $\sigma_1,\sigma_2,\sigma_3\in { \mathfrak{G}}^{\check{\mathfrak{S}}}$ distinct, such that $\sigma_i\notin \ell_{(\sigma_j\sigma_k)}$ for any $\{i,j,k\}=\{1,2,3\}$, we define the plane $\wp_{(\sigma_1,\sigma_2,\sigma_3)}$ to be the set of elements $\sigma_4$ in ${ \mathfrak{G}}^{\check{\mathfrak{S}}}$ such that there exists $i,j,k$ with $\{i,j,k\}=\{1,2,3\}$ and 
$\lambda\in { \mathfrak{G}}^{\check{\mathfrak{S}}}$ and $\lambda\in \ell_{(\sigma_i\sigma_4)}$ and $\lambda\in \ell_{(\sigma_j\sigma_k)}$. \\
The set of planes on ${ \mathfrak{G}}^{\check{\mathfrak{S}}}$ will be denoted $\wp^{\check{\mathfrak{S}}}$.\\
Any element of $\wp^{\check{\mathfrak{S}}}$ has a structure of a projective sheaf.
\end{definition}
\begin{remark}
Any plane is stable under the application of the third Veblen-Young's property.  In other words,  if there exists $U\in { \mathfrak{E}}^{\check{\mathfrak{S}}}$ and $\lambda, \sigma_1,\sigma_2,\sigma_3,\sigma_4\in { \mathfrak{G}}^{\check{\mathfrak{S}}}$ such that $\lambda,\sigma_1,\sigma_2,\sigma_3,\sigma_4\in U$  and such that ${\mathfrak{r}}^U(\lambda,\sigma_1,\sigma_2)$ and ${\mathfrak{r}}^U(\lambda,\sigma_3,\sigma_4)$ and $\forall \mu,\nu,\rho\in \{\sigma_1,\sigma_2,\sigma_3,\sigma_4\}, \;\neg\; {\mathfrak{r}}^U(\mu,\nu,\rho)$. Then,   $\lambda,\sigma_1,\sigma_2,\sigma_3,\sigma_4$ are in a same plane $\wp$. Moreover, if we introduce $U'\in { \mathfrak{E}}^{\check{\mathfrak{S}}}$ and the element $\mu\in U'$  such that ${\mathfrak{r}}^{U'}(\mu,\sigma_1,\sigma_3)$ and ${\mathfrak{r}}^{U'}(\mu,\sigma_2,\sigma_4)$, and if we introduce  $U''\in { \mathfrak{E}}^{\check{\mathfrak{S}}}$
 and $\chi\in U''$ such that ${\mathfrak{r}}^{U''}(\chi,\sigma_2,\sigma_3)$ and ${\mathfrak{r}}^{U''}(\chi,\sigma_1,\sigma_4)\,)$ we have necessarily $\mu,\chi\in \wp$.
 \end{remark}

Let us now consider the properties of the orthogonality relation induced on ${ \mathfrak{G}}^{\check{\mathfrak{S}}}$ from the orthogonality relation defined on ${\mathfrak{S}}$ (see subsection \ref{firstresultsrealstructures}).\\

Even though we have found a very interesting projective structure on ${ \mathfrak{E}}^{\check{\mathfrak{S}}}$, we must now proceed with caution.  Indeed, the problem with ${ \mathfrak{G}}^{\check{\mathfrak{S}}}$ is due to the following simple fact (see Remark \ref{remarkstarredplanes})
\begin{eqnarray}
\hspace{-1cm}\forall U\in { \mathfrak{E}}^{\check{\mathfrak{S}}},\forall \chi\in { \mathfrak{G}}^{\check{\mathfrak{S}}}\cap U, \forall \sigma\in \overline{\mathfrak{S}}{}^{{}^{pure}}\cap U,&&(\,\exists \lambda\in U\;\vert\; {\mathfrak{r}}^{U}(\chi,\sigma,\lambda)\;\textit{\rm and}\; \chi\underline{\perp}\lambda\,)\Rightarrow (\,\chi\in { \mathfrak{G}}^{\widecheck{\mathfrak{S}}}\,)\;\;\;\;\;\;\;\;\;\;\;\;\;
\end{eqnarray}  




The following sequence of results is the consequence of our attempt to deal with this annoying problem and with the necessity to restrict ourselves to ${ \mathfrak{G}}^{\widecheck{\mathfrak{S}}}$.

\begin{definition}\label{completenessU}
A consistent subset of states $U\subseteq { \mathfrak{G}}^{\check{\mathfrak{S}}}$ is said to be {\em orthogonally complete} iff
\begin{eqnarray}
&&\hspace{-2.2cm}\forall \sigma_1,\sigma_2,\sigma_3\in U,\;\textit{\rm distinct},\;\; {\mathfrak{r}}^U(\sigma_1,\sigma_2,\sigma_3) \;\Rightarrow \;(\,\sigma_1\underline{\perp}\sigma_2\;\textit{\rm or}\;\sigma_1\underline{\perp}\sigma_3\;\textit{\rm or}\;\sigma_2\underline{\perp}\sigma_3\,).\;\;\;\;\;\;\;\;\;\;\;\;\;\;\;\;\label{orthogcompletecondition1}\\
&&\hspace{-2.2cm}\forall \sigma_1,\sigma_2,\sigma_3,\sigma_4\in U\;\vert\; (\,\exists \lambda\in U, \;{\mathfrak{r}}^U(\lambda,\sigma_1,\sigma_2),\;{\mathfrak{r}}^U(\lambda,\sigma_3,\sigma_4)\;\textit{\rm and}\; \forall\mu,\nu,\rho\in \{\sigma_1,\sigma_2,\sigma_3,\sigma_4\}, \;\neg\; {\mathfrak{r}}^U(\mu,\nu,\rho)\,)\nonumber\\
&&\hspace{2cm}\exists \alpha,\beta,\gamma\in \{\sigma_1,\sigma_2,\sigma_3,\sigma_4\},\;
\alpha\underline{\perp}\beta\;\;\textit{\rm and}\;\; \alpha\underline{\perp}\gamma
.\;\;\;\;\;\;\;\;\;\;\;\label{orthogcompletecondition2}
\end{eqnarray}
\end{definition}

\begin{definition}
For any $\sigma_1,\sigma_2,\sigma_3\in \overline{ \mathfrak{S}}{}^{{}{pure}}$ satisfying the following conditions
\begin{enumerate}
\item there exists $I\subseteq\{1,\cdots,N\}, Card(I)= N-2\;\vert\;$ $\forall m\in I,\;\;\zeta^{{}^{{\mathfrak{S}}_{A_1}\cdots{\mathfrak{S}}_{A_N}}}_{(m)}(\sigma_1) = \zeta^{{}^{{\mathfrak{S}}_{A_1}\cdots{\mathfrak{S}}_{A_N}}}_{(m)}(\sigma_2)=\zeta^{{}^{{\mathfrak{S}}_{A_1}\cdots{\mathfrak{S}}_{A_N}}}_{(m)}(\sigma_3)$, 
\item if we denote $j,k\in \{1,\cdots,N\}\smallsetminus I$ with $j<k$,  there exists $\alpha\in \overline{\mathfrak{S}}{}_{A_j}^{{}^{pure}}$ and $\beta\in \overline{\mathfrak{S}}{}_{A_k}^{{}^{pure}}$ such that
$\Upsilon^{(j,k)}_{\sigma_1}=\alpha\widetilde{\otimes}\beta$, $\Upsilon^{(j,k)}_{\sigma_2}=\alpha^\star\widetilde{\otimes}\beta$, $\Upsilon^{(j,k)}_{\sigma_3}=\alpha\widetilde{\otimes}\beta^\star$.
\end{enumerate}
we say that the plane $\wp_{(\sigma_1,\sigma_2,\sigma_3)}$ is a {\em starred plane}. The set of starred planes on ${ \mathfrak{G}}^{\check{\mathfrak{S}}}$ will be denoted $\wp_\star^{\check{\mathfrak{S}}}$.
\end{definition}

Let us now analyze the structure of orthogonally complete consistent subsets in the following lemmas.

\begin{lemma}\label{natureorthocompleteU}
Let us consider a consistent subset of states $U\in { \mathfrak{E}}^{\check{\mathfrak{S}}}$ and let us assume that there exists $\chi\in ({\check{\mathfrak{S}}}\smallsetminus \overline{ \mathfrak{S}})$ and that the consistent subset $U$ is of {\em consistent type 2}.\\ 
Let us now assume that $U$ is orthogonally complete. Then, we have necessarily $\chi\in ({\widecheck{\mathfrak{S}}}\smallsetminus \overline{ \mathfrak{S}})$.\\
Conversely,  let us assume that $\chi\in ({\widecheck{\mathfrak{S}}}\smallsetminus \overline{ \mathfrak{S}})$ and that $U$ is maximal,  then we have
\begin{eqnarray}
&&\hspace{-0.5cm}U =  \{\chi\}\cup\{\,\varphi_\alpha\;\vert\; \alpha\in \Theta^{\overline{ \mathfrak{S}}}(\chi)\,\}\cup\{\,\psi_\alpha\;\vert\; \alpha\in \Theta^{\overline{ \mathfrak{S}}}(\chi)\,\}\label{orthocompletestructure0}\\
&& \forall \alpha\in \Theta^{\overline{ \mathfrak{S}}}(\chi),\;\varphi_\alpha,\psi_\alpha\in\overline{ \mathfrak{S}}{}^{{}^{pure}}  \;\textit{\rm and}\; \varphi_\alpha,\psi_\alpha \sqcoversupset_{{}_{\overline{ \mathfrak{S}}}}\alpha\;\;\textit{\rm and}\;\;\varphi_\alpha\not=\psi_\alpha\label{orthocompletestructure1}\;\;\;\;\;\;\;\;\;\;\;\;\;\;\;\;\;\;\;\;\;\;\;\;\\
&& \existunique\, \gamma\in \Theta^{\overline{ \mathfrak{S}}}(\chi)\;\vert\; \varphi_\gamma \underline{\perp} \chi,\;\;\varphi_\gamma \not\!\!\!\underline{\perp} \psi_\gamma,\;\;\forall \alpha\in \Theta^{\overline{ \mathfrak{S}}}(\chi)\smallsetminus \{\gamma\},\; \varphi_\gamma \underline{\perp} \varphi_\alpha,\;\;\varphi_\gamma \underline{\perp} \psi_\alpha,\;\;\;\;\;\;\;\;\;\;\;\;\;\;\;\;\;\;\;\label{orthocompletestructure2}\\
&&\forall \alpha\in \Theta^{\overline{ \mathfrak{S}}}(\chi)\smallsetminus \{\gamma\}, \varphi_{\alpha}\underline{\perp}\psi_{\alpha},\;\;\varphi_{\alpha} \not\!\!\!\underline{\perp}\chi,\;\;\psi_{\alpha} \not\!\!\!\underline{\perp}\chi,\label{orthocompletestructure3}\\
&&\forall \alpha,\beta\in \Theta^{\overline{ \mathfrak{S}}}(\chi)\smallsetminus \{\gamma\},\alpha\not=\beta,\; \varphi_{\alpha}\underline{\perp}\varphi_{\beta},\;\;\varphi_{\alpha}\not\!\!\!\underline{\perp}\psi_{\beta},\label{orthocompletestructure4}\\
&&\forall \alpha,\beta\in \Theta^{\overline{ \mathfrak{S}}}(\chi)\smallsetminus \{\gamma\},\alpha\not=\beta,\exists \lambda_{\alpha\beta}\in \overline{ \mathfrak{S}}{}^{{}^{pure}}\;\vert\; 
(\varphi_{\beta}\sqcap_{{}_{\overline{ \mathfrak{S}}}}\psi_{\alpha}), (\psi_{\beta}\sqcap_{{}_{\overline{ \mathfrak{S}}}}\varphi_{\alpha})\sqcoversubset_{{}_{\overline{ \mathfrak{S}}}}\lambda_{\alpha\beta}
\;\;\;\;\;\;\;\;\;\;\;\;\;\;\;\;\;\nonumber\\
&&\hspace{6.5cm}\textit{\rm and}\;\lambda_{\alpha\beta}\underline{\perp}\varphi_{\alpha},\;\;\lambda_{\alpha\beta}\underline{\perp}\varphi_{\beta},\label{orthocompletestructure5}\\
&&\forall \alpha\in \Theta^{\overline{ \mathfrak{S}}}(\chi)\smallsetminus \{\gamma\},\exists \mu_{\alpha\gamma}\in \overline{ \mathfrak{S}}{}^{{}^{pure}}\;\vert\; 
(\varphi_{\gamma}\sqcap_{{}_{\overline{ \mathfrak{S}}}}\varphi_{\alpha}), (\psi_{\gamma}\sqcap_{{}_{\overline{ \mathfrak{S}}}}\psi_{\alpha})\sqcoversubset_{{}_{\overline{ \mathfrak{S}}}}\mu_{\alpha\gamma}
\;\;\;\;\;\;\;\;\;\;\;\;\nonumber\\
&&\hspace{6.5cm}\textit{\rm and}\;\mu_{\alpha\gamma}\underline{\perp}\psi_{\alpha},\label{orthocompletestructure6}
\end{eqnarray}
and then necessarily $U$ is orthogonally complete.
\end{lemma}
\begin{proof}
As long as $\lambda\in ({\check{\mathfrak{S}}}\smallsetminus \overline{ \mathfrak{S}})$, $\Theta^{\overline{\mathfrak{S}}}(\lambda) $ is given as in Theorem \ref{lambdawr}.\\
Now we use the condition of orthogonal completeness to conclude that necessarily $(\alpha''=\alpha,\beta''=\beta)$ or $(\alpha'=\alpha,\beta'=\beta)$. This concludes the proof of the first assertion, i.e. $\lambda\in ({\widecheck{\mathfrak{S}}}\smallsetminus \overline{ \mathfrak{S}})$.\\
Then, in (\ref{orthocompletestructure0})--(\ref{orthocompletestructure6}), we adopt the same notation as in (\ref{thirdcoveringpropertySbar0}) (\ref{thirdcoveringpropertySbar0bis}) (\ref{thirdcoveringpropertySbarviic}) (\ref{thirdcoveringpropertySbarviib}) for $\varphi_{\alpha}$ and $\psi_{\alpha}$ with $\alpha$ any element of $\Theta^{\overline{ \mathfrak{S}}}(\chi)$. 
With these formulas, the last assertion is trivial to check.
\end{proof}

\begin{lemma}\label{natureorthocompleteUbis}
Let us consider an orthogonally complete consistent subset such that there exists $\chi\in ({\widecheck{\mathfrak{S}}}\smallsetminus \overline{ \mathfrak{S}})$ with $\chi\in U$.  We adopt the same notation as in (\ref{thirdcoveringpropertySbar0}) (\ref{thirdcoveringpropertySbar0bis}) (\ref{thirdcoveringpropertySbarviic}) (\ref{thirdcoveringpropertySbarviib}) for $\gamma$ and for $\varphi_{\alpha}$ and $\psi_{\alpha}$ with $\alpha$ any element of $\Theta^{\overline{ \mathfrak{S}}}(\chi)$. \\ 
Let us now assume that $U$ is of {\em consistent type 1}, i.e. $\existunique \delta\in \Theta^{\overline{ \mathfrak{S}}}(\chi)$ such that $\forall \sigma\in U$ we have $(\sigma\sqcap_{{}_{\overline{ \mathfrak{S}}}}\chi)=\delta$.  Let us fix this element $\delta$. \\
Then,  we are in one of the two following cases \\
(1) $\delta\in \Theta^{\overline{ \mathfrak{S}}}(\chi)\smallsetminus \{\gamma\}$. Then, there exists $\chi'\in {\widecheck{\mathfrak{S}}}\smallsetminus {\overline{\mathfrak{S}}}$ such that
\begin{eqnarray}
&&\hspace{-0.5cm}U =  \{\chi,\varphi_\delta,\psi_\delta,\chi'\,\}\label{orthocompletestructure0bis}\\
&&\varphi_\delta,\psi_\delta\in\overline{ \mathfrak{S}}{}^{{}^{pure}}  \;\textit{\rm and}\; \varphi_\delta,\psi_\delta \sqcoversupset_{{}_{\overline{ \mathfrak{S}}}}\delta\;\;\textit{\rm and}\;\;\varphi_\delta\not=\psi_\delta\label{orthocompletestructure1bis}\;\;\;\;\;\;\;\;\;\;\;\;\;\;\;\;\;\;\;\;\;\;\;\;\\
&&\varphi_{\delta}\underline{\perp}\psi_{\delta},\;\;\varphi_{\delta} \not\!\!\!\underline{\perp}\chi,\;\;\psi_{\delta} \not\!\!\!\underline{\perp}\chi,\;\;\varphi_{\delta} \not\!\!\!\underline{\perp}\chi',\;\;\psi_{\delta} \not\!\!\!\underline{\perp}\chi',\;\;\chi\underline{\perp}\chi'\label{orthocompletestructure3bis}
\end{eqnarray}
(2) $\delta=\gamma$. Then, there exists $\chi'\in {\widecheck{\mathfrak{S}}}\smallsetminus {\overline{\mathfrak{S}}}$ such that
\begin{eqnarray}
&&\hspace{-0.5cm}U =  \{\chi,\varphi_\gamma,\psi_\gamma,\chi'\,\}\label{orthocompletestructure0ter}\\
&&\varphi_\gamma,\psi_\gamma\in\overline{ \mathfrak{S}}{}^{{}^{pure}}  \;\textit{\rm and}\; \varphi_\gamma,\psi_\gamma \sqcoversupset_{{}_{\overline{ \mathfrak{S}}}}\gamma\;\;\textit{\rm and}\;\;\varphi_\gamma\not=\psi_\gamma\label{orthocompletestructure1ter}\;\;\;\;\;\;\;\;\;\;\;\;\;\;\;\;\;\;\;\;\;\;\;\;\\
&&\varphi_{\gamma}\not\!\!\!\underline{\perp}\psi_{\gamma},\;\;\varphi_{\gamma} \underline{\perp}\chi,\;\;\psi_{\gamma} \not\!\!\!\underline{\perp}\chi,\;\;\varphi_{\gamma} \not\!\!\!\underline{\perp}\chi',\;\;\psi_{\gamma} \underline{\perp}\chi',\;\;\chi\not\!\!\!\underline{\perp}\chi'\label{orthocompletestructure3ter}
\end{eqnarray}
\end{lemma}
\begin{proof}
We also adopt the same notation as in (\ref{thirdcoveringpropertySbar0}) (\ref{thirdcoveringpropertySbar0bis}) (\ref{thirdcoveringpropertySbarviic}). Without the constraint imposed by the Lemma \ref{nondegenerateallsigmapure}, and for the maximality of $U$,  we can consider the existence of an element $\chi'\in  {\widecheck{\mathfrak{S}}}\smallsetminus {\overline{\mathfrak{S}}}$ distinct from $\chi$ such that $\chi'\in U$. In order to guaranty that $U$ is orthogonality complete, we have to require new orthogonality conditions.\\
In the first configuration, i.e. $\delta\in \Theta^{\overline{ \mathfrak{S}}}(\chi)\smallsetminus \{\gamma\}$. We choose $\chi':=\delta \sqcup_{{}_{\mathfrak{S}}}\gamma^\star$. Such an element does effectively exist in ${\widecheck{\mathfrak{S}}}\smallsetminus {\overline{\mathfrak{S}}}$ because of property (\ref{thirdcoveringpropertySbarvi}). The announced orthogonality relations are easy to check.\\
In the second configuration, i.e. $\delta=\gamma$. We choose $\chi':=\psi_\gamma^\star \sqcup_{{}_{\mathfrak{S}}}\gamma$. Such an element does effectively exist in ${\widecheck{\mathfrak{S}}}\smallsetminus {\overline{\mathfrak{S}}}$. The announced orthogonality relations are easy to check.
\end{proof}

\begin{definition}
We will denote by 
${ \mathfrak{E}}_{\underline{\perp}}^{\check{\mathfrak{S}}}$ the following consistency cover of ${ \mathfrak{G}}^{\widecheck{\mathfrak{S}}}$ (Note that we have a cover of ${ \mathfrak{G}}^{\widecheck{\mathfrak{S}}}$ and not ${ \mathfrak{G}}^{\check{\mathfrak{S}}}$ because of Lemma \ref{natureorthocompleteU}) :
\begin{eqnarray}
{ \mathfrak{E}}_{\underline{\perp}}^{\check{\mathfrak{S}}} &:=& \{\,U \;\vert\;\exists V\in { \mathfrak{E}}^{\check{\mathfrak{S}}}, U\subseteq V,  \;\textit{\rm and $U$ is orthogonally complete} \,\}.
\end{eqnarray}
\end{definition}

\begin{lemma}\label{caseanalysisUSpure}
Let us consider $U\in { \mathfrak{E}}_{\underline{\perp}}^{\check{\mathfrak{S}}}$ such that $U\subseteq \overline{ \mathfrak{S}}{}^{{}^{pure}}$.  Let us moreover consider $\lambda,\sigma_1,\sigma_2,\sigma_3,\sigma_4\in U$, all distinct, such that ${\mathfrak{r}}^U(\lambda,\sigma_1,\sigma_2)$ and ${\mathfrak{r}}^U(\lambda,\sigma_3,\sigma_4)$ and $\forall\mu,\nu,\rho\in \{\sigma_1,\sigma_2,\sigma_3,\sigma_4\}, \;\neg\; {\mathfrak{r}}^U(\mu,\nu,\rho)$.  Endly, we will suppose that the plane to which $\lambda,\sigma_1,\sigma_2,\sigma_3,\sigma_4$ belong is not a starred plane.\\
Then, there exists $\mu\in { \mathfrak{G}}^{\widecheck{\mathfrak{S}}}\smallsetminus \overline{ \mathfrak{S}}$ and there exists $U'\in { \mathfrak{E}}_{\underline{\perp}}^{\check{\mathfrak{S}}}$ with $\mu, \sigma_1,\sigma_2,\sigma_3,\sigma_4\in U'$ such that ${\mathfrak{r}}^{U'}(\mu,\sigma_1,\sigma_3)$ and ${\mathfrak{r}}^{U'}(\mu,\sigma_2,\sigma_4)$, and there exists $\nu\in { \mathfrak{G}}^{\widecheck{\mathfrak{S}}}\smallsetminus \overline{ \mathfrak{S}}$ and there exists $U'\in { \mathfrak{E}}_{\underline{\perp}}^{\check{\mathfrak{S}}}$ with $\nu, \sigma_1,\sigma_2,\sigma_3,\sigma_4\in U'$ such that ${\mathfrak{r}}^{U'}(\nu,\sigma_2,\sigma_3)$ and ${\mathfrak{r}}^{U'}(\nu,\sigma_1,\sigma_4)$.
\end{lemma}
\begin{proof}
From Theorem \ref{zetalambda}, we know that the conditions ${\mathfrak{r}}^U(\lambda,\sigma_1,\sigma_2)$ and ${\mathfrak{r}}^U(\lambda,\sigma_3,\sigma_4)$ and $\forall\mu,\nu,\rho\in \{\sigma_1,\sigma_2,\sigma_3,\sigma_4\}, \;\neg\; {\mathfrak{r}}^U(\mu,\nu,\rho)$ with $\lambda,\sigma_1,\sigma_2,\sigma_3,\sigma_4\in U\subseteq \overline{ \mathfrak{S}}{}^{{}^{pure}}$ imply $\exists I\subseteq\{1,\cdots,N\}, Card(I)=N-2\;\vert\; \forall i\in I, \forall \sigma,\kappa\in \{\lambda,\sigma_1,\sigma_2,\sigma_3,\sigma_4\},\;\;\zeta^{{}^{{\mathfrak{S}}_{A_1}\cdots{\mathfrak{S}}_{A_N}}}_{(i)}(\sigma) = \zeta^{{}^{{\mathfrak{S}}_{A_1}\cdots{\mathfrak{S}}_{A_N}}}_{(i)}(\kappa)$.  Let us now denote $j$ and $k$ such that $\{j,k\}=\{1,\cdots,N\}\smallsetminus I$ and $j<k$ and let us now adopt the notations of Definition \ref{definUpsilonDelta}.\\
The condition (\ref{orthogcompletecondition1}) imposed on $U$ by the requirement of orthogonal completeness ensures that $\lambda,\sigma_1,\sigma_2,\sigma_3,\sigma_4$ can be written in terms of four elements $\alpha,\alpha'\in \overline{ \mathfrak{S}}{}_A^{{}^{pure}},\beta,\beta'\in \overline{ \mathfrak{S}}{}_B^{{}^{pure}}$ with $\alpha'\not=\alpha,\alpha^\star$ and $\beta'\not=\beta,\beta^\star$, in one of the following configurations
\begin{eqnarray}
&&\hspace{-1.5cm}(i)\;\;\; \Upsilon^{(j,k)}_{\lambda}:=\alpha\widetilde{\otimes}\beta,\;\;\Upsilon^{(j,k)}_{\sigma_1}:=\alpha'\widetilde{\otimes}\beta,\;\;\Upsilon^{(j,k)}_{\sigma_2}:=\alpha^\star\widetilde{\otimes}\beta,\;\;\Upsilon^{(j,k)}_{\sigma_3}:=\alpha\widetilde{\otimes}\beta',\;\;\Upsilon_{\sigma_4}:=\alpha\widetilde{\otimes}\beta^\star\;\;\;\;\;\;\;\;\;\;\;\;\;\;\;\;\;\;\\
&&\hspace{-1.5cm}(ii)\;\;\; \Upsilon^{(j,k)}_{\lambda}:=\alpha'\widetilde{\otimes}\beta,\;\;\Upsilon^{(j,k)}_{\sigma_1}:=\alpha\widetilde{\otimes}\beta,\;\;\Upsilon^{(j,k)}_{\sigma_2}:=\alpha^\star\widetilde{\otimes}\beta,\;\;\Upsilon^{(j,k)}_{\sigma_3}:=\alpha'\widetilde{\otimes}\beta',\;\;\Upsilon^{(j,k)}_{\sigma_4}:=\alpha'\widetilde{\otimes}\beta^\star\;\;\;\;\;\;\;\;\;\;\;\;\;\;\;\;\;\;\\
&&\hspace{-1.5cm}(iii)\;\;\; \Upsilon^{(j,k)}_{\lambda}:=\alpha\widetilde{\otimes}\beta',\;\;\Upsilon^{(j,k)}_{\sigma_1}:=\alpha'\widetilde{\otimes}\beta',\;\;\Upsilon^{(j,k)}_{\sigma_2}:=\alpha^\star\widetilde{\otimes}\beta',\;\;\Upsilon^{(j,k)}_{\sigma_3}:=\alpha\widetilde{\otimes}\beta,\;\;\Upsilon^{(j,k)}_{\sigma_4}:=\alpha\widetilde{\otimes}\beta^\star\;\;\;\;\;\;\;\;\;\;\;\;\;\;\;\;\;\;\\
&&\hspace{-1.5cm}(iv)\;\;\; \Upsilon_{\lambda}:=\alpha'\widetilde{\otimes}\beta',\;\;\Upsilon_{\sigma_1}:=\alpha\widetilde{\otimes}\beta',\;\;\Upsilon_{\sigma_2}:=\alpha^\star\widetilde{\otimes}\beta',\;\;\Upsilon_{\sigma_3}:=\alpha'\widetilde{\otimes}\beta,\;\;\Upsilon_{\sigma_4}:=\alpha'\widetilde{\otimes}\beta^\star.\;\;\;\;\;\;\;\;\;\;\;\;\;\;\;\;\;\;
\end{eqnarray}
The condition (\ref{orthogcompletecondition2}) imposed on $U$ by the requirement of orthogonal completeness ensures that the configuration $(iv)$ is excluded.\\
Endly, the condition requiring that the plane to which $\lambda,\sigma_1,\sigma_2,\sigma_3,\sigma_4$ belong is not a starred plane ensures that the configuration $(i)$ is excluded.\\
On another part we note the following properties.\\
In the configuration $(ii)$ 
\begin{eqnarray}
&&(\sigma_1\sqcap_{{}_{\overline{ \mathfrak{S}}}}\sigma_3)\sqcup_{{}_{{ \mathfrak{S}}}}(\sigma_2\sqcap_{{}_{\overline{ \mathfrak{S}}}}\sigma_4)=(\Delta^{(j,k);\sigma_1}_{\alpha\widetilde{\otimes}\beta})^\star \sqcup_{{}_{{ \mathfrak{S}}}} (\Delta^{(j,k);\sigma_1}_{\alpha\widetilde{\otimes}\beta} \sqcap_{{}_{\overline{ \mathfrak{S}}}} \Delta^{(j,k);\sigma_1}_{\alpha'\widetilde{\otimes}\beta'}) \;\in { \mathfrak{G}}^{\widecheck{\mathfrak{S}}}\smallsetminus \overline{ \mathfrak{S}}\;\;\;\;\;\;\;\;\;\;\;\;\;\;\;\;\;\;\\
&&(\sigma_1\sqcap_{{}_{\overline{ \mathfrak{S}}}}\sigma_4)\sqcup_{{}_{{ \mathfrak{S}}}}(\sigma_2\sqcap_{{}_{\overline{ \mathfrak{S}}}}\sigma_3)=(\Delta^{(j,k);\sigma_1}_{\alpha^\star\widetilde{\otimes}\beta})^\star \sqcup_{{}_{{ \mathfrak{S}}}} (\Delta^{(j,k);\sigma_1}_{\alpha^\star\widetilde{\otimes}\beta} \sqcap_{{}_{\overline{ \mathfrak{S}}}} \Delta^{(j,k);\sigma_1}_{\alpha'\widetilde{\otimes}\beta'}) \;\in { \mathfrak{G}}^{\widecheck{\mathfrak{S}}}\smallsetminus \overline{ \mathfrak{S}}.\;\;\;\;\;\;\;\;\;\;\;\;\;\;\;\;\;\;
\end{eqnarray}
In the configuration $(iii)$ 
\begin{eqnarray}
&&(\sigma_1\sqcap_{{}_{\overline{ \mathfrak{S}}}}\sigma_3)\sqcup_{{}_{{ \mathfrak{S}}}}(\sigma_2\sqcap_{{}_{\overline{ \mathfrak{S}}}}\sigma_4)=(\Delta^{(j,k);\sigma_1}_{\alpha\widetilde{\otimes}\beta})^\star \sqcup_{{}_{{ \mathfrak{S}}}} (\Delta^{(j,k);\sigma_1}_{\alpha\widetilde{\otimes}\beta} \sqcap_{{}_{\overline{ \mathfrak{S}}}} \Delta^{(j,k);\sigma_1}_{\alpha'\widetilde{\otimes}\beta'}) \;\in { \mathfrak{G}}^{\widecheck{\mathfrak{S}}}\smallsetminus \overline{ \mathfrak{S}}\;\;\;\;\;\;\;\;\;\;\;\;\;\;\;\;\;\;\\
&&(\sigma_1\sqcap_{{}_{\overline{ \mathfrak{S}}}}\sigma_4)\sqcup_{{}_{{ \mathfrak{S}}}}(\sigma_2\sqcap_{{}_{\overline{ \mathfrak{S}}}}\sigma_3)=(\Delta^{(j,k);\sigma_1}_{\alpha\widetilde{\otimes}\beta^\star})^\star \sqcup_{{}_{{ \mathfrak{S}}}} (\Delta^{(j,k);\sigma_1}_{\alpha\widetilde{\otimes}\beta^\star} \sqcap_{{}_{\overline{ \mathfrak{S}}}} \Delta^{(j,k);\sigma_1}_{\alpha'\widetilde{\otimes}\beta'}) \;\in { \mathfrak{G}}^{\widecheck{\mathfrak{S}}}\smallsetminus \overline{ \mathfrak{S}}.\;\;\;\;\;\;\;\;\;\;\;\;\;\;\;\;\;\;
\end{eqnarray}
This concludes the proof.
\end{proof}

\begin{remark}\label{remarkstarredplanes}
Let us formulate a simple remark justifying the introduction of starred planes.\\
In the configuration $(i)$ we have the following problem
\begin{eqnarray}
&&\hspace{-1cm}\chi:=(\sigma_1\sqcap_{{}_{\overline{ \mathfrak{S}}}}\sigma_3)\sqcup_{{}_{{ \mathfrak{S}}}}(\sigma_2\sqcap_{{}_{\overline{ \mathfrak{S}}}}\sigma_4)= (\Delta^{(j,k);\sigma_1}_{\alpha^\star\widetilde{\otimes}\beta^\star})^\star \sqcup_{{}_{{ \mathfrak{S}}}} (\Delta^{(j,k);\sigma_1}_{\alpha^\star\widetilde{\otimes}\beta'} \sqcap_{{}_{\overline{ \mathfrak{S}}}} \Delta^{(j,k);\sigma_1}_{\alpha'\widetilde{\otimes}\beta^\star})\;\in { \mathfrak{G}}^{\check{\mathfrak{S}}}\smallsetminus { \mathfrak{G}}^{\widecheck{\mathfrak{S}}}.\label{casproblematique}\;\;\;\;\;\;\;\;\;\;\;\;\;\;\;\;\;\;
\end{eqnarray}
We could still hope to have a property of the type $\forall \sigma\in \overline{\mathfrak{S}}{}^{{}^{pure}}\;\vert\; \sigma\asymp \chi,\;\; \exists \lambda\in {\mathfrak{S}}\;\vert\; \chi\sqcap_{{}_{{ \mathfrak{S}}}}\sigma \sqcoversubset_{{}_{{ \mathfrak{S}}}} \lambda\;\textit{\rm and}\; \chi\underline{\perp}\lambda$. In other words, we could still hope to have $\forall \mu\in \Theta^{\overline{ \mathfrak{S}}}(\chi), \exists \nu\in \Theta^{\overline{ \mathfrak{S}}}(\chi)$ such that $\mu^\star \sqcup_{{}_{{ \mathfrak{S}}}} \nu\in { \mathfrak{S}}$. But this point is false.  Indeed, we have for
\begin{eqnarray}
\Theta^{\overline{ \mathfrak{S}}}(\chi)=\left\{ (\alpha\widetilde{\otimes}\beta')\sqcap_{{}_{{{{ \overline{\mathfrak{S}}}}}}}\!\!\!(\alpha'\widetilde{\otimes}\beta)
\,,\,
(\alpha \widetilde{\otimes}\beta^\star)\sqcap_{{}_{{{{ \overline{\mathfrak{S}}}}}}} (\alpha^\star\widetilde{\otimes} \beta)\,,\,
(\alpha' \widetilde{\otimes}\beta^\star)\sqcap_{{}_{{{{ \overline{\mathfrak{S}}}}}}} (\alpha^\star\widetilde{\otimes} \beta')
\right\}\;\;\;\;\;\;\;\;\;
\end{eqnarray}
the following result
\begin{eqnarray}
&&cl_c^{\overline{ \mathfrak{S}}}\left\{ 
((\alpha\widetilde{\otimes}\beta')\sqcap_{{}_{{{{ \overline{\mathfrak{S}}}}}}}\!\!\!(\alpha'\widetilde{\otimes}\beta))^\star\,,\, ((\alpha \widetilde{\otimes}\beta^\star)\sqcap_{{}_{{{{ \overline{\mathfrak{S}}}}}}} (\alpha^\star\widetilde{\otimes} \beta))\right\}=\nonumber\\
&&= 
cl_c^{\overline{ \mathfrak{S}}}\left\{ 
 ((\alpha^\star\widetilde{\otimes}\bot)\sqcap_{{}_{{{{ \overline{\mathfrak{S}}}}}}}\!\!\!(\bot\widetilde{\otimes}\beta'{}^\star))
 \,,\,
 ((\bot\widetilde{\otimes}\beta^\star)\sqcap_{{}_{{{{ \overline{\mathfrak{S}}}}}}}\!\!\!(\alpha'{}^\star\widetilde{\otimes}\bot))
 \,,\,((\alpha \widetilde{\otimes}\beta^\star)\sqcap_{{}_{{{{ \overline{\mathfrak{S}}}}}}} (\alpha^\star\widetilde{\otimes} \beta))
 \right\}=\nonumber\\
 &&= 
cl_c^{\overline{ \mathfrak{S}}}\left\{ 
((\alpha^\star\widetilde{\otimes}\bot)\sqcap_{{}_{{{{ \overline{\mathfrak{S}}}}}}}\!\!\!(\bot\widetilde{\otimes}\beta'{}^\star))
\,,\,
(\alpha \widetilde{\otimes}\beta^\star)
\right\}
\end{eqnarray}
but we have also easily
\begin{eqnarray}
\left\{(\alpha^\star \widetilde{\otimes}\beta^\star)\,,\,(\alpha \widetilde{\otimes}\beta^\star)
\right\}
\sqsubseteq
cl_c^{\overline{ \mathfrak{S}}}\left\{ 
((\alpha^\star\widetilde{\otimes}\bot)\sqcap_{{}_{{{{ \overline{\mathfrak{S}}}}}}}\!\!\!(\bot\widetilde{\otimes}\beta'{}^\star))
\,,\,
(\alpha \widetilde{\otimes}\beta^\star)
\right\}
\end{eqnarray}
and then, as a consequence,
\begin{eqnarray}
cl_c^{\overline{ \mathfrak{S}}}\left\{ 
((\alpha\widetilde{\otimes}\beta')\sqcap_{{}_{{{{ \overline{\mathfrak{S}}}}}}}\!\!\!(\alpha'\widetilde{\otimes}\beta))^\star\,,\, ((\alpha \widetilde{\otimes}\beta^\star)\sqcap_{{}_{{{{ \overline{\mathfrak{S}}}}}}} (\alpha^\star\widetilde{\otimes} \beta))\right\}&\notin& {\mathcal{Q}}_c(\overline{ \mathfrak{S}}).
\end{eqnarray}
\end{remark}

\begin{lemma}\label{firstlemmawidecheck}
We will consider $\sigma_1,\sigma_2,\sigma_3,\sigma_4\in {\overline{\mathfrak{S}}}{}^{{}^{pure}}$ and $\lambda \in { \mathfrak{G}}^{\widecheck{\mathfrak{S}}}\smallsetminus {\overline{\mathfrak{S}}}$ and we will assume the existence of $U\in { \mathfrak{E}}_{\underline{\perp}}^{\check{\mathfrak{S}}}$ with $\lambda, \sigma_1,\sigma_2,\sigma_3,\sigma_4\in U$ and such that ${\mathfrak{r}}^U(\lambda,\sigma_1,\sigma_2)$ and ${\mathfrak{r}}^U(\lambda,\sigma_3,\sigma_4)$ and $\forall\omega,\nu,\rho\in \{\sigma_1,\sigma_2,\sigma_3,\sigma_4\}, \;\neg\; {\mathfrak{r}}^U(\omega,\nu,\rho)$. Endly, we will assume that the plane to which $\sigma_1,\sigma_2,\sigma_3,\sigma_4$ belong is not a starred plane.\\ 
Then, there exists $\mu\in {\overline{\mathfrak{S}}}{}^{{}^{pure}}$ and there exists $U'\in { \mathfrak{E}}_{\underline{\perp}}^{\check{\mathfrak{S}}}$ with $\mu, \sigma_1,\sigma_2,\sigma_3,\sigma_4\in U'$ such that ${\mathfrak{r}}^{U'}(\mu,\sigma_1,\sigma_3)$ and ${\mathfrak{r}}^{U'}(\mu,\sigma_2,\sigma_4)$ and there exists $\chi\in { \mathfrak{G}}^{\widecheck{\mathfrak{S}}}\smallsetminus {\overline{\mathfrak{S}}}$ and there exists $U''\in { \mathfrak{E}}_{\underline{\perp}}^{\check{\mathfrak{S}}}$ with $\chi, \sigma_1,\sigma_2,\sigma_3,\sigma_4\in U''$ such that ${\mathfrak{r}}^{U''}(\chi,\sigma_2,\sigma_3)$ and ${\mathfrak{r}}^{U''}(\chi,\sigma_1,\sigma_4)$, 
\underline{or} there exists $\mu\in {\overline{\mathfrak{S}}}{}^{{}^{pure}}$ and there exists $U'\in { \mathfrak{E}}_{\underline{\perp}}^{\check{\mathfrak{S}}}$ with $\mu, \sigma_1,\sigma_2,\sigma_3,\sigma_4\in U'$ such that ${\mathfrak{r}}^{U'}(\mu,\sigma_2,\sigma_3)$ and ${\mathfrak{r}}^{U'}(\mu,\sigma_1,\sigma_4)$ and there exists $\chi\in { \mathfrak{G}}^{\widecheck{\mathfrak{S}}}\smallsetminus {\overline{\mathfrak{S}}}$ and there exists $U''\in { \mathfrak{E}}_{\underline{\perp}}^{\check{\mathfrak{S}}}$ with $\chi, \sigma_1,\sigma_2,\sigma_3,\sigma_4\in U''$ such that ${\mathfrak{r}}^{U''}(\chi,\sigma_1,\sigma_3)$ and ${\mathfrak{r}}^{U''}(\chi,\sigma_2,\sigma_4)$.
\end{lemma}
\begin{proof}
As long as $\lambda\in ({\check{\mathfrak{S}}}\smallsetminus \overline{ \mathfrak{S}})$, we have the results of Lemma \ref{natureorthocompleteU} 
and we can adopt the notations of this lemma.  

We note that the conditions ${\mathfrak{r}}^U(\lambda,\sigma_1,\sigma_2)$ and ${\mathfrak{r}}^U(\lambda,\sigma_3,\sigma_4)$ and $\forall\omega,\nu,\rho\in \{\sigma_1,\sigma_2,\sigma_3,\sigma_4\}, \;\neg\; {\mathfrak{r}}^U(\omega,\nu,\rho)$ impose the existence of $\delta,\epsilon\in \Theta^{\overline{\mathfrak{S}}}(\lambda)$ such that $\epsilon\not=\delta$ and $(\sigma_1\sqcap_{{}_{\overline{\mathfrak{S}}}}\sigma_2)=\epsilon$ and $(\sigma_3\sqcap_{{}_{\overline{\mathfrak{S}}}}\sigma_4)=\delta$.\\

Let us now distinguish the different cases exhibited from the condition of orthogonal completeness.  We denote as usual $\gamma$ the unique element of $\Theta^{\overline{ \mathfrak{S}}}(\lambda)$ such that $\lambda=\varphi_\gamma^\star\sqcup\gamma$ and $\gamma:=(\varphi_\gamma\sqcap_{{}_{{\overline{ \mathfrak{S}}}}}\psi_\gamma)\sqcoversubset_{{}_{{\overline{ \mathfrak{S}}}}}\varphi_\gamma,\psi_\gamma$.\\

\noindent If $\sigma_2\sqsupseteq_{{}_{\overline{ \mathfrak{S}}}} (\sigma_1\sqcap_{{}_{\overline{ \mathfrak{S}}}}\sigma_3)^\star$ we are then in one of the two following configurations \\ 
(i) $\gamma=\delta$ and then $\varphi_\gamma=\sigma_3$,  $\psi_\gamma=\sigma_4$ and $\{\varphi_\epsilon,\psi_\epsilon\}=\{\sigma_1,\sigma_2\}$, \\
(ii) $\gamma\not=\delta,\gamma\not=\epsilon$ and then $\varphi_\delta=\sigma_3$,$\psi_\delta=\sigma_4$,$\varphi_\epsilon=\sigma_2$, $\psi_\epsilon=\sigma_1$.\\

In the configuration (i) and if $\sigma_1=\varphi_\epsilon$ , using (\ref{orthocompletestructure6}), we deduce that there exists $\mu:=\mu_{\epsilon\gamma}\in {\overline{\mathfrak{S}}}{}^{{}^{pure}}$ such that $(\sigma_1\sqcap_{{}_{\overline{ \mathfrak{S}}}}\sigma_3)=(\varphi_{\gamma}\sqcap_{{}_{\overline{ \mathfrak{S}}}}\varphi_{\epsilon})\sqcoversubset_{{}_{\overline{ \mathfrak{S}}}}\mu$ and 
$(\sigma_2\sqcap_{{}_{\overline{ \mathfrak{S}}}}\sigma_4)=(\psi_{\gamma}\sqcap_{{}_{\overline{ \mathfrak{S}}}}\psi_{\epsilon})\sqcoversubset_{{}_{\overline{ \mathfrak{S}}}}\mu$. Moreover, we have $\mu_{\alpha\gamma}\underline{\perp}\psi_{\epsilon}$, i.e. $\mu\underline{\perp}\sigma_2$, and $\varphi_\gamma\underline{\perp}\varphi_\epsilon$, i.e. $\sigma_1\underline{\perp}\sigma_3$.  As a result, there exists $U'\in { \mathfrak{E}}_{\underline{\perp}}^{\check{\mathfrak{S}}}$ with $\mu, \sigma_1,\sigma_2,\sigma_3,\sigma_4\in U'$ such that ${\mathfrak{r}}^{U'}(\mu,\sigma_1,\sigma_3)$ and ${\mathfrak{r}}^{U'}(\mu,\sigma_2,\sigma_4)$. Now we can use Lemma \ref{caseanalysisUSpure} to deduce that there exists $U''\in { \mathfrak{E}}_{\underline{\perp}}^{\check{\mathfrak{S}}}$ and $\chi\in { \mathfrak{G}}^{\widecheck{\mathfrak{S}}}\smallsetminus {\overline{\mathfrak{S}}}$ with $\chi, \sigma_1,\sigma_2,\sigma_3,\sigma_4\in U''$ such that ${\mathfrak{r}}^{U''}(\chi,\sigma_2,\sigma_3)$ and ${\mathfrak{r}}^{U''}(\chi,\sigma_1,\sigma_4)$.\\

In the configuration (i) and if $\sigma_1=\psi_\epsilon$ , using (\ref{orthocompletestructure6}), we deduce analogously that there exists $\mu\in {\overline{\mathfrak{S}}}{}^{{}^{pure}}$ and $U'\in { \mathfrak{E}}_{\underline{\perp}}^{\check{\mathfrak{S}}}$ with $\mu, \sigma_1,\sigma_2,\sigma_3,\sigma_4\in U'$ such that ${\mathfrak{r}}^{U'}(\mu,\sigma_2,\sigma_3)$ and ${\mathfrak{r}}^{U'}(\mu,\sigma_1,\sigma_4)$. Now we can use Lemma \ref{caseanalysisUSpure} to deduce that there exists $U''\in { \mathfrak{E}}_{\underline{\perp}}^{\check{\mathfrak{S}}}$ and $\chi\in { \mathfrak{G}}^{\widecheck{\mathfrak{S}}}\smallsetminus {\overline{\mathfrak{S}}}$ with $\chi, \sigma_1,\sigma_2,\sigma_3,\sigma_4\in U''$ such that ${\mathfrak{r}}^{U''}(\chi,\sigma_1,\sigma_3)$ and ${\mathfrak{r}}^{U''}(\chi,\sigma_2,\sigma_4)$.\\

In the configuration (ii), using (\ref{orthocompletestructure5}), we deduce that there exists $\mu:=\lambda_{\epsilon\delta}\in {\overline{\mathfrak{S}}}{}^{{}^{pure}}$ such that $(\sigma_1\sqcap_{{}_{\overline{ \mathfrak{S}}}}\sigma_3)=(\varphi_{\delta}\sqcap_{{}_{\overline{ \mathfrak{S}}}}\psi_{\epsilon})\sqcoversubset_{{}_{\overline{ \mathfrak{S}}}}\mu$ and 
$(\sigma_2\sqcap_{{}_{\overline{ \mathfrak{S}}}}\sigma_4)=(\psi_{\delta}\sqcap_{{}_{\overline{ \mathfrak{S}}}}\varphi_{\epsilon})\sqcoversubset_{{}_{\overline{ \mathfrak{S}}}}\mu$.  Moreover,  we have $\mu=\lambda_{\epsilon\delta}\underline{\perp}\varphi_{\epsilon}=\sigma_2$ and $\mu=\lambda_{\epsilon\delta}\underline{\perp}\varphi_{\delta}=\sigma_3$. In other words, there exists $U'\in { \mathfrak{E}}_{\underline{\perp}}^{\check{\mathfrak{S}}}$ with $\mu, \sigma_1,\sigma_2,\sigma_3,\sigma_4\in U'$ such that ${\mathfrak{r}}^{U'}(\mu,\sigma_1,\sigma_3)$ and ${\mathfrak{r}}^{U'}(\mu,\sigma_2,\sigma_4)$. 
 Now we can use Lemma \ref{caseanalysisUSpure} to deduce that there exists $U''\in { \mathfrak{E}}_{\underline{\perp}}^{\check{\mathfrak{S}}}$ and $\chi\in { \mathfrak{G}}^{\widecheck{\mathfrak{S}}}\smallsetminus {\overline{\mathfrak{S}}}$ with $\chi, \sigma_1,\sigma_2,\sigma_3,\sigma_4\in U''$ such that ${\mathfrak{r}}^{U''}(\chi,\sigma_2,\sigma_3)$ and ${\mathfrak{r}}^{U''}(\chi,\sigma_1,\sigma_4)$.\\

\noindent If we have $\sigma_3\sqsupseteq_{{}_{\overline{ \mathfrak{S}}}} (\sigma_2\sqcap_{{}_{\overline{ \mathfrak{S}}}}\sigma_4)^\star$ or $\sigma_1\sqsupseteq_{{}_{\overline{ \mathfrak{S}}}} (\sigma_2\sqcap_{{}_{\overline{ \mathfrak{S}}}}\sigma_4)^\star$ or $\sigma_4\sqsupseteq_{{}_{\overline{ \mathfrak{S}}}} (\sigma_1\sqcap_{{}_{\overline{ \mathfrak{S}}}}\sigma_3)^\star$, and using the natural symmetries of the assumption ($\sigma_1 \leftrightarrow \sigma_2$, $\sigma_3 \leftrightarrow \sigma_4$, $(\sigma_1,\sigma_2) \leftrightarrow (\sigma_3,\sigma_4)$), we obtain the same result.\\

\noindent If we have $\sigma_1\sqsupseteq_{{}_{\overline{ \mathfrak{S}}}} (\sigma_3\sqcap_{{}_{\overline{ \mathfrak{S}}}}\sigma_4)^\star$, the situation is even simpler. Indeed, we have necessarily $\gamma=\delta$ and $\varphi_\gamma=\sigma_1$,$\psi_\gamma=\sigma_2$,$\{\varphi_\epsilon, \psi_\epsilon\}=\{\sigma_3,\sigma_4\}$.\\
If $\sigma_3=\varphi_\epsilon$ , using (\ref{orthocompletestructure6}), we deduce that there exists $\mu:=\mu_{\epsilon\gamma}\in {\overline{\mathfrak{S}}}{}^{{}^{pure}}$ such that $(\sigma_1\sqcap_{{}_{\overline{ \mathfrak{S}}}}\sigma_3)=(\varphi_{\gamma}\sqcap_{{}_{\overline{ \mathfrak{S}}}}\varphi_{\epsilon})\sqcoversubset_{{}_{\overline{ \mathfrak{S}}}}\mu$ and 
$(\sigma_2\sqcap_{{}_{\overline{ \mathfrak{S}}}}\sigma_4)=(\psi_{\gamma}\sqcap_{{}_{\overline{ \mathfrak{S}}}}\psi_{\epsilon})\sqcoversubset_{{}_{\overline{ \mathfrak{S}}}}\mu$. Moreover, we have $\mu_{\alpha\gamma}\underline{\perp}\psi_{\epsilon}$, i.e. $\mu\underline{\perp}\sigma_4$, and $\varphi_\gamma\underline{\perp}\varphi_\epsilon$, i.e. $\sigma_1\underline{\perp}\sigma_3$.  As a result, there exists $U'\in { \mathfrak{E}}_{\underline{\perp}}^{\check{\mathfrak{S}}}$ with $\mu, \sigma_1,\sigma_2,\sigma_3,\sigma_4\in U'$ such that ${\mathfrak{r}}^{U'}(\mu,\sigma_1,\sigma_3)$ and ${\mathfrak{r}}^{U'}(\mu,\sigma_2,\sigma_4)$. Now we can use Lemma \ref{caseanalysisUSpure} to deduce that there exists $U''\in { \mathfrak{E}}_{\underline{\perp}}^{\check{\mathfrak{S}}}$ and $\chi\in { \mathfrak{G}}^{\widecheck{\mathfrak{S}}}\smallsetminus {\overline{\mathfrak{S}}}$ with $\chi, \sigma_1,\sigma_2,\sigma_3,\sigma_4\in U''$ such that ${\mathfrak{r}}^{U''}(\chi,\sigma_2,\sigma_3)$ and ${\mathfrak{r}}^{U''}(\chi,\sigma_1,\sigma_4)$.\\

If $\sigma_4=\varphi_\epsilon$ , using (\ref{orthocompletestructure6}), we deduce analogously that there exists $\mu\in {\overline{\mathfrak{S}}}{}^{{}^{pure}}$ and $U'\in { \mathfrak{E}}_{\underline{\perp}}^{\check{\mathfrak{S}}}$ with $\mu, \sigma_1,\sigma_2,\sigma_3,\sigma_4\in U'$ such that ${\mathfrak{r}}^{U'}(\mu,\sigma_2,\sigma_3)$ and ${\mathfrak{r}}^{U'}(\mu,\sigma_1,\sigma_4)$. Now we can use Lemma \ref{caseanalysisUSpure} to deduce that there exists $U''\in { \mathfrak{E}}_{\underline{\perp}}^{\check{\mathfrak{S}}}$ and $\chi\in { \mathfrak{G}}^{\widecheck{\mathfrak{S}}}\smallsetminus {\overline{\mathfrak{S}}}$ with $\chi, \sigma_1,\sigma_2,\sigma_3,\sigma_4\in U''$ such that ${\mathfrak{r}}^{U''}(\chi,\sigma_1,\sigma_3)$ and ${\mathfrak{r}}^{U''}(\chi,\sigma_2,\sigma_4)$.\\

\noindent If we have $\sigma_2\sqsupseteq_{{}_{\overline{ \mathfrak{S}}}} (\sigma_3\sqcap_{{}_{\overline{ \mathfrak{S}}}}\sigma_4)^\star$ or $\sigma_3\sqsupseteq_{{}_{\overline{ \mathfrak{S}}}} (\sigma_1\sqcap_{{}_{\overline{ \mathfrak{S}}}}\sigma_2)^\star$ or $\sigma_4\sqsupseteq_{{}_{\overline{ \mathfrak{S}}}} (\sigma_1\sqcap_{{}_{\overline{ \mathfrak{S}}}}\sigma_2)^\star$, and using the natural symmetries of the assumption ($\sigma_1 \leftrightarrow \sigma_2$, $\sigma_3 \leftrightarrow \sigma_4$, $(\sigma_1,\sigma_2) \leftrightarrow (\sigma_3,\sigma_4)$), we obtain the same result.\\
\end{proof}

Let us now summarize the main result of previous lemmas in the following theorem.

\begin{theorem}{\bf [restricted Veblen-Young's third geometric property]}\label{theoremwidecheck}\\
Let us consider $U\in { \mathfrak{E}}_{\underline{\perp}}^{\check{\mathfrak{S}}}$ and $\lambda, \sigma_1,\sigma_2,\sigma_3,\sigma_4\in U$ such that ${\mathfrak{r}}^U(\lambda,\sigma_1,\sigma_2)$ and ${\mathfrak{r}}^U(\lambda,\sigma_3,\sigma_4)$ and $\forall\mu,\nu,\rho\in \{\sigma_1,\sigma_2,\sigma_3,\sigma_4\}, \;\neg\; {\mathfrak{r}}^U(\mu,\nu,\rho)$.  Endly, we will assume that the plane to which $\sigma_1,\sigma_2,\sigma_3,\sigma_4$ belong is not a starred plane.\\
From Lemma \ref{nondegenerateallsigmapure}, we know that $\sigma_1,\sigma_2,\sigma_3,\sigma_4\in \overline{\mathfrak{S}}{}^{{}^{pure}}$. From Lemma \ref{natureorthocompleteU}, we know that, as long as $U$ is orthogonally complete, we have $\lambda\in {\mathfrak{G}}^{\widecheck{\mathfrak{S}}}$. Endly, from Lemma \ref{firstlemmawidecheck} and Lemma \ref{caseanalysisUSpure}, we know that there exists $U'\in { \mathfrak{E}}_{\underline{\perp}}^{\check{\mathfrak{S}}}$ with $\sigma_1,\sigma_2,\sigma_3,\sigma_4\in U'$ and $\lambda'\in U'\cap {\mathfrak{G}}^{\widecheck{\mathfrak{S}}}$ satisfying ${\mathfrak{r}}^{U'}(\lambda',\sigma_1,\sigma_3)$ and ${\mathfrak{r}}^{U'}(\lambda',\sigma_2,\sigma_4)$, and there exists $U''\in { \mathfrak{E}}_{\underline{\perp}}^{\check{\mathfrak{S}}},\sigma_1,\sigma_2,\sigma_3,\sigma_4\in U''$ and $\lambda''\in U''\cap {\mathfrak{G}}^{\widecheck{\mathfrak{S}}}$ satisfying ${\mathfrak{r}}^{U''}(\lambda'',\sigma_1,\sigma_4)$ and ${\mathfrak{r}}^{U''}(\lambda'',\sigma_2,\sigma_3)$.  
\end{theorem}

\begin{theorem}{\bf [Orthogonality properties]}\label{theoremorthoproj}\\
The restriction to ${\mathfrak{G}}^{\widecheck{\mathfrak{S}}}$ of the orthogonality relation $\underline{\perp}$ satisfies the following properties.
\begin{eqnarray}
\hspace{-1cm}\forall \alpha,\beta\in {\mathfrak{G}}^{\widecheck{\mathfrak{S}}},&& \alpha \underline{\perp}\beta \;\;\Rightarrow\;\; \alpha\not=\beta,\label{propO1}\\
\forall \alpha,\beta\in {\mathfrak{G}}^{\widecheck{\mathfrak{S}}},&& \alpha \underline{\perp}\beta \;\;\Rightarrow\;\; \beta \underline{\perp}\alpha,\label{propO2}\\
\forall U\in { \mathfrak{E}}_{\underline{\perp}}^{\check{\mathfrak{S}}}, \forall \alpha,\beta,\epsilon,\delta\in U,&& (\,\alpha\not=\beta,\;\;\alpha \underline{\perp}\epsilon,\;\;\beta \underline{\perp}\epsilon,\;\;{\mathfrak{r}}^U(\delta,\alpha,\beta)\,)\;\; \Rightarrow \;\; (\,\epsilon \underline{\perp} \delta\,),\;\;\;\;\;\;\;\;\;\;\;\;\;\label{propO3}\\
\forall \alpha,\beta\in {\mathfrak{G}}^{\widecheck{\mathfrak{S}}},\;\alpha\not=\beta,\alpha\asymp\beta,&& \exists U\in { \mathfrak{E}}_{\underline{\perp}}^{\check{\mathfrak{S}}}, \alpha,\beta\in U, \exists \epsilon\in U\;\;\vert\;\; {\mathfrak{r}}^U(\epsilon,\alpha,\beta)\;\textit{\rm and}\; \epsilon \underline{\perp} \alpha.\label{propO4}\;\;\;\;\;\;\;\;\;\;\;\;\;\;\;
\end{eqnarray}
\end{theorem}
\begin{proof}
The properties (\ref{propO1}) and (\ref{propO2}) are trivial.\\

Let us now check the property (\ref{propO3}).  As long as $\alpha \underline{\perp}\epsilon$ and $\beta \underline{\perp}\epsilon$, we have $\epsilon^\star \sqsubseteq_{{}_{{ \mathfrak{S}}}} (\alpha\sqcap_{{}_{\overline{\mathfrak{S}}}}\beta)$, but we have also $ (\alpha\sqcap_{{}_{\overline{\mathfrak{S}}}}\beta)\sqcoversubset_{{}_{{ \mathfrak{S}}}}\delta$ because ${\mathfrak{r}}^U(\delta,\alpha,\beta)$. We then have $\epsilon^\star\sqsubseteq_{{}_{{ \mathfrak{S}}}} \delta$, which means $\epsilon \underline{\perp}\delta$.\\

Let us now check the property (\ref{propO4}). We have to distinguish the different cases. \\
Let us begin with $\alpha,\beta\in \overline{\mathfrak{S}}{}^{{}^{pure}}$ with $\alpha\not=\beta$ and $\alpha\asymp\beta$. From Lemma \ref{lambdawr} and $\alpha\asymp\beta$, we know that $\epsilon:=\alpha^\star \sqcup_{{}_{{\mathfrak{S}}}}(\alpha\sqcap_{{}_{\overline{\mathfrak{S}}}}\beta)$ exists in ${\mathfrak{G}}^{\widecheck{\mathfrak{S}}}$.  Moreover, using (\ref{thirdcoveringpropertySbarii}), we have $(\alpha\sqcap_{{}_{\overline{\mathfrak{S}}}}\beta)\sqcoversubset_{{}_{{ \mathfrak{S}}}}\epsilon$, i.e. ${\mathfrak{r}}^U(\epsilon,\alpha,\beta)$. Endly, we have $\alpha^\star\sqsubseteq_{{}_{\overline{ \mathfrak{S}}}} \epsilon$,  i.e. $\epsilon \underline{\perp} \alpha$.\\
Secondly, let us consider the case where $\alpha\in \overline{\mathfrak{S}}{}^{{}^{pure}}$ and $\beta\in {\mathfrak{G}}^{\widecheck{\mathfrak{S}}}\smallsetminus \overline{\mathfrak{S}}{}^{{}^{pure}}$. The same analysis with the same expression for $\epsilon$ gives the same result.\\
Thirdly,  let us consider the case where $\beta\in \overline{\mathfrak{S}}{}^{{}^{pure}}$ and $\alpha\in {\mathfrak{G}}^{\widecheck{\mathfrak{S}}}\smallsetminus \overline{\mathfrak{S}}{}^{{}^{pure}}$.  By definition, there exists $\gamma\in \Theta^{\overline{\mathfrak{S}}}(\alpha)$ and $\varphi_\gamma,\psi_\gamma\in \overline{\mathfrak{S}}{}^{{}^{pure}}$ such that $\gamma=(\varphi_\gamma\sqcap_{{}_{\overline{\mathfrak{S}}}}\psi_\gamma)\sqcoversubset_{{}_{{ \mathfrak{S}}}}\varphi_\gamma,\psi_\gamma$ and $\alpha=\varphi_\gamma^\star \sqcup_{{}_{{\mathfrak{S}}}}\gamma$. Let us denote by $\delta:=\beta\sqcap_{{}_{\overline{\mathfrak{S}}}}\alpha \;\in \Theta^{\overline{\mathfrak{S}}}(\alpha)$. We have to distinguish between two sub-cases :\\
(i) $\delta=\gamma$. In this case, it suffice to choose $\epsilon:=\varphi_\gamma$ and we check immediately $(\alpha\sqcap_{{}_{\overline{\mathfrak{S}}}}\beta)\sqcoversubset_{{}_{{ \mathfrak{S}}}}\epsilon$ and $\epsilon^\star \sqsubseteq_{{}_{{ \mathfrak{S}}}} \alpha$.\\
(ii) $\delta\not=\gamma$. In this case, it suffice to choose $\epsilon:= \delta\sqcup_{{}_{{\mathfrak{S}}}}\gamma^\star $ which is an element of ${\mathfrak{G}}^{\widecheck{\mathfrak{S}}}$ because of property (\ref{thirdcoveringpropertySbarvi}).  We check immediately $(\alpha\sqcap_{{}_{\overline{\mathfrak{S}}}}\beta)\sqcoversubset_{{}_{{ \mathfrak{S}}}}\epsilon$ using property (\ref{thirdcoveringpropertySbarii}). We check also trivially $\gamma^\star \sqsubseteq_{{}_{{ \mathfrak{S}}}} \epsilon$ and $\gamma \sqsubseteq_{{}_{{ \mathfrak{S}}}} \alpha$, which means $\epsilon \underline{\perp} \alpha$.\\
Fourthly, let us consider the case where $\beta,\alpha\in {\mathfrak{G}}^{\widecheck{\mathfrak{S}}}\smallsetminus \overline{\mathfrak{S}}{}^{{}^{pure}}$.  The same construction as above leads to the result.
\end{proof}

The last element to check is the fact that the tensor product respects the property of linearity defined in Definition \ref{linearindeterministicspace}.

\begin{theorem}{\bf [Irreducibility]}\label{theoremirredproj}\\
The space of states ${ \mathfrak{S}}$ equipped with its real structure $(\overline{ \mathfrak{S}},\star)$ satisfies the following irreducibility property
\begin{eqnarray}
\hspace{-1.3cm}\forall \sigma,\lambda\in { \mathfrak{G}}^{\widecheck{\mathfrak{S}}}, \sigma\not=\lambda, \sigma\asymp\lambda,\;\; \exists U\in  { \mathfrak{E}}_{\underline{\perp}}^{\check{\mathfrak{S}}},\sigma,\lambda\in U, \exists \kappa \in U\;\vert\; (\,\kappa\not=\sigma\;\textit{\rm and}\; \kappa\not=\lambda\;\textit{\rm and}\; {\mathfrak{r}}^U(\kappa,\lambda,\sigma)\,).\;\;\;\;\;\;\label{propirreducibility}
\end{eqnarray}
\end{theorem}
\begin{proof}
First of all, we note that, if $\sigma\not\!\!\!\underline{\perp}\lambda$ and $\lambda\not=\sigma$ and $\sigma\asymp\lambda$,  using property (\ref{propO4}), we deduce the existence of $U\in  { \mathfrak{E}}^{\widecheck{\mathfrak{S}}}$ such that $\sigma,\lambda\in U$ and the existence of $\kappa \in U$ such that ${\mathfrak{r}}^U(\kappa,\lambda,\sigma)$ and $\sigma\underline{\perp}\kappa$. As long as $\sigma\underline{\perp}\kappa$ we deduce $\sigma\not=\kappa$, and as long as $\sigma\not\!\!\!\underline{\perp}\lambda$ we deduce $\kappa\not=\lambda$. This proves the property (\ref{propirreducibility}) when $\sigma\not\!\!\!\underline{\perp}\lambda$.\\

It remains to study the case $\sigma\underline{\perp}\lambda$, $\sigma\asymp\lambda$. Let us distinguish the different cases.\\

\noindent (1) Let us first study the case $\sigma,\lambda\in {\overline{\mathfrak{S}}}{}^{{}^{pure}}$ with $\sigma\asymp\lambda$ and $\sigma\underline{\perp}\lambda$.  We will introduce $I\subseteq\{1,\cdots,N\}$ such that $Card(I)=N-2$ and $\forall i\in I,\zeta^{{}^{{\mathfrak{S}}_{A_1}\cdots{\mathfrak{S}}_{A_N}}}_{(i)}(\sigma) = \zeta^{{}^{{\mathfrak{S}}_{A_1}\cdots{\mathfrak{S}}_{A_N}}}_{(i)}(\lambda)$. We will also denote $j$ and $k$ such that $\{j,k\}=\{1,\cdots,N\}\smallsetminus I$ and $j<k$.\\
Then, we are in one of the two following sub-cases :\\
(1A) there exists $\alpha\in \overline{ \mathfrak{S}}{}_{A_j}^{{}^{pure}}$ and $\beta,\beta'\in \overline{ \mathfrak{S}}{}_{A_k}^{{}^{pure}}$ such that $\Upsilon^{(j,k)}_{\sigma}=\alpha\widetilde{\otimes}\beta$ and $\Upsilon^{(j,k)}_{\lambda}=\alpha^\star\widetilde{\otimes}\beta'$.  If $\beta=\beta'$ then for any $\alpha'\in \overline{ \mathfrak{S}}{}_{A_j}^{{}^{pure}}$ such that $\alpha'\not=\alpha,\alpha^\star$, we can define $\kappa:=\alpha'\widetilde{\otimes}\beta$ which satisfies the announced constraints.  If $\beta\not=\beta'$, then, for any $\alpha'\in \overline{ \mathfrak{S}}{}_{A_j}^{{}^{pure}}$ such that $\alpha'\not=\alpha,\alpha^\star$ we can define ${\kappa}:=(\Delta^{(j,k);\sigma}_{\alpha\widetilde{\otimes}\beta^\star})^\star \sqcup_{{}_{{ \mathfrak{S}}}}(\Delta^{(j,k);\sigma}_{\alpha\widetilde{\otimes}\beta^\star}\sqcap_{{}_{\overline{ \mathfrak{S}}}}\Delta^{(j,k);\sigma}_{\alpha'\widetilde{\otimes}\beta'})$ which satisfies indeed $\kappa\not=\sigma$,  $\kappa\not=\lambda$ and $(\sigma\sqcap_{{}_{\overline{ \mathfrak{S}}}}\lambda)\in \Theta^{{\overline{\mathfrak{S}}}}(\kappa)$.\\
(1B) there exists $\alpha,\alpha'\in \overline{ \mathfrak{S}}{}_{A_j}^{{}^{pure}}$ and $\beta\in \overline{ \mathfrak{S}}{}_{A_k}^{{}^{pure}}$ such that $\Upsilon^{(j,k)}_{\sigma}=\alpha\widetilde{\otimes}\beta$ and $\Upsilon^{(j,k)}_{\lambda}=\alpha'\widetilde{\otimes}\beta^\star$.  Same treatment.\\

\noindent (2) Let us now study the case $\sigma\in {\overline{\mathfrak{S}}}{}^{{}^{pure}}$, $\lambda\in { \mathfrak{G}}^{\widecheck{\mathfrak{S}}}\smallsetminus {\overline{\mathfrak{S}}}{}^{{}^{pure}}$ with $\sigma\asymp\lambda$ and $\sigma\underline{\perp}\lambda$. We deduce that $\lambda=\varphi_\gamma^\star\sqcup_{{}_{{ \mathfrak{S}}}}(\varphi_\gamma\sqcap_{{}_{\overline{ \mathfrak{S}}}}\psi_\gamma)$ and $\sigma=\varphi_\gamma$ (we have adopted the same notation as in (\ref{thirdcoveringpropertySbar0}) (\ref{thirdcoveringpropertySbar0bis}) (\ref{thirdcoveringpropertySbarviic}) (\ref{thirdcoveringpropertySbarviib})). It then suffices to choose $\kappa:=\psi_\gamma$ to obtain $\kappa\not=\sigma$,  $\kappa\not=\lambda$ and $(\sigma\sqcap_{{}_{\overline{ \mathfrak{S}}}}\lambda)\in \Theta^{{\overline{\mathfrak{S}}}}(\kappa)$.\\

\noindent (3) Let us now study the case $\sigma,\lambda\in { \mathfrak{G}}^{\widecheck{\mathfrak{S}}}\smallsetminus {\overline{\mathfrak{S}}}{}^{{}^{pure}}$ with $\sigma\asymp\lambda$ and $\sigma\underline{\perp}\lambda$.  In this case, it suffices to choose $\kappa:=\psi_\alpha$ for $\{\alpha\}:=\Theta^{{\overline{\mathfrak{S}}}}(\sigma)\cap\Theta^{{\overline{\mathfrak{S}}}}(\lambda)$.\\

This concludes the proof.
\end{proof}

The result of this section appears to be outstanding in light of the traditional approach to quantum logic. We have reconstructed a structure on the tensor product that is as close as possible to irreducible Hilbert geometry, even though we did not impose anything of the sort {\em ab initio}. This is an absolutely non-trivial consequence of our construction of the tensor product and of the ontic completions.

\section{Fundamental properties of the bipartite indeterministic experiments}\label{sectionfundamental}

This last section is devoted to the study of typical quantum-like properties: the absence of broadcasting and the existence of Bell non-local states. Recall that the core criticism of the quantum logic program \cite{Pullmanova}\cite{HudsonPullmanova}\cite{Pykacz} is the impossibility of constructing entangled states exhibiting Bell non-locality within the framework of a logical approach to quantum mechanics. The outcome of this section is then one of the main achievements of our approach.

\subsection{No-broadcasting theorem}

\begin{definition}{\bf [Broadcasting]}\\
We will say that a given space of states ${ \mathfrak{S}}$ allows the broadcasting iff 
\begin{eqnarray}
\exists \Psi \in \overline{\mathfrak{C}}({ \mathfrak{S}},{ \mathfrak{S}}\widetilde{\otimes}{ \mathfrak{S}})&\vert &\forall \sigma\in \overline{\mathfrak{S}},\;(\zeta^{{{ \mathfrak{S}}}{{ \mathfrak{S}}}}_{(1)}\circ \Psi)(\sigma) = (\zeta^{{{ \mathfrak{S}}}{{ \mathfrak{S}}}}_{(2)}\circ \Psi)(\sigma) = \sigma.\;\;\;\;\;\;\;\;\;\;\;
\end{eqnarray}
\end{definition}

\begin{lemma}\label{lemmabroadcasting}
Let ${ \mathfrak{m}}_1,{ \mathfrak{m}}_2\in  \overline{\mathfrak{M}}_{{}_{{ \mathfrak{S}}}}$ be two real measurements on the space of states ${ \mathfrak{S}}$ (they are real morphisms from ${ \mathfrak{S}}$ to ${ \mathfrak{B}}$).\\
If there exists a real morphism $\Psi \in \overline{\mathfrak{C}}({ \mathfrak{S}},{ \mathfrak{S}}\widetilde{\otimes}{ \mathfrak{S}})$ such that $\forall \sigma\in \overline{\mathfrak{S}},\;(\zeta^{{{ \mathfrak{S}}}{{ \mathfrak{S}}}}_{(1)}\circ \Psi)(\sigma) = (\zeta^{{{ \mathfrak{S}}}{{ \mathfrak{S}}}}_{(2)}\circ \Psi)(\sigma) = \sigma$, then there exists a real morphism $\Psi_{({ \mathfrak{m}}_1,{ \mathfrak{m}}_2)} \in \overline{\mathfrak{C}}({ \mathfrak{S}},{ \mathfrak{B}}\widetilde{\otimes}{ \mathfrak{B}})$ such that $\forall \sigma\in \overline{\mathfrak{S}}, (\zeta^{{{ \mathfrak{B}}}{{ \mathfrak{B}}}}_{(1)}\circ \Psi_{({ \mathfrak{m}}_1,{ \mathfrak{m}}_2)})(\sigma) = { \mathfrak{m}}_1(\sigma)$ and $(\zeta^{{{ \mathfrak{B}}}{{ \mathfrak{B}}}}_{(2)}\circ \Psi_{({ \mathfrak{m}}_1,{ \mathfrak{m}}_2)})(\sigma) = { \mathfrak{m}}_2(\sigma)$.
\end{lemma}
\begin{proof}
It suffices to choose $\Psi_{({ \mathfrak{m}}_1,{ \mathfrak{m}}_2)}:= ({ \mathfrak{m}}_1\widetilde{\otimes}{ \mathfrak{m}}_2)\circ \Psi$.
\end{proof}

\begin{theorem}{\bf [No-broadcasting theorem]}\\
If the space of states ${ \mathfrak{S}}$ (equipped with its real structure $(\overline{ \mathfrak{S}},\star)$) is not deterministic, it does not allow the broadcasting.
\end{theorem}
\begin{proof}
Let us consider the space of states ${ \mathfrak{S}}$, and let us assume that it is NOT deterministic. We know, from the definition of indeterminism, that there exists $\sigma_1,\sigma_2\in  \overline{ \mathfrak{S}}{}^{{}^{pure}}$ such that $\sigma_2\not=\sigma_1$ and $\sigma_2\not\sqsupseteq_{{}_{\overline{ \mathfrak{S}}}} \sigma_1^\star$. Let us then define the following real measurements : ${ \mathfrak{m}}_1:={ \mathfrak{m}}_{{ \mathfrak{l}}_{(\sigma_1,\sigma_1^\star)}}$ and ${ \mathfrak{m}}_2:={ \mathfrak{m}}_{{ \mathfrak{l}}_{(\sigma_2,\sigma_2^\star)}}$. Let us now suppose that there exists a real morphism $\Psi_{({ \mathfrak{m}}_1,{ \mathfrak{m}}_2)} \in \overline{\mathfrak{C}}({ \mathfrak{S}},{ \mathfrak{B}}\widetilde{\otimes}{ \mathfrak{B}})$ such that $\zeta^{{{ \mathfrak{B}}}{{ \mathfrak{B}}}}_{(1)}\circ \Psi_{({ \mathfrak{m}}_1,{ \mathfrak{m}}_2)} = { \mathfrak{m}}_1$ and $\zeta^{{{ \mathfrak{B}}}{{ \mathfrak{B}}}}_{(2)}\circ \Psi_{({ \mathfrak{m}}_1,{ \mathfrak{m}}_2)} = { \mathfrak{m}}_2$, and let us exhibit a contradiction.\\
From ${ \mathfrak{m}}_1(\sigma_1)=\textit{\bf Y}$ and ${ \mathfrak{m}}_2(\sigma_1)=\bot$, we deduce that $\Psi_{({ \mathfrak{m}}_1,{ \mathfrak{m}}_2)}(\sigma_1)=\textit{\bf Y}\widetilde{\otimes} \bot$.\\
From ${ \mathfrak{m}}_1(\sigma_1^\star)=\textit{\bf N}$ and ${ \mathfrak{m}}_2(\sigma_1^\star)=\bot$, we deduce that $\Psi_{({ \mathfrak{m}}_1,{ \mathfrak{m}}_2)}(\sigma_1^\star)=\textit{\bf N}\widetilde{\otimes} \bot$.\\
From ${ \mathfrak{m}}_1(\sigma_2^\star)=\bot$ and ${ \mathfrak{m}}_2(\sigma_2^\star)=\textit{\bf N}$, we deduce that $\Psi_{({ \mathfrak{m}}_1,{ \mathfrak{m}}_2)}(\sigma_2^\star)=\bot\widetilde{\otimes} \textit{\bf N}$.\\

From $\Psi_{({ \mathfrak{m}}_1,{ \mathfrak{m}}_2)}(\sigma_1)=\textit{\bf Y}\widetilde{\otimes} \bot$ and $\Psi_{({ \mathfrak{m}}_1,{ \mathfrak{m}}_2)}(\sigma_1^\star)=\textit{\bf N}\widetilde{\otimes} \bot$ and the homomorphic properties satisfied by $\Psi_{({ \mathfrak{m}}_1,{ \mathfrak{m}}_2)}$ we deduce that
$\Psi_{({ \mathfrak{m}}_1,{ \mathfrak{m}}_2)}(\bot_{{}_{\mathfrak{S}}})=\Psi_{({ \mathfrak{m}}_1,{ \mathfrak{m}}_2)}(\sigma_1\sqcap_{{}_{\mathfrak{S}}}\sigma_1^\star)=\textit{\bf Y}\widetilde{\otimes} \bot \sqcap_{{}_{{ \mathfrak{B}}^{\widetilde{\otimes}2}}} \textit{\bf N}\widetilde{\otimes} \bot=\bot\widetilde{\otimes} \bot$.\\
On another part, from $\Psi_{({ \mathfrak{m}}_1,{ \mathfrak{m}}_2)}(\sigma_2^\star)=\bot\widetilde{\otimes} \textit{\bf N}$ and $\Psi_{({ \mathfrak{m}}_1,{ \mathfrak{m}}_2)}(\sigma_1^\star)=\textit{\bf N}\widetilde{\otimes} \bot$ and the homomorphic properties satisfied by $\Psi_{({ \mathfrak{m}}_1,{ \mathfrak{m}}_2)}$ we deduce that
$\Psi_{({ \mathfrak{m}}_1,{ \mathfrak{m}}_2)}(\bot_{{}_{\mathfrak{S}}})=\Psi_{({ \mathfrak{m}}_1,{ \mathfrak{m}}_2)}(\sigma_1^\star\sqcap_{{}_{\mathfrak{S}}}\sigma_2^\star)=\bot\widetilde{\otimes}\textit{\bf N}\sqcap_{{}_{{ \mathfrak{B}}^{\widetilde{\otimes}2}}} \textit{\bf N}\widetilde{\otimes} \bot$.\\
We have then obtained a contradiction.\\
As a conclusion,  and using Lemma \ref{lemmabroadcasting}, we deduce that there cannot exist a real morphism $\Psi \in \overline{\mathfrak{C}}({ \mathfrak{S}},{ \mathfrak{S}}\widetilde{\otimes}{ \mathfrak{S}})$ such that $\forall \sigma\in \overline{\mathfrak{S}},\;(\zeta^{{{ \mathfrak{S}}}{{ \mathfrak{S}}}}_{(1)}\circ \Psi)(\sigma) = (\zeta^{{{ \mathfrak{S}}}{{ \mathfrak{S}}}}_{(2)}\circ \Psi)(\sigma) = \sigma$. Hence, there is no broadcasting.
\end{proof}



\subsection{Bell non-locality}

During this subsection, we will consider the two copies of one-dimensional indeterministic spaces of states ${{{ \overline{\mathfrak{S}}}}}_{A}={\mathfrak{Z}}'_{N_A}$ and ${{{ \overline{\mathfrak{S}}}}}_{B}={\mathfrak{Z}}'_{N_B}$ (with $N_A,N_B\geq 2$) and we consider the space of states ${{{ \overline{S}}}}_{AB}:={{{ \overline{\mathfrak{S}}}}}_{A}\widetilde{\otimes}{{{ \overline{\mathfrak{S}}}}}_{B}$. We denote ${ \mathfrak{S}}_A={ \mathfrak{J}}^c_{{\overline{ \mathfrak{S}}_A}}$ and ${ \mathfrak{S}}_B={ \mathfrak{J}}^c_{{\overline{ \mathfrak{S}}_B}}$.  The space of states (called two-dimensional indeterministic space of states) has been build, according to the subsection \ref{subsectiontensorcompleterealspace}, as the ontic completion ${{{ \widetilde{S}}}}_{AB}:={ \mathfrak{S}}_{A}\widetilde{\otimes} { \mathfrak{S}}_{B}={ \mathfrak{J}}^c_{{\overline{ \mathfrak{S}}_A}\widetilde{\otimes}{\overline{ \mathfrak{S}}_B}}$. If necessary, in order to shorten our formulas, we will eventually replace the notation ${{{ \overline{S}}}}{}_{AB}$ by $\overline{\mathfrak{S}}$ and the notation ${{{ \widetilde{S}}}}{}_{AB}$ by ${\mathfrak{S}}$.\\

First of all, let ${\Sigma}$ be an element of ${ \mathfrak{S}}_{A}{\widehat{\otimes}}{ \mathfrak{S}}_{B}$.  Let ${\phi}_1\in \overline{ \mathfrak{C}}({ \mathfrak{S}}_A,{ \mathfrak{B}})$ et ${\phi}_2\in \overline{ \mathfrak{C}}({ \mathfrak{S}}_A,{ \mathfrak{B}})$ be a pair of compatible real measurements, and ${\rho}_1\in \overline{ \mathfrak{C}}({ \mathfrak{S}}_B,{ \mathfrak{B}})$ et ${\rho}_2\in \overline{ \mathfrak{C}}({ \mathfrak{S}}_B,{ \mathfrak{B}})$ another pair of compatible real measurements. Then, there exists an element $\Lambda \in  { \mathfrak{B}}^{{\widetilde{\otimes}}4}$ such that
\begin{eqnarray}
\zeta^{{ \mathfrak{B}}^{{\widetilde{\otimes}}4}}_{(1)(3)}(\Lambda) &=&(\phi_1\widetilde{\otimes} \rho_1)({\Sigma})\label{bell1}\\
\zeta^{{ \mathfrak{B}}^{{\widetilde{\otimes}}4}}_{(1)(4)}(\Lambda) &=&(\phi_1\widetilde{\otimes} \rho_2)({\Sigma})\label{bell2}\\
\zeta^{{ \mathfrak{B}}^{{\widetilde{\otimes}}4}}_{(2)(3)}(\Lambda) &=&(\phi_2\widetilde{\otimes} \rho_1)({\Sigma})\label{bell3}\\
\zeta^{{ \mathfrak{B}}^{{\widetilde{\otimes}}4}}_{(2)(4)}(\Lambda) &=&(\phi_2\widetilde{\otimes} \rho_2)({\Sigma})\label{bell4}
\end{eqnarray}
Indeed, it suffices to define
\begin{eqnarray*}
\Lambda &:=& (\Psi_{(\phi_1,\phi_2)}\widetilde{\otimes} \Psi_{(\rho_1,\rho_2)})({\Sigma}),
\end {eqnarray*}
where $\Psi_{(\phi_1,\phi_2)}\in \overline{ \mathfrak{C}}({ \mathfrak{S}}_A,{ \mathfrak{B}}{\widetilde{\otimes}}{ \mathfrak{B}})$ is the joint channel associated to $\phi_1$ and $\phi_2$, and $\Psi_{(\rho_1,\rho_2)}\in \overline{ \mathfrak{C}}({ \mathfrak{S}}_D,{ \mathfrak{B}}{\widetilde{\otimes}}{ \mathfrak{B}})$ is the joint channel associated to $\rho_1$ and $\rho_2$.\\

Secondly, if ${\Sigma}$ is a real state (i.e. ${\Sigma}\in \overline{{S}}_{AB}=\overline{ \mathfrak{S}}_{A}{\widetilde{\otimes}}\overline{ \mathfrak{S}}_{B}$),  then there exists an element $\Lambda \in  { \mathfrak{B}}^{{\widetilde{\otimes}}4}$ such that the four relations (\ref{bell1}) (\ref{bell2}) (\ref{bell3}) (\ref{bell4}) are simultaneously satisfied. Indeed,  if ${\Sigma}=\bigsqcap{}^{{}^{ \mathfrak{S}_A\widetilde{\otimes}\mathfrak{S}_B}}_{i\in I} \sigma_{i,A}\widetilde{\otimes} \sigma_{i,B}$ with $\{(\sigma_{i,A}, \sigma_{i,B})\;\vert\;i\in I\}\subseteq \overline{ \mathfrak{S}}{}_{A}^{{}^{pure}}{\times}\overline{ \mathfrak{S}}{}^{{}^{pure}}_{B}$, it suffices to define
\begin{eqnarray}
\Lambda & = & \bigsqcap{}^{{}^{{ \mathfrak{B}}^{{\widetilde{\otimes}}4}}}_{i\in I} \phi_1(\sigma_{i,A})\widetilde{\otimes}\phi_2(\sigma_{i,A})\widetilde{\otimes}\rho_1(\sigma_{i,B})\widetilde{\otimes} \rho_2(\sigma_{i,B}).
\end{eqnarray}

On the contrary, we can now introduce the notion of Bell non-locality in a perspective that is adapted to our purpose.  Note that it has been shown by M. Plávala \cite{PlavalaA} that the standard notion of Bell non-locality given as in the CHSH scenario \cite{CHSH} is a particular case of the general framework on Bell non-locality adapted to Generalized Probabilistic Theories (GPT) (see the introduction by M. Plávala \cite{Plavala}).\\
\begin{lemma}
If there exists two pairs of real measurements ${\phi}_1\in \overline{ \mathfrak{C}}({ \mathfrak{S}}_A,{ \mathfrak{B}})$ and ${\phi}_2\in \overline{ \mathfrak{C}}({ \mathfrak{S}}_A,{ \mathfrak{B}})$,  ${\rho}_1\in \overline{ \mathfrak{C}}({ \mathfrak{S}}_D,{ \mathfrak{B}})$ and ${\rho}_2\in \overline{ \mathfrak{C}}({ \mathfrak{S}}_D, { \mathfrak{B}})$ such that there exists no state $\Lambda \in  { \mathfrak{B}}^{{\widetilde{\otimes}}4}$ satisfying  simultaneously the four relations  (\ref{bell1}) (\ref{bell2}) (\ref{bell3}) (\ref{bell4}), then the state ${\Sigma}$ is {\em a Bell non-local state}.
\end{lemma}
\begin{proof}
see \cite{PlavalaA}.
\end{proof}

We have the following fundamental result.

\begin{theorem}
Bell non-local states do exist in ${{{ \widetilde{S}}}}_{AB}$.
\end{theorem}
\begin{proof}
We will consider $\sigma_1,\sigma_2\in \overline{\mathfrak{S}}{}^{{}^{pure}}_{A}$ and $\tau_1,\tau_2\in \overline{\mathfrak{S}}{}^{{}^{pure}}_{B}$ such that $\sigma_1\not=\sigma_2,\sigma_2^\star$ and $\tau_1\not=\tau_2,\tau_2^\star$.  Let us now introduce a particular hidden-state. 

\begin{eqnarray}
\Sigma & := & (\sigma_1 \widetilde{\otimes} \tau_1 \sqcap_{{}_{\overline{S}_{AB}}} \sigma_2 \widetilde{\otimes} \tau_2) \sqcup_{{}_{\hat{S}_{AB}}} (\sigma_1^\star  \widetilde{\otimes} \bot_{{}_{ \mathfrak{S}_B}} \sqcap_{{}_{\overline{S}_{AB}}} \bot_{{}_{ \mathfrak{S}_A}} \widetilde{\otimes} \tau_1^\star).
\end{eqnarray}
This state do exist in the ontic completion $\hat{S}_{AB}:={ \mathfrak{J}}^c_{\overline{{S}}_{AB}}$ of the real space $\overline{ {S}}_{AB}:={\overline{ \mathfrak{S}}_A}\widetilde{\otimes}{\overline{ \mathfrak{S}}_B}$.  More precisely, we have $\Sigma\in \hat{S}_{AB}\smallsetminus \overline{S}_{AB}$.

Indeed, as shown in subsection \ref{subsectionpreliminaryNdim}, we have
\begin{eqnarray}
&&\hspace{-2cm} cl_c^{\overline{{S}}_{AB}}\left(
\left\{(\sigma_1 \widetilde{\otimes} \tau_1 \sqcap_{{}_{\overline{S}_{AB}}} \sigma_2 \widetilde{\otimes} \tau_2)\; ,\; (\sigma_1^\star  \widetilde{\otimes} \bot_{{}_{ \mathfrak{S}_B}} \sqcap_{{}_{\overline{S}_{AB}}} \bot_{{}_{ \mathfrak{S}_A}} \widetilde{\otimes} \tau_1^\star)\right\}\right)=\nonumber\\
&&\hspace{-1cm} =\left\{(\sigma_1 \widetilde{\otimes} \tau_1 \sqcap_{{}_{\overline{S}_{AB}}} \sigma_2 \widetilde{\otimes} \tau_2)\;,\; (\sigma_1 \widetilde{\otimes} \tau_1^\star \sqcap_{{}_{\overline{S}_{AB}}} \sigma_1^\star \widetilde{\otimes} \tau_2) \;,\; (\sigma_2 \widetilde{\otimes} \tau_1^\star \sqcap_{{}_{\overline{S}_{AB}}} \sigma_1^\star \widetilde{\otimes} \tau_1)\right\}\;\in {\mathcal{K}}(\overline{{S}}_{AB})\;\;\;\;\;\;\;\;\;\;\;\;
\end{eqnarray}
Let us now consider the following real measurement operators 
\begin{eqnarray}
{\phi}_1:={ \mathfrak{m}}_{{ \mathfrak{l}}_{(\sigma_1,\sigma_1^\star)}}, && {\phi}_2:={ \mathfrak{m}}_{{ \mathfrak{l}}_{(\sigma_2,\sigma_2^\star)}}, \\
{\rho}_1:={ \mathfrak{m}}_{{ \mathfrak{l}}_{(\tau_1,\tau_1^\star)}}, &&{\rho}_2:={ \mathfrak{m}}_{{ \mathfrak{l}}_{(\tau_2,\tau_2^\star)}}.
\end{eqnarray}  
We now intent to prove that $\Sigma$ is a Bell non-local state with respect to the pairs of incompatible real measurements $(\phi_1,\phi_2)$ and $(\rho_1,\rho_2)$. In the following computation we use obviously the defining property (\ref{morphismtensorindeterminism}).\\

The equality
\begin{eqnarray}
\Phi_{13} :=(\phi_1\widetilde{\otimes} \rho_1)({\Sigma})= \textit{\bf N} \,\widetilde{\otimes} \bot \sqcap_{{}_{{ \mathfrak{B}}^{\widetilde{\otimes}2}}} \textit{\bf Y} \widetilde{\otimes}\, \textit{\bf N}.
\end{eqnarray} 
and the property $\zeta^{{ \mathfrak{B}} { \mathfrak{B}} { \mathfrak{B}} { \mathfrak{B}}}_{(1)(3)}(\Lambda) = \Phi_{13}$ imply
\begin{eqnarray}
\Lambda =  \bigsqcap{}^{{}^{{ \mathfrak{B}}^{\widetilde{\otimes}4}}}_i\textit{\bf N} \,\widetilde{\otimes} \alpha_i \,\widetilde{\otimes}\textit{\bf N} \,\widetilde{\otimes} \beta_i \; \sqcap_{{}_{{ \mathfrak{B}}^{\widetilde{\otimes}4}}}
\bigsqcap{}^{{}^{{ \mathfrak{B}}^{\widetilde{\otimes}4}}}_j\textit{\bf N} \,\widetilde{\otimes} \gamma_j \,\widetilde{\otimes}\textit{\bf Y} \,\widetilde{\otimes} \delta_j \; \sqcap_{{}_{{ \mathfrak{B}}^{\widetilde{\otimes}4}}} \bigsqcap{}^{{}^{{ \mathfrak{B}}^{\widetilde{\otimes}4}}}_k\textit{\bf Y} \,\widetilde{\otimes} \epsilon_k \,\widetilde{\otimes}\textit{\bf N} \,\widetilde{\otimes} \varphi_k \;\;\;\;\;\;\;\;\;\;\;
\end{eqnarray} 
for certain collections $(\alpha_i,\beta_i)_i$, $(\gamma_j,\delta_j)_j$, $(\epsilon_k,\varphi_k)_k$ of pairs of elements of $\{\,\textit{\bf Y}\,,\, \textit{\bf N}\,\}$.\\

The equality
\begin{eqnarray}
\Phi_{14} :=(\phi_1\widetilde{\otimes} \rho_2)({\Sigma})= \textit{\bf Y} \,\widetilde{\otimes} \bot \sqcap_{{}_{{ \mathfrak{B}}^{\widetilde{\otimes}2}}} \textit{\bf N} \widetilde{\otimes}\, \textit{\bf Y}.
\end{eqnarray} 
and the property $\zeta^{{ \mathfrak{B}} { \mathfrak{B}} { \mathfrak{B}} { \mathfrak{B}}}_{(1)(4)}(\Lambda) = \Phi_{14}$ imply
\begin{eqnarray}
\bigsqcap{}^{{}^{{ \mathfrak{B}}^{\widetilde{\otimes}2}}}_i\textit{\bf N}  \,\widetilde{\otimes} \beta_i \; \sqcap_{{}_{{ \mathfrak{B}}^{\widetilde{\otimes}2}}}
\bigsqcap{}^{{}^{{ \mathfrak{B}}^{\widetilde{\otimes}2}}}_j\textit{\bf N}  \,\widetilde{\otimes} \delta_j \; \sqcap_{{}_{{ \mathfrak{B}}^{\widetilde{\otimes}2}}} \bigsqcap{}^{{}^{{ \mathfrak{B}}^{\widetilde{\otimes}2}}}_k\textit{\bf Y}  \,\widetilde{\otimes} \varphi_k  &=& \textit{\bf Y} \,\widetilde{\otimes} \bot \sqcap_{{}_{{ \mathfrak{B}}^{\widetilde{\otimes}2}}} \textit{\bf N} \widetilde{\otimes}\, \textit{\bf Y},\;\;\;\;\;\;\;\;\;\;\;\;\;\;\;\;\;\; 
\end{eqnarray} 
and then
\begin{eqnarray}
\bigwedge{}_i \beta_i \wedge
\bigwedge{}_j \delta_j = \textit{\bf Y}, & & \bigwedge{}_k \varphi_k = \bot.
\end{eqnarray}

The equality
\begin{eqnarray}
\Phi_{23}:=(\phi_2\widetilde{\otimes} \rho_1)({\Sigma}) = \textit{\bf Y} \,\widetilde{\otimes} \textit{\bf N} \sqcap_{{}_{{ \mathfrak{B}}^{\widetilde{\otimes}2}}} \bot \widetilde{\otimes}\, \textit{\bf Y}.
\end{eqnarray} 
and $\zeta^{{ \mathfrak{B}} { \mathfrak{B}} { \mathfrak{B}} { \mathfrak{B}}}_{(2)(3)}(\Lambda) = \Phi_{23}$ imply 
\begin{eqnarray}
\bigsqcap{}^{{}^{{ \mathfrak{B}}^{\widetilde{\otimes}2}}}_i \alpha_i  \,\widetilde{\otimes} \textit{\bf N} \; \sqcap_{{}_{{ \mathfrak{B}}^{\widetilde{\otimes}2}}}
 \bigsqcap{}^{{}^{{ \mathfrak{B}}^{\widetilde{\otimes}2}}}_j \gamma_j  \,\widetilde{\otimes} \textit{\bf Y} \; \sqcap_{{}_{{ \mathfrak{B}}^{\widetilde{\otimes}2}}} \bigsqcap{}^{{}^{{ \mathfrak{B}}^{\widetilde{\otimes}2}}}_k \epsilon_k  \,\widetilde{\otimes} \textit{\bf N}  &=& \textit{\bf Y} \,\widetilde{\otimes} \textit{\bf N} \sqcap_{{}_{{ \mathfrak{B}}^{\widetilde{\otimes}2}}} \bot \widetilde{\otimes}\, \textit{\bf Y},\;\;\;\;\;\;\;\;\;\;\;\;\;\;\;\;\;\; 
\end{eqnarray} 
and then
\begin{eqnarray}
\bigwedge{}_i \alpha_i \wedge \bigwedge{}_k \epsilon_k = \textit{\bf Y}, &  & \bigwedge{}_j \gamma_j = \bot.
\end{eqnarray}

If we summarize what we have obtained, we have
\begin{eqnarray}
\Lambda &=& \textit{\bf N} \,\widetilde{\otimes} \textit{\bf Y} \,\widetilde{\otimes}\textit{\bf N} \,\widetilde{\otimes} \textit{\bf Y} \; \sqcap_{{}_{{ \mathfrak{B}}^{\widetilde{\otimes}4}}}
\textit{\bf N} \,\widetilde{\otimes} \bot \,\widetilde{\otimes}\textit{\bf Y} \,\widetilde{\otimes} \textit{\bf Y} \; 
\sqcap_{{}_{{ \mathfrak{B}}^{\widetilde{\otimes}4}}} 
\textit{\bf Y} \,\widetilde{\otimes} \textit{\bf Y} \,\widetilde{\otimes}\textit{\bf N} \,\widetilde{\otimes} \bot 
\end{eqnarray}

Now we use the equality
\begin{eqnarray}
\Phi_{24} :=(\phi_2\widetilde{\otimes} \rho_2)({\Sigma})= \bot \,\widetilde{\otimes} \bot.
\end{eqnarray} 
and  $\zeta^{{ \mathfrak{B}} { \mathfrak{B}} { \mathfrak{B}} { \mathfrak{B}}}_{(2)(4)}(\Lambda) = \Phi_{24}$ to obtain the announced contradiction. Indeed, we have
\begin{eqnarray}
\textit{\bf Y} \,\widetilde{\otimes} \bot \sqcap_{{}_{{ \mathfrak{B}}^{\widetilde{\otimes}2}}} \bot \widetilde{\otimes} \textit{\bf Y} \not= \bot \,\widetilde{\otimes}\, \bot.
\end{eqnarray} 

As a  conclusion,  the state $\Sigma$ is Bell non-local and the pairs $(\phi_1,\phi_2)$ et $(\rho_1,\rho_2)$ are pairs of incompatible real measurements. 
\end{proof}

\section{Conclusion}

My ongoing research program \cite{Buffenoir2021}\cite{Buffenoir2022}\cite{Buffenoir2025} seeks to devise a method for overcoming the previously mentioned no-go theorems \cite{Pullmanova}\cite{HudsonPullmanova}\cite{Pykacz}, which concern the impossibility of constructing a tensor product suitable for describing composite systems starting from the logical formulation of single systems. The approach pursued relies on a formalism closely related to that originally developed by G. Birkhoff and J. von Neumann, 
 but give a central role to the existence of a tensor product to describe experiments on compound systems.\\
In Section 2, after presenting the formalism of States/Effects Chu spaces together with their morphisms, we introduced the notion of real structure, distinguishing between {\em real states} and {\em hidden states}. We then specified how determinism and indeterminism are defined in this setting. Sections 3 and 5 refined the definition of real structure by showing that generalized state spaces including hidden states arise (i) through a natural completion process applied to real \underline{indeterministic} state spaces, and (ii) in connection with \underline{contextual} empirical models derived from the operational description of such real spaces. Section 4 established the conditions tensor products of state spaces must satisfy to allow the construction of compound systems from their components, and highlighted that the simplest solution also accounts for composite deterministic (classical) systems. Section 6 extended the discussion with a proposal for representing compound indeterministic systems, making use of the ontic completions developed in Section 3. This proposal remained general, and Section 7 offered a more precise study of the iterated tensor product of basic indeterministic systems
. The resulting structure proved to approximate as closely as possible that of \underline{irreducible Hilbert geometries} although this type of structure had not been imposed ab initio in our construction. Finally, Section 8 demonstrated that our framework explicitly reproduces key quantum-like phenomena, such as the \underline{no-broadcasting theorem} and the \underline{existence of {\em Bell non-local} states} in our framework. Our primary objective has therefore been fully achieved, which is a result of paramount importance in quantum logic. 

\section{Appendix}

Let us first consider $\sigma_1,\sigma_2,\sigma_3,\sigma_4$ all distinct elements in $({\check{\mathfrak{S}}}\smallsetminus \overline{\mathfrak{S}})$.\\
Using properties (\ref{Rek1}) and $\sigma_1\asymp\sigma_2, \sigma_2\asymp\sigma_3, \sigma_3\asymp\sigma_4, \sigma_1\asymp\sigma_4$ and (\ref{thirdcoveringpropertySbarv}), we note that we can always choose $\pi_i,\rho_i\in \overline{\mathfrak{S}}$ for $i=1,2,3,4$ such that
\begin{eqnarray}
\forall i=1,2,3,4,&&\pi_i,\rho_i\in \Theta^{\overline{\mathfrak{S}}}(\sigma_i),\;\;\sigma_i=\pi_i\sqcup_{{}_{\overline{ \mathfrak{S}}}} \rho_i \\
&&\pi_1=\pi_2,\;\;\pi_3=\pi_4,\;\; \rho_2=\rho_3,\;\;\rho_1=\rho_4\label{iddefREC}
\end{eqnarray}
Let us now assume that there exists $\lambda\in {\mathfrak{G}}^{\check{\mathfrak{S}}}$ such that ${\mathfrak{r}}^U(\lambda,\sigma_1,\sigma_2)$ and ${\mathfrak{r}}^U(\lambda,\sigma_3,\sigma_4)$.  In other words,
\begin{eqnarray}
\lambda\sqcoversupset_{{}_{{ \mathfrak{S}}}} \pi_2,&& \lambda\sqcoversupset_{{}_{{ \mathfrak{S}}}}\pi_3.\label{useful1aREC}
\end{eqnarray}
The consistency relation $\sigma_2 \asymp \sigma_4$ means that there exists $\epsilon_{24}\in \Theta^{\overline{ \mathfrak{S}}}(\sigma_2)\cap \Theta^{\overline{ \mathfrak{S}}}(\sigma_4)$ and $\epsilon_{24}=\sigma_2\sqcap_{{}_{{\mathfrak{S}}}}\sigma_4$ with
$\sigma_2:=\pi_2\sqcup_{{}_{{ \mathfrak{S}}}}\rho_2$ and 
$\sigma_4:=\pi_3\sqcup_{{}_{{ \mathfrak{S}}}}\rho_1$. We will consider that $\epsilon_{24}\not=\pi_2$, $\epsilon_{24}\not=\rho_2$, $\epsilon_{24}\not=\pi_3$ and $\epsilon_{24}\not=\rho_1$. 
From $\epsilon_{24}=\sigma_2\sqcap_{{}_{{\mathfrak{S}}}}\sigma_4$ and $\sigma_2\sqcoversupset_{{}_{{ \mathfrak{S}}}}\pi_2$ and $\sigma_4\sqcoversupset_{{}_{{ \mathfrak{S}}}}\pi_3$, we deduce 
\begin{eqnarray}
\epsilon_{24}\sqsupseteq_{{}_{{ \mathfrak{S}}}}(\pi_2\sqcap_{{}_{\overline{ \mathfrak{S}}}}\pi_3).
\end{eqnarray}
Moreover, from (\ref{thirdcoveringpropertySbariii}) (if $\lambda\in  ({\widecheck{\mathfrak{S}}}\smallsetminus \overline{\mathfrak{S}})$) or from (\ref{secondcoveringpropertySbar}) (if $\lambda\in \overline{\mathfrak{S}}{}^{{}^{pure}}$)
, we deduce
\begin{eqnarray}
(\pi_2\sqcap_{{}_{\overline{ \mathfrak{S}}}}\pi_3) \sqcoversubset_{{}_{\overline{ \mathfrak{S}}}} \pi_2,\pi_3
\end{eqnarray}
and then, using $\epsilon_{24}\not=\pi_2$ and $\epsilon_{24}\not=\pi_3$
\begin{eqnarray}
(\epsilon_{24}\sqcap_{{}_{\overline{ \mathfrak{S}}}}\pi_2)=(\epsilon_{24}\sqcap_{{}_{\overline{ \mathfrak{S}}}}\pi_3)=(\pi_2\sqcap_{{}_{\overline{ \mathfrak{S}}}}\pi_3)
\end{eqnarray}
 From $\pi_2\sqcap_{{}_{\overline{ \mathfrak{S}}}}\epsilon_{24}=\pi_3\sqcap_{{}_{\overline{ \mathfrak{S}}}}\epsilon_{24}$, we deduce, using (\ref{thirdcoveringpropertySbariii}), that $\epsilon_{24}=(\pi_2\sqcap_{{}_{\overline{ \mathfrak{S}}}}\epsilon_{24})\sqcup_{{}_{\overline{ \mathfrak{S}}}}(\rho_2\sqcap_{{}_{\overline{ \mathfrak{S}}}}\epsilon_{24})=(\pi_3\sqcap_{{}_{\overline{ \mathfrak{S}}}}\epsilon_{24})\sqcup_{{}_{\overline{ \mathfrak{S}}}}(\rho_2\sqcap_{{}_{\overline{ \mathfrak{S}}}}\epsilon_{24})$. As a consequence,  $\epsilon_{24}$ is an element of $ \Theta^{\overline{ \mathfrak{S}}}(\sigma_2)\cap \Theta^{\overline{ \mathfrak{S}}}(\sigma_3)$. However, we have already $\rho_2\in \Theta^{\overline{ \mathfrak{S}}}(\sigma_2)\cap \Theta^{\overline{ \mathfrak{S}}}(\sigma_3)$ and $\epsilon_{24}\not=\rho_2$. Using (\ref{Rek1}) we then deduce that $\sigma_2=\sigma_3$.  This case has been excluded by assumption. \\
As a conclusion, this configuration is excluded.\\

Let us now consider that $\sigma_1,\sigma_2,\sigma_3\in ({\check{\mathfrak{S}}}\smallsetminus \overline{\mathfrak{S}})$ and $\sigma_4\in {\overline{\mathfrak{S}}}{}^{{}^{pure}}$. 
Using properties (\ref{Rek1}) and $\sigma_1\asymp\sigma_2, \sigma_2\asymp\sigma_3, \sigma_3\asymp\sigma_4, \sigma_1\asymp\sigma_4$ and (\ref{thirdcoveringpropertySbarv})(\ref{thirdcoveringpropertySbariii}), we note that we can always choose $\pi_i,\rho_i\in \overline{\mathfrak{S}}$ for $i=1,2,3$ such that
\begin{eqnarray}
\forall i=1,2,3&&\pi_i,\rho_i\in \Theta^{\overline{\mathfrak{S}}}(\sigma_i),\;\;\sigma_i=\pi_i\sqcup_{{}_{\overline{ \mathfrak{S}}}} \rho_i \\
&&\pi_1=\pi_2,\;\;\rho_2=\rho_3,\;\;\sigma_4\sqcoversupset_{{}_{\overline{ \mathfrak{S}}}}\pi_3,\;\; \sigma_4\sqcoversupset_{{}_{\overline{ \mathfrak{S}}}}\rho_1.
\end{eqnarray}
Let us now assume that there exists $\lambda\in {\mathfrak{G}}^{\check{\mathfrak{S}}}$ such that ${\mathfrak{r}}^U(\lambda,\sigma_1,\sigma_2)$ and ${\mathfrak{r}}^U(\lambda,\sigma_3,\sigma_4)$. We have then
\begin{eqnarray}
\lambda\sqcoversupset_{{}_{{ \mathfrak{S}}}} \pi_2,&& \lambda\sqcoversupset_{{}_{{ \mathfrak{S}}}}\pi_3.
\end{eqnarray}
The consistency relation $\sigma_2 \asymp \sigma_4$ means that there exists $\epsilon_2\in \Theta^{\overline{ \mathfrak{S}}}(\sigma_2)$ such that $\epsilon_2\sqcoversubset_{{}_{\overline{\mathfrak{S}}}}\sigma_4$.
Analogously,  the consistency relation $\sigma_1 \asymp \sigma_3$ means that there exists $\epsilon_1\in \Theta^{\overline{ \mathfrak{S}}}(\sigma_1)\cap \Theta^{\overline{ \mathfrak{S}}}(\sigma_3)$ such that $\epsilon_1=\sigma_1\sqcap_{{}_{\overline{\mathfrak{S}}}}\sigma_3$.\\
Let us now choose $\chi:=\sigma_1$ and $\chi:=\sigma_2$ and $\chi:=\sigma_3$ with $\kappa:=\sigma_4$ in (\ref{fundamentalformula}). We obtain
\begin{eqnarray}
&&\rho_1 \sqsupseteq_{{}_{\overline{\mathfrak{S}}}} (\sigma_4\sqcap_{{}_{\overline{ \mathfrak{S}}}}\pi_2),(\sigma_4\sqcap_{{}_{\overline{ \mathfrak{S}}}}\epsilon_1)\\
&&\epsilon_2 \sqsupseteq_{{}_{\overline{\mathfrak{S}}}} (\sigma_4\sqcap_{{}_{\overline{ \mathfrak{S}}}}\pi_2),(\sigma_4\sqcap_{{}_{\overline{ \mathfrak{S}}}}\rho_2)\\
&&\pi_3 \sqsupseteq_{{}_{\overline{\mathfrak{S}}}} (\sigma_4\sqcap_{{}_{\overline{ \mathfrak{S}}}}\rho_2),(\sigma_4\sqcap_{{}_{\overline{ \mathfrak{S}}}}\epsilon_1).
\end{eqnarray}
From $\rho_1 \sqsupseteq_{{}_{\overline{\mathfrak{S}}}} (\sigma_4\sqcap_{{}_{\overline{ \mathfrak{S}}}}\pi_2)$ and $\sigma_4\sqsupseteq_{{}_{\overline{ \mathfrak{S}}}}\rho_1$ we deduce $(\rho_1\sqcap_{{}_{\overline{ \mathfrak{S}}}}\pi_2)=(\sigma_4\sqcap_{{}_{\overline{ \mathfrak{S}}}}\pi_2)$. From $\epsilon_2 \sqsupseteq_{{}_{\overline{\mathfrak{S}}}} (\sigma_4\sqcap_{{}_{\overline{ \mathfrak{S}}}}\pi_2)$ and $\epsilon_2\sqcoversubset_{{}_{\overline{\mathfrak{S}}}}\sigma_4$, we deduce $(\epsilon_2\sqcap_{{}_{\overline{ \mathfrak{S}}}}\pi_2)=(\sigma_4\sqcap_{{}_{\overline{ \mathfrak{S}}}}\pi_2)$.\\
From $\rho_1 \sqsupseteq_{{}_{\overline{\mathfrak{S}}}} (\sigma_4\sqcap_{{}_{\overline{ \mathfrak{S}}}}\epsilon_1)$ and $\sigma_4\sqsupseteq_{{}_{\overline{ \mathfrak{S}}}}\rho_1$ we deduce $(\rho_1\sqcap_{{}_{\overline{ \mathfrak{S}}}}\epsilon_1)=(\sigma_4\sqcap_{{}_{\overline{ \mathfrak{S}}}}\epsilon_1)$. From $\pi_3 \sqsupseteq_{{}_{\overline{\mathfrak{S}}}} (\sigma_4\sqcap_{{}_{\overline{ \mathfrak{S}}}}\epsilon_1)$ and $\pi_3\sqcoversubset_{{}_{\overline{\mathfrak{S}}}}\sigma_4$, we deduce $(\pi_3\sqcap_{{}_{\overline{ \mathfrak{S}}}}\epsilon_1)=(\sigma_4\sqcap_{{}_{\overline{ \mathfrak{S}}}}\epsilon_1)$.\\
From $\epsilon_2 \sqsupseteq_{{}_{\overline{\mathfrak{S}}}} (\sigma_4\sqcap_{{}_{\overline{ \mathfrak{S}}}}\rho_2)$ and $\sigma_4\sqsupseteq_{{}_{\overline{ \mathfrak{S}}}}\epsilon_2$ we deduce $(\rho_2\sqcap_{{}_{\overline{ \mathfrak{S}}}}\epsilon_2)=(\sigma_4\sqcap_{{}_{\overline{ \mathfrak{S}}}}\rho_2)$. From $\pi_3 \sqsupseteq_{{}_{\overline{\mathfrak{S}}}} (\sigma_4\sqcap_{{}_{\overline{ \mathfrak{S}}}}\rho_2)$ and $\pi_3\sqcoversubset_{{}_{\overline{\mathfrak{S}}}}\sigma_4$, we deduce $(\pi_3\sqcap_{{}_{\overline{ \mathfrak{S}}}}\rho_2)=(\sigma_4\sqcap_{{}_{\overline{ \mathfrak{S}}}}\rho_2)$.\\
Let us summarize
\begin{eqnarray}
&&(\rho_1\sqcap_{{}_{\overline{ \mathfrak{S}}}}\pi_2)=(\epsilon_2\sqcap_{{}_{\overline{ \mathfrak{S}}}}\pi_2)\\
&&(\rho_1\sqcap_{{}_{\overline{ \mathfrak{S}}}}\epsilon_1)=(\pi_3\sqcap_{{}_{\overline{ \mathfrak{S}}}}\epsilon_1)\\
&&(\rho_2\sqcap_{{}_{\overline{ \mathfrak{S}}}}\epsilon_2)=(\pi_3\sqcap_{{}_{\overline{ \mathfrak{S}}}}\rho_2).
\end{eqnarray}
Using $(\rho_1\sqcap_{{}_{\overline{ \mathfrak{S}}}}\epsilon_1)\sqcoversubset_{{}_{\overline{ \mathfrak{S}}}}\rho_1$ inherited from (\ref{thirdcoveringpropertySbariii}), and $(\rho_1\sqcap_{{}_{\overline{ \mathfrak{S}}}}\epsilon_1)=(\pi_3\sqcap_{{}_{\overline{ \mathfrak{S}}}}\epsilon_1)$, we deduce 
\begin{eqnarray}
(\rho_1\sqcap_{{}_{\overline{ \mathfrak{S}}}}\epsilon_1)=(\pi_3\sqcap_{{}_{\overline{ \mathfrak{S}}}}\epsilon_1)=(\pi_3\sqcap_{{}_{\overline{ \mathfrak{S}}}}\rho_1)
\end{eqnarray}
From $\epsilon_1=\sigma_1\sqcap_{{}_{\overline{\mathfrak{S}}}}\sigma_3$ we know that 
\begin{eqnarray}
\epsilon_1\sqsupseteq_{{}_{\overline{\mathfrak{S}}}} (\pi_2\sqcap_{{}_{\overline{\mathfrak{S}}}}\pi_3)
\end{eqnarray}
On another part, from $\lambda= \pi_2\sqcup_{{}_{\overline{ \mathfrak{S}}}}\pi_3$ and using (\ref{thirdcoveringpropertySbariii}) (if $\lambda\in  ({\widecheck{\mathfrak{S}}}\smallsetminus \overline{\mathfrak{S}})$) or using (\ref{secondcoveringpropertySbar}) (if $\lambda\in \overline{\mathfrak{S}}{}^{{}^{pure}}$)
, we deduce
\begin{eqnarray}
(\pi_2\sqcap_{{}_{\overline{\mathfrak{S}}}}\pi_3) \sqcoversubset_{{}_{\overline{\mathfrak{S}}}}\pi_2,\pi_3.
\end{eqnarray}
and then $(\epsilon_1\sqcap_{{}_{\overline{\mathfrak{S}}}}\pi_3)=(\pi_2\sqcap_{{}_{\overline{\mathfrak{S}}}}\pi_3)$. Let us summarize
\begin{eqnarray}
(\rho_1\sqcap_{{}_{\overline{ \mathfrak{S}}}}\epsilon_1)=(\pi_3\sqcap_{{}_{\overline{ \mathfrak{S}}}}\rho_1)=(\epsilon_1\sqcap_{{}_{\overline{\mathfrak{S}}}}\pi_3)=(\pi_2\sqcap_{{}_{\overline{\mathfrak{S}}}}\pi_3)\sqcoversubset_{{}_{\overline{\mathfrak{S}}}}\pi_2,\pi_3.
\end{eqnarray}
Now, from $(\rho_1\sqcap_{{}_{\overline{ \mathfrak{S}}}}\epsilon_1)\sqcoversubset_{{}_{\overline{\mathfrak{S}}}}\pi_2$ we have immediately 
\begin{eqnarray}
(\rho_1\sqcap_{{}_{\overline{ \mathfrak{S}}}}\epsilon_1)=(\rho_1\sqcap_{{}_{\overline{ \mathfrak{S}}}}\pi_2).
\end{eqnarray}
Now, we can summarize our results. From $(\rho_1\sqcap_{{}_{\overline{ \mathfrak{S}}}}\pi_2)=(\epsilon_2\sqcap_{{}_{\overline{ \mathfrak{S}}}}\pi_2)$, $(\rho_1\sqcap_{{}_{\overline{ \mathfrak{S}}}}\epsilon_1)=(\pi_3\sqcap_{{}_{\overline{ \mathfrak{S}}}}\epsilon_1)$ and $(\rho_1\sqcap_{{}_{\overline{ \mathfrak{S}}}}\epsilon_1)=(\rho_1\sqcap_{{}_{\overline{ \mathfrak{S}}}}\pi_2)$, we obtain $(\epsilon_2\sqcap_{{}_{\overline{ \mathfrak{S}}}}\pi_2)=(\pi_3\sqcap_{{}_{\overline{ \mathfrak{S}}}}\epsilon_1)$. Hence, we have obtained
\begin{eqnarray}
(\epsilon_2\sqcap_{{}_{\overline{ \mathfrak{S}}}}\pi_2)=(\pi_3\sqcap_{{}_{\overline{ \mathfrak{S}}}}\epsilon_1),&&
(\epsilon_2\sqcap_{{}_{\overline{ \mathfrak{S}}}}\rho_2)=(\pi_3\sqcap_{{}_{\overline{ \mathfrak{S}}}}\rho_2)
\end{eqnarray}
and then, in particular,
\begin{eqnarray}
(\epsilon_2\sqcap_{{}_{\overline{ \mathfrak{S}}}}\pi_2)\sqcup_{{}_{{ \mathfrak{S}}}}(\epsilon_2\sqcap_{{}_{\overline{ \mathfrak{S}}}}\rho_2) &=& (\pi_3\sqcap_{{}_{\overline{ \mathfrak{S}}}}\epsilon_1)\sqcup_{{}_{{ \mathfrak{S}}}}(\pi_3\sqcap_{{}_{\overline{ \mathfrak{S}}}}\rho_2).
\end{eqnarray}
On another part, we recall that, using properties (\ref{thirdcoveringpropertySbarv}) (\ref{thirdcoveringpropertySbariii}) and $\epsilon_2\in \Theta^{\overline{\mathfrak{S}}}(\pi_2\sqcup_{{}_{\overline{ \mathfrak{S}}}}\rho_2)=\Theta^{\overline{\mathfrak{S}}}(\sigma_2)$ and $\pi_3\in \Theta^{\overline{\mathfrak{S}}}(\epsilon_1\sqcup_{{}_{{ \mathfrak{S}}}}\rho_2)=\Theta^{{\mathfrak{S}}}(\sigma_3)$, we have 
\begin{eqnarray}
\epsilon_2=(\epsilon_2\sqcap_{{}_{\overline{ \mathfrak{S}}}}\pi_2)\sqcup_{{}_{\overline{ \mathfrak{S}}}}(\epsilon_2\sqcap_{{}_{\overline{ \mathfrak{S}}}}\rho_2),&&\pi_3=(\pi_3\sqcap_{{}_{\overline{ \mathfrak{S}}}}\epsilon_1)\sqcup_{{}_{\overline{ \mathfrak{S}}}}(\pi_3\sqcap_{{}_{\overline{ \mathfrak{S}}}}\rho_2)
\end{eqnarray}
We then conclude that
\begin{eqnarray}
(\sigma_2\sqcap_{{}_{\overline{ \mathfrak{S}}}}\sigma_4)=\epsilon_2=\pi_3=(\sigma_3\sqcap_{{}_{\overline{ \mathfrak{S}}}}\sigma_4)
\end{eqnarray}
but this degenerate case has been excluded by assumption.  As a conclusion, this configuration is excluded.\\

Let us now consider that $\sigma_1,\sigma_2\in ({\check{\mathfrak{S}}}\smallsetminus \overline{\mathfrak{S}})$ and $\sigma_3,\sigma_4\in {\overline{\mathfrak{S}}}{}^{{}^{pure}}$.   
Using properties (\ref{Rek1}) and $\sigma_1\asymp\sigma_2, \sigma_2\asymp\sigma_3, \sigma_3\asymp\sigma_4, \sigma_1\asymp\sigma_4$ and (\ref{thirdcoveringpropertySbarv})(\ref{thirdcoveringpropertySbariii}), we note that we can always choose $\pi_i,\rho_i\in \overline{\mathfrak{S}}$ for $i=1,2$ such that
\begin{eqnarray}
\forall i=1,2&&\pi_i,\rho_i\in \Theta^{\overline{\mathfrak{S}}}(\sigma_i),\;\;\sigma_i=\pi_i\sqcup_{{}_{\overline{ \mathfrak{S}}}} \rho_i \\
&&\pi_1=\pi_2,\;\;\sigma_3\sqcoversupset_{{}_{\overline{ \mathfrak{S}}}}\rho_2,\;\; \sigma_4\sqcoversupset_{{}_{\overline{ \mathfrak{S}}}}\rho_1
.
\end{eqnarray}
Let us now assume that there exists $\lambda\in {\mathfrak{G}}^{\check{\mathfrak{S}}}$ such that ${\mathfrak{r}}^U(\lambda,\sigma_1,\sigma_2)$ and ${\mathfrak{r}}^U(\lambda,\sigma_3,\sigma_4)$. We have then
\begin{eqnarray}
\lambda\sqcoversupset_{{}_{{ \mathfrak{S}}}} \pi_2, && \lambda\sqcoversupset_{{}_{\overline{ \mathfrak{S}}}}(\sigma_3\sqcap_{{}_{{ \mathfrak{S}}}}\sigma_4).
\end{eqnarray}
The consistency relation $\sigma_2 \asymp \sigma_4$ means that there exists $\epsilon_2\in \Theta^{\overline{ \mathfrak{S}}}(\sigma_2)$ such that $\epsilon_2\sqcoversubset_{{}_{\overline{\mathfrak{S}}}}\sigma_4$.
Analogously,  the consistency relation $\sigma_1 \asymp \sigma_3$ means that there exists $\epsilon_1\in \Theta^{\overline{ \mathfrak{S}}}(\sigma_1)$ such that $\epsilon_1\sqcoversubset_{{}_{\overline{\mathfrak{S}}}}\sigma_3$.\\
Once again, we insist on the fact that $\epsilon_1\not=\pi_1$ (resp. $\epsilon_1\not=\rho_1$) because this would mean 
$\sigma_1\sqcap_{{}_{\overline{ \mathfrak{S}}}}\sigma_3=\epsilon_1=\pi_1=\sigma_1\sqcap_{{}_{\overline{ \mathfrak{S}}}}\sigma_2$ (resp. $\sigma_1\sqcap_{{}_{\overline{ \mathfrak{S}}}}\sigma_3=\epsilon_1=\rho_1=\sigma_1\sqcap_{{}_{\overline{ \mathfrak{S}}}}\sigma_4$). These degenerate cases have been excluded by assumption. We will also assume obviously $\epsilon_2\not=\pi_2$ and $\epsilon_2\not=\rho_2$.\\
Let us now choose $\chi:=\sigma_2$ and $\chi:=\sigma_1$ with $\kappa:=\sigma_3$ in (\ref{fundamentalformula}), and simultaneously $\chi:=\sigma_1$ and $\chi:=\sigma_2$ with $\kappa:=\sigma_4$ in (\ref{fundamentalformula}).  We obtain
\begin{eqnarray}
&&\rho_2 \sqsupseteq_{{}_{\overline{\mathfrak{S}}}} (\sigma_3\sqcap_{{}_{\overline{ \mathfrak{S}}}}\pi_2),(\sigma_3\sqcap_{{}_{\overline{ \mathfrak{S}}}}\epsilon_2)\\
&&\epsilon_1 \sqsupseteq_{{}_{\overline{\mathfrak{S}}}} (\sigma_3\sqcap_{{}_{\overline{ \mathfrak{S}}}}\pi_2),(\sigma_3\sqcap_{{}_{\overline{ \mathfrak{S}}}}\rho_1)\\
&&\rho_1 \sqsupseteq_{{}_{\overline{\mathfrak{S}}}} (\sigma_4\sqcap_{{}_{\overline{ \mathfrak{S}}}}\pi_2),(\sigma_4\sqcap_{{}_{\overline{ \mathfrak{S}}}}\epsilon_1)\\
&&\epsilon_2 \sqsupseteq_{{}_{\overline{\mathfrak{S}}}} (\sigma_4\sqcap_{{}_{\overline{ \mathfrak{S}}}}\pi_2),(\sigma_4\sqcap_{{}_{\overline{ \mathfrak{S}}}}\rho_2)
\end{eqnarray}
From $\epsilon_1 \sqsupseteq_{{}_{\overline{\mathfrak{S}}}} (\sigma_3\sqcap_{{}_{\overline{ \mathfrak{S}}}}\pi_2)$ and $\epsilon_1\sqcoversubset_{{}_{\overline{\mathfrak{S}}}}\sigma_3$, we deduce $(\sigma_3\sqcap_{{}_{\overline{ \mathfrak{S}}}}\pi_2)=(\epsilon_1\sqcap_{{}_{\overline{ \mathfrak{S}}}}\pi_2)$. \\
From $\rho_1 \sqsupseteq_{{}_{\overline{\mathfrak{S}}}} (\sigma_4\sqcap_{{}_{\overline{ \mathfrak{S}}}}\pi_2)$ and $\sigma_4\sqsupseteq_{{}_{\overline{ \mathfrak{S}}}}\rho_1$, we deduce $(\sigma_4\sqcap_{{}_{\overline{ \mathfrak{S}}}}\pi_2)=(\rho_1\sqcap_{{}_{\overline{ \mathfrak{S}}}}\pi_2)$.\\
On another part, from (\ref{thirdcoveringpropertySbariii}) (if $\lambda\in  ({\widecheck{\mathfrak{S}}}\smallsetminus \overline{\mathfrak{S}})$) or from (\ref{secondcoveringpropertySbar}) (if $\lambda\in \overline{\mathfrak{S}}{}^{{}^{pure}}$)
, we deduce
\begin{eqnarray}
(\sigma_3\sqcap_{{}_{\overline{ \mathfrak{S}}}}\pi_2)\sqcap_{{}_{\overline{ \mathfrak{S}}}}(\sigma_4\sqcap_{{}_{\overline{ \mathfrak{S}}}}\pi_2) \sqcoversubset_{{}_{\overline{ \mathfrak{S}}}} \pi_2.
\end{eqnarray}
As long as $\sigma_2\not=\sigma_3$, we have $\sigma_3\not\sqsupseteq_{{}_{\overline{ \mathfrak{S}}}}\pi_2$. Analogously, we have $\sigma_4\not\sqsupseteq_{{}_{\overline{ \mathfrak{S}}}}\pi_2$. Hence,
\begin{eqnarray}
(\sigma_3\sqcap_{{}_{\overline{ \mathfrak{S}}}}\pi_2)=(\sigma_4\sqcap_{{}_{\overline{ \mathfrak{S}}}}\pi_2)=
(\sigma_3\sqcap_{{}_{\overline{ \mathfrak{S}}}}\pi_2)\sqcap_{{}_{\overline{ \mathfrak{S}}}}(\sigma_4\sqcap_{{}_{\overline{ \mathfrak{S}}}}\pi_2) \sqcoversubset_{{}_{\overline{ \mathfrak{S}}}} \pi_2.
\end{eqnarray}
As a result, we obtain
\begin{eqnarray}
(\epsilon_1\sqcap_{{}_{\overline{ \mathfrak{S}}}}\pi_2) = 
(\rho_1\sqcap_{{}_{\overline{ \mathfrak{S}}}}\pi_2).
\end{eqnarray}
Hence, we have obtained a contradiction with the property (\ref{thirdcoveringpropertySbariv}).\\
As a conclusion, this configuration is excluded.\\

Let us now consider that $\sigma_1,\sigma_4\in ({\check{\mathfrak{S}}}\smallsetminus \overline{\mathfrak{S}})$ and $\sigma_2,\sigma_3\in {\overline{\mathfrak{S}}}{}^{{}^{pure}}$.   
Using properties (\ref{Rek1}) and $\sigma_1\asymp\sigma_2, \sigma_2\asymp\sigma_3, \sigma_3\asymp\sigma_4, \sigma_1\asymp\sigma_4$ and (\ref{thirdcoveringpropertySbarv})(\ref{thirdcoveringpropertySbariii}), we note that we can always choose $\pi_i,\rho_i\in \overline{\mathfrak{S}}$ for $i=1,2$ such that
\begin{eqnarray}
\forall i=1,4&&\pi_i,\rho_i\in \Theta^{\overline{\mathfrak{S}}}(\sigma_i),\;\;\sigma_i=\pi_i\sqcup_{{}_{\overline{ \mathfrak{S}}}} \rho_i\\
&&\rho_1=\rho_4,\;\;\sigma_2\sqcoversupset_{{}_{\overline{ \mathfrak{S}}}}\pi_1,\;\; \sigma_3\sqcoversupset_{{}_{\overline{ \mathfrak{S}}}}\pi_4
.
\end{eqnarray}
Let us now assume that there exists $\lambda\in {\mathfrak{G}}^{\check{\mathfrak{S}}}$ such that ${\mathfrak{r}}^U(\lambda,\sigma_1,\sigma_2)$ and ${\mathfrak{r}}^U(\lambda,\sigma_3,\sigma_4)$. We have then
\begin{eqnarray}
\lambda\sqcoversupset_{{}_{{ \mathfrak{S}}}} \pi_1&& \lambda\sqcoversupset_{{}_{{ \mathfrak{S}}}}\pi_4.
\end{eqnarray}
The consistency relation $\sigma_2 \asymp \sigma_4$ means that there exists $\epsilon_4\in \Theta^{\overline{ \mathfrak{S}}}(\sigma_4)$ such that $\epsilon_4\sqcoversubset_{{}_{\overline{\mathfrak{S}}}}\sigma_2$.
Analogously,  the consistency relation $\sigma_1 \asymp \sigma_3$ means that there exists $\epsilon_1\in \Theta^{\overline{ \mathfrak{S}}}(\sigma_1)$ such that $\epsilon_1\sqcoversubset_{{}_{\overline{\mathfrak{S}}}}\sigma_3$.\\
Once again, we insist on the fact that $\epsilon_1\not=\pi_1$ (resp. $\epsilon_1\not=\rho_1$) because this would mean 
$\sigma_1\sqcap_{{}_{\overline{ \mathfrak{S}}}}\sigma_3=\epsilon_1=\pi_1=\sigma_1\sqcap_{{}_{\overline{ \mathfrak{S}}}}\sigma_2$ (resp. $\sigma_1\sqcap_{{}_{\overline{ \mathfrak{S}}}}\sigma_3=\epsilon_1=\rho_1=\sigma_1\sqcap_{{}_{\overline{ \mathfrak{S}}}}\sigma_4$). These degenerate cases have been excluded by assumption. We will also assume obviously $\epsilon_4\not=\pi_4$ and $\epsilon_4\not=\rho_1$.\\
Let us now choose $\chi:=\sigma_1$ and $\chi:=\sigma_4$ with $\kappa:=\sigma_2$ in (\ref{fundamentalformula}), and simultaneously $\chi:=\sigma_1$ and $\chi:=\sigma_4$ with $\kappa:=\sigma_3$ in (\ref{fundamentalformula}).  We obtain
\begin{eqnarray}
&&\pi_1 \sqsupseteq_{{}_{\overline{\mathfrak{S}}}} (\sigma_2\sqcap_{{}_{\overline{ \mathfrak{S}}}}\rho_1),(\sigma_2\sqcap_{{}_{\overline{ \mathfrak{S}}}}\epsilon_1)\\
&&\epsilon_4 \sqsupseteq_{{}_{\overline{\mathfrak{S}}}} (\sigma_2\sqcap_{{}_{\overline{ \mathfrak{S}}}}\pi_4),(\sigma_2\sqcap_{{}_{\overline{ \mathfrak{S}}}}\rho_1)\\
&&\epsilon_1 \sqsupseteq_{{}_{\overline{\mathfrak{S}}}} (\sigma_3\sqcap_{{}_{\overline{ \mathfrak{S}}}}\pi_1),(\sigma_3\sqcap_{{}_{\overline{ \mathfrak{S}}}}\rho_1)\\
&&\pi_4 \sqsupseteq_{{}_{\overline{\mathfrak{S}}}} (\sigma_3\sqcap_{{}_{\overline{ \mathfrak{S}}}}\rho_1),(\sigma_3\sqcap_{{}_{\overline{ \mathfrak{S}}}}\epsilon_4).
\end{eqnarray}
From $\pi_1 \sqsupseteq_{{}_{\overline{\mathfrak{S}}}} (\sigma_2\sqcap_{{}_{\overline{ \mathfrak{S}}}}\rho_1)$ and $\pi_1\sqcoversubset_{{}_{\overline{\mathfrak{S}}}}\sigma_2$, we deduce $(\sigma_2\sqcap_{{}_{\overline{ \mathfrak{S}}}}\rho_1)=(\pi_1\sqcap_{{}_{\overline{ \mathfrak{S}}}}\rho_1)$. \\
From $\pi_4 \sqsupseteq_{{}_{\overline{\mathfrak{S}}}} (\sigma_3\sqcap_{{}_{\overline{ \mathfrak{S}}}}\rho_1)$ and $\pi_4\sqcoversubset_{{}_{\overline{\mathfrak{S}}}}\sigma_3$, we deduce $(\sigma_3\sqcap_{{}_{\overline{ \mathfrak{S}}}}\rho_1)=(\pi_4\sqcap_{{}_{\overline{ \mathfrak{S}}}}\rho_1)$.\\
We have then obtained
\begin{eqnarray}
\epsilon_4 \sqsupseteq_{{}_{\overline{\mathfrak{S}}}}(\pi_1\sqcap_{{}_{\overline{ \mathfrak{S}}}}\rho_1),&&
\epsilon_1 \sqsupseteq_{{}_{\overline{\mathfrak{S}}}}(\pi_4\sqcap_{{}_{\overline{ \mathfrak{S}}}}\rho_1).
\end{eqnarray}
From
\begin{eqnarray}
(\pi_1\sqcap_{{}_{\overline{ \mathfrak{S}}}}\rho_1) \sqcoversubset_{{}_{\overline{ \mathfrak{S}}}} \pi_1,\rho_1,&&(\pi_4\sqcap_{{}_{\overline{ \mathfrak{S}}}}\rho_1) \sqcoversubset_{{}_{\overline{ \mathfrak{S}}}} \pi_4,\rho_1
\end{eqnarray}
we deduce
\begin{eqnarray}
(\rho_1\sqcap_{{}_{\overline{ \mathfrak{S}}}}\epsilon_4) =(\pi_1\sqcap_{{}_{\overline{ \mathfrak{S}}}}\rho_1)\sqcoversubset_{{}_{\overline{ \mathfrak{S}}}} \rho_1,&&
(\rho_1\sqcap_{{}_{\overline{ \mathfrak{S}}}}\epsilon_1)=(\pi_4\sqcap_{{}_{\overline{ \mathfrak{S}}}}\rho_1)\sqcoversubset_{{}_{\overline{ \mathfrak{S}}}} \rho_1.
\end{eqnarray}
On another part, from $\lambda\sqcoversupset_{{}_{{ \mathfrak{S}}}} \pi_1,\pi_4$, we know that $\lambda\not\sqsupseteq_{{}_{{ \mathfrak{S}}}}\rho_1$ because, using (\ref{thirdcoveringpropertySbarv}), we would deduce that $\lambda=(\rho_1\sqcup_{{}_{{ \mathfrak{S}}}} \pi_1)=\sigma_1$ and $\lambda=(\rho_1\sqcup_{{}_{{ \mathfrak{S}}}} \pi_4)=\sigma_4$ which has been excluded. Then, we deduce, using $(\pi_1\sqcap_{{}_{\overline{ \mathfrak{S}}}}\rho_1)\sqcoversubset_{{}_{{ \mathfrak{S}}}}\rho_1$ and $(\pi_4\sqcap_{{}_{\overline{ \mathfrak{S}}}}\rho_1)\sqcoversubset_{{}_{{ \mathfrak{S}}}}\rho_1$
\begin{eqnarray}
(\pi_1\sqcap_{{}_{\overline{ \mathfrak{S}}}}\rho_1)=(\pi_4\sqcap_{{}_{\overline{ \mathfrak{S}}}}\rho_1)=(\lambda\sqcap_{{}_{\overline{ \mathfrak{S}}}}\rho_1)
\end{eqnarray}
As a result, we obtain
\begin{eqnarray}
(\rho_1\sqcap_{{}_{\overline{ \mathfrak{S}}}}\epsilon_4) =(\pi_4\sqcap_{{}_{\overline{ \mathfrak{S}}}}\rho_1),&&(\rho_1\sqcap_{{}_{\overline{ \mathfrak{S}}}}\epsilon_1)=(\pi_1\sqcap_{{}_{\overline{ \mathfrak{S}}}}\rho_1).
\end{eqnarray}
Hence, we have obtained a contradiction with the property (\ref{thirdcoveringpropertySbariv}).\\
As a conclusion, this configuration is excluded.\\

Let us now consider that $\sigma_1,\sigma_3\in ({\check{\mathfrak{S}}}\smallsetminus \overline{\mathfrak{S}})$ and $\sigma_2,\sigma_4\in {\overline{\mathfrak{S}}}{}^{{}^{pure}}$.   
Using properties (\ref{Rek1}) and $\sigma_1\asymp\sigma_2, \sigma_2\asymp\sigma_3, \sigma_3\asymp\sigma_4, \sigma_1\asymp\sigma_4$ and (\ref{thirdcoveringpropertySbarv}), we note that we can always choose $\pi_i,\rho_i\in \overline{\mathfrak{S}}$ for $i=1,2$ such that
\begin{eqnarray}
\forall i=1,3&&\pi_i,\rho_i\in \Theta^{\overline{\mathfrak{S}}}(\sigma_i),\;\;\sigma_i=\pi_i\sqcup_{{}_{\overline{ \mathfrak{S}}}} \rho_i \\
&&\sigma_4\sqcoversupset_{{}_{\overline{ \mathfrak{S}}}}\rho_1,\;\; \sigma_4\sqcoversupset_{{}_{\overline{ \mathfrak{S}}}}\pi_3\\
&&\sigma_2\sqcoversupset_{{}_{\overline{ \mathfrak{S}}}}\rho_3,\;\; \sigma_2\sqcoversupset_{{}_{\overline{ \mathfrak{S}}}}\pi_1
.
\end{eqnarray}
Let us now assume that there exists $\lambda\in  {\mathfrak{G}}^{\check{\mathfrak{S}}}$   such that ${\mathfrak{r}}^U(\lambda,\sigma_1,\sigma_2)$ and ${\mathfrak{r}}^U(\lambda,\sigma_3,\sigma_4)$. 
We will exclude the trivial case $\pi_1=\sigma_1\sqcap_{{}_{\overline{ \mathfrak{S}}}}\sigma_2= \sigma_3\sqcap_{{}_{\overline{ \mathfrak{S}}}}\sigma_4=\pi_3$, i.e. we will assume $\pi_1\not=\pi_3$.  We will also exclude the trivial cases $\pi_1=\sigma_1\sqcap_{{}_{\overline{ \mathfrak{S}}}}\sigma_2= \sigma_2\sqcap_{{}_{\overline{ \mathfrak{S}}}}\sigma_3=\rho_3$ or $\pi_3=\sigma_3\sqcap_{{}_{\overline{ \mathfrak{S}}}}\sigma_4= \sigma_1\sqcap_{{}_{\overline{ \mathfrak{S}}}}\sigma_4=\rho_1$, i.e. we will then assume $\pi_1\not=\rho_3$ and $\pi_3\not=\rho_1$. \\
We have then
\begin{eqnarray}
&&\lambda\sqcoversupset_{{}_{{ \mathfrak{S}}}} \pi_1,\pi_3.
\end{eqnarray} 
The consistency relation $\sigma_2 \asymp \sigma_4$ contains nothing. 
On another part, the consistency relation $\sigma_1 \asymp \sigma_3$ means that there exists $\epsilon\in \Theta^{\overline{ \mathfrak{S}}}(\sigma_1)\cap\Theta^{\overline{ \mathfrak{S}}}(\sigma_3)$.\\
Let us now choose $\chi:=\sigma_1$ and $\chi:=\sigma_3$ with $\kappa:=\sigma_2$ in (\ref{fundamentalformula}), and simultaneously $\chi:=\sigma_3$ and $\chi:=\sigma_1$ with $\kappa:=\sigma_4$ in (\ref{fundamentalformula}).  
\begin{eqnarray}
&&(\sigma_2\sqcap_{{}_{\overline{ \mathfrak{S}}}}\rho_1),(\sigma_2\sqcap_{{}_{\overline{ \mathfrak{S}}}}\epsilon)\sqsubseteq_{{}_{\overline{ \mathfrak{S}}}} \pi_1\\
&&(\sigma_2\sqcap_{{}_{\overline{ \mathfrak{S}}}}\pi_3),(\sigma_2\sqcap_{{}_{\overline{ \mathfrak{S}}}}\epsilon)\sqsubseteq_{{}_{\overline{ \mathfrak{S}}}} \rho_3\\
&&(\sigma_4\sqcap_{{}_{\overline{ \mathfrak{S}}}}\rho_3),(\sigma_4\sqcap_{{}_{\overline{ \mathfrak{S}}}}\epsilon)\sqsubseteq_{{}_{\overline{ \mathfrak{S}}}} \pi_3\\
&&(\sigma_4\sqcap_{{}_{\overline{ \mathfrak{S}}}}\pi_1),(\sigma_4\sqcap_{{}_{\overline{ \mathfrak{S}}}}\epsilon)\sqsubseteq_{{}_{\overline{ \mathfrak{S}}}} \rho_1
\end{eqnarray}
From $(\sigma_2\sqcap_{{}_{\overline{ \mathfrak{S}}}}\epsilon)\sqsubseteq_{{}_{\overline{ \mathfrak{S}}}} \pi_1$, $(\sigma_2\sqcap_{{}_{\overline{ \mathfrak{S}}}}\epsilon)\sqsubseteq_{{}_{\overline{ \mathfrak{S}}}} \rho_3$ and $\sigma_2\sqcoversupset_{{}_{\overline{ \mathfrak{S}}}}\rho_3$ and $\sigma_2\sqcoversupset_{{}_{\overline{ \mathfrak{S}}}}\pi_1$ (and also $(\sigma_4\sqcap_{{}_{\overline{ \mathfrak{S}}}}\epsilon)\sqsubseteq_{{}_{\overline{ \mathfrak{S}}}} \pi_3$, $(\sigma_4\sqcap_{{}_{\overline{ \mathfrak{S}}}}\epsilon)\sqsubseteq_{{}_{\overline{ \mathfrak{S}}}} \rho_1$ and $\sigma_4\sqcoversupset_{{}_{\overline{ \mathfrak{S}}}}\rho_1$ and $\sigma_4\sqcoversupset_{{}_{\overline{ \mathfrak{S}}}}\pi_3$) we deduce
\begin{eqnarray}
&&(\sigma_2\sqcap_{{}_{\overline{ \mathfrak{S}}}}\epsilon)= (\pi_1\sqcap_{{}_{\overline{ \mathfrak{S}}}}\epsilon)=(\rho_3\sqcap_{{}_{\overline{ \mathfrak{S}}}}\epsilon)
\\
&&(\sigma_4\sqcap_{{}_{\overline{ \mathfrak{S}}}}\epsilon)= (\pi_3\sqcap_{{}_{\overline{ \mathfrak{S}}}}\epsilon)=(\rho_1\sqcap_{{}_{\overline{ \mathfrak{S}}}}\epsilon)
.
\end{eqnarray}
On another part, from $\lambda\sqcoversupset_{{}_{{ \mathfrak{S}}}} \pi_1,\pi_3$, we know that $\lambda\not\sqsupseteq_{{}_{{ \mathfrak{S}}}}\epsilon$ because, using (\ref{thirdcoveringpropertySbarv}), we would deduce that $\sigma_1=(\epsilon\sqcup_{{}_{{ \mathfrak{S}}}} \pi_1)=(\epsilon\sqcup_{{}_{{ \mathfrak{S}}}} \pi_3)=\sigma_3$ which has been excluded. Then, we deduce, using $(\pi_1\sqcap_{{}_{\overline{ \mathfrak{S}}}}\epsilon)\sqcoversubset_{{}_{{ \mathfrak{S}}}}\epsilon$ and $(\pi_3\sqcap_{{}_{\overline{ \mathfrak{S}}}}\epsilon)\sqcoversubset_{{}_{{ \mathfrak{S}}}}\epsilon$
\begin{eqnarray}
(\pi_1\sqcap_{{}_{\overline{ \mathfrak{S}}}}\epsilon)=(\pi_3\sqcap_{{}_{\overline{ \mathfrak{S}}}}\epsilon)=(\lambda\sqcap_{{}_{\overline{ \mathfrak{S}}}}\epsilon)
\end{eqnarray}
We obtain finally
\begin{eqnarray}
(\pi_1\sqcap_{{}_{\overline{ \mathfrak{S}}}}\epsilon)=(\rho_1\sqcap_{{}_{\overline{ \mathfrak{S}}}}\epsilon),&&(\pi_3\sqcap_{{}_{\overline{ \mathfrak{S}}}}\epsilon)=(\rho_3\sqcap_{{}_{\overline{ \mathfrak{S}}}}\epsilon)
\end{eqnarray}
which contradicts (\ref{thirdcoveringpropertySbariv}).\\
As a conclusion, this configuration is excluded.\\

Let us now consider that $\sigma_1\in ({\check{\mathfrak{S}}}\smallsetminus \overline{\mathfrak{S}})$ and $\sigma_2,\sigma_3,\sigma_4\in {\overline{\mathfrak{S}}}{}^{{}^{pure}}$.   
Using properties (\ref{Rek1}) and $\sigma_1\asymp\sigma_2, \sigma_2\asymp\sigma_3, \sigma_3\asymp\sigma_4, \sigma_1\asymp\sigma_4$ and (\ref{thirdcoveringpropertySbarv})(\ref{thirdcoveringpropertySbariii}), we note that we can always choose $\pi,\rho\in \overline{\mathfrak{S}}$ such that
\begin{eqnarray}
&&\pi,\rho\in \Theta^{\overline{\mathfrak{S}}}(\sigma_1),\;\;\sigma_1=\pi\sqcup_{{}_{\overline{ \mathfrak{S}}}} \rho \\
&&\sigma_4\sqcoversupset_{{}_{\overline{ \mathfrak{S}}}}\rho,\;\; \sigma_2\sqcoversupset_{{}_{\overline{ \mathfrak{S}}}}\pi\\
.
\end{eqnarray}
Let us now assume that there exists $\lambda\in {\mathfrak{G}}^{\check{\mathfrak{S}}}$ such that ${\mathfrak{r}}^U(\lambda,\sigma_1,\sigma_2)$ and ${\mathfrak{r}}^U(\lambda,\sigma_3,\sigma_4)$. We have then
\begin{eqnarray}
\lambda\sqcoversupset_{{}_{{ \mathfrak{S}}}} \pi, &&
\lambda\sqcoversupset_{{}_{{ \mathfrak{S}}}}(\sigma_3\sqcap_{{}_{\overline{ \mathfrak{S}}}}\sigma_4).
\end{eqnarray}
The consistency relation $\sigma_1 \asymp \sigma_3$ means that there exists $\epsilon\in \Theta^{\overline{ \mathfrak{S}}}(\sigma_1)$ such that $\epsilon\sqcoversubset_{{}_{\overline{\mathfrak{S}}}}\sigma_3$.
The consistency relation $\sigma_2 \asymp \sigma_4$ contains nothing.\\
Once again, we insist on the fact that $\epsilon\not=\pi$ (resp. $\epsilon\not=\rho$) because this would mean 
$\sigma_1\sqcap_{{}_{\overline{ \mathfrak{S}}}}\sigma_3=\epsilon=\pi=\sigma_1\sqcap_{{}_{\overline{ \mathfrak{S}}}}\sigma_2$ (resp. $\sigma_1\sqcap_{{}_{\overline{ \mathfrak{S}}}}\sigma_3=\epsilon=\rho=\sigma_1\sqcap_{{}_{\overline{ \mathfrak{S}}}}\sigma_4$). These degenerate cases have been excluded by assumption. \\
Let us now choose$\chi:=\sigma_1$ with $\kappa:=\sigma_2$ in (\ref{fundamentalformula}), and simultaneously $\chi:=\sigma_1$ with $\kappa:=\sigma_3$ in (\ref{fundamentalformula}), and simultaneously $\chi:=\sigma_1$ with $\kappa:=\sigma_4$ in (\ref{fundamentalformula}).  We obtain
\begin{eqnarray}
&&\pi \sqsupseteq_{{}_{\overline{\mathfrak{S}}}} (\sigma_2\sqcap_{{}_{\overline{ \mathfrak{S}}}}\rho),(\sigma_2\sqcap_{{}_{\overline{ \mathfrak{S}}}}\epsilon)\\
&&\epsilon \sqsupseteq_{{}_{\overline{\mathfrak{S}}}} (\sigma_3\sqcap_{{}_{\overline{ \mathfrak{S}}}}\pi),(\sigma_3\sqcap_{{}_{\overline{ \mathfrak{S}}}}\rho)\\
&&\rho \sqsupseteq_{{}_{\overline{\mathfrak{S}}}} (\sigma_4\sqcap_{{}_{\overline{ \mathfrak{S}}}}\pi),(\sigma_4\sqcap_{{}_{\overline{ \mathfrak{S}}}}\epsilon).
\end{eqnarray}
From $\pi \sqsupseteq_{{}_{\overline{\mathfrak{S}}}} (\sigma_2\sqcap_{{}_{\overline{ \mathfrak{S}}}}\rho)$ and $\pi\sqcoversubset_{{}_{\overline{\mathfrak{S}}}}\sigma_2$, we deduce $(\sigma_2\sqcap_{{}_{\overline{ \mathfrak{S}}}}\rho)=(\rho\sqcap_{{}_{\overline{ \mathfrak{S}}}}\pi)$. \\
From $\rho \sqsupseteq_{{}_{\overline{\mathfrak{S}}}} (\sigma_4\sqcap_{{}_{\overline{ \mathfrak{S}}}}\pi)$ and $\sigma_4\sqsupseteq_{{}_{\overline{ \mathfrak{S}}}}\rho$, we deduce $(\sigma_4\sqcap_{{}_{\overline{ \mathfrak{S}}}}\pi)=(\rho\sqcap_{{}_{\overline{ \mathfrak{S}}}}\pi)$.\\
From $\epsilon \sqsupseteq_{{}_{\overline{\mathfrak{S}}}} (\sigma_3\sqcap_{{}_{\overline{ \mathfrak{S}}}}\pi)$ and $\sigma_3\sqsupseteq_{{}_{\overline{ \mathfrak{S}}}}\epsilon$, we deduce $(\sigma_3\sqcap_{{}_{\overline{ \mathfrak{S}}}}\pi)=(\epsilon\sqcap_{{}_{\overline{ \mathfrak{S}}}}\pi)$.\\
On another part, from (\ref{thirdcoveringpropertySbariii}) (if $\lambda\in  ({\widecheck{\mathfrak{S}}}\smallsetminus \overline{\mathfrak{S}})$) or from (\ref{secondcoveringpropertySbar}) (if $\lambda\in \overline{\mathfrak{S}}{}^{{}^{pure}}$)
, we deduce
\begin{eqnarray}
(\sigma_3\sqcap_{{}_{\overline{ \mathfrak{S}}}}\pi)\sqcap_{{}_{\overline{ \mathfrak{S}}}}(\sigma_4\sqcap_{{}_{\overline{ \mathfrak{S}}}}\pi) \sqcoversubset_{{}_{\overline{ \mathfrak{S}}}} \pi.
\end{eqnarray}
As long as $\sigma_2\not=\sigma_3$, we have $\sigma_3\not\sqsupseteq_{{}_{\overline{ \mathfrak{S}}}}\pi_2$. Analogously, we have $\sigma_4\not\sqsupseteq_{{}_{\overline{ \mathfrak{S}}}}\pi_2$. Hence,
\begin{eqnarray}
(\sigma_3\sqcap_{{}_{\overline{ \mathfrak{S}}}}\pi)=(\sigma_4\sqcap_{{}_{\overline{ \mathfrak{S}}}}\pi)=
(\sigma_3\sqcap_{{}_{\overline{ \mathfrak{S}}}}\pi)\sqcap_{{}_{\overline{ \mathfrak{S}}}}(\sigma_4\sqcap_{{}_{\overline{ \mathfrak{S}}}}\pi) \sqcoversubset_{{}_{\overline{ \mathfrak{S}}}} \pi.
\end{eqnarray}
As a result, we obtain
\begin{eqnarray}
(\epsilon\sqcap_{{}_{\overline{ \mathfrak{S}}}}\pi) = 
(\rho\sqcap_{{}_{\overline{ \mathfrak{S}}}}\pi).
\end{eqnarray}
Hence, we have obtained a contradiction with the property (\ref{thirdcoveringpropertySbariv}).\\
As a conclusion, this configuration is excluded.\\

As a conclusion of the previous analysis we obtain that $\forall i\in \{1,2,3,4\}, \sigma_i\in \overline{\mathfrak{S}}{}^{{}^{pure}}$.

\section*{Statements and Declarations}
The author did not receive any support from any organization for the present work. The author declares that there is no conflict of interest. There is no associated data to this article.



\end{document}